\documentclass{amsart}

\usepackage{amsmath, amssymb, amsthm}
\usepackage{fullpage}
\usepackage{slashed}
\usepackage{mathtools}

\newtheorem{theorem}{Theorem}[section]
\newtheorem{lemma}[theorem]{Lemma}
\newtheorem{corollary}[theorem]{Corollary}
\newtheorem{proposition}[theorem]{Proposition}

\newtheorem{remark}[theorem]{Remark}

\numberwithin{equation}{section}

\mathtoolsset{showonlyrefs=true}

\newcommand{\supp}{\mathrm{supp}}
\newcommand{\TM}{T\mathcal{M}}
\newcommand{\gs}{\mathbf{X}}
\newcommand{\ms}{\mathcal{P}}
\newcommand{\Xb}{\overline{X}}
\newcommand{\Pb}{\overline{P}}
\newcommand{\Phib}{\overline{\Phi}}
\newcommand{\tr}{\mathrm{tr}}

\def\newhat#1{{\,\widehat{\!#1\!}\,\,}}

\newcommand{\ph}{\newhat{p}}
\newcommand{\Ph}{\newhat{P}}
\newcommand{\Xh}{\newhat{X}}
\newcommand{\Gh}{\newhat{\Gamma}}
\newcommand{\Zb}{\overline{Z}}

\def\ga{\gamma}\def\de{\delta}
\def\Si{\Sigma}
\def\bm{\left( \begin{array}{cc}}
\def\endm{\end{array}\right)}\newcommand{\eq}{\end{equation}}
\def\tr{\text{tr}}
\def\ga{\gamma}\def\de{\delta}\def\Box{\square}\def\pa{\partial}\def\pab{\bar\pa}

\def \rectangle#1#2{\hbox{\vrule\vbox to #2 {\hrule\hbox to #1{\hfil}\vfil\hrule}\vrule}}
\def\Lb{\underline{L}}
\def\tr{\text{tr}}\def\ga{\gamma}
\def\de{\delta}\def\Box{\square}\def\Boxr{\widetilde{\square}}\def\pa{\partial}
\def\pab{\bar\pa}\def\Lb{\underline{L}}

\def\trs{\,\slash{\!\!\!\!\tr}\,}

\newcommand{\les}{\lesssim}
\newcommand{\bea}{\begin{eqnarray}}\newcommand{\eea}{\end{eqnarray}}
\newcommand{\beq}{\begin{equation}}\newcommand{\ee}{\end{equation}}

\newcommand{\R}{{\mathbb R}}

\def\nn{\nonumber}

\begin{document}
\title[Stability of Minkowski space for Einstein--Vlasov]{Global stability of Minkowski space for the Einstein--Vlasov system in the harmonic gauge}
\author{Hans Lindblad}
\address{Johns Hopkins University, Department of Mathematics, 3400 N.\@ Charles Street, Baltimore, MD 21218, USA}
\email{lindblad@math.jhu.edu}
\author{Martin Taylor}
\address{Imperial College London, Department of Mathematics, South Kensington Campus, London, SW7 2AZ, UK}
\email{martin.taylor@imperial.ac.uk}

\date{\today}

\begin{abstract}
	Minkowski space is shown to be globally stable as a solution to the massive Einstein--Vlasov system.  The proof is based on a harmonic gauge in which the equations reduce to a system of quasilinear wave equations for the metric, satisfying the weak null condition, coupled to a transport equation for the Vlasov particle distribution function.  Central to the proof is a collection of vector fields used to control the particle distribution function, a function of both spacetime and momentum variables.  The vector fields are derived using a general procedure, are adapted to the geometry of the solution and reduce to the generators of the symmetries of Minkowski space when restricted to acting on spacetime functions.  Moreover, when specialising to the case of vacuum, the proof provides a simplification of previous stability works.
\end{abstract}

\maketitle

\tableofcontents

\section{Introduction}
\subsection{The Einstein--Vlasov system}
The Einstein--Vlasov system provides a statistical description of a collection of collisionless particles, interacting via gravity as described by Einstein's general theory of relativity.  A fundamental problem is to understand the long time dynamics of solutions of this system.  The problem is a great challenge even in the absence of particles, and global works on the \emph{vacuum Einstein equations} all either involve simplifying symmetry assumptions or solutions arising from small data.  In the asymptotically flat setting, small data solutions of the vacuum Einstein equations were first shown to exist globally and disperse to Minkowski space in the monumental work of Christodoulou--Klainerman \cite{ChKl}.  An alternative proof of the stability of Minkowski space was later given by Lindblad--Rodnianski \cite{LR3} in which a global \emph{harmonic coordinate} system was constructed, described below.  For the Einstein--Vlasov system, the global properties of small data solutions were first understood when the initial data were assumed to be spherically symmetric, an assumption under which the equations simplify dramatically, by Rein--Rendall \cite{ReRe} in the \emph{massive case}, when all particles are assumed to have mass 1, and by Dafermos \cite{Da} in the \emph{massless case}, when all particles are assumed to have mass 0.  Without the simplifying assumption of spherical symmetry, small data solutions of the massless Einstein--Vlasov system were later understood by Taylor \cite{Ta}.

In this work the problem of the stability of Minkowski space for the Einstein--Vlasov system, without any symmetry assumptions, is addressed in the case that all particles have mass 1 (and can easily be adapted to the case that all particles have any fixed mass $m>0$).  The system takes the form,
\begin{equation} \label{eq:Einstein}
	Ric(g)_{\mu \nu} - \frac{1}{2} R(g) g_{\mu \nu} = T_{\mu \nu},
\end{equation}
\begin{equation} \label{eq:Tmunu}
	T^{\mu \nu} (t,x) = \int_{\ms_{(t,x)}} f(t,x,p) p^{\mu} p^{\nu} \frac{\sqrt{\vert \det g \vert}}{p^0} dp^1 dp^2 dp^3,
\end{equation}
\begin{equation} \label{eq:Vlasov}
	\gs(f) = 0,
\end{equation}
where the unknown is a $3+1$ dimensional manifold $\mathcal{M}$ with  Lorentzian metric $g$, together with a particle distribution function $f: \ms \to [0,\infty)$, where the \emph{mass shell} $\mathcal{P}$ is defined by
\[
	\ms = \{ (t,x,p_0,p) \in \TM \mid (p^0,p) \textrm{ future directed, } g_{\mu \nu}p^{\mu}p^{\nu} = -1 \}.
\]
Here $Ric(g)$ and $R(g)$ denote the Ricci and scalar curvature of $g$ respectively.  A coordinate system $(t,x)$ for $\mathcal{M}$, with $t$ a \emph{time function} (i.\@e.\@ the one form $dt$ is timelike with respect to the metric $g$), defines a coordinate system $(t,x,p^0,p)$ for the tangent bundle $\TM$ of $\mathcal{M}$ \emph{conjugate} to $(t,x)$, where $(t,x^i,p^0,p^i)$ denotes the point
\[
	p^0 \partial_t \vert_{(t,x)} + p^i \partial_{x^i} \vert_{(t,x)} \in \TM.
\]
The mass shell relation in the definition of $\ms$,
\begin{equation} \label{eq:massshell}
	g_{\mu \nu}p^{\mu}p^{\nu} = -1,
\end{equation}
should be viewed as defining $p^0$ as a function of $(t,x^1,x^2,x^3,p^1,p^2,p^3)$.  Here Greek indices run over $0,1,2,3$, lower case latin indices run over $1,2,3$, and often the notation $t=x^0$ is used.  The vector $\gs$ is the generator of the geodesic flow of $\mathcal{M}$ which, with respect to the coordinate system $(t,x,p)$ for $\ms$, takes the form,
\begin{equation*}
	\gs = p^{\mu} \partial_{x^{\mu}} - p^{\alpha} p^{\beta} \Gamma_{\alpha \beta}^{i} \partial_{p^{i}}.
\end{equation*}
The volume form $(\sqrt{\vert \det g \vert}/{p^0}) dp^1 dp^2 dp^3$ in \eqref{eq:Tmunu} is the induced volume form of the spacelike hypersurface $\ms_{(t,x)} \!\subset\! T_{(t,x)}\mathcal{M}$ when the tangent space $T_{(t,x)} \mathcal{M}$ is endowed with the metric $g_{\mu\nu}(t,x) dp^{\mu} dp^{\nu}$ induced by $g$ on $\mathcal{M}$.

\subsection{The global existence theorem}
\subsubsection{The initial value problem and the global existence theorem}
In the Cauchy problem for the system \eqref{eq:Einstein}--\eqref{eq:Vlasov} one prescribes an \emph{initial data set}, which consists of a Riemannian 3 manifold $(\Sigma, \overline{g})$ together with a symmetric $(0,2)$ tensor $k$ on $\Sigma$, and an initial particle distribution $f_0$, satisfying the \emph{constraint equations},
\begin{equation*}
	\overline{\mathrm{div}} k_j - (d \overline{\tr} k)_j = T_{0j},
	\qquad
	\overline{R} + (\overline{\tr} k)^2 - \vert k\vert^2_{\overline{g}} = 2 T_{00},
\end{equation*}
for $j=1,2,3$.  Here $\overline{\mathrm{div}}$, $\overline{\tr}$, $\overline{R}$ denote the divergence, trace and scalar curvature of $\overline{g}$ respectively, and $T_{00}, T_{0j}$ denote (what will become) the $00$ and $0j$ components of the energy momentum tensor.  The topology of $\Sigma$ will here always be assumed to be that of $\mathbb{R}^3$.  A theorem of Choquet-Bruhat \cite{ChBr71}, based on previous work by Choquet-Bruhat \cite{ChBr} and Choquet-Bruhat--Geroch \cite{ChBrGe} on the vacuum Einstein equations, (see also the recent textbook of Ringstr\"{o}m \cite{Ri}) guarantees that, for any such initial data set as above, there exists a globally hyperbolic solution $(\mathcal{M},g,f)$ of the system \eqref{eq:Einstein}--\eqref{eq:Vlasov} which attains the initial data, in the sense that there exists an imbedding $\iota : \Sigma \to \mathcal{M}$ under the pullback of which the induced first and second fundamental form of $g$ are $\overline{g}$ and $k$ respectively, and the restriction of $f$ to the mass shell over $\iota(\Sigma)$ is given by $f_0$.  As in \cite{LR3}, the proof is based on the \emph{harmonic gauge} (or \emph{wave gauge}), i.\@e.\@ a system of coordinates $x^{\mu}$ satisfying
\begin{equation} \label{eq:harmonicgauge}
	\Box_g x^{\mu} = 0,
\end{equation}
where $\Box_g$ denotes the geometric wave operator of the metric $g$.

The initial data set
\begin{equation} \label{eq:Minkowskidata}
	\left( \mathbb{R}^3,e, k \equiv 0, f_0 \equiv 0\right),
\end{equation}
where $e$ denotes the Euclidean metric on $\mathbb{R}^3$, satisfies the constraint equations and gives rise to Minkowski space, the trivial solution of the system \eqref{eq:Einstein}--\eqref{eq:Vlasov}.  The main result of this work concerns solutions arising from initial data sufficiently close to the trivial initial data set \eqref{eq:Minkowskidata}.  The initial data will be assumed to be \emph{asymptotically flat} in the sense that $\Sigma$ is diffeomorphic to $\mathbb{R}^3$ and there exists a global coordinate chart $(x^1,x^2,x^3)$ of $\Sigma$ and $M \geq 0$ and $0<\gamma<1$ such that
\begin{equation} \label{eq:af}
	\overline{g}_{ij}
	=
	\left( 1 + \frac{M}{r} \right) \delta_{ij} + o(r^{-1-\gamma}),
	\qquad
	k_{ij} = o(r^{-2-\gamma}),
	\qquad
	\text{as }
	r = \vert x \vert \to \infty.
\end{equation}
For such $\overline{g}$, $k$, write
\[
	\overline{g}_{ij}
	=
	\delta_{ij}
	+
	\overline{h}^0_{ij}
	+
	\overline{h}^1_{ij},\qquad\text{where}\qquad
\overline{h}^0_{ij}(x) = \chi(r) \delta_{ij}M/r,
\]
and $\chi$ is a smooth cut off function such that $0 \leq \chi \leq 1$, $\chi(s) = 1$ for $s \geq {3}/{4}$ and $\chi(s) = 0$ for $s \leq {1}/{2}$.

For given $\gamma>0$ and such an initial data set, define the initial energy,
\begin{equation}\label{eq:initialmetricnorm}
	\mathcal{E}_N
	=
	{\sum}_{\vert I \vert \leq N} \left(
	\Vert (1+r)^{{1}/{2} + \gamma + \vert I \vert} \nabla \nabla^I \overline{h}^1 \Vert_{L^2(\Sigma)}
	+
	\Vert (1+r)^{{1}/{2} + \gamma + \vert I \vert} \nabla^I k \Vert_{L^2(\Sigma)}
	\right),
\end{equation}
where $\nabla = (\partial_{x^1},\partial_{x^2}, \partial_{x^3})$ denotes the coordinate gradient, and the initial norms for $f_0$,
\begin{equation}\label{eq:initialvlasovnorm}
	\mathbb{D}_K
	=
	{\sum}_{k + \ell\leq K}
	\| \partial_{x}^{k} \partial_{p}^{\ell} f_0 \|_{L^\infty},
	\qquad
	\mathcal{V}_K
	=
	{\sum}_{k + \ell\leq K}
	\| \partial_{x}^{k} \partial_{p}^{\ell} f_0 \|_{L^2_x L^2_p}.
\end{equation}
By the Sobolev inequality there exists a constant $C$ such that $\mathbb{D}_{K-4} \leq C \mathcal{V}_K$ for any $K \geq 0$.

The main result of this work is the following.

\begin{theorem} \label{thm:main1}
	Let $(\Sigma, \overline{g}, k,f_0)$ be an initial data set for the Einstein--Vlasov system \eqref{eq:Einstein}--\eqref{eq:Vlasov}, asymptotically flat in the above sense, such that $f_0$ is compactly supported.  For any $0<\gamma <1$ and $N \geq 11$, there exists $\varepsilon_0>0$ such that, for all $\varepsilon \leq \varepsilon_0$ and all initial data satisfying
	\[
		\mathcal{E}_N + \mathcal{V}_N + M \leq \varepsilon,
	\]
	there exists a unique future geodesically complete solution of \eqref{eq:Einstein}--\eqref{eq:Vlasov}, attaining the given data, together with a global system of harmonic coordinates $(t,x^1,x^2,x^3)$, relative to which the solution asymptotically decays to Minkowski space.
\end{theorem}

The precise sense in which the spacetimes of Theorem \ref{thm:main1} asymptotically decay to Minkowski space is captured in the estimates \eqref{eq:main2energy}, \eqref{eq:main2pointwise}, \eqref{eq:wavecoordinatederivativeLiefirst}, \eqref{eq:wavecoordinatefunctionLiefirst}, \eqref{eq:sharpHhigherlowderbeginsec61}, \eqref{eq:sharpHhigherlowderbeginsec6}, \eqref{eq:weakdecayderh1}, \eqref{eq:weakdecayh1} below.

A similar stability result to Theorem \ref{thm:main1} has been shown independently by Fajman--Joudioux--Smulevici \cite{FaJoSm172}.

There have been a number of related stability works on the Einstein--Vlasov system.  In addition to those discussed earlier, there has been related work by Ringstr\"{o}m \cite{Ri} on the Einstein--Vlasov system in the presence of a positive cosmological constant, where the analogue of the Minkowski solution is the de Sitter spacetime.  A stability result for a class of cosmological spacetimes was shown in $2+1$ dimensions with vanishing cosmological constant by Fajman \cite{Fa}.  See also the recent work of Moschidis \cite{Mo} on the instability of the anti-de Sitter spacetime for a related spherically symmetric model in the presence of a negative cosmological constant.  A much more comprehensive overview of work on the Einstein--Vlasov system can be found in the review paper of Andr\'{e}asson \cite{An}.

There has also been work on the problem of the stability of Minkowski space for the Einstein equations coupled to various other matter models \cite{Bi}, \cite{Zi}, \cite{Fr}, \cite{Hu}, \cite{KlNi}, \cite{LeMa}, \cite{Lo}, \cite{Sp}.

Small data solutions of the Vlasov--Poisson system, the non-relativistic analogue of the Einstein--Vlasov system, were studied in \cite{BaDe}, \cite{HwReVe}, and those of the Vlasov--Maxwell system in \cite{GlSt}.  We note in particular works of Smulevici \cite{Sm} and Fajman--Joudioux--Smulevici \cite{FaJoSm17} (following \cite{FaJoSm}) on the asymptotic properties of small data solutions of the Vlasov--Poisson and the related Vlasov--Nordstr\"{o}m systems respectively, where issues related to those described in Section \ref{subsec:introvfs2} arise and are resolved using an alternative approach.  Moreover, there are global existence results for general data for Vlasov--Poisson \cite{Pf}, \cite{LiPe} and Vlasov--Nordstr\"{o}m \cite{Ca}.

\subsubsection{Small data global existence for the reduced Einstein--Vlasov system}

Following \cite{LR3}, the proof of Theorem \ref{thm:main1} is based on a harmonic gauge \eqref{eq:harmonicgauge}, relative to which the Einstein equations\footnote{Note that the Einstein equations \eqref{eq:Einstein} are equivalent to $Ric(g)_{\alpha\beta}=T_{\alpha\beta}- \frac{1}{2}g_{\alpha\beta}\operatorname{tr}_g T$.} \eqref{eq:Einstein} take the form of a system of quasilinear wave equations,
\begin{equation}\label{eq:RE1}
	\Boxr_g \, g_{\mu\nu}
	=
	F_{\mu\nu} (g)(\pa g, \pa g)+\widehat{T}_{\mu\nu} ,
	\qquad
	\text{where}
	\qquad
	\Boxr_g=g^{\alpha\beta}\pa_\alpha\pa_\beta,
\end{equation}
where $\widehat{T}_{\mu \nu} := T_{\mu \nu} - \frac{1}{2}g_{\mu \nu}\operatorname{tr}_g T$ and $F_{\mu \nu}(u)(v,v)$ depends quadratically on $v$.
The system \eqref{eq:Tmunu}, \eqref{eq:Vlasov}, \eqref{eq:RE1} is known as the \emph{reduced Einstein--Vlasov system}.
The condition that the coordinates $x^{\mu}$ satisfy the harmonic gauge condition \eqref{eq:harmonicgauge} is equivalent to the metric in the coordinates $x^{\mu}$ satisfying the \emph{wave coordinate condition},
\begin{equation} \label{eq:wcc}
	g^{\alpha \beta} \partial_{\alpha} g_{\beta \mu}
	=
	\frac{1}{2} g^{\alpha \beta} \partial_{\mu} g_{\alpha \beta},
\qquad\text{for}\qquad \mu = 0,1,2,3
\end{equation}

Let $m = \text{diag}\,(-1,1,1,1)$ denote the Minkowski metric in Cartesian coordinates and, for a solution $g$ of the reduced Einstein equations \eqref{eq:RE1}, write
\begin{equation} \label{eq:metricdecomp}
	g
	=
	m + h^0 + h^1,
	\qquad
	\text{where}
	\quad
	h^0_{\mu \nu}(t,x)
	=
	\chi \left( \frac{r}{t} \right) \chi(r) \frac{M}{r} \delta_{\mu \nu},
\end{equation}
and define, for given $0< \gamma < 1$, $0<\mu < 1-\gamma$, the energy at time $t$,
\begin{equation} \label{eq:energydef}
	E_N(t)
	=
	\sum_{\vert I \vert \leq N}
	\Vert
	w^{\frac{1}{2}}
	\partial Z^I h^1 (t,\cdot)
	\Vert_{L^2}^2,
	\qquad \qquad
	\text{where}
	\qquad
	w(t,x)
	=
	\begin{cases}
		(1+|r-t|)^{1+2\gamma},\,\,\,
		&
		r>t
		\\
		1+(1+|r-t|)^{-2\mu},
		\quad
		&
		r\leq t.
	\end{cases}
\end{equation}
Here $I$ denotes a multi index and $Z^I$ denotes a combination of $|I|$ of the vector fields,
\[
	\Omega_{ij} = x^i \partial_{x^j} - x^j \partial_{x^i},
	\quad
	B_{i} = x^i \partial_t + t \partial_{x^i},
	\quad
	S = t\partial_t + x^k \partial_{x^k},\quad\text{and}\quad\partial_\alpha,
\]
for $i,j=1,2,3$ and $\alpha=0,1,2,3$.
Let $\vert \cdot \vert$ denote the norm,
$	\vert x \vert = {\big( (x^1)^2 + (x^2)^2 + (x^3)^2 \big)}^{{1}\!/{2}}$,
and
$
	\vert h(t,x) \vert
	=
	\sum_{\alpha,\beta=0}^3 \vert h_{\alpha \beta}(t,x)\vert$,
	$\left\vert  \Gamma (t,x) \right\vert
	=
	\sum_{\alpha,\beta,\gamma=0}^3
	\big\vert \Gamma^{\alpha}_{\beta \gamma} (t,x) \big\vert $,
and similarly $\Vert h(t,\cdot) \Vert_{L^2} = \sum_{\alpha,\beta=0}^3 \Vert h_{\alpha \beta}(t,\cdot)\Vert_{L^2}$ etc.  The notation $A \lesssim B$ will be used if there exists a universal constant $C$ such that $A \leq C B$.

Theorem \ref{thm:main1} follows as a Corollary of the following theorem.

\begin{theorem} \label{thm:main2}
	For any $0<\gamma<1$ and $N\geq 11$, there exists $\varepsilon_0>0$ such that, for any data $(g_{\mu \nu}, \partial_t g_{\mu \nu},f) \vert_{t=0}$ for the reduced Einstein--Vlasov system \eqref{eq:Tmunu}, \eqref{eq:Vlasov}, \eqref{eq:RE1} which satisfy the smallness condition
	\[
		E_N(0)^{\frac{1}{2}} + \mathcal{V}_N + M < \varepsilon,
	\]
	for any $\varepsilon \leq \varepsilon_0$ and the wave coordinate condition \eqref{eq:wcc}, and such that $f\vert_{t=0}$ is compactly supported, there exists a global solution attaining the data such that
	\begin{equation} \label{eq:main2energy}
		\left(E_N(t) \right)^{\frac{1}{2}} + {\sum}_{\vert I \vert \leq N} (1+t) \Vert Z^I T^{\mu \nu}(t,\cdot) \Vert_{L^2}
		\leq
		C_N \varepsilon (1+t)^{C_N' \varepsilon},
	\end{equation}
	for all $t\geq 0$, along with the decay estimates
\begin{equation} \label{eq:main2pointwise}
|Z^I h^1(t,x)|\leq \frac{C_N^\prime\varepsilon(1+t)^{C_N'\varepsilon}}{(1+t+r) (1+q_+)^{\gamma}},
\qquad |I|\leq N-3,\qquad q_+=\begin{cases}
		r-t,\,\,\,
		&
		r>t
		\\
		0,
		\quad
		&
		r\leq t,
	\end{cases}
\end{equation}
and the estimates \eqref{eq:wavecoordinatederivativeLiefirst}, \eqref{eq:wavecoordinatefunctionLiefirst}, \eqref{eq:sharpHhigherlowderbeginsec61}, \eqref{eq:sharpHhigherlowderbeginsec6}, \eqref{eq:weakdecayderh1}, \eqref{eq:weakdecayh1} stated in Section \ref{section:Einstein}.
\end{theorem}

Note that the wave coordinate condition \eqref{eq:wcc}, if satisfied initially, is propagated in time by the reduced Einstein--Vlasov system \eqref{eq:Tmunu}, \eqref{eq:Vlasov}, \eqref{eq:RE1} (this fact is standard; see, for example, Section 4 of \cite{LR2}, which requires only minor modifications for the presence of matter).

Given an initial data set $(\Sigma,\overline{g},k,f_0)$ as in Theorem \ref{thm:main1}, define initial data for the reduced equations,
\[
	g_{ij} \vert_{t=0} = \overline{g}_{ij},
	\quad
	g_{00} \vert_{t=0} = - a^2,
	\quad
	g_{0i} \vert_{t=0} = 0,\qquad
a(x)^2 = \left( 1 - \chi(r)  {M}/{r}\right),
\]
and
\[
	\partial_t g_{ij} \vert_{t=0} = -2a k_{ij},
	\quad
	\partial_t g_{00} \vert_{t=0} = 2 a^3 \overline{g}^{ij} k_{ij},
	\quad
	\partial_t g_{0i} \vert_{t=0}
	=
	a^2 \overline{g}^{jk} \partial_j \overline{g}_{ik}
	-
	\frac{a^2}{2} \overline{g}^{jk} \partial_i \overline{g}_{jk}
	-
	a \partial_i a.
\]
One can show that, with this choice,
$
	\left( E_N(0) \right)^{{1}/{2}} \!\!\lesssim \mathcal{E}_N,
$
where $\mathcal{E}_N$ given by \eqref{eq:initialmetricnorm} is the norm of geometric data,
and moreover $(g_{\mu \nu}, \partial_t g_{\mu \nu}) \vert_{t=0}$ satisfy the wave coordinate condition \eqref{eq:wcc}, see, for example, \cite{LR3,LR2}.

It is therefore clear that Theorem \ref{thm:main1} follows from Theorem \ref{thm:main2} (the future causal geodesic completeness can be shown as in \cite{LR2}) and so the goal of the paper is to establish the proof of Theorem \ref{thm:main2}.

\subsection{Estimates for the Vlasov matter}
In what follows it is convenient, instead of parameterising the mass shell $\ms$ by $(t,x,p)$, to instead parameterise it by $(t,x,\ph)$, where
\[
	\ph^i = {p^i}/{p^0},
\]
for $i=1,2,3$.  Note that, by the mass shell relation \eqref{eq:massshell}, in Minkowski space $(p^0)^2 = 1 + (p^1)^2 + (p^2)^2 + (p^3)^2$ and, under a mild smallness condition on $g - m$, $\vert \ph \vert < 1$.  Abusing notation slightly, we will write $f(t,x,\ph)$ for the solution of the Vlasov equation \eqref{eq:Vlasov}.

Let $\{ \Sigma_t \}$ denote the level hypersurfaces of the time coordinate $t$, and let $X(s,t,x,\ph)^i$, $\Ph(s,t,x,\ph)^i$ denote solutions of the geodesic equations,
\begin{equation} \label{eq:geodesictimenormalized}
	\frac{d X^i}{d s}(s,t,x,\ph) = \Ph^i(s,t,x,\ph),
	\qquad
	\frac{d \Ph^i}{ds}(s,t,x,\ph) = \Gh^i (s, X(s,t,x,\ph), \Ph(s,t,x,\ph)),
\end{equation}
normalised so that $(s,X(s,t,x,\ph)) \in \Sigma_s$, with
\[
	X^i(t,t,x,\ph) = x^i,
	\qquad
	\Ph^i(t,t,x,\ph) = \ph^i.
\]
Here
\[
	\Gh^{\mu} (t,x,\ph)
	=
	\Gamma_{\alpha \beta}^0(t,x)\ph^{\alpha} \ph^{\beta} \ph^{\mu}
	-
	\Gamma_{\alpha\beta}^{\mu} (t,x) \ph^{\alpha}\ph^{\beta},\qquad \ph^0=1,
\]
where  $\Gamma^{\alpha}_{\beta \gamma}$ are the Christoffel symbols of the metric $g$ with respect to a given coordinate chart $(t,x^1,x^2,x^3)$.
Define $X(s,t,x,\ph)^0 = s$ and $\Ph(s,t,x,\ph)^0 = 1$.  The notation $X(s)$, $\Ph(s)$ will sometimes be used for $X(s,t,x,\ph)$, $\Ph(s,t,x,\ph)$ when it is clear from the context which point $(t,x,\ph)$ is meant, and the notation $\newhat{X}(s)=(s,X(s))$ will sometimes be used.

It follows that the Vlasov equation \eqref{eq:Vlasov} can be rewritten as,
\begin{equation} \label{eq:Vlasov2}
	f(t,x,\ph) = f(s,X(s,t,x,\ph),\Ph(s,t,x,\ph))=f_0(X(0,t,x,\ph),\Ph(0,t,x,\ph)),
\end{equation}
for all $s$.  The notation $(y,q)$ will be used to denote points in the mass shell over the initial hypersurface, $\ms\vert_{t=0}$.  In Theorem \ref{thm:main2} it is assumed that $f_0$ has compact support; $|y|\leq K$ and $|q|\leq K'$ for $(y,q) \in \supp(f_0)$, for some constants $K$, $K'$. Under a
relatively mild smallness assumptions on $h=g-m$, see Proposition \ref{prop:suppf}, it follows that there exists $c<1$, depending only on $K'$, such that solutions of the Vlasov equation satisfy
\beq\label{eq:supportofmatterintro}
\supp f\subset\{ (t,x,p);\, |x|\leq K+ct,\, \,|p\,|\leq K'+1,\, \,|\widehat{p}\,|\leq c\}.
\eq

The main new difficulties in the proof of Theorem \ref{thm:main2}, arising from the coupling to the Vlasov equation, are resolved in the following theorem, which is appealed to in the proof of Theorem \ref{thm:main2}.
\begin{theorem} \label{thm:mainL2}
	For a given $t\geq 0$ and $N\geq 1$, suppose that g is a Lorentzian metric such that the Christoffel symbols of $g$ with respect to a global coordinate system $(t,x^1,x^2,x^3)$, for some ${1}/{2} <a<1$, satisfy
	\begin{equation} \label{eq:GammamainL2}
		\left\vert Z^I \Gamma(t',x) \right\vert
		\leq
		\frac{C_N^\prime \varepsilon}{(1+t')^{1+a}},
		\qquad
		\text{for }
		\vert x \vert \leq ct' + K,
		\quad
		\vert I \vert \leq \left\lfloor \frac{N}{2} \right\rfloor +2,
	\end{equation}
	for all $t'\!\in[0,t]$. Then there exists $\varepsilon_0>0$ such that, for $\varepsilon<\varepsilon_0$ and any solution $f$ of the Vlasov equation \eqref{eq:Vlasov} satisfying \eqref{eq:supportofmatterintro}, and metric satisfying \eqref{eq:GammamainL2} and $\vert g -m \vert \leq \varepsilon$, the components of the energy momentum tensor $T^{\mu \nu} (t,x)$ satisfy
	\begin{align*}
		\Vert \left(  Z^I T^{\mu \nu} \right) (t,\cdot) \Vert_{L^{1}}
		\leq
		D_k\mathcal{V}_{k }
		+D_k\mathbb{D}_{k^\prime}
		\Big(\!\!\!
		\sum_{\vert J \vert \leq \vert I \vert  -1}\!\!\!\!\!
		\frac{\Vert ( Z^J \Gamma)(t,\cdot) \Vert_{L^2}\!\!\!}{(1+t)^{a-1\!/2}}
		+\!\!\!\!
		\sum_{\vert J \vert \leq \vert I \vert  + 1}
		\int_{0}^t
		\frac{\Vert ( Z^J \Gamma)(s,\cdot) \Vert_{L^2}\!\!\!}{(1+s)^{{1}\!/{2}+a}}\,
		ds
		\Big),	
	\end{align*}
	for $\vert I \vert \leq N-1$, where $k=|I|$ and $k^\prime=\left\lfloor {k }/{2} \right\rfloor +1$, and
	\begin{equation*}
	\Vert \left(  Z^I T^{\mu \nu} \right) (t,\cdot) \Vert_{L^{2}}
		\leq
		\frac{D_k\mathcal{V}_{k}}{(1+t)^{{3}/{2}}}
		+D_k\mathbb{D}_{k^\prime}
		\Big(\!\!\!
		\sum_{\vert J \vert \leq \vert I \vert  -1}\!\!\!\!\!
		\frac{\Vert ( Z^J \Gamma)(t,\cdot) \Vert_{L^2}\!\!}{(1+t)^{1+a}}
		+\!\!
		\sum_{\vert J \vert \leq \vert I \vert}
		\frac{1}{(1+t)^{{3}/{2}}}
		\int_{0}^t
		\frac{\Vert ( Z^J \Gamma)(s,\cdot) \Vert_{L^2}\!\!\!}{(1+s)^{{1}\!/{2}}}\,
		ds
		\Big),
	\end{equation*}
	for $\vert I \vert \leq N$.  Here the constants $D_k$ depend only on $C_N'$, $K$, $K'$ and $c$, and $\varepsilon_0$ depends only on $c$.
\end{theorem}

\begin{remark}
	The proof of Theorem \ref{thm:mainL2} still applies when $a\!\geq \!1$, though the theorem is only used in the proof of Theorem \ref{thm:main2} for some fixed $\frac{1}{2} \!< \!a<\!\!1$.  The case of $a\!=\!1$ is omitted in order to avoid logarithmic factors.  The proof of an appropriate result when $a\!>\!1$ is much simpler, although, when Theorem \ref{thm:mainL2} is used in the proof of Theorem \ref{thm:main2}, one could not hope for the assumptions \eqref{eq:GammamainL2} to hold with $a\!>\!1$, see  Section \ref{subsec:introvfs2}.
\end{remark}

\begin{remark}
	In Section \ref{section:Testimates} a better $L^2$ estimate in terms of $t$ behaviour, compared with the $L^2$ estimate of Theorem \ref{thm:mainL2}, is shown to hold for $Z^I T^{\mu \nu}$, which involves one extra derivative of $\Gamma$.  See Proposition \ref{prop:mainL21}.  It is important however to use the $L^2$ estimate which does not lose a derivative in the proof of Theorem \ref{thm:main2}.
\end{remark}

\subsection{Overview of the global existence theorem for the reduced Einstein equations}

It should be noted from the outset that the reduced Einstein--Vlasov system \eqref{eq:Tmunu}, \eqref{eq:Vlasov}, \eqref{eq:RE1} is a system of quasilinear wave equations coupled to a transport equation.  It is well known that the general quasilinear wave equation does not necessarily admit global solutions for all small data \cite{J1,J2}.  The \emph{null condition}, an algebraic condition on the nonlinearity of such equations, was introduced by Klainerman \cite{K1}, and used independently by Klainerman \cite{Kl} and Christodoulou \cite{C1}, as a sufficient condition that small data solutions exist globally in time and are asymptotically free.
However, as was noticed by Lindblad \cite{L2,L3} and Alinhac \cite{A}, there are quasilinear equations that do not satisfy the null condition but still admit global solutions for all sufficiently small data.  In fact, the classical null condition fails to be satisfied by the vacuum Einstein equations in the harmonic gauge (\eqref{eq:RE1} with $T\equiv 0$), though it was noticed by Lindblad--Rodnianski \cite{LR1} that they satisfy a \emph{weak null condition}, which they used to prove a small data global existence theorem \cite{LR2}, \cite{LR3}.

The proof of Theorem \ref{thm:main2} follows the strategy adopted in \cite{LR3}.  The new difficulties, of course, arise from the coupling to the Vlasov equation.  A fundamental feature of the problem arises from the fact that, whilst the slowest decay of solutions to wave equations occurs in the \emph{wave zone}, where $t \sim r$, the slowest decay of solutions of the massive Vlasov equation occurs in the interior region $t>r$.  The most direct way to exploit this fact is to impose that $f_0$ has compact support, in which case, as will be shown in Proposition \ref{prop:suppf}, the support condition \eqref{eq:supportofmatterintro} holds and $f$ actually vanishes in the wave zone at late times.

Since they have been described at length elsewhere, the difficulties associated to the failure of \eqref{eq:RE1} to satisfy the classical null condition are only briefly discussed here.  Suffice it to say that there is a rich structure in the equations \eqref{eq:RE1} which is exploited heavily (see the further discussion in the introductions to \cite{LR2,LR3}).  The main new features of this work are contained in the proof of Theorem \ref{thm:mainL2}.  Indeed, for a given inhomogeneous term $\widehat{T}$ in \eqref{eq:RE1} which satisfies the support conditions and estimates of Theorem \ref{thm:mainL2}, the small data global existence theorem of \cite{LR3} mostly goes through unchanged.  An outline is given below, including a discussion of some observations which lead to simplifications compared with the proof in \cite{LR3}.

The proof of Theorem \ref{thm:main2} is based on a continuity argument.  One assumes the bounds
\begin{equation} \label{eq:introba}
	E_N(t)^{\frac{1}{2}}
	\leq
	C_N \varepsilon (1+t)^{\delta},
	\qquad
	{\sum}_{\vert I \vert \leq N-1} \Vert Z^I T^{\mu \nu} (t,\cdot) \Vert_{L^1}
	\leq
	C_N \varepsilon,
\end{equation}
hold for all $t\in [0,T_*]$ for some time $T_*>0$ and some fixed constants $C_N$ and $\delta$, and the main objective is to use the Einstein equations to prove that the bounds \eqref{eq:introba} in fact hold with better constants, provided the initial data are sufficiently small.

\subsubsection{The contribution of the mass}
The first step in the proof of Theorem \ref{thm:main2} is to identify the contribution of the mass $M$.  Recall the decomposition of the metric \eqref{eq:metricdecomp} and note that the energy $E_N$ is defined in terms of $h^1$.  Had the energy been defined with $h=g-m$ in place of $h^1$, it would not be finite unless $M=0$, in which case it follows from the Positive Mass Theorem \cite{ShYa}, \cite{Wi} that the constraint equations imply that the solution is trivial.  The contribution of the mass is therefore identified explicitly using the decomposition \eqref{eq:metricdecomp}, and the reduced Einstein equations are recast as a system of equations for $h^1$,
\begin{equation} \label{eq:RE2}
	\Boxr_g \, h^1_{\mu\nu}
	=
	F_{\mu\nu} (h)(\pa h, \pa h)+\widehat{T}_{\mu\nu} - \Boxr_g \, h^0_{\mu\nu}.
\end{equation}
The term $\Boxr_g \, h^0_{\mu\nu}$ is treated as an error term.  Note that $h^0$ is defined so that $\Box \, h^0_{\mu\nu}$, which is good approximation to $\Boxr_g \, h^0_{\mu\nu}$, is supported away from the wave zone $t \sim r$ and so only contributes in the interior region, where $h^1$ will be shown to have fast decay.

\subsubsection{Energy inequality with weights}
An important ingredient in the procedure to recover the assumption on the energy \eqref{eq:introba} is the energy inequality with weights,
\begin{multline} \label{eq:introenergyineq}
	\int_{\Si_{t}}\!\! |\pa\phi|^{2}\,w dx + \!\int_{0}^{t}\!\! \int_{\Si_{\tau}}
	\!\!|\overline{\pa}\phi|^{2}\,w^{\,\prime} dx d\tau \leq
	8\!\int_{\Si_{0}}\!\! |\pa \phi|^{2}\,w dx+
	\!\int_0^t\!\!\frac{C\varepsilon}{1+\tau}\int_{\Si_{\tau}} \!\!  |\pa\phi|^{2}
	\, w \,dx \, d\tau\\
	+16 \int_0^t\Big(\int_{\Si_{\tau}}\!\!\! |\Boxr_g \phi|^2 wdx\Big)^{\!1/2}\Big(\int_{\Si_{\tau}} \!\!\! |\pa \phi|^2 w dx\Big)^{\!1/2}\! d\tau,
\end{multline}
which holds under mild assumptions on the metric $g$ (see Lemma \ref{lem:Decayenergy}), where the weight $w$ is as in \eqref{eq:energydef}  and $\overline{\partial}_{\mu} = \partial_{\mu} - \frac{1}{2}L_{\mu}(\partial_r - \partial_t)$, with $L = \partial_t + \partial_r$, denotes the derivatives tangential to the outgoing Minkowski light cones.  The inequality will be applied to the system \eqref{eq:RE2} after commuting with vector fields.

It is in the energy inequality \eqref{eq:introenergyineq} that the $L^2$ estimates of Theorem \ref{thm:mainL2} are used.  The proof of Theorem \ref{thm:main2}, using Theorem \ref{thm:mainL2}, is given in detail in Section \ref{section:cty}, but we  briefly illustrate the use of the $L^2$ estimates of Theorem \ref{thm:mainL2} here. Setting $Q_N(t) = \sup_{0\leq s \leq t} \sum_{\vert I \vert \leq N} \Vert w^{\frac{1}{2}} \partial Z^I h^1 (s,\cdot) \Vert_{L^2}$, it follows from Theorem \ref{thm:mainL2} that,
\begin{equation} \label{eq:introThat}
	\Vert Z^I \newhat{T} (\tau,\cdot) \Vert_{L^2}
	\leq
	\frac{C\varepsilon}{(1+\tau)^{\frac{3}{2}}} + \frac{C\varepsilon}{(1+\tau)} Q_N(t),
\end{equation}
(see Section \ref{section:cty} for details of how estimates for $\newhat{T}^{\mu \nu}$ follow from estimates for $T^{\mu \nu}$).  By the reduced Einstein equations \eqref{eq:RE2} and the energy inequality \eqref{eq:introenergyineq},
\begin{equation} \label{eq:introQN}
	Q_N(t)
	\lesssim
	Q_N(0)
	+
	\sum_{\vert I \vert \leq N}
	\int_0^t
	\Vert [Z^I, \Boxr_g] h^1 (s,\cdot)\Vert_{L^2}
	+
	\Vert Z^I \Boxr_g h^0 (s,\cdot)\Vert_{L^2}
	+
	\Vert Z^I F_{\mu\nu} (s,\cdot)\Vert_{L^2}
	+
	\Vert Z^I \newhat{T} (s,\cdot)\Vert_{L^2}
	ds.
\end{equation}
The first four terms on the right hand side arise already in \cite{LR3} and so, combining estimates which will be shown for these terms in Section \ref{section:Einstein} (see also the discussion below) with \eqref{eq:introThat}, the bound \eqref{eq:introQN} implies that
\[
	Q_N(t) \leq C \varepsilon (1+t)^{C\varepsilon} + C \varepsilon \int_0^t \frac{Q_N(\tau)}{1+\tau} d \tau,
\]
and so the Gr\"{o}nwall inequality yields
\[
	Q_N(t) \leq C\varepsilon (1+t)^{2C\varepsilon}.
\]
In the proof of Theorem \ref{thm:main2}, such a bound for $Q_N$ will lead to a recovery of the assumption \eqref{eq:introba} on $E_N$ with better constants provided $C_N$ is chosen to be sufficiently large and $\varepsilon$ is sufficiently small.

\subsubsection{The structure of the nonlinear terms} \label{subsubsec:nonlinearstructure}
As discussed above, whether a given quasilinear wave equation admits global solutions for small data or not depends on the structure of the nonlinear terms (moreover, the main analysis of the nonlinear terms is relevant in the wave zone, where ${\newhat{T}}^{\mu \nu}\!\!$ vanishes at late times and so plays no role in this discussion).  A closer inspection of the nonlinearity in \eqref{eq:RE2} reveals (see \cite{LR3}, \cite{ChBr00}) that
\begin{equation} \label{eq:introFmunu}
	F_{\mu\nu}(h)(\partial h,\partial h)
	=
	P(\pa_{\mu} h,\pa_{\nu} h) + Q_{\mu\nu}(\pa h,\pa h)+G_{\mu\nu}(h)(\pa h,\pa h),
\end{equation}
where $Q_{\mu\nu}(\pa h,\pa h)$ is a linear combination of \emph{null forms} (satisfing $\vert Q_{\mu\nu}(\pa h,\pa h) \vert \lesssim \vert \overline{\partial} h \vert \vert \partial h \vert$), $\vert G_{\mu\nu}(h)(\pa h,\pa h) \vert \lesssim \vert h \vert \vert \partial h \vert^2$ denote cubic terms and,
\begin{equation} \label{eq:introPmunu}
	P(\pa_{\mu}h,\pa_{\nu} h)
	=
	\frac{1}{2}  m^{\alpha\alpha^\prime}m^{\beta\beta^\prime} \pa_\mu h_{\alpha\beta}
	\,
	\pa_\nu h_{\alpha^\prime\beta^\prime}
	-
	\frac{1}{4}  m^{\alpha\alpha^\prime}\pa_\mu
	h_{\alpha\alpha^\prime} \,  m^{\beta\beta^\prime}\pa_\nu
	h_{\beta\beta^\prime}.
\end{equation}

Clearly the failure of the semilinear terms of the system \eqref{eq:RE2} to satisfy the classical null condition arises in the $P(\partial_{\mu} h, \partial_{\nu} h)$ terms.  In \cite{LR1} it was observed that the semilinear terms of \eqref{eq:RE2}, after being
decomposed with respect to a null frame $\mathcal{N} = \{ \underline{L}, L, S_1, S_2 \}$, where
\begin{equation} \label{eq:frameintro}
	\underline{L}=\pa_t-\pa_r,
	\quad
	L=\pa_t+\pa_r,
	\quad
	S_1,S_2\in\mathbf{S}^2,
	\quad
	\langle S_i,S_j\rangle =\delta_{ij},
\end{equation}
possess a \emph{weak null structure}. It is well known
that, for solutions of wave equations, derivatives tangential to the outgoing light
cones $\overline{\pa }\in\mathcal{T}=\{L,S_1,S_2\}$ decay faster and so, neglecting such $\overline{\partial}  h$ derivatives of $h$,
\beq\label{eq:transversalderivativeprojection}
	\pa_\mu h\sim L_\mu \pa_q h,\qquad\text{where}\quad \pa_q=(\pa_r-\pa_t)/2,
	\quad L_\mu=m_{\mu\nu} L^\nu,
\eq
and, neglecting cubic terms and quadratic terms involving at least one tangential derivative,
\beq\label{eq:einsteinfirstapproximationintro}
	\widetilde{\Box}_g h_{\mu\nu}\sim L_\mu L_\nu  P(\pa_q h,\pa_q h).
\eq
For vectors $U,V$, define $(\Boxr_g h)_{UV} = U^{\mu} V^{\nu} \Boxr_g h_{\mu \nu}$.  With respect to the null frame \eqref{eq:frameintro}, the Einstein equations \eqref{eq:RE2} become
\beq
	(\widetilde{\Box}_g h)_{TU}\sim 0, \quad T\in \mathcal{T},U\in\mathcal{N}\qquad
	(\widetilde{\Box}_g h)_{\underline{L}\underline{L}}\sim 4 P(\pa_q h,\pa_q h),
	\label{eq:simplifiedEinstein}
\eq
since $T^\mu L_\mu=0$ for $T\in\mathcal{T}$.  Decomposing with respect to the null frame \eqref{eq:frameintro}, the term $P(\pa_q h,\pa_q h)$ is equal to $P_N(\partial_{q} h, \partial_{q} h)$, where
\begin{multline}\label{eq:Pnullframe}
	P_{\!\mathcal{N}}(D,E)=-\big(D_{LL} E_{\underline{L}\underline{L}}
	+D_{\underline{L}\underline{L}} E_{{L}{L}}\big)/8
	-\big(2D_{AB}E^{AB}-
	D_{\!\! A}^{\, A} E_{B}^{\,\,\,B}\big)/4\\
	+\big(2D_{A L}E^{A}_{\,\,\, \underline{L}} +2D_{A
	\underline{L}}E^A_{\,\,\,{L}}- D_{\!\!A}^{\, A}
	E_{L\underline{L}}-D_{L\underline{L}}E_{\!A}^{\,\,A}\big)/4,
\end{multline}
(see \cite{LR2}).  Except for the $\partial_q h_{LL} \partial_q h_{\underline{L} \underline{L}}$ term, $P_{\!\mathcal{N}}( \partial_{q} h, \partial_{q} h)$ only involves the non $\underline{L} \underline{L}$ components of $h$.  The wave coordinate condition \eqref{eq:wcc} with respect to the null frame becomes
\begin{equation}\label{eq:wavecordhormander12}
	\pa_q h_{L T}\sim 0,\quad T\in\mathcal{T},\qquad
	\delta^{AB}\partial_q  h_{AB}\sim 0,\quad A,B\in\mathcal{S}=\{S_1,S_2\},
\end{equation}
 neglecting tangential derivatives and quadratic terms, see \cite{LR2}.  In particular the $\partial_q h_{LL} \partial_q h_{\underline{L} \underline{L}}$ term in $P_{\!\mathcal{N}}( \partial_{q} h, \partial_{q} h)$ can be neglected.  The asymptotic identity \eqref{eq:wavecordhormander12} moreover implies (see equation \eqref{eq:tanP}) that the leading order behaviour of $P_N(\partial_{q} h, \partial_{q} h)$ is contained in $P_{\!\mathcal{S}}(\pa_q h,\pa_q h)$, where
\begin{equation} \label{eq:PSintro}
	P_{\mathcal{S}} (D,E)= -\widehat{D}_{AB}\, \widehat{E}^{AB}\!/2,\quad A,B\in\mathcal{S},\quad
	\text{where} \quad \widehat{D}_{AB}=D_{AB}-\delta_{AB}\trs D/2,\quad
	\trs D=\delta^{AB}D_{AB}.
\end{equation}
A decoupling therefore occurs in the semilinear terms of \eqref{eq:RE2}, modulo terms which are cubic or involve at least one ``good'' $\overline{\partial}$ derivative, and the right hand side of the second identity in \eqref{eq:simplifiedEinstein}
only depends on components we have better control on by the first identity in
\eqref{eq:simplifiedEinstein}.

A further failure of \eqref{eq:RE2} to satisfy the classical null condition arises in the \emph{quasilinear terms}.  Expressing the inverse of the metric $g_{\mu \nu}$ as
\begin{equation*}
	g^{\mu\nu}
	=
	m^{\mu\nu}+H^{\mu\nu},
\end{equation*}
the reduced wave operator takes the form
\[
	\Boxr_g = m^{\alpha \beta} \partial_{\alpha} \partial_{\beta} + H^{\alpha \beta} \partial_{\alpha} \partial_{\beta},
\]
which differs from the Minkowski wave operator $\Box$ only by the term $H^{\underline{L} \underline{L}} \partial_q^2$, plus terms which involve at least one tangential $\overline{\partial}$ derivative.  This main quasilinear term is controlled by first rewriting the wave coordinate condition \eqref{eq:wcc} as
\begin{equation} \label{eq:introwccnull}
\pa_\mu \widehat{H}^{\mu\nu\!}\!
= W^{\nu}(h,\pa h)\qquad
\text{where}\quad  \widehat{H}^{\mu\nu}\!=\!H^{\mu\nu}\!-m^{\mu\nu}
\tr_m H_{\!}/2,\quad \tr H\!=m_{\alpha\beta} H^{\alpha\beta},
\end{equation}
where $\vert W^{\nu}(h,\pa h) \vert \lesssim \vert h \vert \vert \partial h \vert$ is quadratic, and using the formula
\[
	\pa_\mu F^{\mu \nu}
	=
	L_\mu\pa_q F^{\mu \nu} -\Lb_{\,\mu} \pa_s F^{\mu \nu} + A_\mu \pa_A F^{\mu \nu},
\]
for any $F^{\mu \nu}$, to rewrite $\pa_\mu \widehat{H}^{\mu\nu}$ in terms of the null frame.  This gives,
\begin{equation*}
	\partial_q H^{\Lb \Lb} = L( H^{L \Lb} ) - \partial_A H^{A \Lb} + W^{\Lb}(h,\partial h),
\end{equation*}
i.\@e.\@ $\partial_q H^{\Lb \Lb}$ is equal to quadratic terms plus terms involving only tangential $\overline{\partial}$ derivatives.  Integrating $\partial_q H^{\Lb \Lb}$ from initial data $\{t=0\}$ then gives that $H^{\underline{L} \underline{L}}$ is approximately equal to the main contribution of its corresponding initial value, $2M/r$.

\subsubsection{Commutation} \label{subsubsec:introcommutation}
In order to apply the energy inequality \eqref{eq:introenergyineq} to improve the higher order energy bounds \eqref{eq:introba}, it is necessary to commute the system \eqref{eq:RE2} with the vector fields $Z$.  Instead of commuting with the vector fields $Z$ directly as in \cite{LR3}, notice that, for any function $\phi$,
\[
	\Boxr_g Z \phi
	=
	Z \left( \Boxr_g \phi \right)
	+
	2 g^{\alpha \beta} \partial_{\alpha} Z^{\mu} \partial_{\beta} \partial_{\mu} \phi - Z(g^{\alpha \beta}) \partial_{\alpha} \partial_{\beta} \phi
	=
	Z \left( \Boxr_g \phi \right)
	-
	(\mathcal{L}_Z g^{\alpha \beta}) \partial_{\alpha} \partial_{\beta} \phi,
\]
where the fact that $\partial_{\mu} \partial_{\nu} Z^{\lambda} = 0$ for each $Z$ and $\mu, \nu, \lambda = 0,1,2,3$ has been used.  Here $\mathcal{L}_Z$ denotes the Lie derivative along the vector field $Z$ (see Section \ref{subsec:commutation} for a coordinate defintion).  The procedure of commuting the system \eqref{eq:RE2} therefore becomes computationally much simpler if it is instead commuted with the Lie derivatives along the vector fields, $\mathcal{L}_Z$.  In fact, the procedure simplifies further by commuting with a modified Lie derivative $\widehat{\mathcal{L}}$, defined in the $(t,x)$ coordinates by the formula
\[
	\widehat{\mathcal L}_Z K^{\alpha_1\dots \alpha_r}_{\beta_1\dots \beta_s}
	=
	{\mathcal L}_Z K^{\alpha_1\dots \alpha_r}_{\beta_1\dots \beta_s}
	+
	\tfrac{r-s}{4}(\pa_\gamma Z^\gamma)K^{\alpha_1\dots \alpha_r}_{\beta_1\dots \beta_s}.
\]
The modified Lie derivative has the property that $\widehat{\mathcal{L}}_Z m = 0$ for each of the vector fields $Z$ and moreover a computation shows that, in the case that $\phi_{\mu \nu}$ is a $(0,2)$ tensor, the commutation property
\begin{equation} \label{eq:intromodLiecomm}
	\Boxr_g \widehat{\mathcal{L}}_Z \phi_{\mu \nu}
	=
	\mathcal{L}_Z \left( \Boxr_g \phi_{\mu \nu} \right)
	-
	(\widehat{\mathcal{L}}_Z H^{\alpha \beta}) \partial_{\alpha} \partial_{\beta} \phi_{\mu \nu},
\end{equation}
holds for each of the vector fields $Z$.

The commutation error in \eqref{eq:intromodLiecomm} can be controlled by $(\widehat{\mathcal{L}}_Z H^{\Lb \Lb}) \partial^2 \phi_{\mu \nu}$, plus terms which involve at least one tangential derivative of $\phi_{\mu \nu}$,
\[
	\vert (\widehat{\mathcal{L}}_Z H^{\alpha \beta}) \partial_{\alpha} \partial_{\beta} \phi_{\mu \nu} \vert
	\lesssim
	\vert \widehat{\mathcal{L}}_Z H^{\Lb \Lb} \vert \vert \partial^2 \phi \vert
	+
	\vert \widehat{\mathcal{L}}_Z H \vert \vert \overline{\partial} \partial \phi \vert.
\]
The Lie derivative along any of the vector fields $Z$ commutes with partial derivatives $\partial$ (see Proposition \ref{prop:Liecommutation}).  This fact leads to the commutation formula,
\[
	 \pa_\mu  \widehat{\mathcal L}_Z \widehat{H}^{\mu\nu}=\big(\widehat{\mathcal L}_Z +\tfrac{\pa_\gamma Z^\gamma}{2}\big)\pa_\mu \widehat{H}^{\mu\nu},
\]
involving the modified Lie derivative.  The term $\widehat{\mathcal{L}}_Z H^{\Lb \Lb}$ is then controlled easily by using the formula \eqref{eq:introwccnull} and repeating the argument, described in Section \ref{subsubsec:nonlinearstructure}, used to control $H^{\Lb \Lb}$ itself.

When applying the commutation formula \eqref{eq:intromodLiecomm} to the reduced Einstein equations \eqref{eq:RE2}, it in particular becomes necessary to estimate the Lie derivative of the nonlinear terms, $\mathcal{L}_Z^I \left( F_{\mu \nu}(h)(\partial h, \partial h) \right)$.  Recall the nonlinear terms take the form \eqref{eq:introFmunu}.  The modified Lie derivative also simplifies the process of understanding derivatives of the nonlinear terms due to the following product rule.  Let $h_{\alpha\beta}$ and $k_{\alpha\beta}$ be $(0,2)$ tensors and let $S_{\mu\nu}(\pa h,\pa k)$ be a $(0,2)$ tensor which is a quadratic form in the $(0,3)$ tensors $\pa h$ and $\pa k$ with two contractions with the Minkowski metric (in particular $P(\pa_\mu h,\pa_\nu k)$ or $Q_{\mu\nu}(\pa h,\pa k)$). Then
\[
 {\mathcal L}_Z\,\big(  S_{\mu\nu}(\pa h,\pa k)\big)
 =S_{\mu\nu}(\pa \widehat{\mathcal L}_Z h,k)+S_{\mu\nu}(\pa h,\pa \widehat{\mathcal L}_Z k),
\]
and so the desirable structure of the nonlinear terms $P(\pa_\mu h,\pa_\nu h)$ and $Q_{\mu\nu}(\pa h,\pa k)$ described in Section \ref{subsubsec:nonlinearstructure} is preserved after applying Lie derivatives.

The Lie derivatives of the energy momentum tensor are controlled by Theorem \ref{thm:mainL2} since, for any function $\phi$, the quantities $\sum_{\vert I \vert \leq N} \vert Z^I \phi \vert$ and $\sum_{\vert I \vert \leq N} \vert \widehat{\mathcal L}_Z^I \phi \vert$ are comparable.

\subsubsection{The Klainerman--Sobolev Inequality with weights}
In order to control the derivatives of the nonlinear terms and the error terms arising from commuting the system \eqref{eq:RE2}, described in Section \ref{subsubsec:introcommutation}, when using the energy inequality \eqref{eq:introenergyineq}, pointwise estimates for lower order derivatives of the solution are first shown to hold.  The \emph{Klainerman--Sobolev Inequality} can be used to derive non-sharp bounds for $\vert \partial Z^I h^1 \vert$ for $\vert I \vert \leq N-2$ (see equation \eqref{eq:weakdecayderh1}) directly from the bound on the energy \eqref{eq:introba}.  These pointwise bounds can be integrated from $\{t=0\}$ to also give pointwise bounds for $\vert Z^I h^1\vert$ for $\vert I \vert \leq N-2$ which, using the fact that $\vert \overline{\partial} \phi \vert \leq \frac{C}{1+t+r} \sum_{\vert I \vert = 1} \vert Z^I \phi \vert$ for any function $\phi$, lead to strong pointwise estimates for all components of $\overline{\partial} Z^I h^1$ for $\vert I \vert \leq N-3$.  See Section \ref{subsec:weakdecay} for more details.

Since, without restricting to tangential derivatives, it is only true that $\vert \partial \phi \vert \leq \frac{C}{1+\vert t-r\vert} \sum_{\vert I \vert = 1} \vert Z^I \phi \vert$ for any function $\phi$, the Klainerman--Sobolev Inequality does not directly lead to good pointwise estimates for all derivatives of $h^1$.  Some further improvement is necessary to control the terms in the energy estimate and recover the inequality \eqref{eq:introba}.

\subsubsection{$L^{\infty}$--$L^{\infty}$ estimate for the wave equation}
The pointwise decay obtained for the transverse derivative of certain components of $h^1$ ``for free'' from the wave coordinate condition is not sufficient in the wave zone $t\sim r$ to close the energy estimate.  The decay in this region is further improved by an $L^{\infty}$--$L^{\infty}$ estimate, obtained by integrating the equations along the outgoing characteristics of the wave equation.  In fact, instead of using the estimate for the full wave operator $\Boxr_g$, as in \cite{LR3}, it suffices to use the estimate for the operator $\Box_0 = \big( m^{\alpha \beta} - \frac{M}{r} \chi \big( \frac{r}{1+t} \big) \delta^{\alpha \beta} \big) \partial_{\alpha} \partial_{\beta}$.  Moreover, using the pointwise decay obtained from the Klainerman--Sobolev inequality and the wave coordinate condition, it can be seen that, for the purposes of this estimate, the essential contribution of the failure of \eqref{eq:RE1} to satisfy the classical null condition is present in the $P_{\!\mathcal{S}}(\pa_q h,\pa_q h)$ terms, defined by \eqref{eq:PSintro}.  See Proposition \ref{prop:approxwaveequation} for a precise statement of this, and Lemma \ref{lem:transversalder} for a proof of the $L^{\infty}$--$L^{\infty}$ inequality.  The fact that $f$ is supported away from the wave zone can be shown using only the decay obtained from the Klainerman--Sobolev Inequality, and so the $\widehat{T}$ term in \eqref{eq:RE2} plays no role in Lemma \ref{lem:transversalder}.  The pointwise decay of higher order Lie derivatives of $h^1$ is similarly improved in Section \ref{subsec:sharpdecayhighorder}.

\subsubsection{The H\"{o}rmander $L^1$--$L^{\infty}$ Inequality}
Whilst the pointwise decay for lower order derivatives of $h$ described above is sufficient to recover the assumptions \eqref{eq:introba} in the vacuum (when $T^{\mu \nu} \equiv 0$), the interior decay is not sufficient to satisfy the assumptions of Theorem \ref{thm:mainL2}.  In Proposition \ref{prop:weakdecayhormander} the H\"{o}rmander $L^1$--$L^{\infty}$ inequality, Lemma \ref{lemma:hormander}, is used, together with the assumptions \eqref{eq:introba} on the $L^1$ norms of $Z^I T^{\mu \nu}(t,\cdot)$ and on the energy $E_N(t)$ of $h^1$, in order to improve the interior decay of $h^1$ and lower order derivatives.  This improved decay ensures that the assumptions of Theorem \ref{thm:mainL2} are satisfied and hence the theorem can be appealed to in order to recover the assumptions \eqref{eq:introba} on the $L^1$ norms of $Z^I T^{\mu \nu}$, and to control the $L^2$ norms of $Z^I T^{\mu \nu}$ arising when the energy inequality \eqref{eq:introenergyineq} is used to improve the assumptions \eqref{eq:introba} on the energy $E_N$.  See Section \ref{section:cty} for further details on the completion of the proof of Theorem \ref{thm:main2}.

\subsection{Vector fields for the Vlasov equation}
\label{subsec:Vlasovvectors}

The remaining difficulty in the proof of Theorem \ref{thm:main2} is in establishing the $L^1$ and $L^2$ estimates of the vector fields applied to components of the energy momentum tensor, $Z^I T^{\mu \nu}(t,x)$, of Theorem \ref{thm:mainL2}.  For simplicity, we outline here how bounds are obtained for $Z \rho(t,x)$, for $Z = \Omega_{ij}, B_i, S$, where $\rho(t,x)$ is the momentum average of $f$, defined by
\[
	\rho(t,x) := \int f(t,x,\ph) d \ph.
\]
The bounds for $Z \rho(t,x)$ will follow from bounds of the form
\begin{equation} \label{eq:goodfbounds}
	\left\vert
	(\Zb f)(t, x, \ph)
	\right\vert
	\leq
	C,
\end{equation}
for a suitable collection of vector fields $\Zb$, which reduce to the $Z = \Omega_{ij}, B_i, S$ vector fields when acting on spacetime functions, i.\@e.\@ functions of $(t,x)$ only.

Throughout this section, and in Sections \ref{section:geodesics} and \ref{section:Testimates}, it is convenient, instead of considering initial data to be given at $t=0$, to consider initial data for the Vlasov equation to be given at $t=t_0$ for some $t_0\geq 1$.\footnote{In fact in Sections \ref{section:geodesics} and \ref{section:Testimates} a new translated time coordinate $\tilde{t} = t + t_0$ is introduced, which has the property that $\{ \tilde{t} = t_0 \} = \{ t = 0\}$.  It is more convenient to use the vector fields defined with respect to $\tilde{t}$ than the vector fields defined with respect to $t$.  See Section \ref{subsec:transtime} for details.}  It follows from the form of the Vlasov equation \eqref{eq:Vlasov2} that
\begin{equation} \label{eq:ZbVlasov}
	\Zb f(t,x,\ph)
	=
	\Zb \big( X(t_0,t,x,\ph)^i \big) (\partial_{x^i} f) (t_0,X(t_0), \Ph(t_0))
	+
	\Zb \big( \Ph(t_0,t,x,\ph)^i \big) (\partial_{\ph^i} f) (t_0,X(t_0), \Ph(t_0)),
\end{equation}
(where $X(t_0,t,x,\ph), \Ph(t_0,t,x,\ph)$ are abbreviated to $X(t_0), \Ph(t_0)$ respectively) for any vector $\Zb$. Since derivatives of $f|_{t=t_0}$ are explicitly determined by initial data, an estimate for $\vert \Zb f(t,x,\ph) \vert$ will follow from appropriate estimates for $\vert \Zb \left( X(t_0,t,x,\ph)^i \right) \vert$ and $\vert \Zb \big( \Ph(t_0,t,x,\ph)^i \big) \vert$.

\subsubsection{General procedure and vector fields in Minkowski space} \label{subsec:Minkowskivectors}

A natural way to extend a given vector field $Z$ on $\mathcal{M}$ to a vector field on $\ms$, which by construction will have the property that $\Zb(X(t_0)^i)$ satisfy good bounds, is as follows.
For a given vector field $Z$ on $\mathcal{M}$, let $\Phi^Z_{\lambda} : \mathcal{M} \to \mathcal{M}$ denote the associated one parameter family of diffeomorphisms, so that
\[
	\frac{d \Phi^Z_{\lambda}(t,x)}{d\lambda} \bigg\vert_{\lambda = 0} = Z\vert_{(t,x)}.
\]
Under a mild assumption on g, for fixed $\tau$ any point $(t,x,\ph) \in \ms$ with $t >\tau$ can be uniquely described by a pair of points $\{ (t,x), (\tau,y) \}$ in $\mathcal{M}$, where
\begin{equation} \label{eq:Minkowskiy}
	y=X(\tau,t,x,\ph),
\end{equation}
is the point where the geodesic emanating from $(t,x)$ with velocity $\ph$ intersects the hypersurface $\Sigma_{\tau}$ (recall that $\{ \Sigma_t \}$ denotes the level hypersurfaces of the function $t$), i.\@e.\@ $(t,x,\ph) \in \ms$ can be parameterised by $\{ (t,x), (\tau,y) \}$,
\[
	(t,x,\ph) = (t,x,\ph_X(t,x,\tau,y)).
\]
Now the action of $\Phi^Z_{\lambda}$ on $(t,x)$ and $(\tau,y)$ induces an action on $\ms$ at time $t$, given by
\[
	\Phib^{Z,X}_{\lambda,\tau}(t,x,\ph)
	:=
	\left( \Phi^Z_{\lambda}(t,x), \ph_X\left( \Phi^Z_{\lambda}(t,x), \Phi^Z_{\lambda}(\tau,y) \right) \right).
\]
 For fixed $t_0$ we define the vector field $\Zb$ by
\[
	\Zb \vert_{(t,x,\ph)} = \frac{d \Phib^{Z,X}_{\lambda,\tau}(t,x,\ph)}{d\lambda} \bigg\vert_{\lambda = 0,\,\tau=t_0}.
\]
A computation shows that
\begin{equation} \label{eq:introZbX0}
	\Zb\vert_{(t,x,\ph)} \left( X(t_0,t,x,\ph)^i \right)
	=
	Z^i \vert_{(t_0,X(t_0))} - Z^0 \vert_{(t_0,X(t_0))} \Ph(t_0,t,x,\ph)^i,
\end{equation}
which results in a good bound for $\vert \Zb (X(t_0)^i) \vert$.  In particular, if $X^i(t_0,t,x,\ph)$ and  $\Ph^i(t_0,t,x,\ph)$ are bounded in the support of $f(t_0,X(t_0), \Ph(t_0))$, equation \eqref{eq:introZbX0} guarantees that $\vert \Zb (X(t_0)^i) \vert$ is bounded by a constant.

To see that \eqref{eq:introZbX0} indeed holds, first note that the left hand side is the derivative of $ X\big(t_0,\Phib^{Z,X}_{\lambda,t_0}(t,x,\ph)\big)^i$ with respect to $\lambda$ at
$\lambda=0$.  Also the first
term on the right hand side is the derivative of $\Phi^Z_{\lambda}(t_0,y)^i$ at $\lambda=0$.
The equality \eqref{eq:introZbX0} follows from taking the derivative with respect to $\lambda$, and setting $\lambda=0$, of both sides of the identity
\begin{equation}\label{eq:introZbX0computation}
X\big(\Phi^Z_{\lambda}(t_0,y)^0,\Phib^{Z,X}_{\lambda,t_0}(t,x,\ph)\big)^i
=\Phi^Z_{\lambda}(t_0,y)^i.
\end{equation}

In Minkowski space, $y$ has the explicit form $y=x-(t-\tau)\newhat{p}$ and, when $Z$ is chosen to be $\Omega_{ij}, B_i, S$, a straightforward computation, see Section 3, shows that the resulting vectors $\Zb^M\!$ take the form,
\begin{align*}
	\overline{\Omega}_{ij}^M
	=
	x^i \partial_{x^j} - x^j \partial_{x^i} + \ph^i \partial_{\ph^j} - \ph^j \partial_{\ph^i},
	\qquad
	\overline{B}_i^M
	=
	x^i \partial_t + t \partial_{x^i} + \left( \delta_i^j - \ph^i \ph^j \right) \partial_{\ph^j},
	\qquad
	\overline{S}^M
	&
	=
	t\partial_t + x^k \partial_{x^k}.
\end{align*}

By differentiating the equality
\[
	\Ph \left( \Phi_{\lambda}^Z (t_0,y)^0, \Phib^{Z,X}_{\lambda,t_0}(t,x,\ph) \right)^i
	=
	\ph_X \left( \Phi_{\lambda}^Z(t,x), \Phi_{\lambda}^Z(t_0,y) \right)^i,
\]
with respect to $\lambda$, one can similarly obtain an estimate for $\Zb(\Ph^i(t_0))$.  In the simple case of Minkowski space, $\ph_X$ has the explicit form
\[
	\ph_X \left( \Phi_{\lambda}^Z(t,x), \Phi_{\lambda}^Z(t_0,y) \right)^i
	=
	\frac{\Phi_{\lambda}^Z(t,x)^i - \Phi_{\lambda}^Z(t_0,y)^i}{\Phi_{\lambda}^Z(t,x)^0 - \Phi_{\lambda}^Z(t_0,y)^0},
\]
and so
\[
	\Zb\vert_{(t,x,\ph)} \left( \Ph(t_0,t,x,\ph)^i \right)
	=
	\frac{Z^i \vert_{(t,x)} - Z^i\vert_{(t_0,X(t_0))}}{t-t_0}
	-
	\frac{Z^0 \vert_{(t,x)} - Z^0\vert_{(t_0,X(t_0))}}{t-t_0} \ph^i,
\]
since ${d\Ph}\!/{ds} = 0$ in Minkowski space, which leads to a good bound for $\vert \Zb(\Ph^i(t_0)) \vert$ and, together with \eqref{eq:introZbX0}, results in bounds of the form \eqref{eq:goodfbounds} for solutions of the Vlasov equation on Minkowski space.  Such bounds lead to bounds on $Z \rho(t,x)$ since,
\[
	\Omega_{ij} \rho (t,x) = \int \overline{\Omega}^M_{ij} f (t,x,\ph) d\ph,
	\ \
	B_i \rho (t,x) = \int \overline{B}^M_i f (t,x,\ph) - 4 \ph^i f (t,x,\ph) d\ph,
	\ \
	S \rho (t,x) = \int \overline{S}^M f(t,x,\ph) d\ph.
\]

The rotation vector fields $\overline{\Omega}^M_{ij}$ and a form of the scaling vector field $\overline{S}^M$ were used in \cite{Ta} for small data solutions of the massless Einstein--Vlasov system (note though that the above procedure of using $Z$ to define $\Zb$ breaks down when the mass shell $\ms$ becomes the set of null vectors, as is the case for the massless Einstein--Vlasov system).\footnote{The proof in \cite{Ta} is based on a double null foliation, and an associated double null coordinate system $(u,v,\theta^1,\theta^2)$, of the spacetimes which are constructed, and so the language used there is slightly different.  In the coordinate system $(u,v,\theta^1,\theta^2,p^{\theta^1},p^{\theta^2},p^v)$ conjugate to the double null coordinate system for $\mathcal{M}$, the vector fields $\partial_{\theta^A}$, for $A=1,2$, are used.  Defining appropriate Cartesian coordinates, one can show that $\partial_{\theta^A}$ take the form of $\overline{\Omega}_{ij}^M$.  The proof in \cite{Ta} in fact reduces to a semi global problem since the matter is shown, as part of the bootstrap argument in the proof, to be supported in a strip of finite retarded $u$ length.  The vector $(v-u) \partial_v$ is also used which, since $u$ remains of size 1 in the support of the matter, agrees to leading order with the vector field $u\partial_u + v\partial_v$ which, when written with respect to an appropriate Cartersian coordinate system, is seen to be equal to $\overline{S}^M$.}  The vector fields $\overline{\Omega}^M_{ij}$, $\overline{B}^M_i$, $\overline{S}^M$ were also used in the work \cite{FaJoSm} on the Vlasov--Nordstr\"{o}m system, where the authors notice that the rotations $\overline{\Omega}^M_{ij}$ and the boosts $\overline{B}^M_i$ are the \emph{complete lifts} of the spacetime rotations and boosts, and hence generate symmetries of the tangent bundle.

\subsubsection{Vector fields used in the proof of Theorem \ref{thm:mainL2}}
\label{subsec:introvfs2}

In the proof of Theorem \ref{thm:main2}, in order to obtain good estimates for $Z^I T^{\mu \nu}(t,x)$, it is necessary to obtain bounds of the form \eqref{eq:goodfbounds}, now for solutions of the Vlasov equation on the spacetimes being constructed.  The sharp interior decay rate of the Christoffel symbols in the spacetimes which are constructed, as we plan to show in forthcoming work, is
\begin{equation} \label{eq:Gammasharp}
	\left\vert \Gamma^{\alpha}_{\beta \gamma}(t,x) \right\vert
	\leq
	\frac{C}{t^2}
	\quad
	\text{for }
	\vert x \vert \leq c t +K,
\end{equation}
where $0<c<1$ and $K \geq 0$.\footnote{
	It should be noted that there are two contributions to this slow interior decay.  The first arises from the failure of the Einstein equations in the harmonic gauge to satisfy the classical null condition of \cite{K1}.  Indeed, it was recently shown by Lindblad \cite{L4} that small data solutions of the vacuum Einstein equations in the harmonic gauge satisfy this decay rate (compare with \cite{ChKl} where the Ricci coefficients associated to the maximal--null foliation decay in the interior at a faster rate).  The second contribution arises from the presence of the Vlasov matter, in the form of the energy momentum tensor as a source term in the Einstein equations.  This fact can be more easily seen in a simplified setting.  Indeed, if $T(t,x)$ denotes a function which decays at rate $t^{-3}$ for $\vert x \vert \leq ct +K$ and vanishes for $\vert x \vert \geq c t +K$ -- the sharp behaviour of the components of the energy momentum tensor associated to solutions of the Vlasov equation on Minkowski space -- the sharp interior behaviour of solutions of $\Box \phi = T$ is
	$
		\vert \partial \phi (t,x) \vert \les t^{-2}$, for $\vert x \vert \leq c t + K.
	$
	}
On a spacetime whose Christoffel symbols decay as such, it can be shown that the Minkowksi vector fields $\overline{Z}^M$ of the previous section only satisfy
\[
	\big\vert \Zb^M f(t,x,\ph) \big\vert
	\leq
	C \log t,
\]
for solutions $f$ of the Vlasov equation.  This logarithmic loss compounds at higher orders, and cannot be used to recover the sharp bounds \eqref{eq:Gammasharp} in the context of Theorem \ref{thm:main2}.

The proof of Theorem \ref{thm:main2} is therefore based on a different collection of vector fields, $\overline{Z}$, which are adapted to the geometry of the background spacetime and again reduce to $Z = \Omega_{ij}, B_i, S$ when acting on spacetime functions, and satisfy a good bound of the form \eqref{eq:goodfbounds} when applied to solutions $f$ of the Vlasov equation.  The vector fields can be derived using the procedure described in Section \ref{subsec:Minkowskivectors}, which in fact did not rely on the background spacetime being Minkowski space.  Instead of the expression \eqref{eq:Minkowskiy}, the components $y^i$ are defined using approximations\footnote{The notation $X_1$ is later used to denote a cruder approximation to the geodesics.} $X_2(s,t,x,\ph)$ to the true geodesics $X(s,t,x,\ph)$ of the spacetimes to be constructed.  One could also use the geodesics themselves, but we choose not to in order to avoid a regularity issue which would arise in Theorem \ref{thm:main2}, described below.  A derivation of the vector fields obtained using this procedure is given in Section \ref{section:newsection}, but here the failure of the Minkowski vector fields $\Zb^M$ are identified explicitly and shown how to be appropriately corrected.  The two procedures agree up to lower order terms.

Instead of the sharp interior bounds \eqref{eq:Gammasharp}, the proof of Theorem \ref{thm:mainL2} requires only the weaker bounds
\begin{equation} \label{eq:Gammanotsharp}
	\left\vert Z^I \Gamma^{\alpha}_{\beta \gamma} (t,x) \right\vert
	\leq
	\frac{C}{t^{1+a}},
\end{equation}
for $\vert I \vert \leq \left\lfloor N/2 \right\rfloor + 2$, where $\frac{1}{2} < a < 1$.  Consider first the rotation vector fields, and recall the expression \eqref{eq:ZbVlasov}.  The rotation vector fields $\overline{\Omega}_{ij}$ are defined using approximations to the geodesics.
 The geodesics take the form
\begin{equation*}
	X(s,t,x,\ph)^k
	=
	x^k - (t-s)\ph^k
	-
	\int_s^t (s' - s) \Gh^k(s',X(s'),\Ph(s')) ds'.
\end{equation*}
Using the fact that
\[
	\Ph(s,t,x,\ph)^k
	\sim
	\ph^k
	\sim
	\frac{x^k}{t},
	\qquad
	X(s,t,x,\ph)^k
	\sim
	x^k - (t-s) \ph^k
	\sim
	x^k - (t-s) \frac{x^k}{t}
	=
	s \frac{x^k}{t},
\]
where each of the first approximations arise by replacing $\Ph(s,t,x,\ph)^k$ and $X(s,t,x,\ph)^k$ by their respective values in Minkowski space, and the second arise from the fact that $\frac{x^i}{t} \sim \ph^i$, which holds asymptotically along each given geodesic (see Proposition \ref{prop:sec22} below), the approximations to the geodesics are defined as
\begin{equation} \label{eq:X2intro}
	X_2(s,t,x,\ph)^k
	=
	x^k - (t-s) \ph^k
	-
	\int_s^t (s'-s) \Gh^k \left( s', s'\frac{x}{t}, \frac{x}{t} \right) ds',
\end{equation}
for $t_0 \leq s \leq t$.  It will be shown in Section \ref{section:suppf} that $X_2(s,t,x,\ph)^k$ are good approximations to the geodesics $X(s,t,x,\ph)^k$ in the sense that,
\begin{equation} \label{eq:introX2X}
	\left\vert X_2(s,t,x,\ph) - X(s,t,x,\ph) \right\vert
	\leq
	C,
\end{equation}
for all $t_0\leq s \leq t$ and $k=1,2,3$.  The idea is now to construct vector fields
so that the vector fields applied to $X_2$ are bounded. Then one can show that \eqref{eq:introX2X} is true with $X_2 - X$ replaced by $\Zb(X_2-X)$.  See Section \ref{section:newsection} for more details. The approximations $X_2$ have the desirable property, which will be exploited below, that
\beq\label{eq:desirablepathproperty}
	\partial_{\ph^l} \left( X_2(s,t,x,\ph)^k - \big( x^k - (t-s)\ph^k \big) \right)=0,
\eq
vanishes (and in particular does not involve derivatives of $\Gamma$).

Applying the Minkowski rotation vector fields to the approximations $X_2$ gives,
\begin{align} \label{eq:minkowskirotationofdesirablepath}
	\overline{\Omega}_{ij}^M \left( X_2(t_0,t,x,\ph)^k \right)
	=
	&
	\
	\big( x^i - (t-t_0)\ph^i \big) \delta_j^k
	-
	\big( x^j - (t-t_0)\ph^j \big) \delta_i^k
	-
	\overline{\Omega}_{ij}^M \left( \int_{t_0}^t (s' - t_0) \Gh^k(s',s'\frac{x}{t},\frac{x}{t}) ds' \right)
	\nonumber
	\\
	=
	&
	\
	X_2(t_0,t,x,\ph)^i \delta_j^k
	-
	X_2(t_0,t,x,\ph)^j \delta_i^k\qquad\qquad
	\\
	&
	+
	\int_{t_0}^t (s' - t_0)
	\left[
	\Gh^i(s',s'\frac{x}{t},\frac{x}{t}) \delta_j^k
	-
	\Gh^j(s',s'\frac{x}{t},\frac{x}{t}) \delta_i^k
	-
	\Omega_{ij} \left( \Gh^k \left(s',s'\frac{x}{t},\frac{x}{t}\right) \right) \right]
	ds'.
	\nonumber
\end{align}
In the final equality, for $(t,x,\ph) \in \supp(f)$, the first two terms are bounded (as can be seen from \eqref{eq:introX2X} and the fact that $f\vert_{t=t_0}$ has compact support). However on a spacetime only satisfying the bounds \eqref{eq:Gammanotsharp}, the terms on the last line in general grow in $t$.  The vector fields $\overline{\Omega}_{ij}$ are defined so that these terms are removed:
\[
	\overline{\Omega}_{ij}
	=
	\overline{\Omega}_{ij}^M
	+
	\mathring{\Omega}_{ij}^l \partial_{\ph^l},
\]
where if the functions $\mathring{\Omega}_{ij}^l$ are defined as
\begin{align*}
	\mathring{\Omega}_{ij}^l(t,x)
	=
	\int_{t_0}^t \frac{s' - t_0}{t-t_0}
	\Big[
	\Gh^i\left( s',s'\frac{x}{t} ,\frac{x}{t} \right) \delta_j^k
	-
	\Gh^j\left( s',s'\frac{x}{t} ,\frac{x}{t} \right) \delta_i^k
	-
	\overline{\Omega}_{ij}^M \left( \Gh^k\left( s',s'\frac{x}{t} ,\frac{x}{t} \right) \right)
	\Big]
	ds',
\end{align*}
it follows from \eqref{eq:desirablepathproperty} and \eqref{eq:minkowskirotationofdesirablepath} that
\beq\label{eq:modifiedvectorfieldinitialcond}
	\overline{\Omega}_{ij} \big( X_2(t_0,t,x,\ph)^k \big)\!
	=
	\overline{\Omega}_{ij}^M \big( X_2(t_0,t,x,\ph)^k \big)
	-
	(t-t_0) \mathring{\Omega}_{ij}^k(t,x)
	=
	X_2(t_0,t,x,\ph)^i \delta_j^k
	-
	X_2(t_0,t,x,\ph)^j \delta_i^k,
\eq
and so
\begin{equation} \label{eq:OmegaX2intro}
	\left\vert
	\overline{\Omega}_{ij} \left( X_2(t_0,t,x,\ph)^k \right)
	\right\vert
	\leq
	C,
\end{equation}
for $(t,x,\ph) \in \supp(f)$ and $k=1,2,3$.  It can similarly be shown, see Section \ref{subsec:rotationexample}, that
\[
	\big\vert
	\overline{\Omega}_{ij} \big( X(t_0,t,x,\ph)^k \big)
	\big\vert
	+
	\big\vert
	\overline{\Omega}_{ij} \big( \Ph(t_0,t,x,\ph)^k \big)
	\big\vert
	\leq
	C,
\]
which, by \eqref{eq:ZbVlasov}, leads to a good bound for $\overline{\Omega}_{ij} f(t,x,\ph)$.

We remark that the identity \eqref{eq:modifiedvectorfieldinitialcond} can be expressed
\beq\label{eq:modifiedvectorfieldinitialcondshorter}
	\overline{\Omega}_{ij} \big( X_2(t_0,t,x,\ph)^k \big)\!
	=
	\Omega^k\big|_{(t_0,X_2(t_0,t,x,\ph))}
\eq
which is exactly \eqref{eq:introZbX0} when $Z = \Omega_{ij}$ and $X$ is replaced by $X_2$. One could therefore alternatively obtain the rotation vector fields $\overline{\Omega}_{ij}$ by following
the procedure in Section \ref{subsec:Minkowskivectors} with $X$ replaced by $X_2$, see Section 3.

Since the functions $\mathring{\Omega}_{ij}^l$ do not depend on $\ph$,
\[
	\Omega_{ij} \left( \int f(t,x,\ph) d \ph \right)
	=
	\int \left( \overline{\Omega}_{ij} - \mathring{\Omega}_{ij}^l \partial_{\ph^l} \right) f(t,x,\ph) d\ph
	=
	\int \overline{\Omega}_{ij} f(t,x,\ph) d\ph,
\]
and so the good bounds for $\overline{\Omega}_{ij} f(t,x,\ph)$ lead to good bounds for $\Omega_{ij} \rho (t,x)$.

Had the true geodesics been used to define the vector fields $\overline{\Omega}_{ij}$, the functions $\mathring{\Omega}_{ij}^l$
would have involved a term of the form $\int_{t_0}^t ({s' - t_0})({t-t_0})^{-1} \overline{\Omega}_{ij}^M \big( \Gh^k\big( s',X(s'),\Ph(s') \big) \big) ds'$ and so $\partial_{\ph^k} \mathring{\Omega}_{ij}^l$ would involve second order derivatives of $\Gamma$.  The average $\int \mathring{\Omega}_{ij}^l \partial_{\ph^l} f (t,x,\ph) d \ph$ has better decay than the term $\mathring{\Omega}_{ij}^l \partial_{\ph^l} f (t,x,\ph)$ does pointwise and so, to exploit this fact when obtaining a bound for $\Omega_{ij} \rho (t,x)$, it is necessary to integrate this term by parts $\int \mathring{\Omega}_{ij}^l \partial_{\ph^l} f (t,x,\ph) d \ph = - \int \big( \partial_{\ph^l} \mathring{\Omega}_{ij}^l \big) f (t,x,\ph) d \ph$.  If the true geodesics are used to define $\overline{\Omega}_{ij}$ then, in the setting of Theorem \ref{thm:main2}, the fact that $\partial_{\ph^k} \mathring{\Omega}_{ij}^l$ involves second order derivatives of $\Gamma$ would mean that this integration by parts cannot be used when estimating derivatives of $T^{\mu \nu}$ at the top order.

The approximations to the geodesics $X_2(s,t,x,\ph)$ are used in a similar manner in Section \ref{subsec:vectorfields} to define vector fields, $\overline{B}_i$ and $\overline{S}$.\footnote{For the boosts, $\overline{B}_i$, it is slightly more convenient to allow the functions $\mathring{B}_i^l$ to depend on $\ph$.  This dependence is in such a way, however, so that $\partial_{\ph^l} \mathring{B}_i^l$ involves only one derivative of $\Gamma^{\alpha}_{\beta \gamma}$.}

\subsection{Outline of the Paper}
In Section \ref{section:suppf} the system \eqref{eq:geodesictimenormalized} is used to prove the property \eqref{eq:supportofmatterintro} regarding the support of the matter, and the maps \eqref{eq:X2intro} are shown to be good approximations to the geodesics.  Section \ref{section:newsection} is not required for the proof of Theorem \ref{thm:main2}, but the discussion of the vector fields in Section \ref{subsec:Vlasovvectors} is expanded on and an outline is given of how one controls one rotation vector field applied to the solution of the Vlasov equation.  In Section \ref{section:geodesics} higher order combinations of vector fields applied to the geodesics are estimated.  In Section \ref{section:Testimates} the estimates of Section \ref{section:geodesics} are used to obtain estimates for vector fields applied to the components of the energy momentum tensor and hence prove Theorem \ref{thm:mainL2}.  In Section \ref{section:Einstein} the solution of the reduced Einstein equations is estimated in terms of the components of the energy momentum tensor.  The results of the previous sections are combined in Section \ref{section:cty} to give the proof of Theorem \ref{thm:main2}.

\subsection{Acknowledgements}
We thank Alan Rendall and Igor Rodnianski for helpful discussions.  H.\@L.\@ is supported in part by NSF grant DMS--1500925.  M.\@T.\@ acknowledges the support of a grant from the European Research Council.

\section{The support of the matter and approximations to the geodesics}
\label{section:suppf}
In this section the Vlasov equation on a fixed spacetime is considered.  It is shown that under some assumptions on the metric the solution is supported, for large times, away from the \emph{wave zone} $x \sim t$ provided $f_0$ is compactly supported.  Curves $X_2$, which approximate the timelike geodesics $X$, are introduced, which are later used to define the $\Zb$ vector fields.

Recall that
\begin{equation} \label{eq:geodesictimenormalizedsec2}
	\frac{d {X}^{\mu}}{d s} = {\newhat{P}^{\mu}},
	\qquad
	\frac{d \newhat{P}^\mu}{ds} = \widehat{\Gamma}\big( s,X(s),\newhat{P}(s) \big)^\mu.
\end{equation}
where $X^0=s$, $\newhat{P}^0=1$ and
\begin{equation}
	\widehat{\Gamma}(t,x,\newhat{p})^\mu
	=
	\newhat{p}^\mu\Gamma_{\alpha \beta}^0(t,x)\newhat{p}^\alpha\newhat{p}^\beta
	-
	\Gamma_{\alpha\beta}^\mu(t,x)\newhat{p}^{\alpha}\newhat{p}^\beta.
\end{equation}
Let functions $\Lambda^{\alpha \beta, \mu}_{\gamma}$ be defined so that,
\begin{equation} \label{eq:GammaH}
	\Gamma(t,x)\cdot \Lambda(\newhat{p})^\mu
	=
	\widehat{\Gamma}(t,x,\newhat{p})^\mu,
\end{equation}
where
\[
	\Gamma(t,x)\cdot \Lambda(\newhat{p})^\mu
	:=
	\Gamma^{\gamma}_{\alpha \beta}(t,x) \Lambda^{\alpha \beta, \mu}_{\gamma} (\ph).
\]

\subsection{Properties of the support of the matter}

The results of this section will rely on the Christoffel symbols satsifying the assumptions
\begin{equation}\label{eq:Gammasuppfsec2Nprimeprime}
 |Z^I \Gamma|\leq \frac{c_{N^{\prime\prime}}^{\prime\prime}}{(1+t)^{1-\delta}(1+|t-r|)^{a+\delta}},
 \qquad\text{for}\quad |I|\leq N^{\prime\prime},
\end{equation}
 for various small $N^{\prime\prime}$. Then if \eqref{eq:xtphsuppf} indeed holds in $\supp(f)$, the assumption \eqref{eq:Gammasuppfsec2Nprimeprime} implies that
 \begin{equation}\label{eq:Gammasuppfsec2Nprimeprimeprime}
 |Z^I \Gamma|\leq \frac{c_{N^{\prime\prime}}^{\prime\prime\prime}}{(1+t)^{1+a}},
 \qquad\text{for}\quad |I|\leq N^{\prime\prime},\qquad\text{where}\quad
  c_{N^{\prime\prime}}^{\prime\prime\prime}
  =\frac{c_{N^{\prime\prime}}^{\prime\prime}K^{a+\delta}}{(1-c)^{a+\delta}},
 \end{equation}
in $\supp(f)$.  The notation $(y,q)$ will be used for points in the mass shell over the initial hypersurface $\{t=0\}$, so $y \in \Sigma_{0}$, $q \in \ms_{(0,y)}$.

The following result guarantees that
\begin{equation} \label{eq:xtphsuppf}
	\left\vert x \right\vert \leq c t + K,\qquad K\geq 1
	\qquad
	\left\vert \ph \right\vert \leq c <1,
\end{equation}
for $(t,x,\ph) \in \supp(f)$, for constants $K$ and $c<1$.

\begin{proposition} \label{prop:suppf}
Suppose that $|y|\leq K$ and $|q|\leq K^\prime$ in the support of $f_0(y,q)$, for some $K,K'\geq 1$, and that for some fixed $a,\delta>0$
\begin{equation}\label{eq:Gammasuppfsec2}
|g-m|\leq c',\quad\text{and}\quad |\Gamma|\leq \frac{c^{\prime\prime}}{(1+t)^{1-\delta}(1+|t-r|)^{a+\delta}},\end{equation}
where $c'={1}/({16(1+4{K^{\prime}}^2)})$ and $c^{\prime\prime}=\min{\big({c'}^{a+\delta}\!/a,(1+2K/c')^{-\delta}/\delta\big)}/4$. Then with $c=1-c'$ we have
	\begin{equation}\label{eq:PhatPXest}
		\left\vert P(s,0,y,q) \right\vert \leq K'+1,
        \qquad
		\vert {P}(s,0,y,\newhat{q}) \vert \leq c\,\vert {P}^0(s,0,y,\newhat{q}) \vert,
\qquad 	\left\vert X(s,0,y,q) \right\vert \leq c s + K,\qquad s\geq 0.
\end{equation}
\end{proposition}
\begin{proof}
Let $s_1$ be the largest number such that $| P(s,0,y,q)|\leq 2K'$ for $0\leq s\leq s_1$. We will show that then it follows that  $| P(s,0,y,q)|\leq 3K'/8$, for $0\leq s\leq s_1$, contradicting the maximality of $s_1$.
 Let $p=P(s,0,y,q)$. Since $g_{\alpha\beta}p^\alpha p^\beta=-1$ it follows that
 $\big| 1\!+|p_{\,}|^2\!-|p^0|^2\big|\leq |g-m|(|p^0|^2\!+|p_{\,}|^2)$ and hence
 \begin{equation*}
 \frac{|p_{\,}|^2}{|p^0|^2\!}(1-c')\leq {1+c'}-\frac{1}{|p^0|^2\!\!}\quad\Longrightarrow \quad\frac{|p_{\,}|^2}{|p^0|^2\!}\leq\frac{1+c'}{1-c'}-\frac{1}{1\!+4{K'}^2\!}\leq
 1+4c'-\frac{1}{1\!+4{K'}^2\!}\leq (1-c')^2=c^2.
 \end{equation*}
 Hence
 \begin{equation*}\label{eq:lightconedistance}
 \frac{d}{ds} \big(s-|X|\big)=1-\frac{X \cdot P}{\vert X \vert P^0} \geq 1-c,
 \end{equation*}
and it follows that
$s-|X| \geq (1-c)s-K$ along a characteristic. Therefore
$|\Gamma|\leq 2c^{\prime\prime}{c'}^{-a-\delta} s^{-1-a}$, when $s\geq 2K/c'$, and $|\Gamma|\leq 2c^{\prime\prime}(1+s)^{-1+\delta}$, when $s\leq 2K/c'$ along a characteristic. We have
 \begin{equation*}\label{eq:pbound}
 \frac{d}{ds} (1+|P|)\leq  |\Gamma|(1+|P|),\quad\text{and}\quad
  \frac{d}{ds} |P-q_{\,}|\leq  |\Gamma|(1+|P|)
 \end{equation*}
 where
 \[
 \int_{0}^t |\Gamma| ds\leq \int_0^{2K/c'} \frac{2c^{\prime\prime}}{(1+s)^{1-\delta}} ds+\int_{2K/c'}^\infty \frac{2c^{\prime\prime}{c'}^{-a-\delta} }{s^{1+a}} ds
 \leq  \frac{2c^{\prime\prime}(1+2K/c')^{\delta}}{\delta}+ \frac{2c^{\prime\prime}{c'}^{-a-\delta} }{a(1+K/c')^\delta} \leq\frac{1}{4}
 \]
 Hence
 \begin{equation*}
\frac{ 1+|P|}{1+|q_{\,}|}\leq e^{1\!/4}
\leq \frac{3}{2}
\quad\text{and}\quad
   |P-q_{\,}|\leq  \frac{1 }{4}\big(1+\frac{3}{2}\big)\leq 1.
 \end{equation*}
 It follows that $|P(s,0,y,q)|\leq K'+1$
in the support of $f$,
and \eqref{eq:PhatPXest} follows.
\end{proof}

\begin{proposition} \label{prop:sec22}
	If $\vert \Gamma(t,x) \vert \leq c^{\prime\prime\prime}(1+t)^{-1-a}$ for $\vert x \vert \leq ct +K$ then, for $(t,x,\ph) \in \supp(f)$ with $t\geq 1$,
	\begin{equation}
		\left\vert \frac{x}{t} - \ph \right\vert
		\leq
		(K+c^{\prime\prime\prime})(1+t)^{-a}.
		\label{eq:sec22}
	\end{equation}
\end{proposition}

\begin{proof}
	Note that any $(t,x,\ph) \in \supp(f)$ can be written as $(t,x,\ph) = (t,X(t,0,y,q),\Ph(t,0,y,q))$ for some $(y,q) \in \supp(f_0)$.  For $(y,q)$ such that $(0,y,q) \in \supp(f)$,
	\[
		\left\vert \frac{d}{ds} \left( X(s,0,y,q)- s \Ph(s,0,y,q) \right) \right\vert
		=
		\left\vert s \Gh\left( s, X(s,0,y,q), \Ph(s,0,y,q) \right) \right\vert
		\lesssim
		\frac{c^{\prime\prime\prime}}{(1+s)^a},
	\]
	for all $s\geq 0$, using the bound \eqref{eq:Gammasuppfsec2}.  The proof follows by integrating forwards from $s=0$ and using the fact that $\vert y \vert + \vert q \vert \leq C$,
	\beq\label{eq:XminusshatP}
		\left\vert X(s,0,y,q) - s \Ph(s,0,y,q) \right\vert
		\leq
		K
		+
		c^{\prime\prime\prime} (1+s)^{1-a},
	\eq
	and dividing by $s$.
\end{proof}

\subsection{Translated time coordinate}
\label{subsec:transtime}
Proposition \ref{prop:suppf} implies that $\vert x \vert \leq c t +K$ for $t\geq 0$ in $\supp(T^{\mu\nu})$ or, equivalently, $\vert x \vert \leq c \tilde{t}$ for $\tilde{t} \geq t_0$, where $t_0 := {K}/{c}$, and $\tilde{t} = t_0 + t$.  It is convenient to use this translated time coordinate $\tilde{t}$ in what follows.  In particular, the vector fields of Section \ref{section:geodesics} will be defined using this variable.  The main advantage is that the spacetime vector fields $\tilde{Z} = \tilde{\Omega}_{ij}, \tilde{B}_i, \tilde{S}$ defined by
\[
	\tilde{\Omega}_{ij} = \Omega_{ij},
	\quad
	\tilde{B}_i = \tilde{t} \partial_{x^i} + x^i \partial_t,
	\quad
	\tilde{S} = \tilde{t} \partial_t + x^k \partial_{x^k},
\]
satisfy, for any multi index $I$,
\begin{equation} \label{eq:partialZ}
	\tilde{t}^{\vert I \vert} \partial^I
	=
	\sum_{\vert J \vert \leq \vert I \vert}
	\Lambda_{IJ} \left( \frac{x}{\tilde{t}} \right) \tilde{Z}^J
	\qquad
	\text{if }
	{\vert x \vert}/{\tilde{t}} \leq c <1,
\end{equation}
for some smooth functions $\Lambda_{IJ}$.  Estimates for $\partial^I \tilde{Z}^J T^{\mu \nu}$ will then follow directly from estimates for $\tilde{Z}^I T^{\mu \nu}$, which are less cumbersome to obtain.  See Section \ref{subsec:Tmainestimates}.  For simplicity the $\sim$ will always be omitted, and statements just made for $t\geq t_0$.

\subsection{Approximations to geodesics}

Define, for $(t,x,\ph) \in \supp(f)$ and $t_0 \leq s \leq t$,
\begin{equation} \label{eq:X2P2def}
	X_2^i(s,t,x,\ph)
	=
	x^i - (t-s) \ph^i
	-
	\int_s^t (s' - s) \Gh^i \left( s', s' \frac{x}{t}, \frac{x}{t} \right) ds',
	\quad
	\widehat{P}_2(s,t,x,\ph)
	=\ph^i + \int_s^t \Gh^i \left( s', s' \frac{x}{t}, \frac{x}{t} \right) ds',
\end{equation}
and set
\begin{align*}
	\Xb (s,t,x,\ph)^{i}
	&
	:= X(s,t,x,\ph)^{i} - X_2(s,t,x,\ph)^i,
	\\
	\Pb(s,t,x,\ph)^{i}
	&
	:=
	\frac{d\Xb}{ds}(s,t,x,\ph)
	=
	\Ph(s,t,x,p)^{i} -	\Ph_{\!\!2}(s,t,x,p)^{i} ,
\end{align*}
for $i=1,2,3$.  Note that
\[
	\Xb(t,t,x,\ph)^i = \Pb(t,t,x,\ph)^i = 0,
\]
for $i=1,2,3$.  The geodesic equations \eqref{eq:geodesictimenormalizedsec2} can be used to derive the following equations for $\Xb$ and $\Pb$,
\begin{align}
	\frac{d \Xb^i}{ds}
	=
	&
	\
	\Pb^i;
	\label{eq:Xbar}
	\\
	\frac{d \Pb^i}{ds}
	=
	&
	\
	\Gh^i \left( s, X(s,t,x,\ph), \Ph(s,t,x,\ph) \right)
	-
	\Gh^i \left( s, s\frac{x}{t}, \frac{x}{t} \right).
	\label{eq:Pbar}
\end{align}

It follows from the next Proposition that the curves $s\mapsto X_2(s,t,x,\ph)$ are good approximations to the geodesics $s \mapsto X(s,t,x,\ph)$.

\begin{proposition} \label{prop:sec23}
	Suppose $t\geq t_0$ is such that $(t,x,\ph) \in \supp(f)$ and $\vert \Gamma(t,x) \vert + t \vert \partial \Gamma(t,x)
	\vert \lesssim \varepsilon t^{-1-a}$.  Then
	\begin{equation}
		s^{2a-1} \left\vert \overline{X}(s,t,x,\ph)^i \right\vert
		+
		s^{2a} \left\vert \overline{P}(s,t,x,\ph)^i \right\vert
		\lesssim
		\varepsilon,
		\label{eq:sec23}
	\end{equation}
	for all $t_0 \leq s \leq t$ and $i=1,2,3$.
\end{proposition}

\begin{proof}
	First note that equation \eqref{eq:geodesictimenormalizedsec2} and the bounds \eqref{eq:Gammasuppfsec2} imply
	\beq\label{eq:dhatPds}
		\left\vert \frac{d \Ph^i(s)}{ds} \right\vert
		\lesssim
		\frac{\varepsilon}{s^{1+a}},
	\eq
	for $t_0\leq s \leq t$.  Proposition \ref{prop:sec22}, integration of \eqref{eq:dhatPds} from $s$ to $t$ and \eqref{eq:XminusshatP} then gives
	\[
		\left\vert \frac{x^i}{t} - \frac{X(s)^i}{s} \right\vert
		\leq
		\left\vert \frac{x^i}{t} - \ph^i \right\vert
		+
		\left\vert \ph^i - \Ph(s)^i \right\vert
		+
		\left\vert \Ph(s)^i - \frac{X(s)^i}{s} \right\vert
		\lesssim
		\frac{1}{s^{a}},
	\]
	for $t_0 \leq s \leq t$ and $i=1,2,3$.  Hence
	\[
		\left\vert
		\Gamma^{\mu}_{\alpha \beta} \left( s, s\frac{x}{t} \right)
		-
		\Gamma^{\mu}_{\alpha \beta} \left( s, X(s) \right)
		\right\vert
		\lesssim
		\Vert \partial \Gamma^{\mu}_{\alpha \beta} (s, \cdot) \Vert_{L^{\infty}}
		\left\vert s \frac{x}{t} - X(s) \right\vert
		\lesssim
		\frac{\varepsilon}{s^{1+2a}}.
	\]
	Moreover,
	\[
		\left\vert \Ph(s)^i - \frac{x^i}{t} \right\vert
		\leq
		\left\vert \Ph(s)^i - \ph^i \right\vert
		+
		\left\vert \ph^i - \frac{x^i}{t} \right\vert
		\lesssim
		\frac{1}{s^a},
	\]
	for $t_0 \leq s \leq t$, and so
	\[
		\left\vert
		\Gh^i \left( s, X(s), \Ph(s) \right)
		-
		\Gh^i \left( s, s \frac{x}{t}, \frac{x}{t} \right)
		\right\vert
		\lesssim
		\frac{\varepsilon}{s^{1+2a}}.
	\]
	Integrating equation \eqref{eq:Pbar} backwards from $s=t$, and using the fact that $\Pb(t,t,x,\ph)^i = 0$, gives
	\[
		\left\vert \Pb(s,t,x,\ph)^i \right\vert
		\lesssim
		\frac{\varepsilon}{s^{2a}},
	\]
	and integrating the equation \eqref{eq:Xbar} backwards from $s=t$ and using $\Xb(t,t,x,\ph)^i = 0$ gives,
	\[
		\left\vert \Xb(s,t,x,\ph)^i \right\vert
		\lesssim
		\varepsilon s^{1-2a},
	\]
	since $a> \frac{1}{2}$.
\end{proof}

\begin{corollary} \label{cor:X2bound}
	Suppose $t\geq t_0$ is such that $(t,x,\ph) \in \supp(f)$ and $\vert \Gamma(t,x) \vert + t \vert \partial \Gamma(t,x)
	\vert \lesssim \varepsilon t^{-1-a}$.  Then
	\begin{equation*}
		\left\vert X_2(s,t,x,\ph)^i \right\vert
		\lesssim
		s,
	\end{equation*}
	for $t_0 \leq s \leq t$ and $i=1,2,3$.
\end{corollary}

\begin{proof}
	The proof is an immediate consequence of the first bound of \eqref{eq:xtphsuppf}, and Proposition \ref{prop:sec23}.
\end{proof}

\section{The vector fields}
\label{section:newsection}

In this section the Vlasov equation on a general spacetime is considered.  The procedure for using a vector field $Z$ on $\mathcal{M}$ to define a vector field $\Zb$ on $\ms$, outlined in Section \ref{subsec:Vlasovvectors}, is described in more detail.  In the proof of Theorem \ref{thm:mainL2} it is computationally more convenient to use slightly different vector fields to those given by this procedure.  These vector fields are introduced in Section \ref{subsec:vectorsused}.  In Section \ref{subsec:rotationexample} estimates are obtained for one rotation vector field applied to the geodesics.  This is a special simple case of a computation which is done more generally in Section \ref{section:geodesics}.  This section is independent of the rest of the paper and is included for the purpose of exposition.

\subsection{The general procedure}

Given the first approximations to the geodesics $X_1(s,t,x,\newhat{p})=sx/t$ and $\newhat{P}_{\!\!1}(s,t,x,\newhat{p})=x/t$ we define the second approximations to the geodesic $X_2(s,t,x,\newhat{p})$ and $\newhat{P}_{\!\!2}(s,t,x,\newhat{p})$ to be the solution of the system
\begin{equation}
\frac{d}{ds} {X}_2=\newhat{P}_{\!2},\qquad
\frac{d}{ds} \newhat{P}_{\!2}=-{\Gamma}(X_1)\cdot \Lambda(\newhat{P}_{\!1\!}),\qquad X_2(t)=x,\quad \newhat{P}_{\!2}(t)=\newhat{p}.
\end{equation}
The solution takes the explicit form \eqref{eq:X2P2def}.  Under some mild assumptions on the metric $g$, for fixed $\tau$ any point $(t,x,\ph) \in \ms$ with $t >\tau$ can be described uniquely by the pair of points $\{ (t,x), (\tau,y) \}$ in $\mathcal{M}$, where
\begin{equation} \label{eq:Minkowskiysec2}
	y=X_2(\tau,t,x,\ph),
\end{equation}
is the point where the approximate geodesic $X_2$ emanating from $(t,x)$ with velocity $\ph$ intersects the hypersurface $\Sigma_{\tau}$, i.\@e.\@ $(t,x,\ph) \in \ms$ can be parameterised by $\{ (t,x), (\tau,y) \}$,
\[
	(t,x,\ph) = (t,x,\ph_{X_2}(t,x,\tau,y)).
\]

For a given vector field $Z$ on $\mathcal{M}$, let $\Phi^Z_{\lambda} : \mathcal{M} \to \mathcal{M}$ denote the associated one parameter family of diffeomorphisms, so
$\Phi^Z_{0}=Id$ and so that
\[
	\frac{d \Phi^Z_{\lambda}(t,x)}{d\lambda}  = Z\vert_{\Phi^Z_{\lambda}(t,x)}.
\]
Now the action of $\Phi^Z_{\lambda}$ on $(t,x)$ and $(\tau,y)$ induces an action on $\ms$ at time $t$, given by
\[
	\Phib^{Z,X_2}_{\lambda,\tau}(t,x,\ph)
	:=
	\left( \Phi^Z_{\lambda}(t,x), \ph_{X_2}\left( \Phi^Z_{\lambda}(t,x), \Phi^Z_{\lambda}(\tau,y) \right) \right).
\]
 For fixed $t_0$ we define the vector field $\Zb$ by
\[
	\Zb \vert_{(t,x,\ph)} = \frac{d \Phib^{Z,X_2}_{\lambda,\tau}(t,x,\ph)}{d\lambda} \bigg\vert_{\lambda = 0,\,\tau=t_0}.
\]
A computation shows that
\begin{equation} \label{eq:introZbX0sec3}
	\Zb\vert_{(t,x,\ph)} \left( X_2(t_0,t,x,\ph)^i \right)
	=
	Z^i \vert_{(t_0,X_2(t_0))} - Z^0 \vert_{(t_0,X_2(t_0))} \Ph_{\!2}(t_0,t,x,\ph)^i,
\end{equation}
which as desired is bounded independently of $t>t_0+1$.
To see that \eqref{eq:introZbX0sec3} indeed holds, first note that the left hand side is the derivative of $ X_2\big(t_0,\Phib^{Z,X_2}_{\lambda,t_0}(t,x,\ph)\big)^i$ with respect to $\lambda$ at
$\lambda=0$.  Also the first
term on the right hand side is the derivative of $\Phi^Z_{\lambda}(t_0,y)^i$ at $\lambda=0$.
The equality \eqref{eq:introZbX0sec3} follows from taking the derivative, with respect to $\lambda$ at
$\lambda=0$ of both sides of the identity
\begin{equation*}
X_2\big(\Phi^Z_{\lambda}(t_0,y)^0,\Phib^{Z,X_2}_{\lambda,t_0}(t,x,\ph)\big)^i
=\Phi^Z_{\lambda}(t_0,y)^i.
\end{equation*}

We have
\beq
\Zb f(t,x,\ph)=\frac{d}{d\lambda} f\left( \Phi^Z_{\lambda}(t,x), \ph_{X_2}\left( \Phi^Z_{\lambda}(t,x), \Phi^Z_{\lambda}(t_0,y) \right) \right)\bigg\vert_{\lambda = 0}
=Z^\alpha\partial_{x^\alpha} f(t,x,\ph)+\frac{d}{d\lambda} f\left( t,x, \ph(\lambda)\right)\bigg\vert_{\lambda = 0},
\eq
where
\beq
\ph(\lambda)=\ph_{X_2}\left( \Phi^Z_{\lambda}(t,x), \Phi^Z_{\lambda}(t_0,y) \right).
\eq
Here $\ph_{X_2}(t,x,t_0,y)$ is implicitly given by
 \[
 y=x - (t-t_0)\newhat{p}_{X_2}- (t-t_0)\Theta(t_0,t,x,\newhat{p}_{X_2}),
 \]
 where
 \[
  \Theta(t_0,t,x,\newhat{p}_{X_2})^i=
\int_{t_0}^t \frac{s-t_0}{t-t_0}\,{\Gamma}\big( s, X_1(s,t,x,\newhat{p}_{X_2})\big)\cdot \Lambda\big(\newhat{P}_{\!\!1}(s,t,x,\newhat{p}_{X_2})\big)^i\, ds.
\]

We have made the simplifying assumption that $X_1=sx/t$ and $\Ph_{\!1}=x/t$ which are independent
of $\ph$ so $\Theta$ is independent of $\ph$
\beq\label{eq:Thetadef}
  \Theta(t_0,t,x)^i=
\int_{t_0}^t \frac{s-t_0}{t-t_0}\,{\Gamma}\big( s, sx/t\big)\cdot \Lambda\big(x/t\big)^i\, ds.
\eq
 and we have
\beq
\ph_{X_2}(t,x,t_0,y)=\frac{x-y}{t-t_0}-\Theta(t_0,t,x).
\eq
Then
\beq
\ph(\lambda)^i=\ph_{X_2}\left( \Phi^Z_{\lambda}(t,x), \Phi^Z_{\lambda}(t_0,y) \right)^i=\frac{\Phi^Z_{\lambda}(t,x)^i-\Phi^Z_{\lambda}(t_0,y)^i}
{\Phi^Z_{\lambda}(t,x)^0-\Phi^Z_{\lambda}(t_0,y)^0}
-\Theta\left(\Phi^Z_{\lambda}(t_0,y)^0,\Phi^Z_{\lambda}(t,x)\right)^i.
\eq
We have
\begin{multline}
\frac{ d\ph(\lambda)^i}{d\lambda}\Big|_{\lambda=0}=
	\frac{Z^i \vert_{(t,x)} - Z^i\vert_{(t_0,X_2(t_0))}}{t-t_0}
	-
	\big(\ph^i+\Theta(t_0,t,x)^i\big)\frac{ Z^0 \vert_{(t,x)} - Z^0\vert_{(t_0,X_2(t_0))}}{t-t_0} \\
- Z^0\vert_{(t_0,X_2(t_0))} \pa_{t_0}\Theta(t_0,t,x)
- Z^\mu\vert_{(t,x)} \pa_{x^\mu}\Theta(t_0,t,x)
\end{multline}
The important fact is that, provided the components of $Z$ grow at most like $t$, the right hand side is bounded, independently of $t>t_0+1$.
For the last term in the right above this follows from having good bounds for the vector fields $Z$ applied to $\Gamma$ in \eqref{eq:Thetadef}.

By differentiating the equality
\[
	\Ph_{\!2} \left( \Phi_{\lambda}^Z (t_0,y)^0, \Phib^{Z,X_2}_{\lambda,t_0}(t,x,\ph) \right)^i
	=
	\ph_{X_2} \left( \Phi_{\lambda}^Z(t,x), \Phi_{\lambda}^Z(t_0,y) \right)^i,
\]
with respect to $\lambda$, one can similarly obtain an estimate for $\Zb(\Ph_{\!2}^i(t_0))$:
\beq \label{eq:ZbPh2}
 \Zb\vert_{(t,x,\ph)} \left( \Ph_{\!2}(t_0,t,x,\ph)^i \right)=\frac{ d\ph(\lambda)^i}{d\lambda}\Big|_{\lambda=0}-Z^0 \vert_{(t_0,X_2(t_0))} \frac{d \Ph_{\!2}}{ds}(t_0,t,x,\ph)^i.
 \eq

For the vector fields we are considering $ \Phi^Z_{\lambda}$ is a linear transformation,
a rotation, a scaling or a Lorentz transformation, given by multiplication by a matrix $Q(\lambda)$.

\subsubsection{The rotation vector fields}
Let us first consider the vector field $\Omega_{12}=x^1\partial_{x^2}-x^2\partial_{x_1}$ Let $Q=Q(\lambda)$ be the one parameter subgroup of rotations by a positive angle $\lambda$ in the $x^1-x^2$ plane, i.e.
 $x(\lambda) = \Phi^{\Omega_{12}}_{\lambda}(t,x) =Q(\lambda) x=\big(\cos{\lambda} \,x^1-\sin{\lambda}\, x^2,\sin{\lambda} \, x^1+\cos{\lambda}\, x^2,x^3\big)$. Then $Q(0)=I$ and
$Q^{\,\prime}(0)=E_{12}$, where the action of $E_{12}$ is defined on a general $V=(V^1,V^2,V^3)$ by $(E_{12} V)^i = V^1 \delta_2^i - V_2 \delta_1^i$.
Then
 \[
Q^{-1}\newhat{p} =\frac{x-y}{t-t_0}- Q^{-1}\Theta(t_0,t,Qx),
\]
and hence
 \[
-Q^{\,\prime\!}(0) \newhat{p}+\newhat{p}^\prime\!(0)
= Q^{\,\prime\!}(0)\Theta(t_0,t,x)
-\Theta_x(t_0,t,x)Q^{\,\prime\!}(0)x
\]
and
 \[
\newhat{p}^\prime\!(0)
= E_{12}(\newhat{p}+\Theta(t_0,t,x))
-\Theta_x(t_0,t,x) E_{12} x .
\]
Then we get
 \begin{equation}
 \overline{\Omega}_{12} f(t,x,\newhat{p})
 =f_x(t,x,\newhat{p}) E_{12} x
 +f_{\newhat{p}}(t,x,\newhat{p})\big( E_{12}(\newhat{p}+\Theta)
-\Theta_x E_{12} x\big).
 \end{equation}
Here
\[
\Theta_x(t_0,t,x)^\ell E_{12}=\Theta_{\Omega_{12}}(t_0,t,x)^\ell,
\]
where
\[
\Theta_{\Omega_{12}}(t_0,t,x)
=\int_{t_0}^t \frac{s-t_0}{t-t_0}\Big((\Omega_{12} {\Gamma})(s,sx/t)\cdot \Lambda(x/t)
+ {\Gamma}\big(s,s\frac{x}{t}\big)\cdot\Omega_{12}\Lambda\big(\frac{x}{t}\big)\Big) ds.
\]
We get
\[
\overline{\Omega}_{12}=x^1\partial_{x^2}-x^2\partial_{x^1}
 +(\newhat{p}^1\!+\Theta^1)\partial_{\newhat{p}^2}-(\newhat{p}^2\!+\Theta^2)\partial_{\newhat{p}^1}
 -\Theta_{\Omega_{12}}^\ell \partial_{\newhat{p}^\ell}.
\]
In particular in the Minkowski case when $\Theta=0$
\[
\overline{\Omega}_{12}^M=x^1\partial_{x^2}-x^2\partial_{x^1}
 +\newhat{p}^1\partial_{\newhat{p}^2}-\newhat{p}^2\partial_{\newhat{p}^1}.
\]

\subsubsection{The scaling vector field in Minkowski space}
The scaling is defined by $Q(\lambda)=(1+\lambda)Id$. In the Minkowski case when $\Theta=0$ then $\ph(\lambda)=\ph$ is in fact constant so
\beq
\overline{S}^M=S^{\mu} \partial_{x^{\mu}}.
\eq

\subsubsection{The boosts vector fields in Minkowski space}
Let us define the boost as the generators of the Lorentz transformation
$ \big(t(\lambda),x(\lambda)\big)=Q(\lambda)(t,x)=\big((t+\lambda x^1)\gamma,(x^1+\lambda t)\gamma,x^2,x^3\big)$,
where $\gamma=1/\sqrt{1-\lambda^2}$, and $\big(t_0(\lambda),y(\lambda)\big)=Q(\lambda)(t_0,y)$.
In the Minkowski case $\newhat{p}(\lambda)$ is the solution of
\[
\newhat{p}(\lambda)
=\frac{x(\lambda)-y(\lambda)}{t(\lambda)-t_0(\lambda)},
\]
i.e.
\[
\newhat{p}(\lambda)^1\!
=\frac{x^1\!-y^1\!+\lambda(t-t_0)}{t-t_0+\lambda(x^1\!-y^{1\!})\!\!}\, ,\quad\text{and}\quad
\newhat{p}(\lambda)^k\!
=\frac{(x^k\!-y^k)\gamma}{t-t_0+\lambda(x^1\!-y^{1\!})\!\!}\,\, ,\quad k\!>\!1.
\]
Differentiating with respect to $\lambda$ at $0$ and putting $\newhat{p}=\newhat{p}(0)$
gives
\[
\newhat{p}^\prime(0)^k
=\delta^{k1}-{\newhat{p}^k\newhat{p}^1}.
\]

Hence
\begin{equation}\label{eq:boostmodifiedvectorfields}
\overline{B}_1^M=B_1^{\mu} \partial_{x^{\mu}}+\big(\delta^{k1}-{\newhat{p}^k\newhat{p}^1}\big)\partial_{\ph^k}
\end{equation}

\subsection{The vector fields used for Theorem \ref{thm:mainL2}}
\label{subsec:vectorsused}
In the proof of Theorem \ref{thm:mainL2} a slightly different collection of vector fields are used.  They are computationally slightly simpler to use than the vector fields of the previous section, though are mildly singular for $t$ close to $t_0$.  It is therefore assumed throughout this section that $t\geq t_0+1$, and the vector fields will only be used in the proof of Theorem \ref{thm:mainL2} under this assumption.  The vector fields can be derived by imposing that they have the form
\beq \label{eq:Zanzatz}
\Zb=Z^\mu\partial_{x^\mu}+\check{Z}^i\partial_{\ph^i},
\eq
where $\check{Z}^i$ are to be determined, and insisting that the relation
\begin{equation} \label{eq:introZbX0sec3truedef}
	\Zb\vert_{(t,x,\ph)} \left( X_2(t_0,t,x,\ph)^i \right)
	=
	Z^i \vert_{(t_0,X_2(t_0))} - Z^0 \vert_{(t_0,X_2(t_0))} \ph^i,
\end{equation}
holds, instead of the relation \eqref{eq:introZbX0sec3} which involves $\Ph_2(t_0,t,x,\ph)^i$ instead of $\ph^i$ and is satisfied by the vector fields of the previous section.  Equation \eqref{eq:introZbX0sec3truedef} is in particular true for the Minkowski vector fields $\Zb^M$.

Applying the expression \eqref{eq:Zanzatz} to the approximation
\beq
X_2(t_0,t,x,\ph)=x - (t-t_0)\newhat{p}- (t-t_0)\Theta(t_0,t,x),
\eq
gives
\beq
\Zb \big(X_2(t_0,t,x,\ph)^i\big)=Z^i\vert_{(t,x)} - Z^0\vert_{(t,x)}\newhat{p}^i- Z\big((t-t_0)\Theta(t_0,t,x)^i\big)-(t-t_0)\check{Z}^i\vert_{(t,x,\ph)},
\eq
and so $\check{Z}^i$ is indeed determined if $\Zb \big(X_2(t_0,t,x,\ph)^i\big)$ is prescribed
for each $i$.
With this definition we obtain
\beq\label{eq:modificationofZ}
\check{Z}^i\vert_{(t,x,\ph)}=\frac{1}{t-t_0}
\Big(Z^i\vert_{(t,x)}-Z^i \vert_{(t_0,X_2(t_0))}-\big(Z^0\vert_{(t,x)}-Z^0 \vert_{(t_0,X_2(t_0))} \big)\ph^i
-Z\big((t-t_0)\Theta(t_0,t,x)\big)^i\Big).
\eq

We have
\beq
  \Ph_{\!2}(t_0,t,x,\ph)^i=\ph^i-\Psi(t_0,t,x)^i,\qquad\text{where}\quad \Psi(t_0,t,x)^i=
\int_{t_0}^t \,{\Gamma}\big( s, sx/t\big)\cdot \Lambda\big(x/t\big)^i\, ds.
\eq
so
\begin{equation}
	\Zb\vert_{(t,x,\ph)} \left( \Ph_2(t_0,t,x,\ph)^i \right)
	=
	\check{Z}^i \vert_{(t,x,\ph)} - Z^\mu \vert_{(t,x,\ph)} \pa_{x^\mu} \Psi(t_0,t,x)^i
.
\end{equation}
which again is bounded, provided the components of $Z$ grow at most like $t$.
\subsubsection{The rotation vector fields}
The rotations agree with the definition in the previous section since then $Z^0=0$:
\[
\overline{\Omega}_{12}=x^1\partial_{x^2}-x^2\partial_{x^1}
 +(\newhat{p}^1\!+\Theta^1)\partial_{\newhat{p}^2}-(\newhat{p}^2\!+\Theta^2)\partial_{\newhat{p}^1}
 -\Theta_{\Omega_{12}}^\ell \partial_{\newhat{p}^\ell},
\]
where
\[
\Theta_{\Omega_{12}}(t_0,t,x)
=\int_{t_0}^t \frac{s-t_0}{t-t_0}\Big((\Omega_{12} {\Gamma})(s,sx/t)\cdot \Lambda(x/t)
+ {\Gamma}\big(s,s\frac{x}{t}\big)\cdot\Omega_{12}\Lambda\big(\frac{x}{t}\big)\Big) ds,
\]
and
\beq
  \Theta(t_0,t,x)^i=
\int_{t_0}^t \frac{s-t_0}{t-t_0}\,{\Gamma}\big( s, sx/t\big)\cdot \Lambda\big(x/t\big)^i\, ds.
\eq
\subsubsection{The scaling vector field}
By \eqref{eq:modificationofZ}
we have
\begin{equation*}
\check{S}^i\vert_{(t,x,\ph)}=\frac{1}{t-t_0}
\Big
(x^i-X_2(t_0)^i-\big(t-t_0\big)\ph^i
-S\big((t-t_0)\Theta(t_0,t,x)\big)^i\Big)=\Theta(t_0,t,x)^i-\Theta_S(t_0,t,x)^i,
\end{equation*}
where
\beq
\Theta_{S}(t_0,t,x)^i=t\,\Gamma(t,x)\cdot \Lambda(x/t)^i.
\eq

\subsubsection{The boost vector fields}
By \eqref{eq:modificationofZ}
we have
\begin{multline}
\check{B}_1^i\vert_{(t,x,\ph)}=\frac{1}{t-t_0}
\Big
((t-t_0)\delta^{i1}-\big(x^1-X_2(t_0)^1\big)\ph^i
-B_1\big((t-t_0)\Theta(t_0,t,x)\big)^i\Big)\\
=\delta^{i1}-\big(\ph^1+\Theta(t_0,t,x)^1\big)\ph^i -\Theta_{B_1}(t_0,t,x)^i,
\end{multline}
where
\begin{multline*}
\Theta_{B_1}(t_0,t,x)=x^1\,\Gamma(t,x)\cdot \Lambda(x/t)^i\\
+\int_{t_0}^t \frac{s-t_0}{t-t_0}\Big( \big((B_1 \,\Gamma)\big(s,s\frac{x}{t}\big)-\frac{x^1}{t}(S \,\Gamma)\big(s,s\frac{x}{t}\big)\big)\cdot \Lambda\big(\frac{x}{t}\big)+\Gamma\big(s,s\frac{x}{t}\big)\cdot B_1\big( \Lambda\big(\frac{x}{t}\big)\big) \Big)ds.
\end{multline*}

\subsection{Estimates for one vector field applied to the geodesics}
\label{subsec:rotationexample}
Since the vector fields $\Zb$ in Section \ref{subsec:vectorsused} are defined using the approximations to the geodesics $X_2$, it remains to check how they behave when applied to the true geodesics $X(s,t,x,\ph)^k$ and $\Ph(s,t,x,\ph)^k$.  The detailed computations, including higher order estimates, are given in Section \ref{section:geodesics}.  Here the basic idea is illustrated by considering only one rotation vector field $\overline{\Omega}_{ij}$ applied to the geodesics.

Instead of considering the geodesics themselves directly, it is convenient to first consider, for fixed $(t,x,\ph) \in \supp(f)$,
\begin{equation} \label{eq:dXbdsintro}
	\overline{X}(s,t,x,\ph)^k := X(s,t,x,\ph)^k - X_2(s,t,x,\ph)^k,
	\qquad
	\overline{P}(s,t,x,\ph)^k := \frac{d \overline{X}^k}{ds} (s,t,x,\ph),
\end{equation}
where $X_2$ is defined in equation \eqref{eq:X2intro}.  The geodesic equations \eqref{eq:geodesictimenormalized} imply that
\begin{align} \label{eq:dPbdsintro}
\begin{split}
	\frac{d\overline{P}^k}{ds}
	&
	=
	\hat{\Gamma}^k (s,X(s),\Ph(s)) - \hat{\Gamma}^k \left( s, s\frac{x}{t},\frac{x}{t} \right)
	\\
	&
	=
	\Gamma^0_{\alpha \beta}(s,X(s)) \Ph^k \Ph^{\alpha} \Ph^{\beta}
	-
	\Gamma^0_{\alpha \beta}\left( s,s \frac{x}{t} \right) \frac{x^k}{t} \frac{x^{\alpha}}{t} \frac{x^{\beta}}{t}
	+
	\Gamma^k_{\alpha \beta}\left( s,s \frac{x}{t} \right) \frac{x^{\alpha}}{t} \frac{x^{\beta}}{t}
	-
	\Gamma^k_{\alpha \beta}(s,X(s)) \Ph^{\alpha} \Ph^{\beta}.
\end{split}
\end{align}
It can be shown, and this is outlined below, that the assumptions for $\Gamma$ in Theorem \ref{thm:mainL2} imply that
\begin{equation} \label{eq:Omegaestimateintro}
	\left\vert \overline{\Omega}_{ij} \left( \hat{\Gamma}^k (s,X(s),\Ph(s)) - \hat{\Gamma}^k \left( s, s\frac{x}{t},\frac{x}{t} \right) \right) \right\vert
	\lesssim
	\frac{\varepsilon}{s^{1+2a}}
	+
	\frac{\varepsilon}{s^{1+2a}} \left\vert \overline{\Omega}_{ij} \left( \overline{X}(s) \right) \right\vert
	+
	\frac{\varepsilon}{s^{1+a}} \left\vert \overline{\Omega}_{ij} \left( \overline{P}(s) \right) \right\vert,
\end{equation}
for all $t_0 \leq s \leq t$.  Once this has been shown it follows from applying $\overline{\Omega}_{ij}$ to the equation \eqref{eq:dPbdsintro} and integrating backwards from $s=t$, using the fact that $(\overline{\Omega}_{ij} \Pb)(t,t,x,\ph) = 0$, that
\[
	\left\vert \overline{\Omega}_{ij} ( \Pb(s)^k) \right\vert
	\lesssim
	\frac{\varepsilon}{s^{2a}}
	+
	\int_s^t
	\frac{\varepsilon}{\tilde{s}^{1+2a}}
	\left\vert \overline{\Omega}_{ij} (\overline{X}(\tilde{s})) \right\vert
	+
	\frac{\varepsilon}{\tilde{s}^{1+a}}
	\left\vert \overline{\Omega}_{ij} (\overline{P}(\tilde{s})) \right\vert
	d\tilde{s},
\]
and, summing over $k=1,2,3$, the Gr\"{o}nwall inequality (see Lemma \ref{lem:Gronwall}) gives,
\[
	\left\vert \overline{\Omega}_{ij} ( \Pb(s)^k) \right\vert
	\lesssim
	\frac{\varepsilon}{s^{2a}}
	+
	\int_s^t
	\frac{\varepsilon}{\tilde{s}^{1+2a}}
	\left\vert \overline{\Omega}_{ij} (\overline{X}(\tilde{s})) \right\vert
	d\tilde{s}.
\]
Inserting this bound into the equation \eqref{eq:dXbdsintro} for $\Xb$, after applying $\overline{\Omega}_{ij}$, integrating backwards from $s=t$ and using the fact that $(\overline{\Omega}_{ij} \Xb)(t,t,x,\ph) = 0$ gives,
\[
	\left\vert \overline{\Omega}_{ij} ( \Xb(s)^k) \right\vert
	\lesssim
	\varepsilon s^{1-2a}
	+
	\int_s^t
	\int_{s'}^t
	\frac{\varepsilon}{\tilde{s}^{1+2a}}
	\left\vert \overline{\Omega}_{ij} (\overline{X}(\tilde{s})) \right\vert
	d\tilde{s}
	ds',
\]
since $2a>1$.  For any function $\lambda(\tilde{s})$,
\[
	\int_s^t
	\int_{s'}^t
	\frac{1}{\tilde{s}^{1+2a}}
	\lambda(\tilde{s})
	d\tilde{s}
	ds'
	=
	\int_s^t
	\int_{s}^t
	\chi_{\{ s'\leq \tilde{s} \}}
	ds'
	\frac{1}{\tilde{s}^{1+2a}}
	\lambda(\tilde{s})
	d\tilde{s}
	=
	\int_s^t
	\frac{(\tilde{s} - s)}{\tilde{s}^{1+2a}}
	\lambda(\tilde{s})
	d\tilde{s},
\]
where $\chi_A$ denotes the indicator function of the set $A$, and so
\[
	\left\vert \overline{\Omega}_{ij} ( \Xb(s)^k) \right\vert
	\lesssim
	\varepsilon s^{1-2a}
	+
	\int_s^t
	\frac{\varepsilon}{\tilde{s}^{2a}}
	\left\vert \overline{\Omega}_{ij} (\overline{X}(\tilde{s})) \right\vert
	d\tilde{s},
\]
and another application of the Gr\"{o}nwall inequality gives,
\[
	\left\vert \overline{\Omega}_{ij} ( \Xb(s)^k) \right\vert
	\lesssim
	\varepsilon s^{1-2a}.
\]
Setting $s=t_0$, the bound \eqref{eq:OmegaX2intro} then gives,
\[
	\left\vert \overline{\Omega}_{ij} ( X(t_0,t,x,\ph)^k) \right\vert
	\leq
	C,
\]
for some constant $C$.  One similarly deduces the same bound with $\Ph$ in place of $X$.

In order to see that the bound \eqref{eq:Omegaestimateintro} holds for $t_0 \leq s \leq t$, first write,
\begin{multline*}
	\Gamma^k_{\alpha \beta}(s,X(s)) \Ph^{\alpha} \Ph^{\beta}
	-
	\Gamma^k_{\alpha \beta}\left( s,s \frac{x}{t} \right) \frac{x^{\alpha}}{t} \frac{x^{\beta}}{t}
	=
	\left(
	\Gamma^k_{\alpha \beta}(s,X(s))
	-
	\Gamma^k_{\alpha \beta}\left( s,s \frac{x}{t} \right)
	\right) \frac{x^{\alpha}}{t} \frac{x^{\beta}}{t}
	\\
	+
	\Gamma^k_{\alpha \beta}(s,X(s))
	\left( \Ph^{\alpha} - \frac{x^{\alpha}}{t} \right) \frac{x^{\beta}}{t}
	+
	\Gamma^k_{\alpha \beta}(s,X(s))
	\frac{x^{\alpha}}{t} \left( \Ph^{\beta} - \frac{x^{\beta}}{t} \right).
\end{multline*}
Then note that,
\begin{multline*}
	\overline{\Omega}_{ij}
	\left(
	\Gamma^k_{\alpha \beta}(s,X(s))
	-
	\Gamma^k_{\alpha \beta}\left( s,s \frac{x}{t} \right)
	\right)
	=
	\overline{\Omega}_{ij}
	\left(
	X(s)^l
	-
	s \frac{x^l}{t}
	\right)
	(\partial_{x^l} \Gamma^k_{\alpha \beta}) (s,X(s))
	\\
	+
	\overline{\Omega}_{ij}
	\left(
	s \frac{x^l}{t}
	\right)
	\left(
	(\partial_{x^l} \Gamma^k_{\alpha \beta})(s,X(s))
	-
	(\partial_{x^l} \Gamma^k_{\alpha \beta})\left( s,s \frac{x}{t} \right)
	\right).
\end{multline*}
The bound for the first term follows from writing $X(s)^l - s \frac{x^l}{t} = \Xb(s)^l + X_2(s)^l - s \frac{x^l}{t}$ and, using the expression \eqref{eq:X2intro} for $X_2$ and the definition of $\mathring{\Omega}^l_{ij}$,
\begin{align*}
	&
	\overline{\Omega}_{ij}
	\left(
	X_2(s)^l
	-
	s \frac{x^l}{t}
	\right)
	=
	x^i \delta_j^l - x^j \delta_i^l
	-
	(t-s) \left( \ph^i \delta_j^l - \ph^j \delta_i^l \right)
	-
	\frac{s}{t} \left( x^i \delta_j^l - x^j \delta_i^l \right)
	-
	(t-s) \mathring{\Omega}^l_{ij}
	\\
	&
	\qquad \qquad \qquad
	-
	\Omega_{ij}
	\left( \int_s^t (s'-s) \Gh^l \left( s', s'\frac{x}{t}, \frac{x}{t} \right) ds' \right)
	\\
	&
	=
	x^i \delta_j^l - x^j \delta_i^l
	-
	\left( (t-t_0) - (s-t_0) \right) \left( \ph^i \delta_j^l - \ph^j \delta_i^l \right)
	-
	\frac{s}{t} \left( x^i \delta_j^l - x^j \delta_i^l \right)
	-
	(t-t_0) \mathring{\Omega}^l_{ij}
	+
	(s-t_0) \mathring{\Omega}^l_{ij}
	\\
	&
	\qquad
	-
	\Omega_{ij}
	\left(
	\int_{t_0}^t (s'-t_0 + t_0) \Gh^l \left( s', s'\frac{x}{t}, \frac{x}{t} \right) ds'
	-
	\int_{t_0}^s s' \Gh^l \left( s', s'\frac{x}{t}, \frac{x}{t} \right) ds'
	-
	s \int_s^t \Gh^l \left( s', s'\frac{x}{t}, \frac{x}{t} \right) ds'
	\right)
	\\
	&
	=
	X_2(t_0)^i \delta_j^l - X_2(t_0)^j \delta_i^l
	-
	s \left( \left( \frac{x^i}{t} - \ph^i \right) \delta_j^l - \left( \frac{x^j}{t} - \ph^j \right) \delta_i^l \right)
	+
	t_0 \left( \ph^i \delta_j^l - \ph^j \delta_i^l \right)
	+
	(s-t_0) \mathring{\Omega}^l_{ij}
	\\
	&
	\qquad
	-
	\Omega_{ij}
	\left(
	t_0 \int_{t_0}^t \Gh^l \left( s', s'\frac{x}{t}, \frac{x}{t} \right) ds'
	-
	\int_{t_0}^s s' \Gh^l \left( s', s'\frac{x}{t}, \frac{x}{t} \right) ds'
	-
	s \int_s^t \Gh^l \left( s', s'\frac{x}{t}, \frac{x}{t} \right) ds'
	\right).
\end{align*}

Using the assumptions on $\Gamma$ of Theorem \ref{thm:mainL2}, the fact that
\[
	\left\vert \frac{x^i}{t} - \ph^i \right\vert \leq \frac{C}{t^a},
\]
for $(t,x,\ph) \in \supp(f)$ (see Proposition \ref{prop:sec22}) and the fact that $\vert \mathring{\Omega}^l_{ij}(t,x) \vert \leq Ct^{-a}$, it follows that
\[
	\left\vert
	\overline{\Omega}_{ij}
	\left(
	X_2(s)^l
	-
	s \frac{x^l}{t}
	\right)
	\right\vert
	\leq
	C s^{1-a}.
\]
For the second term, note that,
\[
	\left\vert
	(\partial_{x^l} \Gamma^k_{\alpha \beta})(s,X(s))
	-
	(\partial_{x^l} \Gamma^k_{\alpha \beta})\left( s,s \frac{x}{t} \right)
	\right\vert
	\leq
	\sup_{\vert z \vert \leq c s+K} \left\vert \partial^2 \Gamma^k_{\alpha \beta}(s,z) \right\vert
	\left\vert X(s) - s \frac{x}{t} \right\vert
	\leq
	\frac{C \varepsilon}{s^{3+a}} \left\vert X(s) - s \frac{x}{t} \right\vert,
\]
and, as above,
\begin{multline*}
	\left\vert X(s)^k - s \frac{x^k}{t} \right\vert
	\leq
	\left\vert \Xb(s)^k \right\vert
	+
	\left\vert X_2(s)^k - s \frac{x^k}{t} \right\vert
	\leq
	\left\vert \Xb(s)^k \right\vert
	+
	\left\vert X_2(t_0)^k \right\vert
	+
	s
	\left\vert \frac{x^k}{t} - \ph^k \right\vert
	\\
	+
	t_0 \left\vert \int_{t_0}^t \Gh^k \left( s', s'\frac{x}{t}, \frac{x}{t} \right) ds' \right\vert
	+
	\left\vert \int_{t_0}^s s' \Gh^k \left( s', s'\frac{x}{t}, \frac{x}{t} \right) ds' \right\vert
	+
	s \left\vert \int_{s}^t \Gh^k \left( s', s'\frac{x}{t}, \frac{x}{t} \right) ds' \right\vert
	\leq
	Cs^{1-a}.
\end{multline*}
Hence,
\[
	\left\vert
	\overline{\Omega}_{ij}
	\left( \Gamma^k_{\alpha \beta}(s,X(s))
	-
	\Gamma^k_{\alpha \beta}\left( s,s \frac{x}{t} \right) \right)
	\right\vert
	\lesssim
	\frac{\varepsilon}{s^{1+2a}}
	+
	\frac{\varepsilon}{s^{2+a}} \left\vert \overline{\Omega}_{ij} (\Xb(s)) \right\vert.
\]
Similarly, using the fact that
\[
	\Pb(s)^k = \Ph^k(s) - \ph^k + \int_s^t \Gh^k \left( s', s'\frac{x}{t}, \frac{x}{t} \right) ds',
\]
it follows that
\[
	\left\vert \Ph(s)^k - \frac{x^k}{t} \right\vert
	\lesssim
	\left\vert \Pb(s)^k \right\vert
	+
	\left\vert \frac{x^k}{t} - \ph^k \right\vert
	+
	\int_s^t \left\vert \Gh^k \left( s', s'\frac{x}{t}, \frac{x}{t} \right) \right\vert ds'
	\lesssim
	s^{-a},
	\qquad
	\left\vert \Ph(s)^k \right\vert \lesssim 1,
\]
and similarly,
\[
	\left\vert \overline{\Omega}_{ij} \left( \Ph(s)^k - \frac{x^k}{t} \right) \right\vert
	\lesssim
	\left\vert \overline{\Omega}_{ij} \left( \Pb(s)^k \right) \right\vert
	+
	s^{-a},
	\qquad
	\left\vert \overline{\Omega}_{ij} \left( \Ph(s)^k \right) \right\vert
	\lesssim
	\left\vert \overline{\Omega}_{ij} \left( \Pb(s)^k \right) \right\vert + 1.
\]
Hence,
\[
	\left\vert \overline{\Omega}_{ij} \left(
	\Gamma^k_{\alpha \beta}(s,X(s)) \Ph^{\alpha} \Ph^{\beta}
	-
	\Gamma^k_{\alpha \beta}\left( s,s \frac{x}{t} \right) \frac{x^{\alpha}}{t} \frac{x^{\beta}}{t}
	\right) \right\vert
	\lesssim
	\frac{\varepsilon}{s^{1+2a}}
	+
	\frac{\varepsilon}{s^{1+2a}} \left\vert \overline{\Omega}_{ij} \left( \overline{X}(s) \right) \right\vert
	+
	\frac{\varepsilon}{s^{1+a}} \left\vert \overline{\Omega}_{ij} \left( \overline{P}(s) \right) \right\vert.
\]

In a similar manner it follows that
\[
	\left\vert \overline{\Omega}_{ij} \left(
	\Gamma^0_{\alpha \beta}(s,X(s)) \Ph^k \Ph^{\alpha} \Ph^{\beta}
	-
	\Gamma^0_{\alpha \beta}\left( s,s \frac{x}{t} \right) \frac{x^k}{t} \frac{x^{\alpha}}{t} \frac{x^{\beta}}{t}
	\right) \right\vert
	\lesssim
	\frac{\varepsilon}{s^{1+2a}}
	+
	\frac{\varepsilon}{s^{1+2a}} \left\vert \overline{\Omega}_{ij} \left( \overline{X}(s) \right) \right\vert
	+
	\frac{\varepsilon}{s^{1+a}} \left\vert \overline{\Omega}_{ij} \left( \overline{P}(s) \right) \right\vert,
\]
from which \eqref{eq:Omegaestimateintro} then follows.

\section{Vector fields applied to geodesics}
\label{section:geodesics}

It will be assumed throughout most of this section that $t\geq t_0+1$, and $\vert x \vert \leq c t$, where $0<c<1$ (recall the discussion in Section \ref{subsec:transtime}).  Solutions of the Vlasov equation on a fixed spacetime will be considered and the pointwise bounds
\begin{equation} \label{eq:Gammasuppf}
	\left\vert Z^I \Gamma^{\alpha}_{\beta \gamma}(t',x') \right\vert
	\lesssim
	\frac{\varepsilon}{(t')^{1+a}}
	\qquad
	\text{for }
	t_0 \leq t' \leq t,
	\text{ }
	\vert x' \vert \leq c t',
	\text{ and }
	\vert I \vert \leq \left\lfloor \frac{N}{2} \right\rfloor + 2,
\end{equation}
where $\frac{1}{2} < a < 1$, will often be assumed to hold.  Section \ref{subsec:smalltime} is concerned with the case $t_0 \leq t \leq t_0 + 1$.

Recall from Section \ref{section:suppf} the definitions \eqref{eq:X2P2def} of $X_2(s,t,x,\ph)^i$ and $\Ph_2(s,t,x,\ph)^i$, and
\begin{align*}
	\Xb (s,t,x,\ph)^{i}
	&
	:= X(s,t,x,\ph)^{i} - X_2(s,t,x,\ph)^i,
	\\
	\Pb(s,t,x,\ph)^{i}
	&
	:=
	\frac{d\Xb}{ds}(s,t,x,\ph)
	=
	\Ph(s,t,x,p)^{i} -	\Ph_{\!\!2}(s,t,x,p)^{i}.
\end{align*}

\subsection{Vector fields} \label{subsec:vectorfields}
Recall the vector fields which are used for the Vlasov equation on Minkowski space,
\[
	\overline{\Omega}_{ij}^M
	=
	x^i \partial_{x^j} - x^j \partial_{x^i} + \ph^i \partial_{\ph^j} - \ph^j \partial_{\ph^i},
	\quad
	\overline{B}_i^M
	=
	x^i \partial_t + t \partial_{x^i} + \left( \delta_i^j - \ph^i \ph^j \right) \partial_{\ph^j},
	\quad
	\overline{S}^M
	=
	t\partial_t + x^k \partial_{x^k}.
\]

\begin{remark} \label{rmk:minkowskivectors}
	Note that, in Minkowski space, the maps $X(s,t,x,\ph)^i$ are simply given by
	\[
		X_M(s,t,x,\ph)^i = x^i - (t-s)\ph^i,
	\]
	for $i=1,2,3$.  It is easy to see that the Minkowski vector fields satisfy,
	\[
		\overline{\Omega}_{ij}^M \left( X_M(s)^k \right)
		=
		X_M(s)^i \delta_j^k - X_M(s)^j \delta_i^k,
		\qquad
		\overline{B}_{i}^M \left( X_M(s)^k \right)
		=
		s \delta_i^k - X_M(s)^i \ph^k,
	\]
	\[
		\overline{S}^M \left( X_M(s)^k \right)
		=
		X_M(s)^k - s \ph^k.
	\]
\end{remark}

Recall, for $i=1,2,3$,
\begin{equation}
	\newhat{\Gamma}(t,x,\newhat{p})^i
	=
	\newhat{p}^i\Gamma_{\alpha \beta}^0(t,x)\newhat{p}^\alpha\newhat{p}^\beta
	-
	\Gamma_{\alpha\beta}^i(t,x)\newhat{p}^{\alpha}\newhat{p}^\beta,
\end{equation}
and define,
\[
	\Theta^i(t,x) = \frac{1}{t-t_0} \int_{t_0}^t (s' - t_0) \Gh^i \left( s',s'\frac{x}{t},\frac{x}{t} \right) ds'.
\]
The proof of the main result uses the following collection of vector fields, schematically denoted $\overline{Z}$
\[
	\overline{Z}
	=
	\overline{\Omega}_{ij},
	\overline{B}_i,
	\overline{S},
\]
for $i,j=1,2,3$, $i<j$, where $\overline{\Omega}_{ij}, \overline{B}_i, \overline{S}$ are defined as follows.  First,
\[
	\overline{\Omega}_{ij}\vert_{(t,x,p)}
	=
	x^i \partial_{x^j} - x^j \partial_{x^i} + \ph^i \partial_{\ph^j} - \ph^j \partial_{\ph^i} + \mathring{\Omega}_{ij}^k \partial_{\ph^k},
\]
where,
\begin{align}
	\mathring{\Omega}_{ij}^k\vert_{(t,x,p)}
	=
	\
	&
	\frac{1}{t-t_0}
	\Bigg[
	\int_{t_0}^t (s'-t_0)
	\left(
	\Gh^i \left( s', s' \frac{x}{t}, \frac{x}{t} \right) \delta_j^k
	-
	\Gh^j \left( s', s' \frac{x}{t}, \frac{x}{t} \right) \delta_i^k
	\right)
	ds'
	\nonumber
	\\
	&
	\qquad \qquad \qquad \qquad \qquad
	-
	(x^i \partial_{x^j} - x^j \partial_{x^i})
	\left(
	\int_{t_0}^t (s'-t_0)
	\Gh^k \left( s', s' \frac{x}{t}, \frac{x}{t} \right)
	ds'
	\right)
	\Bigg]
	\nonumber
	\\
	=
	\
	&
	\Theta^i(t,x) \delta_j^k
	-
	\Theta^j(t,x) \delta_i^k
	-
	\Omega_{ij} \left( \Theta^k(t,x) \right),
	\label{eq:Omega0Phi}
\end{align}
for $1 \leq i < j \leq 3$.  Second,
\[
	\overline{B}_i\vert_{(t,x,p)}
	=
	x^i \partial_t + t\partial_{x^i} + \left( \delta_i^j - \ph^i \ph^j \right) \partial_{\ph^j} + \mathring{B}_{i}^k \partial_{\ph^k},
\]
where,
\begin{align}
	\mathring{B}_{i}^k\vert_{(t,x,p)}
	=
	\
	&
	-
	\frac{1}{t - t_0}
	\bigg[
	\int_{t_0}^t (s' - t_0) \Gh^i \left( s', s'\frac{x}{t}, \frac{x}{t} \right) ds' \ph^k
	+
	(x^i \partial_t + t \partial_{x^i})
	\left(
	\int_{t_0}^t (s' - t_0) \Gh^k \left( s', s'\frac{x}{t}, \frac{x}{t} \right) ds'
	\right)
	\bigg]
	\nonumber
	\\
	=
	\
	&
	-
	\Theta^i(t,x) \ph^k
	-
	B_i \left( \Theta^k(t,x) \right)
	-
	\frac{x^i}{t-t_0} \Theta^k(t,x),
	\label{eq:B0Phi}
\end{align}
for $i=1,2,3$. Finally,
\[
	\overline{S}\vert_{(t,x,p)}
	=
	t \partial_t + x^i \partial_{x^i} + \mathring{S}^k \partial_{\ph^k},
\]
where,
\begin{align}
	\mathring{S}^k\vert_{(t,x,p)}
	=
	\
	&
	\frac{1}{t-t_0}
	\bigg[
	\int_{t_0}^t (s' - t_0) \Gh^k \left( s', s'\frac{x}{t}, \frac{x}{t} \right) ds'
	-
	(t \partial_t + x^i \partial_{x^i})
	\left(
	\int_{t_0}^t (s' - t_0) \Gh^k \left( s', s'\frac{x}{t}, \frac{x}{t} \right) ds'
	\right)
	\bigg]
	\nonumber
	\\
	=
	\
	&
	\Theta^k(t,x)
	-
	S \left( \Theta^k(t,x) \right)
	-
	\frac{t}{t-t_0} \Theta^k(t,x).
	\label{eq:S0Phi}
\end{align}

Recall that, by the bound \eqref{eq:sec23}, the maps $X_2(s,t,x,\ph)^i$ are good approximations to the true geodesics, $X(s,t,x,\ph)^i$.  The vector fields $\overline{\Omega}_{ij}$, $\overline{B}_i$, $\overline{S}$ are defined so that, when applied to $X_2(s,t,x,\ph)^i$, the analogue of the bounds in Remark \ref{rmk:minkowskivectors} hold at $s=t_0$ as follows.

\begin{proposition} \label{prop:ZbX2first}
	For $(t,x,p) \in \supp(f)$, the vector fields $\overline{\Omega}_{ij}$, $\overline{B}_i$, $\overline{S}$ defined above satisfy,
	\begin{align*}
		\overline{\Omega}_{ij} \left( X_2(t_0,t,x,\ph)^k \right)
		&
		=
		X_2(t_0,t,x,\ph)^i \delta_j^k - X_2(t_0,t,x,\ph)^j \delta_i^k,
		\\
		\overline{B}_i \left( X_2(t_0,t,x,\ph)^k \right)
		&
		=
		t_0 \delta_i^k - X_2(t_0,t,x,\ph)^i \ph^k,
		\\
		\overline{S} \left( X_2(t_0,t,x,\ph)^k \right)
		&
		=
		X_2(t_0,t,x,\ph)^k - t_0 \ph^k,
	\end{align*}
	for $i,j,k = 1,2,3$.
\end{proposition}

\begin{proof}
	The proof is a straightforward computation.
\end{proof}

\subsection{Repeated vector fields applied to the initial conditions for approximations to geodesics}
After applying $\Zb^I$ to the equation \eqref{eq:Pbar}, after replacing $X$ with $\Xb$ the expression
\[
	\Zb^I \left( s\frac{x^i}{t} - X_2^i(s,t,x,\ph) \right),
\]
will appear on the right hand side, and it is therefore necessary to estimate it.  Using a higher order analogue of Proposition \ref{prop:ZbX2first}, it suffices to control $\Zb^I$ applied to
\begin{align}
	&
	s\frac{x^i}{t}
	-
	X_2^i(s,t,x,\ph)
	+
	X_2^i(t_0,t,x,\ph)\label{eq:X1X2difference}
	\\
	&
	\qquad \qquad
	=
	s \left( \frac{x^i}{t} - \ph^i \right)
	+
	t_0 \ph^i
	+
	\int_s^t (s' - s) \Gh^k \left( s', s'\frac{x}{t}, \frac{x}{t} \right) ds'
	-
	\int_{t_0}^t (s' - t_0) \Gh^k \left( s', s'\frac{x}{t}, \frac{x}{t} \right) ds'
	\nonumber
	\\
	&
	\qquad \qquad
	=
	s \left( \frac{x^i}{t} - \ph^i \right)
	+
	t_0 \ph^i
	+
	\int_{t_0}^s (s' - t_0) \Gh^k \left( s', s'\frac{x}{t}, \frac{x}{t} \right) ds'
	-
	(s - t_0) \int_{s}^t \Gh^k \left( s', s'\frac{x}{t}, \frac{x}{t} \right) ds'.
	\nonumber
\end{align}
Similarly, the expression
\begin{equation}
	\Zb^I \left( \frac{x^i}{t} - \Ph^i(s,t,x,p) \right)
	=
	- \Zb^I \left(
	\Pb^i(s,t,x,p)
	\right)
	+
	\Zb^I \left(
	\frac{x^i}{t} -\widehat{P}_2^i(s,t,x,p)
	\right),\label{eq:P1Pdifference}
\end{equation}
will also appear, where
$$
\widehat{P}_2^i(s,t,x,p)= \ph^i+\int_s^t \Gh^i \left( s', s'\frac{x}{t}, \frac{x}{t} \right) ds'
$$

The following generalises Proposition \ref{prop:ZbX2first} to higher orders.

\begin{proposition} \label{prop:ZbX2}
	For any multi index $I$, there exist smooth functions $\Lambda$ such that
	\begin{align*}
		\Zb^I \left( X_2(t_0)^i \right)
		=
		\
		&
		X_2(t_0)^j
		\Bigg(
		\Lambda^i_{I,j} \left( \ph \right)
		+
		\sum_{k=1}^{\vert I\vert} \sum_{\vert J_1\vert + \ldots + \vert J_k \vert \leq \vert I \vert -k}
		\tilde{\Lambda}_{I,j,i_1,\ldots,i_k}^{i,J_1,\ldots,J_k} \left( \ph,t,x \right)
		Z^{J_1} \left( \Theta^{i_1} \right) \ldots Z^{J_k} \left( \Theta^{i_k} \right)
		\Bigg)
		\\
		&
		+
		\Lambda^i_{I} \left( \ph \right)
		+
		\sum_{k=1}^{\vert I\vert} \sum_{\vert J_1\vert + \ldots + \vert J_k \vert \leq \vert I \vert -k}
		\tilde{\Lambda}_{I,i_1,\ldots,i_k}^{i,J_1,\ldots,J_k} \left( \ph,t,x \right)
		Z^{J_1} \left( \Theta^{i_1} \right) \ldots Z^{J_k} \left( \Theta^{i_k} \right)
	\end{align*}
	for $i=1,2,3$, where $X_2(t_0)^i = X_2(t_0,t,x,p)^i$ and $\tilde{\Lambda}_{I,j,i_1,\ldots,i_k}^{i,J_1,\ldots,J_k}$ and $\tilde{\Lambda}_{I,i_1,\ldots,i_k}^{i,J_1,\ldots,J_k}$ satisfy
	\[
		\tilde{\Lambda}_{I,j,i_1,\ldots,i_k}^{i,J_1,\ldots,J_k} \left( \ph,t,x \right)
		=
		\Lambda_{I,j,i_1,\ldots,i_k}^{i,J_1,\ldots,J_k} \left( \ph, \frac{x}{t-t_0}, \frac{t}{t-t_0} \right),
		\qquad
		\tilde{\Lambda}_{I,i_1,\ldots,i_k}^{i,J_1,\ldots,J_k} \left( \ph,t,x \right)
		=
		\Lambda_{I,i_1,\ldots,i_k}^{i,J_1,\ldots,J_k} \left( \ph, \frac{x}{t-t_0}, \frac{t}{t-t_0} \right),
	\]
	for some smooth functions $\Lambda_{I,j,i_1,\ldots,i_k}^{i,J_1,\ldots,J_k}$ and $\Lambda_{I,i_1,\ldots,i_k}^{i,J_1,\ldots,J_k}$.
\end{proposition}

\begin{proof}
	The result is clearly true for $\vert I \vert = 1$ by Proposition \ref{prop:ZbX2first}.  The result for $\vert I \vert \geq 2$ then follows from a straightforward induction argument after noting that, for any multi index $J_i$,
	\[
		\Zb \left( Z^{J_i} \left( \Theta^{i_i} \right) \right)
		=
		Z^{L} \left( \Theta^{i_1} \right),
	\]
	where $\vert L \vert = \vert J_i \vert + 1$, and also noting that
	\begin{align*}
		\overline{\Omega}_{ij} \left( \frac{t}{t-t_0} \right)
		&
		=
		0,
		&
		\overline{\Omega}_{ij} \left( \frac{x^k}{t-t_0} \right)
		=
		\frac{x^i}{t-t_0} \delta_j^k
		-
		\frac{x^j}{t-t_0} \delta_i^k,
		\\
		\overline{B}_{i} \left( \frac{t}{t-t_0} \right)
		&
		=
		\frac{x^i}{t-t_0} - \frac{x^i}{t-t_0} \frac{t}{t-t_0},
		&
		\overline{B}_{i} \left( \frac{x^k}{t-t_0} \right)
		=
		\frac{t}{t-t_0} \delta_i^k
		-
		\frac{x^i}{t-t_0} \frac{x^k}{t-t_0},
		\\
		\overline{S} \left( \frac{t}{t-t_0} \right)
		&
		=
		\frac{t}{t-t_0} - \left( \frac{t}{t-t_0} \right)^2,
		&
		\overline{S} \left( \frac{x^k}{t-t_0} \right)
		=
		\frac{x^k}{t-t_0}
		-
		\frac{x^k}{t-t_0} \frac{t}{t-t_0},
	\end{align*}
	and
	\begin{align}
		\overline{\Omega}_{ij} \left( \ph^k \right)
		&
		=
		\left( \ph^i + \Theta^i \right) \delta_j^k
		-
		\left( \ph^j + \Theta^j \right) \delta_i^k
		-
		\Omega_{ij} \left( \Theta^k \right),
		\label{eq:Omegap}
		\\
		\overline{B}_{i} \left( \ph^k \right)
		&
		=
		\left( \delta_i^k - \ph^i \ph^k \right)
		-
		\Theta^i \ph^k
		-
		\frac{x^i}{t-t_0} \Theta^k
		-
		B_{i} \left( \Theta^k \right),
		\label{eq:Bp}
		\\
		\overline{S} \left( \ph^k \right)
		&
		=
		\Theta^k
		-
		\frac{t}{t-t_0} \Theta^k
		-
		S \left( \Theta^k \right),
		\label{eq:Sp}
	\end{align}
	where the equalities \eqref{eq:Omega0Phi}, \eqref{eq:B0Phi}, \eqref{eq:S0Phi} have been used.
\end{proof}

\begin{proposition} \label{prop:Zbxtph}
	For any multi index $I$, there exist functions $\Lambda$ such that
	\begin{align*}
		\Zb^I \left( \frac{x^i}{t} - \ph^i \right)
		=
		&
		\left( \frac{x^j}{t} - \ph^j \right)
		\Lambda^i_{I,j} \left( \frac{x}{t}, \ph \right)
		+
		\sum_{k=1}^{\vert I\vert} \sum_{\vert J_1\vert + \ldots + \vert J_k \vert \leq \vert I \vert -k+1}
		\!\!\!\!\!\!
		\tilde{\Lambda}_{I,i_1,\ldots,i_k}^{i,J_1,\ldots,J_k} \left( \ph,t,x \right)
		Z^{J_1} \left( \Theta^{i_1} \right) \ldots Z^{J_k} \left( \Theta^{i_k} \right)
		,
	\end{align*}
	for $i=1,2,3$, where
	\[
		\tilde{\Lambda}_{I,i_1,\ldots,i_k}^{i,J_1,\ldots,J_k} \left( \ph,t,x \right)
		=
		\Lambda_{I,i_1,\ldots,i_k}^{i,J_1,\ldots,J_k} \left( \ph, \frac{x}{t}, \frac{x}{t-t_0}, \frac{t}{t-t_0} \right),
	\]
	for some smooth functions $\Lambda_{I,i_1,\ldots,i_k}^{i,J_1,\ldots,J_k}$.
\end{proposition}

\begin{proof}
	Note that,
	\begin{align*}
		\overline{\Omega}_{ij} \left( \frac{x^k}{t} - \ph^k \right)
		=
		&
		\left( \frac{x^i}{t} - \ph^i \right) \delta_j^k
		-
		\left( \frac{x^j}{t} - \ph^j \right) \delta_i^k
		-
		\mathring{\Omega}_{ij}^k
		\\
		\overline{B}_{i} \left( \frac{x^k}{t} - \ph^k \right)
		=
		&
		-
		\left( \frac{x^i}{t} - \ph^i \right) \ph^k
		-
		\frac{x^i}{t} \left( \frac{x^j}{t} - \ph^j \right)
		-
		\mathring{B}_i^k
		\\
		\overline{S} \left( \frac{x^k}{t} - \ph^k \right)
		=
		&
		\mathring{S}^k,
	\end{align*}
	and so, inserting the equalities \eqref{eq:Omega0Phi}, \eqref{eq:B0Phi}, \eqref{eq:S0Phi}, the result is clearly true for $\vert I \vert = 1$.  The result for $\vert I \vert \geq 2$ follows from a straightforward induction, as in the proof of Proposition \ref{prop:ZbX2}, using now also the fact that,
	\begin{equation} \label{eq:Zbxt}
		\overline{\Omega}_{ij} \left( \frac{x^k}{t} \right)
		=
		\frac{x^i}{t} \delta_j^k - \frac{x^j}{t} \delta_i^k,
		\qquad
		\overline{B}_{i} \left( \frac{x^k}{t} \right)
		=
		\delta_i^k - \frac{x^i}{t} \frac{x^k}{t},
		\qquad
		\overline{S} \left( \frac{x^k}{t} \right)
		=
		0.
	\end{equation}
\end{proof}

\begin{proposition} \label{prop:Zbph}
	For any multi index $I$, there exist smooth functions $\Lambda$ such that
	\begin{align*}
		\Zb^I \left( \ph^i \right)
		=
		\Lambda^i_{I} \left( \ph \right)
		&
		+
		\sum_{k=1}^{\vert I\vert} \sum_{\vert J_1\vert + \ldots + \vert J_k \vert \leq \vert I \vert -k+1}
		\tilde{\Lambda}_{I,i_1,\ldots,i_k}^{i,J_1,\ldots,J_k} \left( \ph,t,x \right)
		Z^{J_1} \left( \Theta^{i_1} \right) \ldots Z^{J_k} \left( \Theta^{i_k} \right)
		,
	\end{align*}
	for $i=1,2,3$, where
	\[
		\tilde{\Lambda}_{I,i_1,\ldots,i_k}^{i,J_1,\ldots,J_k} \left( \ph,t,x \right)
		=
		\Lambda_{I,i_1,\ldots,i_k}^{i,J_1,\ldots,J_k} \left( \ph, \frac{x}{t-t_0}, \frac{t}{t-t_0} \right),
	\]
	for some smooth functions $\Lambda_{I,i_1,\ldots,i_k}^{i,J_1,\ldots,J_k}$.
\end{proposition}

\begin{proof}
	The result for $\vert I \vert = 1$ clearly follows from the equalities \eqref{eq:Omegap}, \eqref{eq:Bp}, \eqref{eq:Sp}.  The proof for $\vert I \vert \geq 2$ follows from a straightforward induction argument, as in the proof of Proposition \ref{prop:ZbX2}.
\end{proof}

\subsection{Preliminary estimates for repeated vector fields applied to approximations to geodesics}

\begin{proposition} \label{prop:Phi}
	Suppose $t\geq t_0 + 1$, $\vert x \vert \leq c t$ and the bounds \eqref{eq:Gammasuppf} hold.  Then, for $\vert I \vert \leq N$,
	\[
		\left\vert Z^I \Theta(t,x) \right\vert
		\lesssim
		\sum_{\vert J\vert \leq \vert I \vert -1} t \left\vert (Z^J \Gamma)(t,x) \right\vert
		+
		\sum_{\vert J\vert \leq \vert I \vert}
		\frac{1}{t-t_0} \int_{t_0}^t (s'- t_0) \left\vert (Z^J \Gamma) \left(s', s'\frac{x}{t} \right) \right\vert ds',
	\]
	where $Z^I$ is a product of $\vert I \vert$ of the vector fields $\Omega_{ij}, B_i, S$.  Moreover, if $\vert I \vert \leq \left\lfloor \frac{N}{2} \right\rfloor + 2$, then
	\[
		\left\vert Z^I \Theta(t,x) \right\vert \leq \frac{C}{t^a}.
	\]
\end{proposition}

\begin{proof}
	Note that,
	\[
		Z \left( \Theta^i(t,x) \right)
		=
		- \frac{Z(t)}{t-t_0} \Theta^i(t,x)
		+
		Z(t) \Gh^i \left(t,x,\frac{x}{t} \right)
		+
		\frac{1}{t-t_0} \int_{t_0}^t (s'-t_0) Z \left( \Gh^i \left( s', s'\frac{x}{t}, \frac{x}{t} \right) \right) ds',
	\]
	and that
	\[
		\Omega_{ij} \left( \Gamma^{\mu}_{\alpha \beta} \left( s, s\frac{x}{t} \right) \right)
		=
		s \frac{x^i}{t} \left( \partial_{x^j} \Gamma^{\mu}_{\alpha \beta} \right) \left( s, s\frac{x}{t} \right)
		-
		s \frac{x^j}{t} \left( \partial_{x^i} \Gamma^{\mu}_{\alpha \beta} \right) \left( s, s\frac{x}{t} \right)
		=
		\left( \Omega_{ij} \Gamma^{\mu}_{\alpha \beta} \right) \left( s, s\frac{x}{t} \right),
	\]
	\[
		B_i \left( \Gamma^{\mu}_{\alpha \beta} \left( s, s\frac{x}{t} \right) \right)
		=
		s \left( \partial_{x^i} \Gamma^{\mu}_{\alpha \beta} \right) \left( s, s\frac{x}{t} \right)
		-
		s \frac{x^i}{t} \frac{x^k}{t} \left( \partial_{x^k} \Gamma^{\mu}_{\alpha \beta} \right) \left( s, s\frac{x}{t} \right)
		=
		\left( B_{i} \Gamma^{\mu}_{\alpha \beta} \right) \left( s, s\frac{x}{t} \right)
		-
		\frac{x^i}{t} \left( S \Gamma^{\mu}_{\alpha \beta} \right) \left( s, s\frac{x}{t} \right),
	\]
	and
	\[
		S \left( \Gamma^{\mu}_{\alpha \beta} \left( s, s\frac{x}{t} \right) \right)
		=
		0.
	\]
	Using the equalities \eqref{eq:Zbxt}, the result for $\vert I \vert = 1$ clearly follows from the $L^{\infty}$ bounds \eqref{eq:Gammasuppf} for $\Gamma^{\mu}_{\alpha \beta}$.  The proof for $\vert I \vert \geq 2$ is a straightforward induction.
\end{proof}

\begin{proposition} \label{prop:Zbmisc}
	Suppose $t\geq t_0 + 1$, $\vert x \vert \leq c t$, $(t,x,\ph) \in \supp(f)$ and the bounds \eqref{eq:Gammasuppf} hold, then, for $\vert I \vert \leq N$,
	\begin{align}
		\begin{split}
		\left\vert \Zb^I \left( X_2(t_0)^i \right) \right\vert
		\lesssim & \ 1
		+\sum_{\vert J\vert \leq \vert I \vert -1} t \left\vert (Z^J \Gamma)(t,x) \right\vert
		+\sum_{\vert J\vert \leq \vert I \vert}
		\frac{1}{t-t_0} \int_{t_0}^t (s'- t_0) \left\vert (Z^J \Gamma) \left(s', s'\frac{x}{t} \right) \right\vert ds'
		\end{split} \label{eq:ZbX2}\\
		\begin{split}
		\left\vert \Zb^I \left( \frac{x^i}{t} - \ph^i \right) \right\vert
		\lesssim &\ \frac{1}{t^a}
		+\sum_{\vert J\vert \leq \vert I \vert -1} t \left\vert (Z^J \Gamma)(t,x) \right\vert
		+\sum_{\vert J\vert \leq \vert I \vert}
		\frac{1}{t-t_0} \int_{t_0}^t (s'- t_0) \left\vert (Z^J \Gamma) \left(s', s'\frac{x}{t} \right) \right\vert ds'
		\end{split} \label{eq:Zbxtph}\\
		\begin{split}
		\left\vert \Zb^I \left( \ph^i \right) \right\vert
		\lesssim &\ 1
		+\sum_{\vert J\vert \leq \vert I \vert -1} t \left\vert (Z^J \Gamma)(t,x) \right\vert
		+\sum_{\vert J\vert \leq \vert I \vert}
		\frac{1}{t-t_0} \int_{t_0}^t (s'- t_0) \left\vert (Z^J \Gamma) \left(s', s'\frac{x}{t} \right) \right\vert ds',
		\end{split} \label{eq:Zbph}
	\end{align}
	and
	\begin{align}
		\left\vert \Zb^I \left( \int_{t_0}^s (s' - t_0) \Gh^i \left( s', s'\frac{x}{t}, \frac{x}{t} \right) ds' \right) \right\vert
		\lesssim &\
		\sum_{\vert J\vert \leq \vert I \vert}
		\int_{t_0}^s (s'- t_0) \left\vert (Z^J \Gamma) \left(s', s'\frac{x}{t} \right) \right\vert ds'
		\label{eq:Zbt0s}\\
		\left\vert \Zb^I \left( \int_{s}^t \Gh^i \left( s', s'\frac{x}{t}, \frac{x}{t} \right) ds' \right) \right\vert
		\lesssim &\
		\sum_{\vert J\vert \leq \vert I \vert -1} t \left\vert (Z^J\Gamma)(t,x)\right\vert
		+\sum_{\vert J\vert \leq \vert I \vert}
		\int_{s}^t \left\vert (Z^J \Gamma) \left(s', s'\frac{x}{t} \right) \right\vert ds',\label{eq:Zbst}
	\end{align}
	for $i=1,2,3$.  In particular, for $\vert I \vert \leq \left\lfloor \frac{N}{2} \right\rfloor + 2$,
	\[
		\left\vert \Zb^I \left( X_2(t_0)^i \right) \right\vert
		+
		\left\vert \Zb^I \left( \ph^i \right) \right\vert
		\leq
		C,
		\qquad
		\left\vert \Zb^I \left( \frac{x^i}{t} - \ph^i \right) \right\vert
		\leq
		\frac{C}{t^a},
	\]
	for $i=1,2,3$.
\end{proposition}

\begin{proof}
	The bounds \eqref{eq:ZbX2}, \eqref{eq:Zbxtph} and \eqref{eq:Zbph} follow directly from Propositions \ref{prop:ZbX2}, \ref{prop:Zbxtph}, \ref{prop:Zbph} and Proposition \ref{prop:Phi}, using also Proposition \ref{prop:sec22} and Corollary \ref{cor:X2bound}.  The bounds \eqref{eq:Zbt0s} and \eqref{eq:Zbst} can be shown as in the proof of Proposition \ref{prop:Phi}.
\end{proof}

\begin{corollary} \label{cor:ZbX2}
	Suppose $t\geq t_0 + 1$, $\vert x \vert \leq c t$, $(t,x,\ph) \in \supp(f)$ and the bounds \eqref{eq:Gammasuppf} hold, then
	\begin{align} \label{eq:ZbX2bound}
		\left\vert \Zb^I \left( \frac{X_2^i(s,t,x,\ph)}{s} \right) \right\vert
		\lesssim
		1
		+
		\sum_{\vert J \vert \leq \vert I \vert -1} \!\!\!\!\! t \left\vert (Z^J \Gamma)(t,x) \right\vert
		+
		\sum_{\vert J \vert \leq \vert I \vert }
		\int_{t_0}^t \left\vert (Z^J \Gamma) \left( s', s' \frac{x}{t} \right) \right\vert \frac{s'}{s'+s} \, ds',
	\end{align}
	and
	\begin{align}
		\left\vert \Zb^I \left( \frac{X_2^i(s,t,x,\ph)}{s} - \frac{x^i}{t} \right) \right\vert
		\lesssim
		\frac{1}{s^a}
		+\!\!\!\!\!
		\sum_{\vert J \vert \leq \vert I \vert -1} \!\!\!\!\! t \left\vert (Z^J \Gamma)(t,x) \right\vert
		+
		\sum_{\vert J \vert \leq \vert I \vert }
		\int_{t_0}^t \left\vert (Z^J \Gamma) \left( s', s' \frac{x}{t} \right) \right\vert \frac{s'}{s'+s}\, ds',
		\label{eq:X2minusX1highder}
	\end{align}
	for all $t_0\leq s \leq t$, $i=1,2,3$.  In particular,  for $\vert I \vert \leq \left\lfloor \frac{N}{2} \right\rfloor + 2$,
	\begin{equation}\label{eq:X2minusX1}
		\left\vert \Zb^I \left( \frac{X_2^i(s,t,x,\ph)}{s} - \frac{x^i}{t} \right) \right\vert
		\lesssim
		\frac{1}{s^a}.
	\end{equation}
	Moreover
	\begin{equation}\label{eq:P2minusP1highder}
		\left\vert \Zb^I \Big( \frac{x^i}{t} - \widehat{P}_2^i(s,t,x,\ph) \Big) \right\vert
		\lesssim
		\frac{1}{t^a}
		+
		\!\!\!\!\!\sum_{\vert J\vert \leq \vert I \vert -1}\!\!\!\!\! t \left\vert (Z^J \Gamma)(t,x) \right\vert
		+
		\!\!\sum_{\vert J\vert \leq \vert I \vert}
		 \int_{t_0}^t \left\vert (Z^J \Gamma) \big(s', s'\frac{x}{t} \big) \right\vert \frac{s'}{s'+s} ds'
\end{equation}
In particular,  for $\vert I \vert \leq \left\lfloor \frac{N}{2} \right\rfloor + 2$,
\begin{equation}\label{eq:P2minusP1}
\left\vert \Zb^I \left( \widehat{P}_2^i(s,t,x,\ph) - \frac{x^i}{t} \right) \right\vert
		\lesssim
		\frac{1}{s^a}.
\end{equation}
\end{corollary}

\begin{proof}
	After writing
	\begin{align*}
		X_2^i(s,t,x,\ph)
		=
		&
		\
		X_2^i(t_0,t,x,\ph)
		+
		\left(
		X_2^i(s,t,x,\ph)
		-
		X_2^i(t_0,t,x,\ph)
		\right)
		\\
		=
		&
		\
		X_2^i(t_0,t,x,\ph)
		+
		(s-t_0) \ph^i
		+
		(s-t_0) \int_s^t \Gh^i \left( s', s' \frac{x}{t}, \frac{x}{t} \right) ds'
		+
		\int_{t_0}^s (s' - t_0) \Gh^i \left( s', s' \frac{x}{t}, \frac{x}{t} \right) ds',
	\end{align*}
	Proposition \ref{prop:Zbmisc} implies that
	\begin{align*}
		\left\vert
		\Zb^I \left( X_2^i(s,t,x,\ph) \right)
		\right\vert
		\leq
		\
		&
		C s
		\Bigg[
		1
		+
		\sum_{\vert J \vert \leq \vert I \vert -1} t \left\vert (Z^J \Gamma)(t,x) \right\vert
		+
		\sum_{\vert J \vert \leq \vert I \vert }
		\frac{1}{t-t_0}
		\int_{t_0}^t (s'-t_0) \left\vert (Z^J \Gamma) \left( s', s' \frac{x}{t} \right) \right\vert ds'
		\Bigg]
		\\
		&
		+
		Cs
		\sum_{\vert J \vert \leq \vert I \vert }
		\int_{s}^t \left\vert (Z^J \Gamma) \left( s', s' \frac{x}{t} \right) \right\vert ds'
		+
		C
		\sum_{\vert J \vert \leq \vert I \vert }
		\int_{t_0}^s (s'-t_0) \left\vert (Z^J \Gamma) \left( s', s' \frac{x}{t} \right) \right\vert ds'.
	\end{align*}
	The proof of the bound \eqref{eq:ZbX2bound} follows after dividing by $s$ and using the fact that
	\[
		\int_{t_0}^t \frac{s'-t_0}{t-t_0} \left\vert (Z^J \Gamma) \left( s', s' \frac{x}{t} \right) \right\vert ds'
		\lesssim
		\int_s^t \left\vert (Z^J \Gamma) \left( s', s' \frac{x}{t} \right) \right\vert ds'
		+
		\int_{t_0}^s \frac{s'}{s} \left\vert (Z^J \Gamma) \left( s', s' \frac{x}{t} \right) \right\vert ds',
	\]
	and $1 \lesssim \frac{s'}{s'+s} \lesssim 1$ if $s\leq s'$, and $\frac{s'}{s} \lesssim \frac{s'}{s'+s} \lesssim \frac{s'}{s}$.  The proof of \eqref{eq:X2minusX1highder} follows similarly by writing
	\begin{align*}
		\frac{X_2^i(s,t,x,\ph)}{s} - \frac{x^i}{t}
		=
		\
		&
		\left( \ph^i - \frac{x^i}{t} \right)
		-
		\frac{t_0}{s} \ph^i
		+
		\frac{X_2^i(t_0,t,x,\ph)}{s}
		\\
		&
		+
		\frac{(s-t_0)}{s} \int_s^t \Gh^i \left( s', s' \frac{x}{t}, \frac{x}{t} \right) ds'
		+
		\frac{1}{s} \int_{t_0}^s (s' - t_0) \Gh^i \left( s', s' \frac{x}{t}, \frac{x}{t} \right) ds',
	\end{align*}
	and using the bound \eqref{eq:Zbxtph}.  The bound \eqref{eq:P2minusP1highder} similarly follows from Proposition \ref{prop:Zbmisc} after writing,
	\[
		\frac{x^i}{t} - \Ph^i_2(s,t,x,\ph)
		=
		\frac{x^i}{t} - \ph^i - \int_s^t \Gh^i \left( s', s'\frac{x}{t}, \frac{x}{t} \right) ds'.
	\]
	The lower order estimates \eqref{eq:X2minusX1} and \eqref{eq:P2minusP1} follow from rewriting,
	\[
		\int_{t_0}^t \left\vert (Z^J \Gamma) \big(s', s'\frac{x}{t} \big) \right\vert \frac{s'}{s'+s} ds'
		\lesssim
		\int_s^t \left\vert (Z^J \Gamma) \left( s', s' \frac{x}{t} \right) \right\vert ds'
		+
		\int_{t_0}^s \frac{s'}{s} \left\vert (Z^J \Gamma) \left( s', s' \frac{x}{t} \right) \right\vert ds',
	\]
	and using the pointwise bounds \eqref{eq:Gammasuppf} for lower order derivatives of $\Gamma$.
\end{proof}

\subsection{Schematic notation and repeated vector fields applied to differences}

To make long expressions more concise, the $s$ dependence of many quantities is surpressed throughout this section.

The proofs of Proposition \ref{prop:ZXbarPbarlo} and Proposition \ref{prop:ZXbarPbarho} below follow from applying vector fields to the system \eqref{eq:Xbar}--\eqref{eq:Pbar}.  It is therefore necessary to estimate vector fields applied to the difference
\begin{equation}\label{eq:GammaX1P1XPdifference}
	\widehat{\Gamma}(X,\widehat{P})-\widehat{\Gamma}(X_1,\widehat{P}_1)
	=
	\Gamma(X)\cdot \Lambda(\widehat{P})-\Gamma(X_1)\cdot \Lambda(\widehat{P}_1)
	=
	(\Gamma(X)-\Gamma(X_1))\cdot \Lambda(\widehat{P})+\Gamma(X_1)\cdot(\Lambda(\widehat{P})-\Lambda(\widehat{P}_1)),
\end{equation}
which appears on the right hand side of equation \eqref{eq:Pbar}, where $X_1(s)=sx/t$ and $\widehat{P}_1= \frac{dX_1}{ds} = x/t$, and $\Lambda$ is defined in \eqref{eq:GammaH}.  The following result is straightforward to show.
\begin{lemma}\label{lem:chainruleidentity}
	Given $Y_1,\dots,Y_k,Y\in \mathbb{R}^3$ and $F:\mathbb{R}^3\to \mathbb{R}$, let $Y_1\cdots Y_k \cdot (\partial^k F)(Y)$, denote the sum of $Y_1^{j_1}\cdots Y_k^{j_k} \cdot (\partial_{j_1}
\cdots \partial_{j_k}F)(Y)$ over all components $1\leq j_i\leq n$ for  $i=1,\dots, k$.
	We have
	\begin{multline} \label{eq:ZbJFXFX1}
		\Zb^J (F(X)-F(X_1)) =\!\!\!\!\sum_{\substack{ k<|J|/2,\, J_1+\dots+J_k=J,\,
		1\leq |J_1|\leq \dots \leq |J_k|}\!\!\!\!\!\!\!\!\!\!\!\!\!\!\!\!\!\!\!\!\!\!}
		c_{J_1\dots J_k}\,\,\,\, \Zb^{J_1} X_1\cdots \Zb^{J_{k}}X_1\cdot\big((\partial^k F)(X)-(\partial^k F)(X_1)\big)
		\\
		+
		\!\!\!\!\!\!\!\sum_{\substack{k<|J|/2,\, J_1+\dots+J_k=J,\,
		1\leq |J_1|\leq \dots \leq |J_k|\!\!\!\!\!\!\!\!\!\!\!\!\! }\!\!\!\!\!\!\!}
		c_{J_1\dots J_k}\sum_{1\leq \ell\leq k} \Zb^{J_1} X\cdots \Zb^{J_{\ell-1}}X\,\,  \Zb^{J_\ell}(X-X_1)
		\,\, \Zb^{J_{\ell+1}} X_1\cdots \Zb^{J_{k}}X_1\cdot(\partial^k F)(X),
		\\
+\!\!\!\!\sum_{\substack{k\geq |J|/2,\, J_1+\dots+J_k=J,\,
		1\leq |J_1|\leq \dots \leq |J_k|\!\!\!\!\!\!\!\!\!\!\!\!\!} \!\!\!\!\!\!\!\!\!\!\!}
		c^\prime_{J_1\dots J_k}\,\,\,\, \Zb^{J_1} X\cdots \Zb^{J_{k}}X\cdot\big((\partial^k F)(X)-(\partial^k F)(X_1)\big)
		\\
		+
		\!\!\!\!\!\!\!\sum_{\substack{k\geq |J|/2,\, J_1+\dots+J_k=J,\,
		1\leq |J_1|\leq \dots \leq |J_k|\!\!\!\!\!\!\!\!\!\!\!\!\! }\!\!\!}
		c^\prime_{J_1\dots J_k}\sum_{1\leq \ell\leq k} \Zb^{J_1} X_1\cdots \Zb^{J_{\ell-1}}X_1\,\,  \Zb^{J_\ell}(X-X_1)
		\,\, \Zb^{J_{\ell+1}} X\cdots \Zb^{J_{k}}X\cdot(\partial^k F)(X_1)
	\end{multline}
	where the sums are over all possible partitions of the multi index $J$ into
	into nonempty sub indices $J_1$ to $J_k$.
\end{lemma}
\begin{proof}
	First one differentiates to get a sum of terms of the form
	\[
		\Zb^{I_1} X \cdots \Zb^{I_k} X \cdot (\partial^k F)(X) -\Zb^{I_1} X_1 \cdots \Zb^{I_k} X_1 \cdot (\partial^k F)(X_1),
	\]
	and then one makes a different decomposition depending on the size of $k$.
	Then one proceeds by organizing them in order so $|I_1|$ is smallest.
	If $k<|J|/2$ one writes
	\begin{multline*}
		\Zb^{I_1} X \cdots \Zb^{I_k} X \cdot (\partial^k F)(X) -\Zb^{I_1} X_1 \cdots \Zb^{I_k} X_1 \cdot (\partial^k F)(X_1)
		\\
		= \Zb^{I_1} X \cdots \Zb^{I_k} X \cdot ((\partial^k F)(X)- (\partial^k F)(X_1))
		+(\Zb^{I_1} X \cdots \Zb^{I_k} X -\Zb^{I_1} X_1 \cdots \Zb^{I_k} X_1) \cdot (\partial^k F)(X_1)
	\end{multline*}
	and replaces them one by one
	\begin{multline*}
		\Zb^{I_1} X \cdots \Zb^{I_k} X = \Zb^{I_1} X\cdots \Zb^{I_k} (X-X_1)+
		\Zb^{I_1} X \cdots \Zb^{I_k} X_1\\
		= \Zb^{I_1} X\cdots \Zb^{I_k} (X-X_1)+\Zb^{I_1} X \cdots \Zb^{I_{k-1}} (X-X_1)\Zb^{I_k} X_1
		+\Zb^{I_1} X \cdots \Zb^{I_{k-1}}X_1 \Zb^{I_k} X_1=...
	\end{multline*}
	This produces the first two sums with $k<|J|/2$. For $k\geq |J|/2$ one simply does the same
	thing but with $X_1$ and $X$ interchanged.
\end{proof}

Note that, for each term in the equality \eqref{eq:ZbJFXFX1}, $|J_i|\leq |J|/2$ if $i\neq k$ and therefore, in the applications of Lemma \ref{lem:chainruleidentity} below, the factors $\Zb^{J_i}\! X$
can be estimated by the their $L^\infty$ norms using induction by previous estimates.  Given that $X_1$ is known the expression \eqref{eq:ZbJFXFX1} can be thought of as linear in the unknown $X-X_1$.
However when we prove $L^2$ estimates for high derivatives we also have to take into account
how $Z^{J_k} X_1$ depends on high derivatives of $\Gamma$.  Either $k\leq |J|/2 $ is small, in which case we can use $L^\infty$ estimates for $\partial^k F$ and $\partial^{k+1} F$,
or $k\geq |J|/2$ is large, and as a result $|J_i|\leq |J|/2$ are small for all $i$, and
we can use $L^\infty$ estimates for all the other factors. Applying Lemma \ref{lem:chainruleidentity} we get the following estimates:

\begin{lemma}\label{lem:ZIHPminusHP1}
	Recall the function $\Lambda$ from \eqref{eq:GammaX1P1XPdifference}.  Suppose $|\Zb^{K} \widehat{P}| \leq C$ for $ |K|\leq |L|/2$. Then
	\begin{align}
		\big|\Zb^L \big(\Lambda(\widehat{P})-\Lambda(\widehat{P}_1)\big)\big|
		&\lesssim \sum_{|M|\leq |L|}
		\big|\Zb^{M} (\widehat{P}-\widehat{P}_1)\big|
		,
		\label{eq:ZbHPh1}
		\\
		\big|\Zb^L \big(\Lambda(\widehat{P}_1)\big)\big|&
		\lesssim 1.
		\label{eq:ZbHPh2}
	\end{align}
\end{lemma}

\begin{proof}
	The bound \eqref{eq:ZbHPh1} follows from an appropriate version of Lemma \ref{lem:chainruleidentity} after noting that $\vert \Zb^I ( \Ph_1) \vert \lesssim 1$ for any $I$, using the form of the vector fields $\Zb$ and the fact that $\Ph_1(t,x) = \frac{x}{t}$.  The term $(\partial^k \Lambda)(\Ph) - (\partial^k \Lambda)(\Ph_1)$ is estimated by
	\[
		\left\vert (\partial^k \Lambda)(\Ph) - (\partial^k \Lambda)(\widehat{P}_1) \right\vert
		\lesssim
		\left\vert \Ph - \widehat{P}_1 \right\vert,
	\]
	since $\Lambda$ is smooth.  The bound \eqref{eq:ZbHPh2} is even simpler.
\end{proof}

\begin{lemma}\label{lem:ZIGammaXminuGammaX1}
	Suppose $|\Zb^{K} X| \leq Cs $ for $ |K|\leq |J|/2$. Then
	\begin{align*}
		\big|\Zb^J \big(\Gamma(X)-\Gamma(X_1)\big)\big|
		\lesssim
		\
		&
		\sum_{|K|\leq |J|/2\!\!\!\!\!\!\!\!\!\!\!\!\!\!\!\!}\big|(Z^{K} \Gamma)(X)\big|
		\sum_{|M|\leq |J|\!\!\!\!\!\!\!\!\!\!}
		\frac{|\Zb^M (X-X_1)|}{s}
		+
		\sum_{|M|\leq |J|\!\!\!\!\!\!\!\!}\big|(Z^{M} \Gamma)(X)-(Z^{M} \Gamma)(X_1)\big|
		\\
		&
		+
		\sum_{|K|\leq |J|/2\!\!\!\!\!\!\!\!\!\!\!\!\!\!\!}
		\frac{|\Zb^{K} (X-X_1)|}{s}\sum_{|M|\leq |J|\!\!\!\!\!\!\!\!\!}\big|(Z^{M} \Gamma)(X_1)\big|
	\end{align*}
	and
	\begin{equation*}
		\big|\Zb^J\big(\Gamma(X_1)\big)\big|
		\lesssim
		\sum_{|M|\leq |J|\!\!\!\!\!\!\!\!\!\!\!\!\!\!\!\!}\big|(Z^{M} \Gamma)(X_1)\big|
	\end{equation*}
\end{lemma}

\begin{proof}
	The Lemma is again a straightforward application of Lemma \ref{lem:chainruleidentity}, noting again that $\left\vert \frac{\Zb X_1}{s} \right\vert \lesssim 1$ for any $I$ since $X_1(s,t,x) = s \frac{x}{t}$.  Note that
	\[
		\partial^k \Gamma(t,x)=t^{-k} \sum_{|K|\leq k} A_K(x/t) (Z^K \Gamma)(t,x),
	\]
	for some homogeneous functions $A$ that are smooth when $|x|/t\leq c<1$, and hence,
	\[
		s^k \left\vert (\partial^k \Gamma)(X) - (\partial^k \Gamma)(X_1) \right\vert
		\lesssim
		\sum_{\vert K \vert \leq k} \left\vert (Z^K \Gamma)(X) - (Z^K \Gamma)(X_1) \right\vert
		+
		\frac{|X-X_1|}{s}\sum_{|K|\leq k}\big|(Z^{M} \Gamma)(X_1)\big|.
	\]
\end{proof}

In the application we will estimate $\Gamma$ with its $L^\infty$ norms for low derivatives:
\begin{lemma}\label{lem:ZIGammaXminuGammaX1Linfty}
	Suppose $|\Zb^{K} X| \leq Cs$ for $ |K|\leq |J|/2$. Then
\begin{align*}
	\big|\Zb^J \big(\Gamma(X)-\Gamma(X_1)\big)\big|
	\lesssim &
	\sum_{|K|\leq |J|/2+1\!\!\!\!\!\!\!\!\!\!\!\!\!\!\!\!\!\!\!\!\!} \big\|(Z^{K} \Gamma)(s,\cdot)\big\|_{L^{\infty}}
	\sum_{|M|\leq 	|J|\!\!\!\!\!\!\!\!\!\!}
	\frac{|\Zb^M (X-X_1)|}{s}
	+\sum_{|M|\leq |J|\!\!\!\!\!\!\!\!}\big|(Z^{M} \Gamma)(X)-(Z^{M} \Gamma)(X_1)\big|
	\\
	&
	+\sum_{|K|\leq |J|/	2\!\!\!\!\!\!\!\!\!\!\!\!\!\!\!}
	\frac{|\Zb^{K} (X-X_1)|}{s}
	\sum_{|M|\leq |J|\!\!\!\!\!\!\!\!\!}\big|(Z^{M} \Gamma)(X_1)\big|,\\
	\big|\Zb^J\big(\Gamma(X_1)\big)\big|
	\lesssim
	&
	\sum_{|M|\leq |J|\!\!\!\!\!\!}\big|(Z^{M} \Gamma)(X_1)\big|.
\end{align*}
\end{lemma}
Instead of applying vector fields to the decomposition \eqref{eq:GammaX1P1XPdifference}
we first differentiate and then apply this decomposition if more derivatives fall on $\Gamma$ but if more fall on $\Lambda$ we apply the decomposition with $\Lambda(\widehat{P})$ interchanged with
$\Lambda(\widehat{P}_1)$ and $\Gamma(X)$ with $\Gamma(X_1)$:
\begin{multline*}
	 \Zb^I\big(\widehat{\Gamma}(X,\widehat{P})-\widehat{\Gamma}(X_1,\widehat{P}_1)\big)
	=
	\!\!\!\!\!\sum_{J+L=I,\, |J|\geq |I|/2, |L|\leq |I|/2} \!\!\!\!\!
	\Zb^J(\Gamma(X)-\Gamma(X_1))\cdot \Zb^L \Lambda(\widehat{P})
	+
	\Zb^J\Gamma(X_1)\cdot \Zb^L(\Lambda(\widehat{P})-	\Lambda(\widehat{P}_1))
	\\
	\sum_{J+L=I,\, |J|<|I|/2,|L|>|I|/2} \Zb^J(\Gamma(X)-\Gamma(X_1))\cdot \Zb^L \Lambda(\widehat{P}_1)
	+
	\Zb^J\Gamma(X)\cdot \Zb^L(\Lambda(\widehat{P})-\Lambda(\widehat{P}_1)).
\end{multline*}
Using this decomposition and the previous lemmas we obtain:

\begin{proposition} \label{prop:GammaXPminusGammaX1P1high}
	Suppose $|\Zb^{K} X|/s + |\Zb^{K} \widehat{P}| \leq C$ for $ |K|\leq |I|/2$. Then
	\begin{multline*}
		 \big|\Zb^I\big(\widehat{\Gamma}(X,\widehat{P})-\widehat{\Gamma}(X_1,\widehat{P}_1)\big)\big|
		\lesssim
		\!\sum_{|K|\leq |I|/2+1\!\!\!\!\!\!\!\!\!\!\!\!}
		\!\big\|(Z^{K} \Gamma)(s,\cdot)\big\|_{L^{\infty}}
		\sum_{|J|\leq |I|\!\!\!\!\!\!}\!\!
		\frac{|\Zb^J \!(X\!\!-\!\!X_1)|\!}{s}+\big|\Zb^{J} \!(\widehat{P}\!-\!\widehat{P}_1)\big|
		\\
		+\sum_{|J|\leq |I|\!\!\!\!\!\!\!\!}\Big(\big|(Z^{J} \Gamma)(X)-(Z^{J} \Gamma)(X_1)\big|
		+\sum_{|K|\leq |I|/2\!\!\!\!\!\!\!\!\!\!\!\!\!\!\!}
		\Big(\frac{|\Zb^{K} \!(X-X_1)|}{s}+\big|\Zb^{K}\! (\widehat{P}-\widehat{P}_1)\big|\Big)
		\big|(Z^{J} \Gamma)(X_1)\big|\Big),
	\end{multline*}
	and
	\begin{equation*}
		\big|\Zb^J\big(\widehat{\Gamma}(X_1,P_1)\big)\big|
		\lesssim
		\sum_{|M|\leq |J|\!\!\!\!\!\!\!\!\!\!}\big|(Z^{M} \Gamma)(X_1)\big|.
	\end{equation*}
\end{proposition}

\begin{proof}
The above decomposition and Lemma \ref{lem:ZIHPminusHP1} give
\begin{align*}
	\Big|\Zb^I \Big(\Gamma(X)\cdot \Lambda(\widehat{P})-\Gamma(X_1)\cdot \Lambda(\widehat{P}_1)\Big)\Big|
	\lesssim
	\
	&
	\sum_{|J|\leq |I|\!\!\!\!\!\!} \big|\Zb^J(\Gamma(X)-\Gamma(X_1))\big|+
	\big|\Zb^J(\Gamma(X_1))| \sum_{|L|\leq |J|/2\!\!\!\!\!\!\!\!\!\!} \big| \Zb^L\big(\widehat{P}-\widehat{P}_1\big)\big|
	\\
	&
	+\sum_{|J|\leq |I|/2\!\!\!}
	\big|\Zb^J(\Gamma(X))|
	\sum_{|L|\leq |I|\!\!\!\!\!\!\!\!}
	\big|\Zb^{L} (\widehat{P}-\widehat{P}_1)\big|
	.
\end{align*}
The first bound then follows from Lemma \ref{lem:ZIGammaXminuGammaX1Linfty}.  The proof of the second bound is straightforward.
\end{proof}

\begin{corollary} \label{cor:GammaXPminusGammaX1P1low}
Suppose that $|\Zb^{K}\! X|/s+|\Zb^{K} \!\widehat{P}|\leq C$, for $ |K|\leq |I|/2$. Then
\begin{align*}
\big|\Zb^I\big(\widehat{\Gamma}(X,\widehat{P})-\widehat{\Gamma}(X_1,\widehat{P}_1)\big)\big|
&\lesssim
\sum_{|K|\leq |I|+1\!\!\!\!\!\!\!\!\!\!\!\!\!}\big\|(Z^{K} \Gamma)(s,\cdot)\big\|_{L^{\infty}}  \sum_{|J|\leq |I|\!\!\!}
\Big(\frac{|\Zb^J (X\!-\!X_1)|\!}{s}+\big|\Zb^{J} (\widehat{P}\!-\!\widehat{P}_1)\big|
\Big),\\
\big|\Zb^I\big(\widehat{\Gamma}(X_1,\widehat{P}_1)\big)\big|
&\lesssim
\sum_{|K|\leq |I|\!\!\!\!\!\!\!\!\!\!\!\!\!}\big\|(Z^{K} \Gamma)(s,\cdot)\big\|_{L^{\infty}}  .
\end{align*}
\end{corollary}

 Apart from controlling vector fields applied to $X_1$ and $\widehat{P}_1$ we also control vector fields applied to $X_2$ and $\widehat{P}_2$, and moreover the differences $X_1-X_2$ and $\widehat{P}_1-\widehat{P}_2$ decay,
see \eqref{eq:X1X2difference} and \eqref{eq:P1Pdifference} and the following propositions.
 Furthermore because of the definition of $X_2$ and $\widehat{P}_2$ in terms of $X_1$ and $\widehat{P}_1$ the differences $\overline{X}=X-X_2$ and $\overline{P}=\widehat{P}-\widehat{P}_2$ are small, see Proposition \ref{prop:sec23}.
It may be tempting to write this as differences with $\widehat{\Gamma}$ evaluated at
$(X_2,P_2)$. However we do not want to involve an estimate of $\Gamma$ applied to $X_2$ so instead we will first differentiate \eqref{eq:GammaX1P1XPdifference} and
use the decomposition $X_1-X=X_1-X_2-\overline{X}$ to the factors that come out
when we differentiated.

\subsection{Parameter derivatives of the equations and vector fields applied to their differences}
\label{subsec:parameterderivatives}
The vector fields applied to the system \eqref{eq:Xbar}--\eqref{eq:Pbar} will be estimated by integrating from the final time $t$ and in order to do this we need to control the final conditions for $\Zb^I\overline{X}$ and $\Zb^I\overline{P}$ at time $t$.
If $\Zb$ only consist of space derivatives these vanish, but some of our vector fields
also have time components. Therefore we need to estimate higher $\frac{d}{ds}$ derivatives of the
system, which can be recast as spacetime derivatives.  Recall that $\newhat{X}(s) = (s,X(s))$.
We have
\begin{equation}
\frac{d\newhat{X}}{ds}=\Ph,\qquad \frac{d\Ph}{ds}=\Gamma(\Xh)\cdot \Lambda(\Ph ).
\end{equation}
The structure of higher order $\frac{d}{ds}$ derivatives is very simple. Either the derivative falls on $\Lambda(\Ph)$, in which case we can substitute the second equation for
$d\Ph\!/ds$ and get another factor of $\Gamma(\Xh)$, or the derivative falls $\Gamma(\Xh)$,
which produces a derivative $\partial \Gamma$. Hence we get the system
\begin{equation}
\frac{d\widehat{X}^{(k)}\!\!\!\!}{ds}=\widehat{P}^{(k)},\qquad
\frac{d\widehat{P}^{(k)\!\!\!\!}}{ds}=\widehat{\Gamma}^{(k)}(\Xh,\Ph ),\quad\text{where}\quad
\widehat{X}^{(k)} =\frac{d^k\widehat{X}}{ds^k},\quad \widehat{P}^{(k)} =\frac{d^k\newhat{P}}{ds^k},
\end{equation}
where
\begin{equation}\label{eq:Gammak}
\widehat{\Gamma}^{(k)}(\Xh,\Ph )={\Gamma}^{(k)}(\Xh)\cdot  \Lambda^{(k)}(\Ph ):=\sum_{\substack{k_1+\cdot+k_m+m=k+1,\\ 0\leq k_1\leq \dots\leq k_m \leq k}\!\!\!\!}
\,\,(\partial^{k_1}\Gamma)(\Xh)\cdots (\partial^{k_m}\Gamma)(\Xh)\cdot
\Lambda_{k_1,\dots,k_m}(\Ph).
\end{equation}
Here the first dot product is schematic notation to be interpreted as dot products of elements $\Gamma^{(k)}$ and $\Lambda^{(k)}$ in some larger dimensional space whose components corresponds to the terms in the sum.
We now also want to take $\frac{d}{ds}$ derivatives of
\begin{equation}
\frac{d\widehat{X}_2}{ds}=\widehat{P}_2,\quad \frac{d\widehat{P}_2}{ds}=\Gamma(\widehat{X}_1)\cdot \Lambda(\widehat{P}_1 ),\quad\text{and}\quad
\frac{d\widehat{X}_1}{ds}=\widehat{P}_1,\quad \frac{d\widehat{P}_1}{ds}=0.
\end{equation}
Then
\begin{equation}\label{eq:timederivativeequations}
\frac{d\overline{X}^{(k)}\!\!\!\!}{ds}=\overline{P}^{(k)}\!\!\!\!,\quad
\frac{d\overline{P}^{(k)\!\!\!\!}}{ds}
=\widehat{\Gamma}^{(k)}(\Xh\!,\Ph)-\widehat{\Gamma}^{(k)}(\Xh_{\!\!1},\Ph_{\!\!1})
+\widehat{\Gamma}^{(k,2)}(\Xh_{\!\!1},\Ph_{\!\!1}),\quad\text{if}\quad
\overline{X}^{(k)}\!\! =\frac{d^k\overline{X}\!}{ds^k},\quad \overline{P}^{(k)} \!\! =\frac{d^k\overline{P}\!}{ds^k},
\end{equation}
where
\begin{equation}\label{eq:Gammak2}
\widehat{\Gamma}^{(k,2)}(\Xh,\Ph )\!={\Gamma}^{(k,2)}(\Xh)\cdot  \Lambda^{(k)}(\Ph )=\sum_{\substack{k_1+\cdot+k_m+m=k+1,\\ 0\leq k_1\leq \dots\leq k_m \leq k,\, m\geq 2}\!\!\!\!}
\,\,(\partial^{k_1}\Gamma)(\Xh)\cdots (\partial^{k_m}\Gamma)(\Xh)\cdot
\Lambda_{k_1,\dots,k_m}(\Ph).
\end{equation}
Here $\Gamma^{(k)}$ satisfies the same estimates as $\partial^k\Gamma$ whereas
$\Gamma^{(k,2)}$ is at least quadratic in $\Gamma$ (note $m\geq 2$ in the summation) and hence satisfies the same estimates as $\partial^{k-1}\Gamma$ multiplied
with $\Gamma$, which decays $s^{-a}$ faster than another derivative. We have
\begin{lemma} \label{lem:ZGammaK}
Suppose that $t|Z^I \Gamma(t,x)|\lesssim 1$ for $|I|\leq (|L|+k)/2$. Then
\begin{align}
t^k |Z^L \Gamma^{(k)}(t,x)|&\lesssim \sum_{|I|\leq |L|+k} |Z^I \Gamma(t,x)|,\\
 t^k |Z^L \Gamma^{(k,2)}(t,x)|&\lesssim \sum_{|M|\leq (|L|+k)/2} t |Z^I \Gamma(t,x)| \sum_{|I|\leq |L|+k} |Z^I \Gamma(t,x)|.
\end{align}
\end{lemma}
We can now apply the estimates from the previous section
 with $\widehat{\Gamma}^{(k)}$ respectively $\widehat{\Gamma}^{(k,2)}$
in place of $\widehat{\Gamma}$.
\begin{proposition}\label{prop:GammaXPminusGammaX1P1hightimeder}
	Suppose $|\Zb^{K} X|/s+|\Zb^{K} \widehat{P}| \leq C$ for $ |K|\leq |L|/2$. 	Set $n=|L|+k$. Then
	\begin{multline*}
		s^k \big|\Zb^L\big(\widehat{\Gamma}^{(k)}(X,\widehat{P})-\widehat{\Gamma}^{(k)}(X_1,\widehat{P}_1)\big)\big|
		\lesssim
		\!\sum_{|K|\leq n/2+1\!\!\!\!\!\!\!\!}\!\big\|(Z^{K} \Gamma)(s,\cdot)\big\|_{L^{\infty}}
		\!\!\sum_{|J|\leq |L| \!\!\!\!\!}\!\!
		\frac{|\Zb^J \!(X\!\!-\!\!X_1)|\!}{s}+\big|\Zb^{J} \!(\widehat{P}\!-\!\widehat{P}_1)\big|
		\\
		+\sum_{|J|\leq n\!\!\!\!\!}\Big(\big|(Z^{J} \Gamma)(X)-(Z^{J} \Gamma)(X_1)\big|
		+\sum_{|K|\leq |L|/2\!\!\!\!\!\!\!\!\!\!\!\!}
		\Big(\frac{|\Zb^{K} \!(X-X_1)|}{s}+\big|\Zb^{K}\! (\widehat{P}-\widehat{P}_1)\big|\Big)
		\big|(Z^{J} \Gamma)(X_1)\big|\Big),
	\end{multline*}
	and
	\begin{equation*}
		s^k \big|\Zb^L\big(\widehat{\Gamma}^{(k,2)}(X_1,P_1)\big)\big|
		\lesssim
		\sum_{|K|\leq n/2\!\!\!\!\!}s\big\|(Z^{K} \Gamma)(s,\cdot)\big\|_{L^{\infty}}^2
		+
		\sum_{|K|\leq n/2\!\!\!\!\!\!}s\big\|(Z^{K} \Gamma)(s,\cdot)\big\|_{L^{\infty}} \sum_{|M|\leq n\!\!\!\!\!\!\!\!\!\!}\big|(Z^{M} \Gamma)(X_1)\big|.
	\end{equation*}
\end{proposition}

The following corollary of Proposition \ref{prop:GammaXPminusGammaX1P1hightimeder} is used at lower orders.

\begin{corollary} \label{cor:GammaXPminusGammaX1P1lowtimeder}
	Suppose that $|\Zb^{K}\! X|/s+|\Zb^{K} \!\widehat{P}|\leq C$, for $ |K|\leq |L|/2$,
	and $|\Zb^{M}\!X_1|/s+|\Zb^{M}\!\widehat{P}_1|\leq C$, for $ |M|\leq |L|$. Set $n=|L|+k$. Then
	\begin{align*}
		s^k \big|\Zb^L\big(\widehat{\Gamma}^{(k)}(X,\widehat{P})-\widehat{\Gamma}^{(k)}(X_1,\widehat{P}_1)\big)\big|
		&\lesssim
		\sum_{|K|\leq n+1\!\!\!\!\!\!}\big\|(Z^{K} \Gamma)(s,\cdot)\big\|_{L^{\infty}} \sum_{|J|\leq |L|\!\!\!}
		\Big(\frac{|\Zb^J (X\!-\!X_1)|\!}{s}+\big|\Zb^{J} (\widehat{P}\!-\!\widehat{P}_1)\big|
		\Big),\\
		s^k\big|\Zb^L\big(\widehat{\Gamma}^{(k,2)}(X_1,\widehat{P}_1)\big)\big|
		&\lesssim
		\sum_{|K|\leq n\!\!\!}s\big\|(Z^{K} \Gamma)(s,\cdot)\big\|_{L^{\infty}}^2 .
	\end{align*}
\end{corollary}

\subsection{The final conditions}
	Note that, for any $Y(s,t,x,\ph)$,
	\begin{equation} \label{eq:Yid10}
		\Zb \left( Y(t,t,x,\ph) \right)
		=
		\Zb(t) \frac{dY}{ds}(t,t,x,\ph) + (\Zb Y)(t,t,x,\ph),
\quad\text{where}\quad (\Zb Y)(t,t,x,\ph)
		=
		\Zb \left( Y(s,t,x,\ph) \right) \vert_{s=t}.
	\end{equation}
	Repeated application \eqref{eq:Yid10} inductively implies that
	\begin{align} \label{eq:Yidnew0}
		\begin{split}
		&
		(\Zb^I Y)(t,t,x,\ph)
		=
		\Zb^I \left( Y(t,t,x,\ph) \right)
		+
		\!\!\!\sum_{\substack{
		J_1 + \dots +J_{k+1} +J =  I ,\,
		|J_i|\geq 1, \, k\geq 0
		}\!\!\!\!\!\!\!\!\!\!\!
        \!\!\!\!\!\!\!\!\!\!\!\!\!\!\!\!\!\!\!\!\!\!\!\!\!\!\!\!\!\!\!\!\!\!\!\!\!\!\!}
		C_{I,J_1,\dots,J_{k+1},J}\,\,\,
		\Zb^{J_1\!}(t)\cdots \Zb^{J_{k+1}\!}(t)
	    \big(\Zb^{J}Y^{(k+1)}\big)(t,t,x,\ph),
		\end{split}
	\end{align}	
where $Y^{(k)}=d^k Y/ds^k$ and
$\Zb^{J_i}(t)=Z^{J_i}(t)$ are constants time $t$ or $x^j$, for some $j$.
Applying \eqref{eq:Yidnew0} to $Y=\Xb$ and $Y=\Pb$ noting that $\Xb(t,t,x,\ph)=\Pb(t,t,x,\ph)=0$ gives
\begin{align} \label{eq:Pidnew}
		(\Zb^I \Pb)(t,t,x,\ph)
		&=\!\!\!\sum_{\substack{
		J_1 + \dots +J_{k+1} +L =  I ,\,
		|J_i|\geq 1, \, k\geq 0
		}\!\!\!\!\!\!\!\!\!\!\!
        \!\!\!\!\!\!\!\!\!\!\!\!\!\!\!\!\!\!\!\!\!\!\!\!\!\!\!\!\!\!\!}
		C_{I,J_1,\dots,J_{k+1},L}\,\,\,
		\Zb^{J_1\!}(t)\cdots \Zb^{J_{k+1}\!}(t)
	    \big(\Zb^{L}\Pb^{(k+1)}\big)(t,t,x,\ph),\\
		(\Zb^I \Xb)(t,t,x,\ph)
		&=\!\!\!\sum_{\substack{
		J_1 + \dots +J_{k+2} +L =  I ,\,
		|J_i|\geq 1, \, k\geq 0
		}\!\!\!\!\!\!\!\!\!\!\!
        \!\!\!\!\!\!\!\!\!\!\!\!\!\!\!\!\!\!\!\!\!\!\!\!\!\!\!\!\!\!\!}
		C^\prime_{I,J_1,\dots,J_{k+2},L}\,\,\,
		\Zb^{J_1\!}(t)\cdots \Zb^{J_{k+2}\!}(t)
	    \big(\Zb^{L}\Pb^{(k+1)}\big)(t,t,x,\ph).
		\label{eq:Xidnew}
	\end{align}
where for the proof of the last inequality we also used the first.
Hence
\begin{align*}
		|(\Zb^I \Pb)(t,t,x,\ph)|
		&\lesssim
		\sum_{|L|+k\leq |I|-1,\, k\geq 0\!\!\!\!\!\!\!\!\!\!\!\!\!\!\!\!\!\!}
		t^{k+1}
	    \big|\big(\Zb^{L}\Pb^{(k+1)}\big)(t,t,x,\ph)\big|,\\
		|(\Zb^I \Xb)(t,t,x,\ph)|
		&\lesssim
		\sum_{|L|+k\leq |I|-2,\, k\geq 0\!\!\!\!\!\!\!\!\!\!\!\!\!\!\!\!\!\!}
		t^{k+2}
	    \big|\big(\Zb^{L}\Pb^{(k+1)}\big)(t,t,x,\ph)\big|.
	\end{align*}	
and by \eqref{eq:timederivativeequations}:
\begin{lemma} With  $\widehat{\Gamma}^{(k)}$ and $\widehat{\Gamma}^{(k,2)}$ as in
\eqref{eq:Gammak} and \eqref{eq:Gammak2}
we have
\begin{align*}
		|\Zb^I \Pb|
		&\lesssim
		\sum_{|L|+k\leq |I|-1,\, k\geq 0\!\!\!\!\!\!\!\!\!\!\!\!\!}
		t^{k+1}
	    \big|\Zb^{L}
\big(\widehat{\Gamma}^{(k)}(\Xh\!,\Ph)
-\widehat{\Gamma}^{(k)}(\Xh_{\!\!1},\Ph_{\!\!1})\big)\big|
+	t^{k+1}
	    \big|\Zb^{L}
\big(\widehat{\Gamma}^{(k,2)}(\Xh_{\!\!1},\Ph_{\!\!1})\big)\big|,\quad\text{at}\quad s=t,\\
		|\Zb^I \Xb|
		&\lesssim
		\sum_{|L|+k\leq |I|-2,\, k\geq 0\!\!\!\!\!\!\!\!\!\!\!\!\!}
		t^{k+2}
	    \big|\Zb^{L}
\big(\widehat{\Gamma}^{(k)}(\Xh\!,\Ph)
-\widehat{\Gamma}^{(k)}(\Xh_{\!\!1},\Ph_{\!\!1})\big)\big|
+	t^{k+2}
	    \big|\Zb^{L}
\big(\widehat{\Gamma}^{(k,2)}(\Xh_{\!\!1},\Ph_{\!\!1})\big)\big|,\quad\text{at}\quad s=t,
\end{align*}
where everything is evaluated at $(s,t,x,\ph)$ where $s=t$.
\end{lemma}
We can now apply the estimates from the previous section.
\begin{proposition} \label{prop:ZXZPkiclowder}
	Suppose $t\geq t_0 + 1$, $\vert x \vert \leq c t$, $(t,x,\ph) \in \supp(f)$ and the bounds \eqref{eq:Gammasuppf} hold.  Then, for $\vert I \vert \leq N/2+2$,
	\begin{equation*}
		\big\vert
		\Zb^I \big( \Pb^i(s,t,x,\ph) \big) \big\vert_{s=t}
		\big\vert
		+t^{-1}\big\vert
		\Zb^I \big( \Xb^i(s,t,x,\ph) \big) \big\vert_{s=t}
		\big\vert
		\lesssim
		\varepsilon t^{-2a}.
	\end{equation*}
\end{proposition}
\begin{proof} We will use induction to prove this. Assuming its true for $|I|< m\leq N/2+1$ the assumptions of Corollary \ref{cor:GammaXPminusGammaX1P1lowtimeder} hold at $s=t$ and writing $X-X_1=\Xb+X_2-X_1$
and $\Ph-\widehat{P}_1=\Pb+\widehat{P}_2-\widehat{P}_1$ and using the estimates \eqref{eq:X2minusX1} and \eqref{eq:P2minusP1} for $X_2-X_1$ respectively $\widehat{P}_2-\widehat{P}_1$ respectively we get
\begin{align*}
t^{k+1}\big|\Zb^L\!\big(\widehat{\Gamma}^{(k)}(X,\widehat{P})-\widehat{\Gamma}^{(k)}(X_1,\widehat{P}_1)\big)\big|
&\lesssim\frac{\varepsilon}{t^{a}}\!\!\!\!\sum_{|J|\leq |L|}\!\!\!\!
\Big(\frac{|\Zb^J\!\! (X\!-\!X_1)|\!}{t}+\big|\Zb^{J} \! (\widehat{P}\!-\!\widehat{P}_1)\big|\Big)
\lesssim\frac{\varepsilon}{t^{2a}\!\!}+\frac{\varepsilon}{t^{a}}\!\!\!\!\sum_{|J|\leq |L|}\!\!\!\!
\Big(\frac{|\Zb^J \!\Xb|\!}{s}+\big|\Zb^{J}\! \Pb\big|\Big)\\
t^{k+1}\big|\Zb^L\big(\widehat{\Gamma}^{(k,2)}(X_1,\widehat{P}_1)\big)\big|
&\lesssim \frac{\varepsilon}{t^{2a}} .
\end{align*}
Assuming that the proposition is true for $|I|<m$ it now follows from this and previous lemma that its also true for $|I|=m$.
\end{proof}

\begin{proposition} \label{prop:ZXZPkichighder}
	Suppose $t\geq t_0 + 1$, $\vert x \vert \leq c t$, $(t,x,\ph) \in \supp(f)$ and the bounds \eqref{eq:Gammasuppf} hold.  Then, for $|I|\leq N$,
	\begin{align*}
		\left\vert
		\Zb^J \left( \Pb^i(s,t,x,\ph) \right) \vert_{s=t}
		\right\vert
		\lesssim
		\
		&
		\frac{\varepsilon}{t^{2a}}
		+
		\frac{1}{t^a}\sum_{\vert J \vert \leq \vert I \vert -1}
		\bigg(
		t \left\vert (Z^J \Gamma) (t,x) \right\vert
		+\int_{t_0}^t \frac{s' - t_0}{t-t_0} \left\vert (Z^J \Gamma) \left(s', s'\frac{x}{t} \right) \right\vert ds'
		\bigg),
		\\
		\left\vert
		\Zb^I \left( \Xb^i(s,t,x,\ph) \right) \vert_{s=t}
		\right\vert
		\lesssim
		\
		&
		\varepsilon t^{1-2a}
		+
		t^{1-a}\sum_{\vert J \vert \leq \vert I \vert -1}
		\bigg(
		t \left\vert (Z^J \Gamma) (t,x) \right\vert
		+
		 \int_{t_0}^t \frac{s' - t_0}{t-t_0} \left\vert (Z^J \Gamma) \left(s', s'\frac{x}{t} \right) \right\vert ds'
		\bigg).
	\end{align*}
\end{proposition}
\begin{proof} We will use induction to prove this. Assuming its true for $|I|< m\leq N/2+1$ the assumptions of Proposition \ref{prop:GammaXPminusGammaX1P1hightimeder} hold at $s=t$
\begin{equation*}
t^{k+1} \big|\Zb^L\!\big(\widehat{\Gamma}^{(k)\!}(X,\!\widehat{P})
-\widehat{\Gamma}^{(k)\!}(X_{\!1},\!\widehat{P}_{1})\big)\big|\!
\lesssim
\!\frac{1}{t^{a}}\!\!\sum_{|J|\leq |L| \!\!\!\!\!\!}
\Big(\varepsilon \frac{|\Zb^J \!\!(X\!\!-\!\!X_{\!1})|\!}{t}+ \varepsilon \big|\Zb^{J} \! \!(\widehat{P}-\widehat{P}_{1})\big|\!
+\frac{\varepsilon}{t^{a}\!}
+\!\sum_{|J|\leq |L|+k\!\!\!\!\!\!\!\!\!\!\!\!\!\!\!}t\big|(Z^{J} \Gamma)(X_{\!1})\big|\Big)
\end{equation*}
and
\begin{equation*}
t^{k+1} \big|\Zb^L\big(\widehat{\Gamma}^{(k,2)}(X_1,P_1)\big)\big|
\lesssim  \frac{\varepsilon}{t^{2a}}
 + \frac{1}{t^a} \sum_{|M|\leq |L|+k\!\!\!\!\!}t\big|(Z^{M} \Gamma)(X_1)\big|,
\end{equation*}
since $\frac{|\Zb^I \!\!X_{\!1}|\!}{t}+|\Zb^I \!\!\widehat{P}_{1}| \leq C$ for any multi index $I$, using the form of the vector fields $\Zb$.
Hence writing $X-X_1=\Xb+X_2-X_1$
and $\Ph-\widehat{P}_1=\Pb+\widehat{P}_2-\widehat{P}_1$ we get
\begin{multline*}
t^{k+1} \big|\Zb^L\!\big(\widehat{\Gamma}^{(k)\!}(X,\!\widehat{P})
-\widehat{\Gamma}^{(k)\!}(X_{\!1},\!\widehat{P}_{\!1})\big)\big|\!
+t^{k+1} \big|\Zb^L\big(\widehat{\Gamma}^{(k,2)}(X_1,P_1)\big)\big|\\
\lesssim
\!\frac{\varepsilon}{t^{a}}\!\!\sum_{|J|\leq |L| \!\!\!\!\!\!}
\Big(\frac{|\Zb^J \Xb|\!}{t}+\big|\Zb^{J} \! \! \Pb\big|\Big)\! +
\!\frac{1}{t^{a}}\!\!\sum_{|J|\leq |L| \!\!\!\!\!\!}
\Big(\varepsilon \frac{|\Zb^J \!\!(X_2\!-\!\!X_{\!1})|\!}{t}+\varepsilon\big|\Zb^{J} \! \!(\widehat{P}_2\!-\widehat{P}_{\!1})\big|\!
+\frac{\varepsilon}{t^{a}\!}
+\!\sum_{|J|\leq |L|+k\!\!\!\!\!\!\!\!\!\!\!\!\!}t\big|(Z^{J} \Gamma)(X_{\!1})\big|\Big).
\end{multline*}
Using induction for the first sum and the estimates \eqref{eq:X2minusX1highder} and \eqref{eq:P2minusP1highder} for $X_2-X_1$ and $\widehat{P}_2-\widehat{P}_1$ respectively, the proposition follows.
\end{proof}

\subsection{$L^{\infty}$ estimates for lower order derivatives of geodesics}

The estimates in the previous sections easily lead to pointwise bounds for lower order derivatives of $\Xb(s,t,x,p)$ and $\Pb(s,t,x,p)$.

\begin{proposition} \label{prop:ZXbarPbarlo}
	Suppose $t\geq t_0 + 1$, $\vert x \vert \leq c t$, $(t,x,\ph) \in \supp(f)$ and the bounds \eqref{eq:Gammasuppf} hold.  Then, for $i=1,2,3$,
	\[
		s^{2a - 1} \left\vert \Zb^I \left( \Xb(s,t,x,p) \right) \right\vert
		+
		s^{2a} \left\vert \Zb^I \left( \Pb(s,t,x,p) \right) \right\vert
		\leq
		C \varepsilon,
	\]
	for all $t_0 \leq s \leq t$, for $\vert I \vert=0,1,\ldots,\left\lfloor \frac{N}{2} \right\rfloor + 1$.
\end{proposition}

\begin{proof}
	The proof proceeds by induction.  Clearly the result is true when $\vert I \vert =0$ by Proposition \ref{prop:sec23}.  Assume the result is true for all $\vert I \vert \leq k$, for some $k \leq \left\lfloor \frac{N}{2} \right\rfloor$.  Then $I$ clearly satisfies the assumptions of Corollary \ref{cor:GammaXPminusGammaX1P1low} and so, by the equations \eqref{eq:Xbar}, \eqref{eq:Pbar} and the pointwise bounds \eqref{eq:Gammasuppf},
	\[
		\Big\vert \frac{d}{ds} \Zb^I \left( \Pb^i(s) \right) \Big\vert
		=
		\left\vert \Zb^I \left( \Gh^i \left( s, X(s), \Ph(s) \right) - \Gh^i \left( s, s\frac{x}{t}, \frac{x}{t} \right) \right) \right\vert
		\lesssim
		\frac{\varepsilon}{s^{1+a}} \sum_{\vert J \vert \leq \vert I \vert}
		\Big(\frac{|\Zb^J\! (X\!-\!X_1)|\!}{s}+\big|\Zb^{J} \! (\widehat{P}\!-\!\widehat{P}_1)\big|\Big),
	\]
	where we recall $X_1(s,t,x) = s \frac{x}{t}$ and $\widehat{P}_1(t,x) = \frac{x}{t}$.  Writing $X-X_1=\Xb+X_2-X_1$ and $\Ph-\widehat{P}_1=\Pb+\widehat{P}_2-\widehat{P}_1$ and using the estimates \eqref{eq:X2minusX1} and \eqref{eq:P2minusP1} for $X_2-X_1$ respectively $\widehat{P}_2-\widehat{P}_1$ respectively, this gives,
	\[
		\Big\vert \frac{d}{ds} \Zb^I \left( \Pb^i(s) \right) \Big\vert
		\lesssim
		\frac{\varepsilon}{s^{1+2a}}
		+
		\frac{\varepsilon}{s^{1+a}}
		\sum_{\vert J \vert \leq \vert I \vert}
		\Big(
		\frac{|\Zb^J\! (\Xb(s))|\!}{s}+\big|\Zb^{J} \! (\Pb(s))\big|
		\Big).
	\]
	Integrating backwards from $s=t$ and using Proposition \ref{prop:ZXZPkiclowder},
	\[
		\left\vert \Zb^I \left( \Pb^i(s) \right) \right\vert
		\lesssim
		\frac{\varepsilon}{s^{2a}}
		+
		\varepsilon
		\sum_{\vert J \vert \leq \vert I \vert}
		\int_s^t
		\Big(
		\frac{|\Zb^J\! (\Xb(\tilde{s}))|\!}{\tilde{s}^{2+a}}
		+
		\frac{\big|\Zb^{J} \! (\Pb(\tilde{s}))\big|}{\tilde{s}^{1+a}}
		\Big)
		d \tilde{s},
	\]
	and so, after summing over $i=1,2,3$ and $I$, the Gr\"{o}nwall inequality, Lemma \ref{lem:Gronwall}, gives
	\[
		\left\vert \Zb^I \left( \Pb(s) \right) \right\vert
		\lesssim
		\frac{\varepsilon}{s^{2a}}
		+
		\varepsilon
		\sum_{\vert J \vert \leq \vert I \vert}
		\int_s^t
		\frac{|\Zb^J\! (\Xb(\tilde{s}))|\!}{\tilde{s}^{2+a}}
		d \tilde{s}.
	\]
	The equation \eqref{eq:Xbar} and Proposition \ref{prop:ZXZPkiclowder} then give, after integrating backwards from $s=t$ again,
	\[
		\left\vert \Zb^I \left( \Xb^i(s) \right) \right\vert
		\lesssim
		\varepsilon
		+
		\varepsilon
		\sum_{\vert J \vert \leq \vert I \vert}
		\int_s^t
		\frac{|\Zb^J\! (\Xb(\tilde{s}))|\!}{\tilde{s}^{1+a}}
		d \tilde{s},
	\]
	where the fact that, for any function $\lambda(s)$,
	\[
		\int_s^t \int_{s'}^t \lambda(\tilde{s}) d\tilde{s} ds'
		=
		\int_s^t \int_s^t \chi_{\{ s'\leq \tilde{s}\}} ds' \lambda(\tilde{s}) d \tilde{s}
		=
		\int_s^t (\tilde{s} - s) \lambda(\tilde{s}) d \tilde{s},
	\]
	has been used (here $\chi_{\{ s'\leq \tilde{s}\}}$ is the indicator function of the interval $[s',\tilde{s}]$).  Another application of the Gr\"{o}nwall inequality, after summing over $i=1,2,3$ and $I$, completes the proof.
\end{proof}

Corollary \ref{cor:ZbX2} and Proposition \ref{prop:ZXbarPbarlo} immediately yield the following sharp pointwise bounds.

\begin{proposition} \label{prop:ZXZPlo}
	Suppose $t\geq t_0 + 1$, $\vert x \vert \leq c t$, $(t,x,\ph) \in \supp(f)$ and the bounds \eqref{eq:Gammasuppf} hold.  Then, for $i=1,2,3$,
	\[
		s^{-1}\big\vert \Zb^I \big( X(s,t,x,p) \big) \big\vert
		+
		\big\vert \Zb^I \big( \Ph(s,t,x,p) \big) \big\vert
		\leq
		C,
	\]
	for all $t_0\leq s \leq t$, for $\vert I \vert=0,1,2,\ldots,\left\lfloor \frac{N}{2} \right\rfloor + 1$.
\end{proposition}

The following form of the Gr\"{o}nwall inequality was used in the proof of Proposition \ref{prop:ZXbarPbarlo} above, and will be used in the proof of Proposition \ref{prop:ZXbarPbarho} below.

\begin{lemma} \label{lem:Gronwall}
	For continuous functions $v,a,b: [t_0 ,t] \to \mathbb{R}$, if
	\[
		v(s) \leq \int_s^t a(s') v(s') ds' + b(s),
	\]
	for $s\in [t_0,t]$, then
	\[
		v(s)
		\leq
		b(s)
		+
		\int_s^t a(s') b(s') e^{\int_s^{s'} a(s'') ds''} ds'.
	\]
\end{lemma}

\subsection{Higher order estimates for derivatives of geodesics}
The main result of this section is the following.

\begin{proposition} \label{prop:ZXbarPbarho}
	Suppose $t\geq t_0 + 1$, $\vert x \vert \leq c t$, $(t,x,\ph) \in \supp(f)$ and the bounds \eqref{eq:Gammasuppf} hold.  Then, for $i=1,2,3$,
	\begin{align}
		\left\vert \Zb^I \left( \Xb(s,t,x,\ph)^i \right) \right\vert
		\lesssim
		\
		&
		\varepsilon s^{1-2a}
		+
		t^{2-a} \sum_{\vert J \vert \leq |I|-1} \left\vert (Z^J \Gamma)(t,x) \right\vert
		+
		\sum_{\vert J \vert \leq |I| }
		\int_{t_0}^t (s')^{1-a}\left\vert (Z^J \Gamma) \left( s', s' \frac{x}{t} \right) \right\vert \, ds'
		\nonumber
		\\
		&
		+
		\sum_{\vert J \vert \leq \vert I \vert}
		\int_s^t
		s'
		\left\vert
		(Z^J \Gamma) \left( s', s'\frac{x}{t} \right)
		-
		(Z^J \Gamma) \left( s', X(s') \right)
		\right\vert
		ds'
		\label{eq:ZXbinduction}
	\end{align}
	and,
	\begin{align}
		\left\vert \Zb^I \left( \Pb(s,t,x,\ph)^i \right) \right\vert
		\lesssim
		\
		&
		\varepsilon s^{-2a}
		+
		\frac{t}{s^{a}} \sum_{\vert J \vert \leq |I|-1} \left\vert (Z^J \Gamma)(t,x) \right\vert
		+
		\sum_{\vert J \vert \leq |I| }
		\int_{t_0}^t (s')^{-a}\left\vert (Z^J \Gamma) \left( s', s' \frac{x}{t} \right) \right\vert \, ds'
		\nonumber
		\\
		&
		+
		\sum_{\vert J \vert \leq \vert I \vert}
		\int_s^t
		\left\vert
		(Z^J \Gamma) \big( s', s'\frac{x}{t} \big)
		-
		(Z^J \Gamma) \left( s', X(s') \right)
		\right\vert
		ds'
		\label{eq:ZPbinduction}
	\end{align}
	for all $t_0 \leq s \leq t$, $\vert I \vert \leq N$.
\end{proposition}

\begin{proof}
	Let $I$ be a multi index with $\vert I \vert \leq N$.  Using the equation \eqref{eq:Pbar} and Proposition \ref{prop:GammaXPminusGammaX1P1high},
	\begin{align*}
		\Big\vert \frac{d \Zb^I \big( \Pb^i(s) \big)}{ds} \Big\vert
		\lesssim
		\
		&
		\sum_{\vert K \vert \leq \left\lfloor \frac{\vert I \vert}{2} \right\rfloor +1}
		\Vert (Z^K\Gamma)(s,\cdot) \Vert_{L^{\infty}}
		\sum_{\vert J \vert \leq \vert I \vert}
		\Big(
		\frac{\vert \Zb^J(X-X_1) \vert}{s}
		+
		\vert \Zb^J(\Ph-\Ph_{\!\!1}) \vert
		\Big)
		\\
		&
		+
		\sum_{\vert J \vert \leq \vert I \vert}
		\left\vert (Z^J\Gamma)\big(s,s \frac{x}{t} \big) \right\vert
		\sum_{\vert K \vert \leq \left\lfloor \frac{\vert I \vert}{2} \right\rfloor +1}
		\Big(
		\frac{\vert \Zb^K(X-X_1) \vert}{s}
		+
		\vert \Zb^K(\Ph-\Ph_{\!\!1}) \vert
		\Big)
		\\
		&
		+
		\sum_{\vert J \vert \leq \vert I \vert}
		\Big\vert
		(Z^J \Gamma) \big( s, s\frac{x}{t} \big)
		-
		(Z^J \Gamma) \left( s, X(s) \right)
		\Big\vert.
	\end{align*}
	Writing,
	\[
		\frac{X(s)-X_1(s)}{s}
		=
		\frac{\Xb(s)}{s}
		+
		\frac{X_2(s)}{s} - \frac{x}{t},
		\qquad
		\Ph(s) - \Ph_{\!\!1}(s)
		=
		\Pb(s) + \Ph_2(s) - \frac{x}{t},
	\]
	and using Corollary \ref{cor:ZbX2}, Proposition \ref{prop:ZXbarPbarlo} and the pointwise bounds \eqref{eq:Gammasuppf} for $\Gamma$ gives,
	\[
		\Big\vert \frac{d \Zb^I \big( \Pb^i(s) \big)}{ds} \Big\vert
		\lesssim
		\varepsilon
		\sum_{\vert J \vert \leq \vert I \vert}
		\Big(
		\frac{\vert \Zb^J\left(\Xb(s)\right) \vert}{s^{2+a}}
		+
		\frac{\vert \Zb^J\left(\Pb(s)\right) \vert}{s^{1+a}}
		\Big)
		+
		F_{\vert I \vert}(s,t,x,\ph),
	\]
	where
	\begin{align*}
		F_{\vert I \vert}(s,t,x,\ph)
		=
		\
		&
		\frac{\varepsilon}{s^{1+2a}}
		+
		\frac{t}{s^{1+a}} \sum_{\vert J \vert \leq |I|-1} \left\vert (Z^J \Gamma)(t,x) \right\vert
		+
		\frac{1}{s^a}
		\sum_{\vert J \vert \leq \vert I \vert}
		\left\vert (Z^J\Gamma)\left(s,s \frac{x}{t} \right) \right\vert
		\\
		&
		+
		\frac{1}{s^{1+a}}
		\sum_{\vert J \vert \leq \vert I \vert}
		\int_{t_0}^t
		\left\vert (Z^J\Gamma)\left(s',s' \frac{x}{t} \right) \right\vert
		\frac{s'}{s'+s}
		ds'
		+
		\sum_{\vert J \vert \leq \vert I \vert}
		\left\vert
		(Z^J \Gamma) \left( s, s\frac{x}{t} \right)
		-
		(Z^J \Gamma) \left( s, X(s) \right)
		\right\vert.
	\end{align*}
	Integrating backwards from $s=t$ gives,
	\begin{align*}
		\left\vert \Zb^I \big( \Pb^i(s) \big) \right\vert
		\lesssim
		\left\vert \Zb^I \big( \Pb^i(s) \big) \vert_{s=t} \right\vert
		+
		\int_s^t
		\varepsilon
		\sum_{\vert J \vert \leq \vert I \vert}
		\Big(
		\frac{\vert \Zb^J\left(\Xb(\tilde{s})\right) \vert}{\tilde{s}^{2+a}}
		+
		\frac{\vert \Zb^J\left(\Pb(\tilde{s})\right) \vert}{\tilde{s}^{1+a}}
		\Big)
		+
		F_{\vert I \vert}(\tilde{s},t,x,\ph)
		d \tilde{s}.
	\end{align*}
	Summing over $i=1,2,3$ and $I$, the Gr\"{o}nwall inequality \ref{lem:Gronwall} gives
	\begin{align*}
		\left\vert \Zb^I \big( \Pb^i(s) \big) \right\vert
		\lesssim
		\left\vert \Zb^I \big( \Pb^i(s) \big) \vert_{s=t} \right\vert
		+
		\int_s^t
		\varepsilon
		\sum_{\vert J \vert \leq \vert I \vert}
		\frac{\vert \Zb^J\left(\Xb(\tilde{s})\right) \vert}{\tilde{s}^{2+a}}
		+
		F_{\vert I \vert}(\tilde{s},t,x,\ph)
		d \tilde{s}.
	\end{align*}
	Integrating backwards from $s=t$ again, the equation \eqref{eq:Xbar} implies
	\begin{align*}
		\left\vert \Zb^I\big( \Xb^i(s) \big) \right\vert
		\lesssim
		\left\vert \Zb^I \big( \Xb^i(s) \big) \vert_{s=t} \right\vert
		+
		(t-s)
		\left\vert \Zb^I \big( \Pb^i(s) \big) \vert_{s=t} \right\vert
		+
		\int_s^t
		\varepsilon
		\sum_{\vert J \vert \leq \vert I \vert}
		\frac{\vert \Zb^J\left(\Xb(\tilde{s})\right) \vert}{\tilde{s}^{1+a}}
		+
		\tilde{s} F_{\vert I \vert}(\tilde{s},t,x,\ph)
		d \tilde{s},
	\end{align*}
	where the fact that, for any function $\lambda(s)$,
	\[
		\int_s^t \int_{\tilde{s}}^t \lambda(s') ds' d \tilde{s}
		=
		\int_s^t \int_s^t \chi_{\{ \tilde{s} \leq s' \}} d \tilde{s} \lambda(s') ds'
		=
		\int_s^t (s'-s) \lambda(s') ds',
	\]
	has been used.  Another application of the Gr\"{o}nwall inequality \ref{lem:Gronwall} gives
	\begin{align*}
		\left\vert \Zb^I \big( \Xb^i(s) \big) \right\vert
		\lesssim
		\left\vert \Zb^I \big( \Xb^i(s) \big) \vert_{s=t} \right\vert
		+
		(t-s)
		\left\vert \Zb^I \big( \Pb^i(s) \big) \vert_{s=t} \right\vert
		+
		\int_s^t
		\tilde{s} F_{\vert I \vert}(\tilde{s},t,x,\ph)
		d \tilde{s}.
	\end{align*}
	The bound \eqref{eq:ZXbinduction} follows from Proposition \ref{prop:ZXZPkichighder}, along with the fact that,
	\begin{multline*}
		\int_s^t \frac{1}{\tilde{s}^{a}}
		\int_{t_0}^t
		\left\vert (Z^J\Gamma)\left(s',s' \frac{x}{t} \right) \right\vert
		\frac{s'}{s'+\tilde{s}}
		ds'
		d\tilde{s}
		\\
		\lesssim
		\int_s^t
		\frac{1}{\tilde{s}^{a}}
		\int_{\tilde{s}}^t
		\left\vert (Z^J\Gamma)\left(s',s' \frac{x}{t} \right) \right\vert
		ds'
		+
		\frac{1}{\tilde{s}^{1+a}}
		\int_{t_0}^{\tilde{s}}
		s'
		\left\vert (Z^J\Gamma)\left(s',s' \frac{x}{t} \right) \right\vert
		ds'
		d\tilde{s},
	\end{multline*}
	and, for any nonnegative function $\lambda(s)$,
	\[
		\int_s^t
		\frac{1}{\tilde{s}^{1+a}}
		\int_{t_0}^{\tilde{s}}
		s'
		\lambda(s')
		ds'
		d\tilde{s}
		\lesssim
		\int_{t_0}^t \int_{t_0}^t
		\frac{1}{\tilde{s}^{1+a}}
		\chi_{\{ s'\leq \tilde{s} \}}
		d \tilde{s}
		s' \lambda(s') ds'
		\lesssim
		\int_{t_0}^t (s')^{1-a} \lambda(s') ds',
	\]
	and
	\[
		\int_s^t
		\frac{1}{\tilde{s}^{a}}
		\int_{\tilde{s}}^t
		\lambda(s')
		ds'
		d\tilde{s}
		\lesssim
		\int_{s}^t \int_{s}^t
		\frac{1}{\tilde{s}^{a}}
		\chi_{\{ \tilde{s} \leq s' \}}
		d \tilde{s}
		\lambda(s') ds'
		\lesssim
		\int_{s}^t (s')^{1-a} \lambda(s') ds'.
	\]
	The bound \eqref{eq:ZPbinduction} follows similarly.
\end{proof}

\begin{corollary} \label{cor:ZXPho}
	Suppose $t\geq t_0 + 1$, $\vert x \vert \leq c t$, $(t,x,\ph) \in \supp(f)$ and the bounds \eqref{eq:Gammasuppf} hold.  Then, for $i=1,2,3$,
	\begin{align*}
		&
		\left\vert \Zb^I \left( X(t_0,t,x,\ph)^i \right) \right\vert
		+
		\left\vert \Zb^I \left( P(t_0,t,x,\ph)^i \right) \right\vert
		\leq
		C \bigg(
		1
		+
		\sum_{\vert J \vert \leq \vert I \vert -1}
		t^{2-a} \vert (Z^J \Gamma)(t,x) \vert
		\\
		&
		\qquad \qquad
		+
		\sum_{\vert J \vert \leq \vert I \vert }
		\int_{t_0}^t (s')^{1-a}
		\left\vert (Z^J \Gamma) \left( s', s'\frac{x}{t} \right) \right\vert
		+
		s'
		\left\vert
		(Z^J \Gamma) \left( s', s'\frac{x}{t} \right)
		-
		(Z^J \Gamma) \left( s', X(s') \right)
		\right\vert
		ds'
		\bigg).
	\end{align*}
\end{corollary}

\begin{proof}
	The corollary is an immediate consequence of Proposition \ref{prop:Zbmisc} and Proposition \ref{prop:ZXbarPbarho}.
\end{proof}

\subsection{Spacetime derivatives and small time}
\label{subsec:smalltime}

Since the vector fields $\Zb$ become singular at time $t=t_0$, in this section the spacetime $\partial_t$ and $\partial_{x^i}$ derivatives of $X(s,t,x,\ph)$ and $\Ph(s,t,x,\ph)$ are estimated for $t_0 \leq t \leq t_0 +1$.  Since the results of this section are local in time they are much simpler than those in previous sections.  In particular, it is not necessary to subtract the approximations $X_2$, $\Ph_2$ from $X$ and $\Ph$ respectively.  Note that $\partial$ always denotes the spacetime gradient $\partial = (\partial_t,\partial_{x^1}, \partial_{x^2}, \partial_{x^3})$.  When applied to functions on $\ms$ the derivatives are, as usual, taken with respect to the $(t,x,\ph)$ coordinate system.

It is first necessary to estimate derivatives of the equations \eqref{eq:geodesictimenormalized}.

\begin{proposition} \label{prop:lowrhsbound}
	Let $I$ be a multi index and suppose $\vert \partial^K X \vert/s + \vert \partial^K \Ph \vert \leq C$ for all $\vert K \vert \leq \frac{\vert I \vert}{2}$.  Then,
	\[
		\left\vert \partial^I \left( \Gh \left( s, X(s), \Ph(s) \right) \right) \right\vert
		\lesssim
		\sum_{\vert K \vert \leq \left\lfloor \frac{\vert I \vert}{2} \right\rfloor + 1}
		\!\!\!\!\!\Vert \partial^K \Gamma (s,\cdot) \Vert_{L^{\infty}}
		\sum_{\vert J \vert \leq \vert I \vert} \left( \vert \partial^J X(s) \vert + \vert \partial^J \Ph(s) \vert \right)
		+
		\sum_{\vert J \vert \leq \vert I \vert} \left\vert \partial^J \Gamma (s,X(s)) \right\vert.
	\]
\end{proposition}

\begin{proof}
	The proof follows exactly as in Proposition \ref{prop:GammaXPminusGammaX1P1high} (and is actually much simpler).
\end{proof}

In order to use the system \eqref{eq:geodesictimenormalized} to estimate $\partial^I X(s,t,x,\ph)$ and $\partial^I \Ph(s,t,x,\ph)$, it is also necessary to estimate the final conditions (note that this is completely straightforward unless $\partial^I$ contains $\partial_t$ derivatives).

\begin{proposition} \label{prop:lowfinalcond}
	Let $I$ be a multi index with $\vert I \vert \geq 1$ and suppose $\vert \partial^J \Gamma (t,x) \vert \leq C$ for all $\vert J \vert \leq \left\lfloor \frac{\vert I \vert}{2} \right\rfloor + 1$.  Then
	\[
		\left\vert (\partial^I \Ph)(t,t,x,\ph) \right\vert
		\lesssim
		\sum_{\vert J \vert \leq \vert I \vert - 1} \left\vert (\partial^J \Gamma)(t,x) \right\vert,
		\qquad
		\left\vert (\partial^I X)(t,t,x,\ph) \right\vert
		\lesssim
		1
		+
		\sum_{\vert J \vert \leq \vert I \vert - 1} \left\vert (\partial^J \Gamma)(t,x) \right\vert.
	\]
\end{proposition}

\begin{proof}
	Recall the notation $\Ph^{(k)}$ and $\Gh^{(k)}$ from Section \ref{subsec:parameterderivatives}.  By the formula \eqref{eq:Yidnew0} it follows that
	\begin{align*}
		(\partial^I \Ph)(t,t,x,\ph)
		&
		=
		\sum_{\substack{
		J_1+ \ldots + J_{k+1} +L=I
		\\
		\vert J_i \vert \geq 1, k \geq 0}}
		C_{IJ_1\ldots J_{k+1},L}
		\partial^{J_1} t \ldots \partial^{J_{k+1}}t \Big( \partial^L \Ph^{(k+1)} \Big) (t,t,x,\ph),
		\\
		(\partial^I X)(t,t,x,\ph)
		&
		=
		\partial^I x
		+
		\sum_{\substack{
		J_1+ \ldots + J_{k+2} +L=I
		\\
		\vert J_i \vert \geq 1, k \geq 0}}
		C_{IJ_1\ldots J_{k+2},L}'
		\partial^{J_1} t \ldots \partial^{J_{k+2}}t \Big( \partial^L \Ph^{(k+2)} \Big) (t,t,x,\ph),
	\end{align*}
	for some constants $C_{IJ_1\ldots J_{k+1},L}, C_{IJ_1\ldots J_{k+2},L}'$, where the proof of the second uses the first and the fact that $\frac{d X^{(k)}}{ds} = \Ph^{(k)}$.  Hence,
	\begin{align*}
		\left\vert (\partial^I \Ph)(t,t,x,\ph) \right\vert
		&
		\lesssim
		\sum_{\substack{
		\vert L\vert + k \leq \vert I \vert - 1,\,\,\,
		k \geq 0}}
		\Big\vert \Big( \partial^L \Ph^{(k+1)} \Big) (t,t,x,\ph) \Big\vert,
		\\
		\left\vert (\partial^I X)(t,t,x,\ph) \right\vert
		&
		\lesssim
		1
		+
		\sum_{\substack{
		\vert L\vert + k \leq \vert I \vert - 2,\,\,\,
		k \geq 0}}
		\Big\vert \Big( \partial^L \Ph^{(k+1)} \Big) (t,t,x,\ph) \Big\vert.
	\end{align*}
	The proof follows by noting that
	\[
		\Big\vert \Big( \partial^L \Ph^{(k+1)} \Big) (t,t,x,\ph) \Big\vert
		=
		\Big\vert \partial^L \Big(  \Gh^{(k)}(s,X(s),\Ph(s)) \Big) \Big\vert_{s=t} \Big\vert
		\lesssim
		\sum_{\vert I \vert \leq \vert L \vert +k}
		\left\vert (\partial^I \Gamma)(t,x) \right\vert,
	\]
	by an appropriate version of Lemma \ref{lem:ZGammaK}.
\end{proof}

\begin{proposition} \label{prop:lowpointwise}
	Suppose $t_0 \leq t \leq t_0+1$, and $\vert \partial^J \Gamma (t',x)\vert \leq C$ for all $t_0 \leq t'\leq t$ and $\vert x \vert \leq ct'$ and $\vert J \vert \leq \left\lfloor \frac{N}{2} \right\rfloor + 2$.  Then, for $\vert I \vert \leq \left\lfloor \frac{N}{2} \right\rfloor + 2$,
	\[
		\left\vert \partial^I X(s,t,x,\ph) \right\vert
		+
		\left\vert \partial^I \Ph(s,t,x,\ph) \right\vert
		\leq
		C,
	\]
	for all $t_0 \leq s \leq t$.
\end{proposition}

\begin{proof}
	Proposition \ref{prop:lowrhsbound} and the equation \eqref{eq:geodesictimenormalized} imply
	\[
		\Big\vert \frac{d \partial^I \Ph^i}{ds}(s) \Big\vert
		\lesssim
		1
		+
		\sum_{\vert J \vert \leq \vert I \vert}
		\left( \left\vert \partial^J X(s) \right\vert + \left\vert \partial^J \Ph(s) \right\vert \right).
	\]
	Integrating backwards from $s=t$, by Proposition \ref{prop:lowfinalcond},
	\[
		\left\vert \partial^I \Ph^i(s) \right\vert
		\lesssim
		1
		+
		\sum_{\vert J \vert \leq \vert I \vert}
		\int_s^t
		\left( \left\vert \partial^J X(s') \right\vert + \left\vert \partial^J \Ph(s') \right\vert \right)
		ds'.
	\]
	Summing over $I$, the Gr\"{o}nwall inequality, Lemma \ref{lem:Gronwall}, gives
	\[
		\left\vert \partial^I \Ph^i(s) \right\vert
		\lesssim
		1
		+
		\sum_{\vert J \vert \leq \vert I \vert}
		\int_s^t
		\left\vert \partial^J X(s') \right\vert
		ds'.
	\]
	The result follows by integrating from $s=t$ again and repeating.
\end{proof}

\begin{proposition} \label{prop:lowhighorder}
	Suppose $t_0 \leq t \leq t_0 +1$ and $\vert \partial^J \Gamma (t',x)\vert \leq C$ for all $t_0 \leq t'\leq t$ and $\vert x \vert \leq ct'$ and $\vert J \vert \leq \left\lfloor \frac{N}{2} \right\rfloor + 2$.  Then, for $\vert I \vert \leq N$,
	\[
		\left\vert \partial^I X(s,t,x,\ph) \right\vert
		+
		\left\vert \partial^I \Ph(s,t,x,\ph) \right\vert
		\lesssim
		1
		+
		\sum_{\vert J\vert \leq \vert I \vert -1} \left\vert \partial^J \Gamma(t,x) \right\vert
		+
		\sum_{\vert J\vert \leq \vert I \vert} \int_s^t \left\vert \partial^J \Gamma(s',X(s')) \right\vert ds',
	\]
	for all $t_0 \leq s \leq t$.
\end{proposition}

\begin{proof}
	Proposition \ref{prop:lowrhsbound}, the equation \eqref{eq:geodesictimenormalized} and Proposition \ref{prop:lowpointwise} now imply
	\[
		\Big\vert \frac{d \partial^I \Ph^i}{ds}(s) \Big\vert
		\lesssim
		1
		+
		\sum_{\vert J \vert \leq \vert I \vert}
		\left(
		\left\vert \partial^J X(s) \right\vert
		+
		\left\vert \partial^J \Ph(s) \right\vert
		+
		\left\vert \partial^J \Gamma(s,X(s)) \right\vert
		\right),
	\]
	and so Proposition \ref{prop:lowfinalcond} implies
	\[
		\left\vert \partial^I \Ph (s) \right\vert
		\lesssim
		1
		+
		\sum_{\vert J \vert \leq \vert I \vert-1}
		\left\vert \partial^J \Gamma(t,x) \right\vert
		+
		\sum_{\vert J \vert \leq \vert I \vert}
		\int_s^t
		\left(
		\left\vert \partial^J X(s') \right\vert
		+
		\left\vert \partial^J \Ph(s') \right\vert
		+
		\left\vert \partial^J \Gamma(s',X(s')) \right\vert
		\right)
		ds'.
	\]
	The proof then proceeds exactly as in Proposition \ref{prop:lowpointwise}.
\end{proof}

\section{Estimates for components of the energy momentum tensor}
\label{section:Testimates}

In this section a proof of Theorem \ref{thm:mainL2} is given.  Recall the discussion in Section \ref{subsec:transtime}.  In order to use the results of Section \ref{section:geodesics} it will again be assumed throughout most of this section that $t\geq t_0+1$, the bounds \eqref{eq:Gammasuppf} hold, and $\pi \left( \supp(f) \right) \subset \{ \vert x \vert \leq c t \}$, where $\pi : \ms \to \mathcal{M}$ is the natural projection.  It is shown how Theorem \ref{thm:mainL2} then follows in Section \ref{subsec:Tmainestimates}.

\subsection{Derivatives of components of the energy momentum tensor in terms of derivatives of $f$}

Recall
\[
	T^{\mu \nu}(t,x)
	=
	\int f(t,x,\ph) p^{\mu} p^{\nu} \frac{\sqrt{-\det g}}{p^0} dp^1 dp^2 dp^3.
\]
The main result of this section is Proposition \ref{prop:ZITpointwise}, which uses the bounds on $\Zb^I X$ and $\Zb^I \Ph$ of Corollary \ref{cor:ZXPho} to give bounds on $Z^I T^{\mu \nu}$.  In order to prove the bounds for $Z^I T^{\mu \nu}$, it is convenient to first rewrite the above integral in terms of the $\ph^i$ variables.

\begin{proposition}
	There exists a function $\Lambda$, smooth provided $\vert \ph^i \vert \leq c <1$ for $i=1,2,3$, such that
	\[
		\det \left( \frac{\partial p^i}{\partial \ph^j}  \right) (t,x,\ph)
		=
		\Lambda(\ph, h(t,x)).
	\]
\end{proposition}

\begin{proof}
	Define $\ph^0 = 1$ and note that, since,
	\[
		g_{\alpha \beta} p^{\alpha} p^{\beta} = -1,
	\]
	it follows that
	\[
		2 g_{\alpha \beta} p^{\alpha} \frac{\partial p^{\beta}}{\partial p^j} = 0,
	\]
	and hence,
	\[
		\frac{\partial p^0}{\partial p^j}
		=
		- \frac{g_{\alpha j} p^{\alpha}}{g_{\beta 0} p^{\beta}}
		=
		- \frac{g_{\alpha j} \ph^{\alpha}}{g_{\beta 0} \ph^{\beta}}.
	\]
	Now, since $\ph^i = \frac{p^i}{p^0}$,
	\[
		\frac{\partial \ph^i}{\partial p^j}
		=
		\frac{1}{p^0}
		\left(
		\delta^i_j - \ph^i \frac{\partial p^0}{\partial p^j}
		\right)
		=
		\frac{1}{p^0}
		\Big(
		\delta^i_j + \ph^i \frac{g_{\alpha j} \ph^{\alpha}}{g_{\beta 0} \ph^{\beta}}
		\Big).
	\]
	The proof follows by writing $g_{\alpha \beta} = m_{\alpha \beta} + h_{\alpha \beta}$, noting that
	\[
		g_{\alpha \beta} p^{\alpha} p^{\beta} = -1
		\Rightarrow
		g_{\alpha \beta} \ph^{\alpha} \ph^{\beta} = - \frac{1}{(p^0)^2}
		\Rightarrow
		p^0 = \sqrt{- \frac{1}{g_{\alpha \beta} \ph^{\alpha} \ph^{\beta}}},
	\]
	and using the fact that $\det (A^{-1}) = ( \det A)^{-1}$ for any matrix $A$.
\end{proof}

In Minkowski space, i.\@e.\@ when $h=0$, it is straightforward to compute
\[
	\det \left( \frac{\partial p^i}{\partial \ph_M^j}  \right)
	=
	(p^0_M)^5
	=
	\left(
	1+ (p^1)^2 + (p^2)^2 + (p^3)^2
	\right)^{\frac{5}{2}},
\]
where $\ph_M^j = \frac{p^j}{p^0_M}$, and $p^0_M$ is defined by the relation $m_{\alpha \beta} p^{\alpha} p^{\beta} = -1$.  It then follows that
\begin{align*}
	&
	\Big\vert
	\det \left( \frac{\partial p^i}{\partial \ph^j}  \right)
	-
	\left(
	1+ (p^1)^2 + (p^2)^2 + (p^3)^2
	\right)^{\frac{5}{2}}
	\Big\vert
	=
	\Big\vert
	\det \left( \frac{\partial p^i}{\partial \ph^j}  \right)
	-
	\det \left( \frac{\partial p^i}{\partial \ph_M^j} \right)
	\Big\vert
	\\
	&
	\qquad \qquad
	=
	\left\vert
	\Lambda(\ph, h(t,x))
	-
	\Lambda(\ph,0)
	\right\vert
	\leq
	\sup \vert \partial \Lambda \vert \vert h(t,x) \vert
	\leq
	C \varepsilon,
\end{align*}
since $\vert \ph^i \vert \leq c <1$, and hence the change of variables $(t,x,p) \mapsto (t,x,\ph)$ is well defined if $\varepsilon$ is sufficiently small.  Moreover, recalling that,
\[
	p^0 = \sqrt{- \frac{1}{g_{\alpha \beta} \ph^{\alpha} \ph^{\beta}}},
\]
it follows that, for each $\mu, \nu = 0,1,2,3$,
\[
	p^{\mu} p^{\nu} \frac{\sqrt{-\det g}}{p^0} \det \left( \frac{\partial p^i}{\partial \ph^j}  \right) (t,x,\ph)
	=
	\Lambda^{\mu \nu}(\ph, h(t,x)),
\]
for some functions $\Lambda^{\mu \nu}$, smooth when $\vert \ph \vert \leq c <1$, and so
\begin{equation} \label{eq:Tmunuschematic}
	T^{\mu \nu}(t,x)
	=
	\int f(t,x,\ph) \Lambda^{\mu \nu}(\ph,h(t,x)) d\ph^1 d\ph^2 d\ph^3.
\end{equation}

For each vector field $\Zb$, recall the corresponding functions $\mathring{Z}^k(t,x,\ph)$ defined in Section \ref{subsec:vectorfields}.  Note that, for each $\Zb$, the $\mathring{Z}^k$ have the form,
\[
	\mathring{Z}^k(t,x,\ph)
	=
	\mathring{Z}^k_{1,l}(t,x) \ph^l + \mathring{Z}^k_1 (t,x),
\]
for some functions $\mathring{Z}^k_{1,l}(t,x), \mathring{Z}^k_{2}(t,x)$.  Explicitly,
\[
	\mathring{\Omega}_{ij, 1,l}^k \equiv 0,
	\qquad
	\mathring{\Omega}_{ij,2}^k(t,x)
	=
	\Theta^i(t,x) \delta_j^k
	-
	\Theta^j(t,x) \delta_i^k
	-
	\Omega_{ij} \left( \Theta^k(t,x) \right),
\]
\[
	\mathring{B}_{i,1,l}^k(t,x)
	=
	-
	\Theta^i(t,x) \delta^k_l
	\qquad
	\mathring{B}_{i,2}^k(t,x)
	=
	-
	B_i \left( \Theta^k(t,x) \right)
	-
	\frac{x^i}{t-t_0} \Theta^k(t,x),
\]
and
\[
	\mathring{S}_{1,l}^k \equiv 0,
	\qquad
	\mathring{S}_2^k (t,x)
	=
	\Theta^k(t,x)
	-
	S \left( \Theta^k(t,x) \right)
	-
	\frac{t}{t-t_0} \Theta^k(t,x).
\]
This notation will be used below.

\begin{proposition} \label{prop:ZTschematic}
	For $\mu,\nu=0,1,2,3$ and any multi index $I$, there exist functions $\Lambda$ such that,
	\begin{align*}
		&
		Z^I T^{\mu \nu} (t,x)
		=
		\sum_{\vert I_1 \vert + \vert I_2 \vert \leq \vert I \vert}
		\int
		\Zb^{I_1} \left( f(t,x,\ph) \right)
		\times
		\sum_{k=0}^{\vert I_2 \vert}\,\,\,
		\sum_{ \substack{\vert J_1 \vert + \ldots + \vert J_k\vert \leq \vert I_2 \vert,
		\,\,\,
		\vert J_i \vert \geq 1
		}}
		\\
		&
		\sum_{ \substack{\vert L_1 \vert + \ldots + \vert L_k\vert \leq \vert I_2 \vert - 1,
		\,\,\,
		\vert L_i \vert \geq 1
		}}
		\left(
		\Zb^{J_1} (\ph^{l_1})
		+
		\Zb^{J_1}(h_{\alpha_1 \beta_1}) \Lambda^{l_1, \alpha_1 \beta_1} (\ph, h)
		+
		\Zb^{L_1}(\mathring{Z}^m_{1,m}) \Lambda^{l_1} (\ph, h)
		\right)
		\times
		\\
		&
		\ldots
		\times
		\left(
		\Zb^{J_k} (\ph^{l_k})
		+
		\Zb^{J_k}(h_{\alpha_k \beta_k}) \Lambda^{l_k, \alpha_k \beta_k} (\ph, h)
		+
		\Zb^{L_k}(\mathring{Z}^m_{1,m}) \Lambda^{l_k} (\ph, h)
		\right)
		\Lambda(\ph,h)
		d \ph,
	\end{align*}
	where $Z^I$ is a product of $\vert I \vert$ of the vector fields $\Omega_{ij}, B_i, S$.
\end{proposition}

\begin{proof}
	Recall that the components of the energy momentum tensor take the form \eqref{eq:Tmunuschematic}.  Note that
	\[
		\Zb \left( \Lambda^{\mu \nu} (\ph,h(t,x)) \right)
		=
		\Zb (\ph^l) (\partial_{\ph^l} \Lambda^{\mu \nu})(\ph, h(t,x))
		+
		\sum_{\alpha \beta} (Zh_{\alpha \beta})(t,x) (\partial_{h_{\alpha \beta}} \Lambda^{\mu \nu})(\ph,h(t,x)),
	\]
	and, for $Z = \Omega_{ij}, B_i, S$.
	\[
		Z \Big(
		\int \eta(t,x,\ph) d\ph
		\Big)
		=
		\int \left( \Zb - \mathring{Z}^k \partial_{\ph^k} \right) \eta(t,x,\ph) d\ph
		=
		\int \Zb \left( \eta(t,x,\ph) \right)
		+
		(\partial_{\ph^k} \mathring{Z}^k)(t,x) \eta(t,x,\ph) d \ph,
	\]
	for any function $\eta(t,x,\ph)$, and
	\[
		(\partial_{\ph^k} \mathring{Z}^k)(t,x)
		=
		\mathring{Z}^k_{1,k}(t,x).
	\]
	Therefore, for $\vert I \vert =1$ and $Z = \Omega_{ij}, B_i, S$,
	\begin{multline*}
		Z T^{\mu \nu} (t,x)
		=
		\int
		\left(
		\Zb \left( f(t,x,\ph) \right)
		+
		\mathring{Z}^m_{1,m}(t,x)
		f(t,x,\ph)
		\right)
		\Lambda^{\mu \nu}(\ph, h(t,x))
		\\
		+
		f(t,x,\ph) \left(
		\Zb (\ph^l) (\partial_{\ph^l} \Lambda^{\mu \nu})(\ph, h(t,x))
		+
		(Zh_{\alpha \beta})(t,x) (\partial_{h_{\alpha \beta}} \Lambda^{\mu \nu})(\ph,h(t,x))
		\right)
		d \ph.
	\end{multline*}
	The proof for $\vert I \vert \geq 1$ follows from a straightforward induction argument.
\end{proof}

\begin{proposition} \label{prop:Zfschematic}
	For any multi index $I$, there exist constants $C_{I,k,J,L}$ such that,
	\begin{align*}
		\Zb^{I} \left( f(t,x,\ph) \right)
		=
		\
		&
		\sum_{k+m=1}^{\vert I\vert}
		\sum_{ \substack{
		\vert J_1\vert + \ldots + \vert J_k \vert
		\\
		+
		\vert L_1\vert + \ldots + \vert L_m \vert
		\leq
		\vert I\vert,
		\\
		\vert J_i \vert \geq 1,
		\quad
		\vert L_i \vert \geq 1
		}}
		\Zb^{J_1} \big( X(t_0)^{i_1} \big)
		\ldots
		\Zb^{J_k} \big( X(t_0)^{i_k} \big)
		\Zb^{L_1} \big( \Ph(t_0)^{l_1} \big)
		\ldots
		\Zb^{L_m} \big( \Ph(t_0)^{l_m} \big)
		\\
		&
		\times
		C_{I,k,J,L}
		\left(
		\partial_{x^{i_1}} \ldots \partial_{x^{i_k}}
		\partial_{\ph^{l_1}} \ldots \partial_{\ph^{l_m}}
		f
		\right)
		(t_0,X(t_0),\Ph(t_0)).
	\end{align*}
\end{proposition}

\begin{proof}
	Using the Vlasov equation to write
	\[
		f(t,x,\ph)
		=
		f(t_0,X(t_0),\Ph(t_0)),
	\]
	it follows that,
	\[
		\Zb \left( f(t,x,\ph) \right)
		=
		\Zb \left( X(t_0)^i \right) (\partial_{x^i} f)(t_0,X(t_0),\Ph(t_0))
		+
		\Zb \left( X(t_0)^l \right) (\partial_{\ph^l} f)(t_0,X(t_0),\Ph(t_0)).
	\]
	The proof for $\vert I \vert \geq 2$ follows from a straightforward induction argument.
\end{proof}

\begin{proposition} \label{prop:ZITpointwise}
	Suppose $t\geq t_0 +1$, $\vert x \vert \leq c t$, and the bounds \eqref{eq:Gammasuppf} hold.  Then, for each $\mu,\nu = 0,1,2,3$ and any multi index $I$ with $\vert I \vert \leq N$ and $Z^I$ equal to a product of $\vert I \vert$ of the vector fields $\Omega_{ij}, B_i, S$,
	\begin{multline*}
		\!\!\!\!\!\!\!\!\left\vert Z^I T^{\mu \nu}(t,x) \right\vert
		\leq
		C \!\!\!\!\!\sum_{n_1+n_2 \leq \vert I \vert}\!
		\int
		\left\vert
		(\partial^{n_1}_x \partial^{n_2}_{\ph} f)(t_0,X(t_0),\Ph(t_0))
		\right\vert
		d\ph
		+\,
		C \!\!\!\!\!\!\!\!\!\!\sum_{n_1+n_2 \leq \left\lfloor \frac{\vert I \vert}{2} \right\rfloor + 1}\!\!\!\!\!\!
		\int
		\left\vert
		(\partial^{n_1}_x \partial^{n_2}_{\ph} f)(t_0,X(t_0),\Ph(t_0))
		\right\vert\times \\
		\times\bigg(
		\sum_{\vert J \vert \leq \vert I \vert -1}\!\!\!\!\!
		t^{2-a} \vert (Z^J \Gamma)(t,x) \vert
		+
		\sum_{\vert J \vert \leq \vert I \vert }
		\int_{t_0}^t s^{1-a}
		\left\vert (Z^J \Gamma) \left( s, s\frac{x}{t} \right) \right\vert
		+
		s
		\left\vert
		(Z^J \Gamma) \left( s, s\frac{x}{t} \right)
		-
		(Z^J \Gamma) \left( s, X(s) \right)
		\right\vert
		ds
		\bigg)
		d\ph.
	\end{multline*}
\end{proposition}

\begin{proof}
	Recall the schematic expression for $Z^I T^{\mu \nu}$ of Proposition \ref{prop:ZTschematic}.  Consider multi indices $I_1$, $I_2$ such that $\vert I_1 \vert + \vert I_2 \vert \leq \vert I \vert$, and suppose first that $\vert I_1 \vert \geq \big\lfloor \frac{\vert I \vert}{2} \big\rfloor + 1$.  It must then be the case that $ \vert I_2 \vert \leq \big\lfloor \frac{\vert I \vert}{2} \big\rfloor + 1 \leq \big\lfloor \frac{N}{2} \big\rfloor + 1$.  If $1 \leq \vert J_i \vert \leq \vert I_2 \vert$ then clearly,
	\[
		\big\vert
		\Zb^{J_i} \left( h_{\alpha_i \beta_i} \right) (t,x)
		\big\vert
		\leq
		C,
	\]
	and Proposition \ref{prop:Zbmisc} implies that
	\[
		\big\vert
		\Zb^{J_i} \big( \ph^i \big)
		\big\vert
		\leq
		C.
	\]
	If $1 \leq \vert L_i \vert \leq \vert I_2 \vert - 1$ then, since
	\[
		\mathring{B}_{i,1,k}^k (t,x) = - 3 \Theta^i(t,x),
		\qquad
		\mathring{\Omega}_{ij,1,k}^k (t,x) = 0,
		\qquad
		\mathring{S}_{1,k}^k = 0,
	\]
	Proposition \ref{prop:Phi} implies
	\[
		\left\vert
		\Zb^{L_i} \left( \mathring{Z}_{1,k}^k \right)
		\right\vert
		\leq
		C,
	\]
	for each $\mathring{Z}_{1,k}^k$.
	Let now $k,m,J_1,\ldots,J_k, L_1,\ldots,L_m$ be such that $1 \leq k+m \leq \vert I_1 \vert$, $\vert J_1 \vert + \ldots + \vert J_k \vert + \vert L_1 \vert + \ldots + \vert L_m\vert \leq \vert I_1 \vert$, $\vert J_i \vert \geq 1$, $\vert L_i \vert \geq 1$.  If $k + m \geq \big\lfloor \frac{\vert I_1 \vert}{2} \big\rfloor + 2$ then it must be the case that $\vert J_i \vert \leq \big\lfloor \frac{\vert I_1 \vert}{2} \big\rfloor + 1 \leq \big\lfloor \frac{\vert N \vert}{2} \big\rfloor + 1$ for $i=1,\ldots,k$, and $\vert L_i \vert \leq \big\lfloor \frac{\vert I_1 \vert}{2} \big\rfloor + 1 \leq \big\lfloor \frac{\vert N \vert}{2} \big\rfloor + 1$ for $i=1,\ldots,m$.  Proposition \ref{prop:ZXZPlo} then implies,
	\begin{align*}
		&
		\Big\vert
		\Zb^{J_1} \left( X(t_0)^{i_1} \right)
		\ldots
		\Zb^{J_k} \left( X(t_0)^{i_k} \right)
		\Zb^{L_1} \big( \Ph(t_0)^{l_1} \big)
		\ldots
		\Zb^{L_m} \big( \Ph(t_0)^{l_m} \big)
		\\
		&
		\qquad
		\times
		\left(
		\partial_{x^{i_1}} \ldots \partial_{x^{i_k}}
		\partial_{\ph^{l_1}} \ldots \partial_{\ph^{l_m}}
		f
		\right)
		(t_0,X(t_0),\Ph(t_0))
		\Big\vert
		\leq
		C \sum_{n_1+n_2 \leq \vert I_1 \vert} \left\vert
		\partial_{x}^{n_1}
		\partial_{\ph}^{n_2}
		f
		(t_0,X(t_0),\Ph(t_0))
		\right\vert.
	\end{align*}
	Similarly, if $k+m \leq \big\lfloor \frac{\vert I_1 \vert}{2} \big\rfloor + 1$, there can be at most one $i$ such that either $\vert J_i \vert \geq \big\lfloor \frac{\vert I_1 \vert}{2} \big\rfloor + 2$ or $\vert L_i \vert \geq \big\lfloor \frac{\vert I_1 \vert}{2} \big\rfloor + 2$, so Proposition \ref{prop:ZXZPlo} and Corollary \ref{cor:ZXPho} imply
	\begin{align*}
		&
		\Big\vert
		\Zb^{J_1} \left( X(t_0)^{i_1} \right)
		\ldots
		\Zb^{J_k} \left( X(t_0)^{i_k} \right)
		\Zb^{L_1} \big( \Ph(t_0)^{l_1} \big)
		\ldots
		\Zb^{L_m} \big( \Ph(t_0)^{l_m} \big)
		\\
		&
		\qquad
		\times
		\left(
		\partial_{x^{i_1}} \ldots \partial_{x^{i_k}}
		\partial_{\ph^{l_1}} \ldots \partial_{\ph^{l_m}}
		f
		\right)
		(t_0,X(t_0),\Ph(t_0))
		\Big\vert
		\\
		&
		\qquad
		\leq
		C \sum_{n_1+n_2 \leq \vert I_1 \vert} \left\vert
		\partial_{x}^{n_1}
		\partial_{\ph}^{n_2}
		f
		(t_0,X(t_0),\Ph(t_0))
		\right\vert
		\bigg(
		1
		+
		\sum_{\vert J \vert \leq \vert I \vert -1}
		t^{2-a} \vert (Z^J \Gamma)(t,x) \vert
		\\
		&
		\qquad \qquad
		+
		\sum_{\vert J \vert \leq \vert I \vert }
		\int_{t_0}^t (s')^{1-a}
		\left\vert (Z^J \Gamma) \left( s', s'\frac{x}{t} \right) \right\vert
		+
		s'
		\left\vert
		(Z^J \Gamma) \left( s', s'\frac{x}{t} \right)
		-
		(Z^J \Gamma) \left( s', X(s') \right)
		\right\vert
		ds'
		\bigg).
	\end{align*}
	It is then clear that
	\begin{multline*}
		\Bigg\vert
		\int
		\Zb^{I_1} \left( f(t,x,\ph) \right)
		\times
		\sum_{k=0}^{\vert I_2 \vert}\,\,\,\,\,
		\sum_{ \substack{\vert J_1 \vert + \ldots + \vert J_k\vert \leq \vert I_2 \vert,\,\,\,
		\vert J_i \vert \geq 1
		}}\,\,\,\,\,
		\sum_{ \substack{\vert L_1 \vert + \ldots + \vert L_k\vert \leq \vert I_2 \vert - 1,
\,\,\,
		\vert L_i \vert \geq 1
		}}
		\\
		\left(
		\Zb^{J_1} (\ph^{l_1})
		+
		\Zb^{J_1}(h_{\alpha_1 \beta_1}) \Lambda^{l_1, \alpha_1 \beta_1} (\ph, h)
		+
		\Zb^{L_1}(\mathring{Z}^m_{1,m}) \Lambda^{l_1} (\ph, h)
		\right)
		\times
		\\
		\ldots
		\times
		\left(
		\Zb^{J_k} (\ph^{l_k})
		+
		\Zb^{J_k}(h_{\alpha_k \beta_k}) \Lambda^{l_k, \alpha_k \beta_k} (\ph, h)
		+
		\Zb^{L_k}(\mathring{Z}^m_{1,m}) \Lambda^{l_k} (\ph, h)
		\right)
		\Lambda(\ph,h)
		d \ph
		\Bigg\vert
		\\
				\leq
		C \sum_{n_1+n_2 \leq \vert I \vert}
		\int
		\left\vert
		(\partial^{n_1}_x \partial^{n_2}_{\ph} f)(t_0,X(t_0),\Ph(t_0))
		\right\vert
		d\ph
		+
		C \sum_{n_1+n_2 \leq \left\lfloor \frac{\vert I \vert}{2} \right\rfloor + 1}
		\int
		\left\vert
		(\partial^{n_1}_x \partial^{n_2}_{\ph} f)(t_0,X(t_0),\Ph(t_0))
		\right\vert\times\\
		\times\Big(
		\sum_{\vert J \vert \leq \vert I \vert -1}
		t^{2-a} \vert (Z^J \Gamma)(t,x) \vert
		+
		\sum_{\vert J \vert \leq \vert I \vert }
		\int_{t_0}^t s^{1-a}
		\left\vert (Z^J \Gamma) \left( s, s\frac{x}{t} \right) \right\vert
		+
		s
		\left\vert
		(Z^J \Gamma) \left( s, s\frac{x}{t} \right)
		-
		(Z^J \Gamma) \left( s, X(s) \right)
		\right\vert
		ds
		\Big)
		d\ph.
	\end{multline*}
	Suppose now $\vert I_2 \vert \geq \big\lfloor \frac{\vert I \vert}{2} \big\rfloor + 1$.  It must then be the case that $\vert I_1 \vert \leq \big\lfloor \frac{\vert I \vert}{2} \big\rfloor + 1 \leq \big\lfloor \frac{\vert N \vert}{2} \big\rfloor + 1$.  Proposition \ref{prop:Zfschematic} then implies that
	\[
		\left\vert \Zb^{I_1} \left( f(t,x,\ph) \right) \right\vert
		\leq
		C
		\sum_{n_1+n_2 \leq \left\lfloor \frac{\vert I \vert}{2} \right\rfloor + 1}
		\left\vert
		(\partial^{n_1}_x \partial^{n_2}_{\ph} f)(t_0,X(t_0),\Ph(t_0))
		\right\vert.
	\]
	If $\vert J_i \vert \leq \vert I_2\vert$ and $\vert L_i \vert \leq \vert I_2 \vert -1$, then Proposition \ref{prop:Phi} and Proposition \ref{prop:Zbmisc} imply,
	\begin{equation*}
		\left\vert Z^{J_i} \big( \ph^k \big) \right\vert
		+
		\left\vert Z^{J_i} \big( h_{\alpha_k \beta_k} \big) \right\vert
		+
		\left\vert Z^{J_i} \big( \mathring{Z}_{1,k}^k \big) \right\vert
		\\
		\leq
		C
		\Big(
		1
		+\!\!\!\!
		\sum_{\vert J \vert \leq \vert I_2 \vert -1}\!\!\!\!
		t \left\vert \left( Z^J \Gamma \right)(t,x) \right\vert
		+\!\!\!
		\sum_{\vert J \vert \leq \vert I_2 \vert}
		\int_{t_0}^t
		\left\vert \left( Z^J \Gamma \right) \left(s',s'\frac{x}{t} \right) \right\vert
		ds'
		\Big).
	\end{equation*}
	Hence,
	\begin{multline*}
		\Bigg\vert
		\int
		\Zb^{I_1} \left( f(t,x,\ph) \right)
		\times
		\sum_{k=0}^{\vert I_2 \vert}\,\,\,\,\,
		\sum_{ \substack{\vert J_1 \vert + \ldots + \vert J_k\vert \leq \vert I_2 \vert,\,\,\,
		\vert J_i \vert \geq 1
		}}\,\,\,\,\,
		\sum_{ \substack{\vert L_1 \vert + \ldots + \vert L_k\vert \leq \vert I_2 \vert - 1,\,\,\,
		\vert L_i \vert \geq 1
		}}
		\\
		\left(
		\Zb^{J_1} (\ph^{l_1})
		+
		\Zb^{J_1}(h_{\alpha_1 \beta_1}) \Lambda^{l_1, \alpha_1 \beta_1} (\ph, h)
		+
		\Zb^{L_1}(\mathring{Z}^m_{1,m}) \Lambda^{l_1} (\ph, h)
		\right)
		\times
		\\
		\ldots
		\times
		\left(
		\Zb^{J_k} (\ph^{l_k})
		+
		\Zb^{J_k}(h_{\alpha_k \beta_k}) \Lambda^{l_k, \alpha_k \beta_k} (\ph, h)
		+
		\Zb^{L_k}(\mathring{Z}^m_{1,m}) \Lambda^{l_k} (\ph, h)
		\right)
		\Lambda(\ph,h)
		d \ph
		\Bigg\vert
		\\
		\leq
		C\!\!\!\!\!\!\!\!\!\!\!
 \sum_{n_1+n_2 \leq \left\lfloor \frac{\vert I \vert}{2} \right\rfloor + 1}
		\!\!\!\int
		\left\vert
		(\partial^{n_1}_x \partial^{n_2}_{\ph} f)(t_0,X(t_0),\Ph(t_0))
		\right\vert
		\Big(
		1
		+\!\!\!\!\!
		\sum_{\vert J \vert \leq \vert I_2 \vert -1}\!\!\!\!\!
		t \left\vert \left( Z^J \Gamma \right)(t,x) \right\vert
		+
		\sum_{\vert J \vert \leq \vert I_2 \vert}
		\!\!\int_{t_0}^t\!\!
		\left\vert \left( Z^J \Gamma \right) \left(s',s'\frac{x}{t} \right) \right\vert
		ds'
		\Big)
		d\ph.
	\end{multline*}
	The proof then follows from Proposition \ref{prop:ZTschematic}.
\end{proof}

\subsection{Determinants and changes of variables}

In the proof of Theorem \ref{thm:mainL2}, the change of variables $(t,x,\ph) \mapsto (t,x,y)$ will be used, where,
\[
	y^i(t,x,\ph) = X(t_0,t,x,\ph)^i,
\]
for $i=1,2,3$, along with several other changes of variables.  See the proof of Proposition \ref{prop:mainL21} and Proposition \ref{prop:mainL22}.  A first step towards controlling the determinants of these changes is contained in the following.

\begin{lemma} \label{lem:changes}
	Suppose $t\geq t_0$, $\vert x \vert \leq ct$, and $(t,x,\ph) \in \supp(f)$.  Then, if the assumption \eqref{eq:Gammasuppf} holds and $\varepsilon$ is sufficiently small,
	\begin{equation} \label{eq:dXdp}
		\left\vert \frac{\partial X^i}{\partial \ph^j}(s,t,x,\ph) + (t-s) \delta^i_j \right\vert
		\lesssim
		\frac{\varepsilon t}{s^a},
		\qquad
		\left\vert \frac{\partial \Ph^i}{\partial \ph^j}(s,t,x,\ph) - \delta^i_j \right\vert
		\lesssim
		\frac{\varepsilon t}{s^{1+a}},
	\end{equation}
	and
	\begin{equation} \label{eq:dXdx}
		\left\vert \frac{\partial X^i}{\partial x^j}(s,t,x,\ph) - \delta^i_j \right\vert
		\lesssim
		\frac{\varepsilon}{s^a},
		\qquad
		\left\vert \frac{\partial \Ph^i}{\partial x^j}(s,t,x,\ph) \right\vert
		\lesssim
		\frac{\varepsilon}{s^{1+a}},
	\end{equation}
	for all $t_0 \leq s \leq t$, and $i,j=1,2,3$.
\end{lemma}

\begin{proof}
	From the equations \eqref{eq:geodesictimenormalized} and the estimates
	\[
		s^{1+a} \left\vert \Gamma^{\alpha}_{\beta \gamma}(s,X(s)) \right\vert
		+
		s^{2+a} \left\vert \partial_k \Gamma^{\alpha}_{\beta \gamma}(s,X(s)) \right\vert
		\lesssim
		\varepsilon,
	\]
	it follows that,
	\[
		\left\vert \frac{d}{ds} \frac{\partial X^i}{\partial \ph^j} \right\vert
		\leq
		\left\vert \frac{\partial \Ph^i}{\partial \ph^j} \right\vert,
		\qquad
		\left\vert \frac{d}{ds} \frac{\partial \Ph^i}{\partial \ph^j}(s) \right\vert
		\lesssim
		\varepsilon \sum_{k,l=1}^3
		\left(
		\frac{1}{s^{2+a}}
		\left\vert \frac{\partial X^k}{\partial \ph^l}(s) \right\vert
		+
		\frac{1}{s^{1+a}}
		\left\vert \frac{\partial \Ph^k}{\partial \ph^l}(s) \right\vert
		\right).
	\]
	Integrating the second inequality backwards from $s=t$, and using the fact that $\frac{\partial \Ph^k}{\partial \ph^l}(t,t,x,\ph) = \delta^i_j$, gives,
	\begin{align*}
		\left\vert \frac{\partial \Ph^i}{\partial \ph^j}(s,t,x,\ph) - \delta^i_j \right\vert
		&
		\lesssim
		\varepsilon \sum_{k,l=1}^3
		\int_s^t
		\frac{1}{\tilde{s}^{2+a}}
		\left\vert \frac{\partial X^k}{\partial \ph^l}(\tilde{s}) \right\vert
		+
		\frac{1}{\tilde{s}^{1+a}}
		\left\vert \frac{\partial \Ph^k}{\partial \ph^l}(\tilde{s}) \right\vert
		d \tilde{s}
		\\
		&
		\lesssim
		\frac{\varepsilon}{s^a}
		+
		\varepsilon \sum_{k,l=1}^3
		\int_s^t
		\frac{1}{\tilde{s}^{2+a}}
		\left\vert \frac{\partial X^k}{\partial \ph^l}(\tilde{s}) \right\vert
		+
		\frac{1}{\tilde{s}^{1+a}}
		\left\vert \frac{\partial \Ph^k}{\partial \ph^l}(\tilde{s}) - \delta^k_l \right\vert
		d \tilde{s}.
	\end{align*}
	Dividing by $s^{1+a}$, integrating again backwards from $s=t$, and summing over $i,j$ gives
	\begin{align*}
		&
		\sum_{k,l=1}^3 \int_s^t \frac{1}{(s')^{1+a}}
		\left\vert \frac{\partial \Ph^k}{\partial \ph^l}(s') - \delta^k_l \right\vert
		d s'
		\\
		&
		\qquad \qquad
		\lesssim
		\int_s^t \frac{1}{(s')^{1+a}}
		\left(
		\frac{\varepsilon}{(s')^a}
		+
		\varepsilon \sum_{k,l=1}^3
		\int_{s'}^t
		\frac{1}{\tilde{s}^{2+a}}
		\left\vert \frac{\partial X^k}{\partial \ph^l}(\tilde{s}) \right\vert
		+
		\frac{1}{\tilde{s}^{1+a}}
		\left\vert \frac{\partial \Ph^k}{\partial \ph^l}(\tilde{s}) - \delta^k_l \right\vert
		d \tilde{s}
		\right)
		ds'
		\\
		&
		\qquad \qquad
		\lesssim
		\left(
		\frac{\varepsilon}{s^a}
		+
		\varepsilon \sum_{k,l=1}^3
		\int_{s}^t
		\frac{1}{\tilde{s}^{2+a}}
		\left\vert \frac{\partial X^k}{\partial \ph^l}(\tilde{s}) \right\vert
		+
		\frac{1}{\tilde{s}^{1+a}}
		\left\vert \frac{\partial \Ph^k}{\partial \ph^l}(\tilde{s}) - \delta^k_l \right\vert
		d \tilde{s}
		\right)
		\int_s^t \frac{1}{(s')^{1+a}}
		ds'
		\\
		&
		\qquad \qquad
		\lesssim
		\frac{\varepsilon}{s^a}
		+
		\varepsilon \sum_{k,l=1}^3
		\int_{s}^t
		\frac{1}{\tilde{s}^{2+a}}
		\left\vert \frac{\partial X^k}{\partial \ph^l}(\tilde{s}) \right\vert
		+
		\frac{1}{\tilde{s}^{1+a}}
		\left\vert \frac{\partial \Ph^k}{\partial \ph^l}(\tilde{s}) - \delta^k_l \right\vert
		d \tilde{s}
	\end{align*}
	
	Taking $\varepsilon$ sufficiently small then gives,
	\[
		\sum_{k,l=1}^3 \int_s^t \frac{1}{\tilde{s}^{1+a}}
		\left\vert \frac{\partial \Ph^k}{\partial \ph^l}(\tilde{s}) - \delta^k_l \right\vert
		d \tilde{s}
		\lesssim
		\frac{\varepsilon}{s^a}
		+
		\varepsilon \sum_{k,l=1}^3
		\int_s^t
		\frac{1}{\tilde{s}^{2+a}}
		\left\vert \frac{\partial X^k}{\partial \ph^l}(\tilde{s}) \right\vert
		d \tilde{s}.
	\]
	Inserting back into the above bound gives
	\begin{equation} \label{eq:dPdpmidproof}
		\left\vert \frac{\partial \Ph^i}{\partial \ph^j}(s,t,x,\ph) - \delta^i_j \right\vert
		\lesssim
		\frac{\varepsilon}{s^a}
		+
		\varepsilon \sum_{k,l=1}^3
		\int_s^t
		\frac{1}{\tilde{s}^{2+a}}
		\left\vert \frac{\partial X^k}{\partial \ph^l}(\tilde{s}) \right\vert
		d \tilde{s}.
	\end{equation}
	Integrating this bound backwards from $s=t$, and using the fact that $\frac{\partial X^i}{\partial \ph^j}(t,t,x,\ph) = 0$, gives
	\begin{align*}
		\left\vert \frac{\partial X^i}{\partial \ph^j}(s,t,x,\ph) + (t-s) \delta^i_j \right\vert
		&
		\lesssim
		\varepsilon t^{1-a}
		+
		\varepsilon \sum_{k,l=1}^3
		\int_s^t
		\frac{1}{\tilde{s}^{1+a}}
		\left\vert \frac{\partial X^k}{\partial \ph^l}(\tilde{s}) \right\vert
		d \tilde{s}
		\\
		&
		\lesssim
		\frac{\varepsilon t}{s^a}
		+
		\varepsilon
		\sum_{k,l=1}^3
		\int_s^t
		\frac{1}{\tilde{s}^{1+a}}
		\left\vert \frac{\partial X^k}{\partial \ph^l}(\tilde{s}) + (t-s) \delta^k_l \right\vert
		d \tilde{s},
	\end{align*}
	where the fact that, for any function $\lambda(s)$,
	\[
		\int_s^t \int_{s'}^t \lambda(\tilde{s}) d\tilde{s} ds'
		=
		\int_s^t \int_s^t \chi_{\{ s'\leq \tilde{s}\}} ds' \lambda(\tilde{s}) d \tilde{s}
		=
		\int_s^t (\tilde{s} - s) \lambda(\tilde{s}) d \tilde{s},
	\]
	has been used (here $\chi_{\{ s'\leq \tilde{s}\}}$ is the indicator function of the interval $[s',\tilde{s}]$).  Again, dividing by $s^{1+a}$, integrating backwards from $s=t$, summing over $i,j$ and taking $\varepsilon$ small gives,
	\[
		\sum_{k,l=1}^3
		\int_s^t
		\frac{1}{\tilde{s}^{1+a}}
		\left\vert \frac{\partial X^k}{\partial \ph^l}(\tilde{s}) + (t-s) \delta^k_l \right\vert
		d \tilde{s}
		\lesssim
		\frac{\varepsilon t}{s^{2a}}.
	\]
	Inserting back into the above bound then gives
	\begin{equation*}
		\left\vert \frac{\partial X^i}{\partial \ph^j}(s,t,x,\ph) + (t-s) \delta^i_j \right\vert
		\lesssim
		\frac{\varepsilon t}{s^a},
	\end{equation*}
	and inserting this into \eqref{eq:dPdpmidproof} gives the second bound of \eqref{eq:dXdp}.
	
	In a similar manner, it is straightforward to show
	\[
		\left\vert \frac{d}{ds} \frac{\partial \Ph^i}{\partial x^j}(s) \right\vert
		\lesssim
		\varepsilon \sum_{k,l=1}^3
		\left(
		\frac{1}{s^{2+a}}
		\left\vert \frac{\partial X^k}{\partial x^l}(s) \right\vert
		+
		\frac{1}{s^{1+a}}
		\left\vert \frac{\partial \Ph^k}{\partial x^l}(s) \right\vert
		\right),
	\]
	and, using the final conditions $\frac{\partial \Ph^i}{\partial x^j}(t,t,x,\ph) = 0$, $\frac{\partial X^i}{\partial x^j}(t,t,x,\ph) = \delta^i_j$, that \eqref{eq:dXdx} holds.
\end{proof}

The properties of these changes are collected in the following proposition.

\begin{proposition} \label{prop:determinants}
	For fixed $t,x$ with $t \geq t_0 +1$, $\vert x \vert \leq c t$, if the assumptions \eqref{eq:Gammasuppf} hold and $\varepsilon$ is sufficiently small then, for $\ph$ such that $(t,x,\ph) \in\supp(f)$, the change of variables $\ph \mapsto y:= X(t_0,t,x,\ph)$ satisfies
	\begin{equation} \label{eq:detdpdy}
		\left\vert
		\det \left( \frac{\partial \ph^i}{\partial y^j} \right)
		\right\vert
		\leq
		\frac{C}{t^3}.
	\end{equation}
	Define
	\[
		z_1^i(s,t,x,y) := X(s,t,x,\ph(t,x,y))^i,
		\qquad
		z_2^i(\sigma,s,t,x,y)
		:=
		\sigma s \frac{x^i}{t}
		+
		(1-\sigma) X(s,t,x,\ph(t,x,y))^i,
	\]
	for $i=1,2,3$.  If $t \geq t_0 +1$ and $t_0 \leq s \leq t_0+\frac{1}{2}$, then the change of variables $(x,y) \mapsto (x,z_1(s,t,x,y))$ satisfies,
	\begin{equation} \label{eq:detdz1dy}
		\left\vert
		\det \left( \frac{\partial z_1^i(s,t)}{\partial y^j} \right)^{-1}
		\right\vert
		\leq
		C.
	\end{equation}
	If $t_0+\frac{1}{2} \leq s \leq t$ then the change of variables $(x,y) \mapsto (z_1(s,t,x,y),y)$ satisfies
	\begin{equation} \label{eq:detdz1dx}
		\left\vert
		\det \left( \frac{\partial z_1^i(s,t)}{\partial x^j} \right)^{-1}
		\right\vert
		\leq
		C \left( \frac{t}{s} \right)^3.
	\end{equation}
	Finally, if $t_0+\frac{1}{2}\leq s \leq t$ and $0 \leq \sigma \leq 1$, then the change of variables $(x,y) \mapsto (z_2(\sigma,s,t,x,y),y)$ satisfies
	\begin{equation} \label{eq:detdz2dx}
		\left\vert
		\det \left( \frac{\partial z_2^i(\sigma,s,t)}{\partial x^j} \right)^{-1}
		\right\vert
		\leq
		C \left( \frac{t}{s} \right)^3.
	\end{equation}
	Moreover, for $t\geq t_0$, the determinant of the 6 by 6 matrix $\kappa$ satisfies,
	\begin{equation} \label{eq:kappamatrix}
		\left\vert
		\det
		\kappa (t_0,t,x,\ph)
		-
		1
		\right\vert
		\lesssim
		\varepsilon,
		\qquad
		\text{where }
		\kappa
		=
		\begin{pmatrix}
		\frac{\partial X}{\partial x} & \frac{\partial X}{\partial \ph} \\
		\frac{\partial \Ph}{\partial x} & \frac{\partial \Ph}{\partial \ph}
		\end{pmatrix}.
	\end{equation}
\end{proposition}

\begin{proof}
	Setting $s=t_0$ in \eqref{eq:dXdp}, it follows that
	\begin{equation} \label{eq:dydp}
		\left\vert \frac{\partial y^i}{\partial \ph^j} + (t-t_0) \delta^i_j \right\vert
		\lesssim
		\varepsilon t,
	\end{equation}
	and, if $\varepsilon$ is sufficiently small,
	\[
		\left\vert \det \left( \frac{\partial y^i}{\partial \ph^j} \right) + (t-t_0)^3 \right\vert
		\lesssim
		\varepsilon t^3.
	\]
	Since $t\geq t_0 +1$, the bound \eqref{eq:detdpdy} follows.
	
	For the remaining bounds of the proposition, it is necessary to consider $\ph^i$ as a function of $t,x,y$ (using the above bound and the Inverse Function Theorem) and estimate $\frac{\partial \ph^i}{\partial x^j}$ and $\frac{\partial \ph^i}{\partial y^j}$.  Clearly the matrix $\left( \frac{\partial \ph^i}{\partial y^j}\right)$ is the inverse of the matrix $\left( \frac{\partial y^i}{\partial \ph^j}\right)$, and hence it follows from \eqref{eq:dydp} that
	\[
		\left\vert \frac{\partial \ph^i}{\partial y^j} + \frac{\delta^i_j}{t-t_0} \right\vert
		\lesssim
		\frac{\varepsilon}{t}.
	\]
	Also,
	\[
		0
		=
		\frac{\partial \ph^i(t,x,y(t,x,\ph))}{\partial x^j}
		=
		\frac{\partial \ph^i}{\partial x^j}(t,x,y) + \frac{\partial \ph^i}{\partial y^k}(t,x,y) \frac{\partial y^k}{\partial x^j},
	\]
	so
	\[
		\frac{\partial \ph^i}{\partial x^j}
		=
		-
		\frac{\partial \ph^i}{\partial y^k} \frac{\partial y^k}{\partial x^j}.
	\]
	Setting $s=t_0$ in \eqref{eq:dXdx} gives
	\[
		\left\vert \frac{\partial y^i}{\partial x^j} - \delta^i_j \right\vert
		\lesssim
		\varepsilon,
	\]
	and so, since
	\[
		\frac{\partial \ph^i}{\partial x^j} - \frac{\delta^i_j}{t-t_0}
		=
		-
		\left( \frac{\partial \ph^i}{\partial y^k} + \frac{\delta_j^k}{t-t_0} \right)
		\left( \frac{\partial y^k}{\partial x^j} - \delta^k_j \right)
		+
		\frac{1}{t-t_0} \left( \frac{\partial y^i}{\partial x^j} - \delta^i_j \right)
		-
		\left( \frac{\partial \ph^i}{\partial y^j} + \frac{\delta^i_j}{t-t_0} \right),
	\]
	it follows that
	\[
		\left\vert \frac{\partial \ph^i}{\partial x^j} - \frac{\delta^i_j}{t-t_0} \right\vert
		\lesssim
		\frac{\varepsilon}{t},
	\]
	since $t \geq t_0 +1$.
	
	Now
	\[
		\frac{\partial z_1^i}{\partial y^j}
		=
		\frac{\partial X^i(s,t,x,\ph(t,x,y))}{\partial y^j}
		=
		\frac{\partial X^i}{\partial \ph^k} \frac{\partial \ph^k}{\partial y^j},
	\]
	and hence, inserting the above bounds,
	\begin{align*}
		&
		\left\vert \frac{\partial z_1^i}{\partial y^j} - \frac{t-s}{t-t_0} \delta^i_j \right\vert
		\\
		&
		\quad
		=
		\left\vert
		\left( \frac{\partial X^i}{\partial \ph^k} + (t-s) \delta^i_k \right)
		\left( \frac{\partial \ph^k}{\partial y^j} - \frac{\delta^k_j}{t-t_0} \right)
		-
		(t-s) \left( \frac{\partial \ph^i}{\partial y^j} - \frac{\delta^i_j}{t-t_0} \right)
		+
		\frac{1}{t-t_0} \left( \frac{\partial X^i}{\partial \ph^j} + (t-s) \delta^i_j \right)
		\right\vert
		\lesssim
		\varepsilon.
	\end{align*}
	It follows that
	\[
		\left\vert \det \left( \frac{\partial z_1^i}{\partial y^j} \right) - \left( \frac{t-s}{t-t_0} \right)^3 \right\vert
		\lesssim
		\varepsilon,
	\]
	and, if $\varepsilon$ is sufficiently small, the bound \eqref{eq:detdz1dy} follows for $t_0 \leq s \leq t_0 + \frac{1}{2}$.  Similarly,
	\[
		\frac{\partial z_1^i}{\partial x^j}
		=
		\frac{\partial X^i(s,t,x,\ph(t,x,y))}{\partial x^j}
		=
		\frac{\partial X^i}{\partial x^j}
		+
		\frac{\partial X^i}{\partial \ph^k} \frac{\partial \ph^k}{\partial x^j},
	\]
	and
	\begin{multline*}
		\left\vert \frac{\partial z_1^i}{\partial x^j} - \frac{s-t_0}{t-t_0} \delta^i_j \right\vert\\
		=
		\Bigg\vert
		\frac{\partial X^i}{\partial x^j} - \delta^i_j
		+
		\left( \frac{\partial X^i}{\partial \ph^k} + (t-s) \delta^i_k \right)
		\bigg( \frac{\partial \ph^k}{\partial x^j} - \frac{\delta^k_j}{t-t_0} \bigg)
		-
		(t-s) \left( \frac{\partial \ph^i}{\partial x^j} - \frac{\delta^i_j}{t-t_0} \right)
		+
		\frac{1}{t-t_0} \left( \frac{\partial X^i}{\partial \ph^j} + (t-s) \delta^i_j \right)
		\Bigg\vert
		\\
		\lesssim
		\varepsilon,
	\end{multline*}
	which, if $\varepsilon$ is suitably small, implies
	\[
		\left\vert \det \left( \frac{\partial z_1^i}{\partial x^j} \right) - \left( \frac{s-t_0}{t-t_0} \right)^3 \right\vert
		\lesssim
		\varepsilon,
	\]
	and the bound \eqref{eq:detdz1dx} follows.
	
	Finally,
	\[
		\frac{\partial z_1^i}{\partial x^j}
		=
		\frac{\sigma s}{t} \delta^i_j
		+
		(1 - \sigma) \frac{\partial z_1^i}{\partial x^j},
	\]
	and so
	\[
		\left\vert
		\frac{\partial z_2^i}{\partial x^j}
		-
		\left( \sigma \frac{s}{t} + (1-\sigma) \frac{s-t_0}{t-t_0} \right) \delta^i_j
		\right\vert
		\lesssim
		\varepsilon,
	\]
	from which the bound \eqref{eq:detdz2dx} follows.
	
	We now prove the bound \eqref{eq:kappamatrix}.
	For $t_0 \leq s \leq t$ denote $W(s) = \left( X(s,t,x,\ph), \Ph(s,t,x,\ph) \right)$ and $w = (x,\ph)$, so that $\kappa = \frac{\partial W}{\partial w}$.  Now
	\[
		\frac{d}{ds}  W(s)
		=F\big(s,W(s)\big),\quad\text{where}\quad
		F\big(s,W\big)=\Big(\Ph,
		\big( \Ph^i \Ph^{\alpha} \Ph^{\beta} \Gamma_{\alpha \beta}^0 (s,X)
     - \Ph^{\alpha} \Ph^{\beta} \Gamma_{\alpha \beta}^i (s,X) \big)\Big),
	\]
With  $M=\partial W/\partial w$ we have
	\begin{equation}
		\frac{d}{ds} \det M(s) = \tr \left( M^{-1} \frac{dM}{ds} (s) \right) \cdot \det M(s),\quad\text{where}\quad
\tr \left( M^{-1} \frac{dM}{ds}\right)
=\frac{\partial w^j}{\partial W^i}\frac{\partial F^i}{\partial w^j}=\frac{\partial F^i}{\partial W^i} .
	\end{equation}
We have
	\[
		\frac{\partial F^i}{\partial W^i}
		=
		3 \Ph^{\alpha} \Ph^{\beta} \Gamma_{\alpha \beta}^0 (s,X)
		+
		2 \Ph^{i} \Ph^{\beta} \Gamma_{i \beta}^0 (s,X)
		-
		2 \Ph^{\beta} \Gamma_{i \beta}^i (s,X),
	\]
	and so,
	\[
		\left\vert \frac{d}{ds} \det \left( \frac{\partial W}{\partial w}(s)  \right) \right\vert
		\leq
		\frac{C\varepsilon}{s^{1+a}} \left\vert \det \left( \frac{\partial W}{\partial w}(s)  \right) \right\vert.
	\]
	The bound \eqref{eq:kappamatrix} then follows from the Gr\"{o}nwall inequality.
\end{proof}

\subsection{$L^1$ and $L^2$ estimates of components of the energy momentum tensor}

The main part of the proof of Theorem \ref{thm:mainL2} is contained in Propositions \ref{prop:mainL21} and \ref{prop:mainL22} below.  The following Lemma will be used.

\begin{lemma} \label{lem:suppfdecay}
	Suppose $\pi(\supp(f_0)) \subset \{\vert x \vert \leq K \}$, $t\geq t_0+1$ and the assumptions \eqref{eq:Gammasuppf} hold.  Then
	\[
		\int \chi_{\supp(f)}(t,x,\ph) d \ph
		\lesssim
		\frac{1}{t^3},
	\]
	where $\chi_{\supp(f)}(t,x,\ph)$ is the characteristic function of $\supp(f)$.
\end{lemma}

\begin{proof}
	Since $f$ solves the Vlasov equation, $\chi_{\supp(f)}(t,x,\ph) = \chi_{\supp(f)}(t_0,X(t_0),\Ph(t_0))$ and
	\[
		\int \chi_{\supp(f)}(t_0,X(t_0),\Ph(t_0)) d \ph
		\leq
		\frac{C}{t^3} \int \chi_{\supp(f)}(t_0,y,\Ph(t_0,t,x,\ph(t,x,y))) d y
		\leq
		\frac{CK}{t^3},
	\]
	where the change of variables $\ph \mapsto y= X(t_0)$ and the bound \eqref{eq:detdpdy} have been used.
\end{proof}

\begin{proposition} \label{prop:mainL21}
	Suppose $\pi(\supp(f)) \subset \{\vert x \vert \leq ct\}$ and consider $t\geq t_0+1$.  If the assumptions \eqref{eq:Gammasuppf} hold and $\varepsilon$ is sufficiently small then, for any multi index $I$ with $\vert I \vert \leq N-1$ and each $\mu, \nu = 0,1,2,3$,
	\begin{align*}
		\Vert Z^I T^{\mu \nu} (t,\cdot) \Vert_{L^{2}}
		&
		\leq
		\frac{C\mathcal{V}_{\vert I \vert}}{t^{\frac{3}{2}}}
		+
		C \mathbb{D}_{\left\lfloor \frac{\vert I \vert}{2} \right\rfloor +1}
		\left(
		\sum_{\vert J \vert \leq \vert I \vert -1}
		\frac{\Vert ( Z^J \Gamma)(t,\cdot) \Vert_{L^2}}{t^{1+a}}
		+
		\sum_{\vert J \vert \leq \vert I \vert + 1}
		\frac{1}{t^{\frac{3}{2}}}
		\int_{t_0}^t
		\frac{\Vert ( Z^J \Gamma)(s,\cdot) \Vert_{L^2}}{s^{\frac{1}{2}+a}}
		ds
		\right),
		\\
		\Vert Z^I T^{\mu \nu} (t,\cdot) \Vert_{L^{1}}
		&
		\leq
		C\mathcal{V}_{\vert I \vert}
		+
		C \mathbb{D}_{\left\lfloor \frac{\vert I \vert}{2} \right\rfloor +1}
		\left(
		\sum_{\vert J \vert \leq \vert I \vert -1}
		t^{\frac{1}{2} - a}
		\Vert ( Z^J \Gamma)(t,\cdot) \Vert_{L^2}
		+
		\sum_{\vert J \vert \leq \vert I \vert + 1}
		\int_{t_0}^t
		\frac{\Vert ( Z^J \Gamma)(s,\cdot) \Vert_{L^2}}{s^{\frac{1}{2}+a}}
		ds
		\right),
	\end{align*}
	where $Z^I$ is a product of $\vert I \vert$ of the vector fields $\Omega_{ij}, B_i, S$.
\end{proposition}

\begin{proof}
	Given any function $F(t,x)$ it follows from Proposition \ref{prop:ZITpointwise} that
	\begin{align*}
		&
		\Vert Z^I T^{\mu \nu} (t,\cdot) F(t,\cdot) \Vert_{L^1}
		\leq
		C \sum_{n_1+n_2 \leq \vert I \vert}
		\int_{\vert x \vert \leq c t +K}
		\int
		\left\vert
		(\partial^{n_1}_x \partial^{n_2}_{\ph} f)(t_0,X(t_0),\Ph(t_0))
		\right\vert
		d\ph
		\vert F(t,x) \vert
		dx
		\\
		&
		\qquad \qquad
		+
		C \sum_{n_1+n_2 \leq \left\lfloor \frac{\vert I \vert}{2} \right\rfloor + 1}
		\int_{\vert x \vert \leq ct + K}
		\int
		\left\vert
		(\partial^{n_1}_x \partial^{n_2}_{\ph} f)(t_0,X(t_0),\Ph(t_0))
		\right\vert
		\Bigg(
		\sum_{\vert J \vert \leq \vert I \vert -1}
		t^{2-a} \vert (Z^J \Gamma)(t,x) \vert
		\\
		&
		\qquad \qquad
		+
		\sum_{\vert J \vert \leq \vert I \vert }
		\int_{t_0}^t s^{1-a}
		\left\vert (Z^J \Gamma) \left( s, s\frac{x}{t} \right) \right\vert
		+
		s
		\left\vert
		(Z^J \Gamma) \left( s, s\frac{x}{t} \right)
		-
		(Z^J \Gamma) \left( s, X(s) \right)
		\right\vert
		ds
		\Bigg)
		d\ph
		\vert F(t,x) \vert
		dx.
	\end{align*}
	Given $n_1+n_2 \leq \vert I \vert$, it follows from Lemma \ref{lem:suppfdecay} and the bound \eqref{eq:kappamatrix} that
	\begin{align*}
		&
		\int_{\vert x \vert \leq c t +K}
		\int
		\left\vert
		(\partial^{n_1}_x \partial^{n_2}_{\ph} f)(t_0,X(t_0),\Ph(t_0))
		\right\vert
		d\ph
		\vert F(t,x) \vert
		dx
		\\
		&
		\qquad \qquad
		\leq
		\int
		\left(\int \chi_{\supp(f)}(t,x,\ph) d\ph \right)^{\frac{1}{2}}
		\left(
		\int
		\left\vert (\partial^{n_1}_x \partial^{n_2}_{\ph} f)(t_0,X(t_0),\Ph(t_0)) \right\vert^2
		d\ph
		\right)^{\frac{1}{2}}
		\vert F(t,x) \vert
		dx
		\\
		&
		\qquad \qquad
		\leq
		\frac{C}{t^{\frac{3}{2}}}
		\left(
		\int
		\int
		\left\vert (\partial^{n_1}_x \partial^{n_2}_{\ph} f)(t_0,X(t_0),\Ph(t_0)) \right\vert^2
		d\ph
		dx
		\right)^{\frac{1}{2}}
		\Vert F(t,\cdot) \Vert_{L^2}
		\leq
		\frac{C \mathcal{V}_{n_1+n_2}}{t^{\frac{3}{2}}} \Vert F(t,\cdot) \Vert_{L^2}.
	\end{align*}
	Similarly, for any multi index $J$,
	\begin{align*}
		&
		\int_{\vert x \vert \leq c t +K}
		\int
		\left\vert
		(\partial^{n_1}_x \partial^{n_2}_{\ph} f)(t_0,X(t_0),\Ph(t_0))
		\right\vert
		d\ph
		t^{2-a}
		\left\vert (Z^J \Gamma)(t,x) \right\vert
		\vert F(t,x) \vert
		dx
		\\
		&
		\quad
		\leq
		\int_{\vert x \vert \leq c t +K}
		\int_{\vert y \vert \leq K}
		\left\vert
		(\partial^{n_1}_x \partial^{n_2}_{\ph} f)
		(t_0,y,\Ph(t_0,t,x,\ph(t,x,y)))
		\right\vert
		\left\vert
		\det \left( \frac{\partial \ph^i}{\partial y^j} \right)
		\right\vert
		d y
		t^{2-a}
		\left\vert (Z^J \Gamma)(t,x) \right\vert
		\vert F(t,x) \vert
		dx
		\\
		&
		\quad
		\leq
		\frac{C \mathbb{D}_{n_1+n_2}}{t^3} t^{2-a}
		\int_{\vert x \vert \leq ct + K} \left\vert (Z^J \Gamma)(t,x) \right\vert \vert F(t,x) \vert dx
		\leq
		\frac{C \mathbb{D}_{n_1+n_2}}{t^{1+a}}
		\Vert (Z^J \Gamma)(t,\cdot) \Vert_{L^2}
		\Vert F(t,\cdot) \Vert_{L^2},
	\end{align*}
	using the bound \eqref{eq:detdpdy}.  Now, for the third term,
	\begin{multline*}
		\int_{\vert x \vert \leq c t +K}
		\int
		\left\vert
		(\partial^{n_1}_x \partial^{n_2}_{\ph} f)(t_0,X(t_0),\Ph(t_0))
		\right\vert
		d\ph
		\int_{t_0}^t s^{1-a}
		\left\vert (Z^J \Gamma) \left( s, s\frac{x}{t} \right) \right\vert
		ds
		\vert F(t,x) \vert
		dx
		\\
		\leq
		\frac{C \mathbb{D}_{n_1+n_2}}{t^3}
		\int_{t_0}^t s^{1-a} \int_{\vert y \vert \leq K} \int_{\vert x \vert \leq c t +K}
		\left\vert (Z^J \Gamma) \left( s, s\frac{x}{t} \right) \right\vert
		\vert F(t,x) \vert
		dx dy ds
		\\
		\leq
		\frac{C \mathbb{D}_{n_1+n_2}}{t^3}
		\int_{t_0}^t s^{1-a}
		\left(
		\int_{\vert x \vert \leq c t +K}
		\left\vert (Z^J \Gamma) \left( s, s\frac{x}{t} \right) \right\vert^2
		dx
		\right)^{\frac{1}{2}}
		\Vert F(t,\cdot) \Vert_{L^2}
		ds
		\\
		\leq
		\frac{C \mathbb{D}_{n_1+n_2}}{t^{\frac{3}{2}}}
		\int_{t_0}^t
		\frac{
		\left\Vert (Z^J \Gamma) \left( s, \cdot \right) \right\Vert_{L^2}
		}
		{
		s^{\frac{1}{2}+a}
		}
		ds
		\Vert F(t,\cdot) \Vert_{L^2},
	\end{multline*}
	where the change of variables $x^i \mapsto z^i := s \frac{x^i}{t}$ has been used, along with the fact that $\det \left( \frac{\partial x^i}{\partial z^j} \right) = \left( \frac{t}{s} \right)^3$.
	
	For the final term, first write,
	\begin{multline*}
		\sum_{\vert J \vert \leq \vert I\vert}
		\int_{t_0}^t
		s
		\left\vert
		(Z^J \Gamma) \left( s, s\frac{x}{t} \right)
		-
		(Z^J \Gamma) \left( s, X(s) \right)
		\right\vert
		ds
		\\
		\leq\!\!
		\sum_{\vert J \vert \leq \vert I\vert}
		\int_{t_0}^{t_0+\frac{1}{2}}\!\!\!\!
		s
		\left\vert
		(Z^J \Gamma) \left( s, s\frac{x}{t} \right)
		\right\vert
		+
		s
		\left\vert
		(Z^J \Gamma) \left( s, X(s) \right)
		\right\vert
		ds
		+\!\!
		\sum_{\vert J \vert \leq \vert I\vert}
		\int_{t_0+\frac{1}{2}}^t\!\!\!\!
		s
		\left\vert
		(Z^J \Gamma) \left( s, s\frac{x}{t} \right)
		-
		(Z^J \Gamma) \left( s, X(s) \right)
		\right\vert
		ds.
	\end{multline*}
	As above,
	\begin{multline*}
		\int_{\vert x \vert \leq c t +K}
		\int
		\left\vert
		(\partial^{n_1}_x \partial^{n_2}_{\ph} f)(t_0,X(t_0),\Ph(t_0))
		\right\vert
		d\ph
		\int_{t_0}^{t_0+\frac{1}{2}} s
		\left\vert (Z^J \Gamma) \left( s, s\frac{x}{t} \right) \right\vert
		ds
		\vert F(t,x) \vert
		dx
		\\
		\leq
		\frac{C \mathbb{D}_{n_1+n_2}}{t^{\frac{3}{2}}}
		\int_{t_0}^{t_0+\frac{1}{2}}
		\left\Vert (Z^J \Gamma) \left( s, \cdot \right) \right\Vert_{L^2}
		ds
		\Vert F(t,\cdot) \Vert_{L^2}
		\leq
		\frac{C \mathbb{D}_{n_1+n_2}}{t^{\frac{3}{2}}}
		\int_{t_0}^t
		\frac{
		\left\Vert (Z^J \Gamma) \left( s, \cdot \right) \right\Vert_{L^2}
		}
		{
		s^{\frac{1}{2}+a}
		}
		ds
		\Vert F(t,\cdot) \Vert_{L^2},
	\end{multline*}
	and, estimating the term with $X(s,t,x,\ph)$ slightly differently,
	\begin{multline*}
		\int_{\vert x \vert \leq c t +K}
		\int
		\left\vert
		(\partial^{n_1}_x \partial^{n_2}_{\ph} f)(t_0,X(t_0),\Ph(t_0))
		\right\vert
		\int_{t_0}^{t_0+\frac{1}{2}} s
		\left\vert (Z^J \Gamma) \left( s, X(s) \right) \right\vert
		ds
		d\ph
		\vert F(t,x) \vert
		dx
		\\
		\leq
		\frac{C \mathbb{D}_{n_1+n_2}}{t^3}
		\int_{t_0}^{t_0+\frac{1}{2}} \int_{\vert x \vert \leq c t +K}
		\vert F(t,x) \vert
		\int_{\vert y \vert \leq K}
		\left\vert (Z^J \Gamma) \left( s, X(s) \right) \right\vert
		dy dx ds
		\\
			\leq
		\frac{C \mathbb{D}_{n_1+n_2}}{t^3}
		\int_{t_0}^{t_0+\frac{1}{2}} \int_{\vert x \vert \leq c t +K}
		\vert F(t,x) \vert
		\left(
		\int_{\vert y \vert \leq K}
		\left\vert (Z^J \Gamma) \left( s, X(s) \right) \right\vert^2
		dy
		\right)^{\frac{1}{2}}
		dx ds
		\\
		\leq
		\frac{C \mathbb{D}_{n_1+n_2}}{t^{\frac{3}{2}}}
		\int_{t_0}^{t_0+\frac{1}{2}}
		\left\Vert (Z^J \Gamma) \left( s, \cdot \right) \right\Vert_{L^2}
		ds
		\Vert F(t,\cdot) \Vert_{L^2}
		\leq
		\frac{C \mathbb{D}_{n_1+n_2}}{t^{\frac{3}{2}}}
		\int_{t_0}^{t}
		\frac{\left\Vert (Z^J \Gamma) \left( s, \cdot \right) \right\Vert_{L^2}}{s^{\frac{1}{2}+a}}
		ds
		\Vert F(t,\cdot) \Vert_{L^2},
	\end{multline*}
	where now the change of variables $y^i \mapsto z_1^i := X(s,t,x,\ph(t,x,y))^i$ has been used, together with Proposition \ref{prop:determinants}, which guarantees that $\left\vert \det \left( \frac{\partial y^i}{\partial z_1^j} \right) \right\vert \leq C$ when $t_0\leq s \leq t_0+\frac{1}{2}$.  Note that this term was estimated slightly differently since the change of variables $x^i \mapsto X(s,t,x,\ph(t,x,y))^i$ breaks down as $s\to t_0$, since then $X(s,t,x,\ph(t,x,y))\to y$.
	
	Finally, write
	\begin{align*}
		Z^J \Gamma \left( s, s\frac{x}{t} \right)
		-
		Z^J \Gamma \left( s, X(s) \right)
		&
		=
		\int_0^1 \frac{d}{d\sigma}
		\left( Z^J \Gamma \left( s, \sigma s\frac{x}{t} + (1-\sigma) X(s) \right) \right)
		d\sigma
		\\
		&
		=
		\left( s \frac{x^l}{t} - X(s)^l \right)
		\int_0^1
		\left( \partial_{x^l} Z^J \Gamma \right) \left( s, \sigma s\frac{x}{t} + (1-\sigma) X(s) \right)
		d\sigma.
	\end{align*}
	Since,
	\[
		\left\vert s \frac{x^l}{t} - X(s)^l \right\vert
		\leq
		\left\vert s \frac{x^l}{t} - X_2(s,t,x,\ph)^l \right\vert
		+
		\left\vert \Xb(s,t,x,\ph)^l \right\vert
		\leq
		Cs^{1-a},
	\]
	by Proposition \ref{prop:sec23} and the bound \eqref{eq:X2minusX1} with $I=0$, it follows that
	\[
		\sum_{\vert J \vert \leq \vert I \vert}
		s
		\left\vert
		Z^J \Gamma \left( s, s\frac{x}{t} \right)
		-
		Z^J \Gamma \left( s, X(s) \right)
		\right\vert
		\leq
		C
		\sum_{\vert J \vert \leq \vert I \vert + 1}
		s^{1-a}
		\int_0^1
		\left\vert
		\left( Z^J \Gamma \right) \left( s, \sigma s\frac{x}{t} + (1-\sigma) X(s) \right)
		\right\vert
		d\sigma.
	\]
	Hence,
	\begin{align*}
		&
		\int_{\vert x \vert \leq ct + K}
		\int
		\left\vert
		(\partial^{n_1}_x \partial^{n_2}_{\ph} f)(t_0,X(t_0),\Ph(t_0))
		\right\vert
		\\
		&
		\quad
		\times
		\sum_{\vert J \vert \leq \vert I \vert }
		\int_{t_0+ \frac{1}{2}}^t
		s
		\left\vert
		(Z^J \Gamma) \left( s, s\frac{x}{t} \right)
		-
		(Z^J \Gamma) \left( s, X(s) \right)
		\right\vert
		ds
		d\ph
		\vert F(t,x) \vert
		dx
		\\
		&
		\quad
		\leq
		\frac{C\mathbb{D}_{n_1+n_2}}{t^3}
		\sum_{\vert J \vert \leq \vert I \vert + 1}
		\int_{\vert x \vert \leq ct + K}
		\int_{\vert y \vert \leq K}
		\int_{t_0+\frac{1}{2}}^t
		s^{1-a}
		\int_0^1
		\left\vert
		\left( Z^J \Gamma \right) \left( s, \sigma s\frac{x}{t} + (1-\sigma) X(s) \right)
		\right\vert
		d\sigma
		ds
		dy
		\vert F(t,x) \vert
		dx
		\\
		&
		\leq
		\frac{C\mathbb{D}_{n_1+n_2}}{t^3}\!\!\!\!\!
		\sum_{\vert J \vert \leq \vert I \vert + 1}
		\int_{t_0+\frac{1}{2}}^t\!\!\!\!\!\!
		s^{1-a}
		\int_0^1
		\int_{\vert y \vert \leq K}
		\Big(
		\int_{\vert x \vert \leq ct + K}
		\left\vert
		\left( Z^J \Gamma \right) \left( s, \sigma s\frac{x}{t} + (1-\sigma) X(s) \right)
		\right\vert^2\!\!
		dx
		\Big)^{\frac{1}{2}}
		dy
		d\sigma
		ds\,
		\Vert F(t,\cdot) \Vert_{L^2}
		\\
		&
		\leq
		\frac{C\mathbb{D}_{n_1+n_2}}{t^{\frac{3}{2}}}
		\sum_{\vert J \vert \leq \vert I \vert + 1}
		\int_{t_0}^t
		\frac{\left\Vert
		\left( Z^J \Gamma \right) \left( s, \cdot \right)
		\right\Vert_{L^2}}
		{s^{\frac{1}{2} + a}}
		ds
		\Vert F(t,\cdot) \Vert_{L^2},
	\end{align*}
	where the change of variables $x^i \mapsto z_2^i := \sigma s \frac{x^i}{t} + (1-\sigma)X(s)^i$ has been used, along with the fact that $\left\vert \det \left( \frac{\partial x^i}{\partial z_2^j} \right) \right\vert \leq C \left( \frac{t}{s} \right)^3$ when $0 \leq \sigma \leq 1$, $t_0 + \frac{1}{2} \leq s \leq t$, by Proposition \ref{prop:determinants}.
	
	It follows that, for any $F(t,x)$,
	\begin{align*}
		&
		\Vert \left( Z^I T^{\mu \nu} \right) (t,\cdot) F(t, \cdot) \Vert_{L^{1}}
		\leq
		\frac{C\mathcal{V}_{\vert I \vert} \Vert F(t,\cdot) \Vert_{L^2}}{t^{\frac{3}{2}}}
		\\
		&
		\qquad
		+
		C \mathbb{D}_{\left\lfloor \frac{\vert I \vert}{2} \right\rfloor +1} \Vert F(t,\cdot) \Vert_{L^2}
		\Big(
		\sum_{\vert J \vert \leq \vert I \vert -1}
		\frac{\Vert ( Z^J \Gamma)(t,\cdot) \Vert_{L^2}}{t^{1+a}}
		+
		\sum_{\vert J \vert \leq \vert I \vert + 1}
		\frac{1}{t^{\frac{3}{2}}}
		\int_{t_0}^t
		\frac{\Vert ( Z^J \Gamma)(s,\cdot) \Vert_{L^2}}{s^{\frac{1}{2}+a}}
		ds
		\Big).
	\end{align*}
	The $L^2$ estimate follows by setting $F(t,x) = Z^I T^{\mu \nu}(t,x)$ and dividing by $\Vert Z^I T^{\mu \nu} (t,\cdot) \Vert_{L^2}$.  The $L^1$ estimate follows by setting $F(t,x) = \mathrm{sign}(Z^I T^{\mu \nu}(t,x)) \chi_{\{ \vert x \vert \leq c t + K \}}$, and using the fact that $\supp(T^{\mu \nu}) \subset \{ \vert x \vert \leq c t + K \}$, and $\Vert \chi_{\{ \vert x \vert \leq c t + K \}} \Vert_{L^2} \leq C t^{\frac{3}{2}}$.
\end{proof}

\begin{proposition} \label{prop:mainL22}
	Suppose $\pi(\supp(f)) \subset \{\vert x \vert \leq ct\}$ and consider $t\geq t_0+1$.  If the assumptions \eqref{eq:Gammasuppf} hold and $\varepsilon$ is sufficiently small then, for any multi index $I$ with $\vert I \vert \leq N$ and each $\mu,\nu = 0,1,2,3$,
	\begin{align*}
		\Vert \left( Z^I T^{\mu \nu} \right) (t,\cdot) \Vert_{L^{2}}
		\leq
		\frac{C\mathcal{V}_{\vert I \vert}}{t^{\frac{3}{2}}}
		+
		C \mathbb{D}_{\left\lfloor \frac{\vert I \vert}{2} \right\rfloor +1}
		\Big(
		\sum_{\vert J \vert \leq \vert I \vert -1}
		\frac{\Vert ( Z^J \Gamma)(t,\cdot) \Vert_{L^2}}{t^{1+a}}
		+
		\sum_{\vert J \vert \leq \vert I \vert}
		\frac{1}{t^{\frac{3}{2}}}
		\int_{t_0}^t
		\frac{\Vert ( Z^J \Gamma)(s,\cdot) \Vert_{L^2}}{s^{\frac{1}{2}}}
		ds
		\Big),
	\end{align*}
	where $Z^I$ is a product of $\vert I \vert$ of the vector fields $\Omega_{ij}, B_i, S$.
\end{proposition}

\begin{proof}
	The proof is very similar to that of Proposition \ref{prop:mainL21}.  Recall that, for any $F(t,x)$,
	\begin{align*}
		&
		\Vert \left( Z^I T^{\mu \nu} \right) (t,\cdot) F(t, \cdot) \Vert_{L^{1}}
		\leq
		\frac{C\mathcal{V}_{\vert I \vert} \Vert F(t,\cdot) \Vert_{L^2}}{t^{\frac{3}{2}}}
		\\
		&
		\qquad
		+
		C \mathbb{D}_{\left\lfloor \frac{\vert I \vert}{2} \right\rfloor +1} \Vert F(t,\cdot) \Vert_{L^2}
		\Big(
		\sum_{\vert J \vert \leq \vert I \vert -1}
		\frac{\Vert ( Z^J \Gamma)(t,\cdot) \Vert_{L^2}}{t^{1+a}}
		+
		\sum_{\vert J \vert \leq \vert I \vert}
		\frac{1}{t^{\frac{3}{2}}}
		\int_{t_0}^t
		\frac{\Vert ( Z^J \Gamma)(s,\cdot) \Vert_{L^2}}{s^{\frac{1}{2}+a}}
		ds
		\Big)
		\\
		&
		\qquad
		+
		C
		\sum_{n_1+n_2 \leq \left\lfloor \frac{\vert I \vert}{2} \right\rfloor +1}
		\sum_{\vert J \vert \leq \vert I \vert}
		\int_{\vert x \vert \leq ct + K}
		\int
		\left\vert
		(\partial^{n_1}_x \partial^{n_2}_{\ph} f)(t_0,X(t_0),\Ph(t_0))
		\right\vert
		\\
		&
		\qquad \qquad
		\times
		\int_{t_0+\frac{1}{2}}^t
		s
		\left\vert
		(Z^J \Gamma) \left( s, s\frac{x}{t} \right)
		-
		(Z^J \Gamma) \left( s, X(s) \right)
		\right\vert
		ds
		d\ph
		\vert F(t,x) \vert
		dx.
	\end{align*}
	It is only the final term which is estimated differently.  In Proposition \ref{prop:mainL21} and extra derivative of $\Gamma$ was used to exploit the cancellation in $(Z^J \Gamma) \left( s, s\frac{x}{t} \right) - (Z^J \Gamma) \left( s, X(s) \right)$.  Now, at the top order, these terms are estimated individually,
	\begin{equation*}
		\sum_{\vert J \vert \leq \vert I \vert}
		\int_{t_0+\frac{1}{2}}^t
		s
		\left\vert
		(Z^J \Gamma) \left( s, s\frac{x}{t} \right)
		-
		(Z^J \Gamma) \left( s, X(s) \right)
		\right\vert
		ds
		\leq
		C\!\!\!
		\sum_{\vert J \vert \leq \vert I \vert}
		\int_{t_0+\frac{1}{2}}^t
		s
		\left\vert
		(Z^J \Gamma) \left( s, s\frac{x}{t} \right)
		\right\vert
		+
		s
		\left\vert
		(Z^J \Gamma) \left( s, X(s) \right)
		\right\vert
		ds.
	\end{equation*}
	The first term is estimated exactly as in Proposition \ref{prop:mainL21} (note that the $s$ power is now slightly worse) to give,
	\begin{align*}
		&
		\int_{\vert x \vert \leq ct + K}
		\int
		\left\vert
		(\partial^{n_1}_x \partial^{n_2}_{\ph} f)(t_0,X(t_0),\Ph(t_0))
		\right\vert
		\int_{t_0+\frac{1}{2}}^t
		s
		\left\vert
		(Z^J \Gamma) \left( s, s\frac{x}{t} \right)
		\right\vert
		ds
		d\ph
		\vert F(t,x) \vert
		dx
		\\
		&
		\qquad \qquad
		\leq
		\frac{C\mathbb{D}_{n_1+n_2}}{t^{{3}/{2}}}
		\Vert F(t,\cdot) \Vert_{L^2}
		\int_{t_0}^t
		\frac{\Vert (Z^JT^{\mu \nu})(s,\cdot)\Vert_{L^2}}{s^{{1}/{2}}}
		ds.
	\end{align*}
	The second term is estimated similarly,
	\begin{align*}
		&
		\int_{\vert x \vert \leq ct + K}
		\int
		\left\vert
		(\partial^{n_1}_x \partial^{n_2}_{\ph} f)(t_0,X(t_0),\Ph(t_0))
		\right\vert
		\int_{t_0+\frac{1}{2}}^t
		s
		\left\vert
		(Z^J \Gamma) \left( s, X(s) \right)
		\right\vert
		ds
		d\ph
		\vert F(t,x) \vert
		dx
		\\
		&
		\qquad \qquad
		\leq
		\frac{C\mathbb{D}_{n_1+n_2}}{t^3}
		\int_{t_0+\frac{1}{2}}^t s \int_{\vert y \vert \leq K}
		\Big(
		\int_{\vert x \vert \leq ct +K} \left\vert (Z^J \Gamma) \left( s, X(s) \right) \right\vert^2 dx
		\Big)^{\frac{1}{2}}
		dy
		ds
		\Vert F(t,\cdot) \Vert_{L^2}
		\\
		&
		\qquad \qquad
		\leq
		\frac{C\mathbb{D}_{n_1+n_2}}{t^{{3}/{2}}}
		\Vert F(t,\cdot) \Vert_{L^2}
		\int_{t_0}^t
		\frac{\Vert (Z^JT^{\mu \nu})(s,\cdot)\Vert_{L^2}}{s^{\frac{1}{2}}}
		ds,
	\end{align*}
	where now the change of variables $x^i \mapsto z_2^i := X(s,t,x,\ph(t,x,y))^i$ has been used, along with the fact that $\vert \det ( {\partial x^i}/{\partial z_2^j} ) \vert \leq C ( {t}/{s} )^3$ for $t_0 + \frac{1}{2} \leq s \leq t$, by Proposition \ref{prop:determinants}.
	
	It follows that
	\begin{align*}
		&
		\Vert \left( Z^I T^{\mu \nu} \right) (t,\cdot) F(t, \cdot) \Vert_{L^{1}}
		\leq
		\frac{C\mathcal{V}_{\vert I \vert} \Vert F(t,\cdot) \Vert_{L^2}}{t^{\frac{3}{2}}}
		\\
		&
		\qquad
		+
		C \mathbb{D}_{\left\lfloor \frac{\vert I \vert}{2} \right\rfloor +1} \Vert F(t,\cdot) \Vert_{L^2}
		\Big(
		\sum_{\vert J \vert \leq \vert I \vert -1}
		\frac{\Vert ( Z^J \Gamma)(t,\cdot) \Vert_{L^2}}{t^{1+a}}
		+
		\sum_{\vert J \vert \leq \vert I \vert}
		\frac{1}{t^{\frac{3}{2}}}
		\int_{t_0}^t
		\frac{\Vert ( Z^J \Gamma)(s,\cdot) \Vert_{L^2}}{s^{\frac{1}{2}}}
		ds
		\Big).
	\end{align*}
	The proof then follows by setting $F(t,x) = Z^IT^{\mu \nu}(t,x)$.
\end{proof}

\subsection{Proof of Theorem \ref{thm:mainL2}}
\label{subsec:Tmainestimates}

First note that Proposition \ref{prop:mainL21} and Proposition \ref{prop:mainL22} can be extended to include $t_0 \leq t \leq t_0 +1$ as follows.

\begin{proposition}
	Suppose $\pi\left( \supp(f) \right) \subset \{ \vert x \vert \leq c t\}$ and consider $t_0 \leq t \leq t_0 +1$.  If the assumptions \eqref{eq:Gammasuppf} hold and $\varepsilon$ is sufficiently small then, for any multi index $I$ with $\vert I \vert \leq N$,
	\[
		\Vert \partial^I T^{\mu \nu}( t,\cdot) \Vert_{L^2}
		+
		\Vert \partial^I T^{\mu \nu}( t,\cdot) \Vert_{L^1}
		\lesssim
		\mathcal{V}_{\vert I \vert}
		+
		\mathbb{D}_{\left\lfloor \frac{\vert I \vert}{2} \right\rfloor + 1}
		\Big(
		\sum_{\vert J \vert \leq \vert I \vert -1} \Vert \partial^J \Gamma(t,\cdot) \Vert_{L^2}
		+
		\sum_{\vert J \vert \leq \vert I \vert} \int_{t_0}^t \Vert \partial^J \Gamma(s,\cdot) \Vert_{L^2} ds
		\Big).
	\]
\end{proposition}

\begin{proof}
	By Proposition \ref{prop:lowpointwise} and Proposition \ref{prop:lowhighorder}, it follows from an appropriate version of Proposition \ref{prop:ZITpointwise} that
	\begin{align*}
		&
		\left\vert \partial^I T^{\mu \nu} (t,x) \right\vert
		\lesssim
		\sum_{n_1+n_2 \leq \vert I \vert}
		\int
		\left\vert (\partial_x^{n_1} \partial_{\ph}^{n_2} f)(t_0,X(t_0),\Ph(t_0)) \right\vert
		d\ph
		\\
		&
		+ \!\!\!\!\!
		\sum_{n_1+n_2 \leq \left\lfloor \frac{\vert I \vert}{2} \right\rfloor + 1}
		\int
		\left\vert (\partial_x^{n_1} \partial_{\ph}^{n_2} f)(t_0,X(t_0),\Ph(t_0)) \right\vert
		\Big(
		\sum_{\vert J \vert \leq \vert I \vert -1}
		\left\vert (\partial^J \Gamma)(t,x) \right\vert
		+
		\sum_{\vert J \vert \leq \vert I \vert}
		\int_{t_0}^t
		\left\vert (\partial^J \Gamma)(s,X(s)) \right\vert
		ds
		\Big)
		d\ph
		\\
		&
		\lesssim
		\sum_{n_1+n_2 \leq \vert I \vert}
		\int
		\left\vert (\partial_x^{n_1} \partial_{\ph}^{n_2} f)(t_0,X(t_0),\Ph(t_0)) \right\vert
		d\ph
		\\
		&
		\qquad
		+
		\mathbb{D}_{\left\lfloor \frac{\vert I \vert}{2} \right\rfloor + 1}
		\Big(
		\sum_{\vert J \vert \leq \vert I \vert -1}
		\left\vert (\partial^J \Gamma)(t,x) \right\vert
		+
		\sum_{\vert J \vert \leq \vert I \vert}
		\int
		\int_{t_0}^t
		\left\vert (\partial^J \Gamma)(s,X(s)) \right\vert
		ds
		d \ph
		\Big).
	\end{align*}
	For any function $F(t,x)$,
	\[
		\int \int
		\left\vert (\partial_x^{n_1} \partial_{\ph}^{n_2} f)(t_0,X(t_0),\Ph(t_0)) \right\vert
		d\ph
		\vert F(t,x) \vert
		dx
		\lesssim
		\mathcal{V}_{n_1+n_2} \Vert F(t,\cdot) \Vert_{L^2},
	\]
	by \eqref{eq:kappamatrix} (as in the proof of Proposition \ref{prop:mainL21}), and
	\begin{align*}
		\int \vert F(t,x) \vert
		\int \int_{t_0}^t
		\left\vert (\partial^J \Gamma)(s,X(s)) \right\vert
		ds
		d \ph dx
		&
		\leq
		\int_{t_0}^t \int
		\left( \int \left\vert (\partial^J \Gamma)(s,X(s)) \right\vert^2 dx \right)^{\frac{1}{2}}
		\Vert F(t,\cdot)\Vert_{L^2} d\ph ds
		\\
		&
		\lesssim
		\int_{t_0}^t \Vert (\partial^J \Gamma)(s,\cdot) \Vert_{L^2} ds \Vert F(t,\cdot)\Vert_{L^2},
	\end{align*}
	where the change of variables $x \mapsto X(s,t,x,\ph)$ and the bound \eqref{eq:dXdx} have been used.  Clearly then
	\begin{align*}
		&
		\int
		\left\vert \partial^I T^{\mu \nu} (t,x) \right\vert \vert F(t,x) \vert
		dx
		\lesssim
		\mathcal{V}_{\vert I \vert} \Vert F(t,\cdot) \Vert_{L^2}
		\\
		&
		\qquad \qquad
		+
		\mathbb{D}_{\left\lfloor \frac{\vert I \vert}{2} \right\rfloor + 1}
		\int
		\vert F(t,x) \vert
		\Big(
		\sum_{\vert J \vert \leq \vert I \vert -1}
		\left\vert (\partial^J \Gamma)(t,x) \right\vert
		+
		\sum_{\vert J \vert \leq \vert I \vert}
		\int
		\int_{t_0}^t
		\left\vert (\partial^J \Gamma)(s,X(s)) \right\vert
		ds
		d \ph
		\Big)
		dx
		\\
		&
		\lesssim
		\Big[
		\mathcal{V}_{\vert I \vert}
		+
		\mathbb{D}_{\left\lfloor \frac{\vert I \vert}{2} \right\rfloor + 1}
		\Big(
		\sum_{\vert J \vert \leq \vert I \vert -1}
		\left\Vert (\partial^J \Gamma)(t,\cdot) \right\Vert_{L^2}
		+
		\sum_{\vert J \vert \leq \vert I \vert}
		\int_{t_0}^t
		\left\Vert (\partial^J \Gamma)(s,\cdot) \right\Vert
		ds
		\Big)
		\Big]
		\Vert F(t,\cdot) \Vert_{L^2}.
	\end{align*}
	The $L^2$ bound then follows by setting $F = \partial^I T^{\mu \nu}$, and the $L^1$ bound follows by setting $F = \mathrm{sign}(\partial^I T^{\mu \nu})$.
\end{proof}

Since, for any function $F(t,x)$ and any multi index $I$, the vector fields $Z$ satisfy
\[
	\sum_{\vert J \vert \leq \vert I \vert} \vert \partial^I F(t,x) \vert
	\lesssim
	\vert Z^I F(t,x) \vert
	\lesssim
	\sum_{\vert J \vert \leq \vert I \vert} \vert \partial^I F(t,x) \vert,
	\qquad
	\text{for }
	t_0 \leq t \leq t_0 + 1,
	\quad
	\vert x \vert \leq c t,
\]
it is clear that Proposition \ref{prop:mainL21} and Proposition \ref{prop:mainL22} in fact hold for $t \geq t_0$.  Moreover, it is clear form \eqref{eq:partialZ} that
\[
	\Vert \partial^I Z^J T^{\mu \nu}(t,\cdot) \Vert
	\lesssim
	\sum_{\vert K \vert \leq \vert I \vert + \vert J \vert}
	\Vert Z^K T^{\mu \nu}(t,\cdot) \Vert,
	\qquad
	\text{for }
	t \geq t_0,
\]
where $Z^J$ is a product of $\vert J \vert$ of the vector fields $\Omega_{ij}, B_i, S$, for $\Vert \cdot \Vert = \Vert \cdot \Vert_{L^1}$ or $\Vert \cdot \Vert_{L^2}$ since $\supp(T^{\mu \nu}) \subset \{ \vert x \vert \leq ct\}$, and so spacetime derivatives $\partial^I$ can be included in Proposition \ref{prop:mainL21} and Proposition \ref{prop:mainL22}.

Suppose now that the assumptions of Theorem \ref{thm:mainL2} hold.  It follows from Proposition \ref{prop:suppf} that the support of $f$ satisfies $\pi(\supp(f)) \subset \{ \vert x \vert \leq ct + K\}$ and so, letting $\tilde{t} = t_0 + t$ where $t_0 = \frac{K}{c}$, it follows from Proposition \ref{prop:mainL21}, Proposition \ref{prop:mainL22} and the above comments that the bounds of Theorem \ref{thm:mainL2} hold with $t$ replaced by $\tilde{t}$ and the vector fields $Z$ replaced by $\tilde{Z}$ for $\tilde{t} \geq t_0$, where the vector fields $\tilde{Z}$ are as in Section \ref{subsec:transtime}.  The proof of Theorem \ref{thm:mainL2} then follows from noting that,
\[
	\tilde{\Omega}_{ij} = \Omega_{ij},
	\qquad
	\tilde{B}_i = B_i + t_0 \partial_{x^i},
	\qquad
	\tilde{S} = S + t_0 \partial_t,
\]
and so
\[
	\partial^I Z^J
	=
	\sum_{\vert I'\vert + \vert J'\vert \leq \vert I \vert + \vert J \vert}
	C_{I'J''} \partial^{I'} \tilde{Z}^{J'},
\]
for some constants $C_{I'J'}$.

\section{The Einstein equations}
\label{section:Einstein}

The main results of this section are the following results

\begin{proposition} \label{thm:Einsteinpointwiseestimates}
	Suppose that $N\geq 4$, $\varepsilon\leq 1$. Consider a solution of the reduced Einstein equations \eqref{eq:RE1} for $t<T_*$ such that the weak decay estimates
	\begin{equation}\label{eq:weakenergyboundstheorem61}
		E_N(t)^{1/2} \leq C_N \varepsilon (1+t)^{\delta},
		\qquad
		\sum_{\vert J \vert \leq N-1} \Vert  Z^J \widehat{T}(t,\cdot) \Vert_{L^1}
		\leq
		C_N \varepsilon,\qquad\text{and}\quad  M\leq \varepsilon
	\end{equation}
	hold for all $t\in [0,T_*]$, for some $\delta$ such that
\beq\label{eq:deltacondition}
0<8\delta < \gamma < 1-8\delta, \qquad M\leq \varepsilon.
\eq
Then for some constant $C_N^\prime$ depending only on $C_N$ the weak decay estimates
\begin{equation}\label{eq:hormanderweakdecayh1gheorem61}
|Z^I h^1(t,x)|\leq \frac{C_N^\prime\varepsilon(1+t)^{2\delta}}{(1+t+r) (1+q_+)^{\gamma}},
\qquad |I|\leq N-3,
\end{equation}
where $q=r-t$, $q_+=\max\{q,0\}$ and  $q_-=\max\{-q,0\}$,
hold for all $t\in [0,T_*]$.
\end{proposition}
Note that the inverse of the metric $g_{\mu \nu}$ can be expressed as
\begin{equation*}
	g^{\mu\nu}
	=
	m^{\mu\nu}+H^{\mu\nu},
	\quad \text{and}\quad
	H^{\mu\nu}=H_0^{\mu\nu}+H_1^{\mu\nu},\quad\text{where}\quad
	H_0^{\mu\nu}
	=
	-{\chi}\big(\tfrac{r}{1+t}\big)\tfrac{M}{r}\de^{\mu\nu}.
\end{equation*}
Then $m^{\mu\nu}\!+H_0^{\mu\nu}\!-h^{1\mu\nu} $, where $h^{1\mu\nu}=m^{\mu\alpha}m^{\nu\beta}h^1_{\alpha\beta}$, is an approximate
inverse to $g_{\mu\nu}\!=m_{\mu\nu}\!+h^0_{\mu\nu}\!+h^1_{\mu\nu}$ up
to $O(h^2)$ so
$H_1^{\mu\nu}\!\!=\!-h^{1\, \mu\nu}\!\! +O(h^2)$.
Therefore $H_1$ will satisfy the same estimates as $h_1$.

We have the following sharp decay estimates from the wave coordinate condition
for certain tangential components expressed in the null frame.
\begin{proposition}\label{prop:wavecoordinatedecayprop}
Suppose the conditions of Proposition \ref{thm:Einsteinpointwiseestimates} hold.  Let $\mathcal{N}=\{L,\underline{L},S_1,S_2\}$ be the null frame defined by \eqref{eq:frameintro}.
The modified Lie derivative $\widehat{\mathcal L}_Z$ defined by  \eqref{eq:modifiedlie}
satisfies
\begin{align}
|\pa_q (\widehat{\mathcal L}_Z^I{H}_1)_{LT}|
+|\pa_q  \trs\widehat{\mathcal L}_Z^I {H}_1 |&\leq C_N^{\prime\prime}
\varepsilon(1+t+r)^{-2+2\delta}(1\!+q_+)^{-\gamma},
\label{eq:wavecoordinatederivativeLiefirst}
\\
|(\widehat{\mathcal L}_Z^I{H}_1)_{L T}|
+|\trs\widehat{\mathcal L}_Z^I {H}_1 |
&\leq C_N^{\prime\prime} \varepsilon(1\!+t\!+r)^{-1-\gamma+2\delta} (1\!+q_-)^{\gamma}.
\label{eq:wavecoordinatefunctionLiefirst}
\end{align}
for $\vert I \vert \leq N-4$, where $\trs H_1=H_{1S_1 S_1}+H_{1S_2 S_2}$ and  $T\in\mathcal{T}=\{L,S_1,S_2\}$, the subset that span the tangent space of the outgoing light cones.
Here the constant $C_N^\prime$ depend only on $C_N^\prime$ in \eqref{eq:hormanderweakdecayh1gheorem61} and on $N$.
\end{proposition}

\begin{proposition} \label{prop:higherordersharpdecay} Suppose that $N\geq 5$ and the weak decay estimates
 \eqref{eq:hormanderweakdecayh1gheorem61} hold
 for some $8\delta \leq \gamma \leq 1-8\delta$,  $M\leq \varepsilon\leq 1$
and that there is a constant $0<c<1$ such that
\begin{equation}\label{eq:supportofT}
\supp \, \newhat{T}(t,x)\subset \{(t,x);\,|x|\leq K+ct\}, \qquad c<1.
\end{equation}
Then the following sharp decay estimates hold.

For any $-1\leq \gamma^\prime<\gamma-2\delta$,
and $|I|=k\leq N-5$ there are constants $c_k$ such that
\beq\label{eq:sharpHhigherlowderbeginsec61}
\big| \pa Z^I h^1\big|\leq c_k \varepsilon(1+t)^{c_k\varepsilon}  (1+t+r)^{-1}(1+q_+)^{-1-\gamma^\prime}.
\eq
In addition we have the following estimates for certain tangential components expressed in a null frame
\beq\label{eq:sharpHhigherlowderbeginsec6}
\big| \pa h^1_{TU}\big|\leq c_0\varepsilon   (1+t+r)^{-1}(1+q_+)^{-1-\gamma^\prime},\qquad
T\in \mathcal{T},\quad U\in \mathcal{N}.
\eq
Here all constants depend only on $C_N^\prime$ in \eqref{eq:hormanderweakdecayh1gheorem61}, on $N$ and on $c$, $K$ in \eqref{eq:supportofT}.
\end{proposition}

\begin{theorem} \label{thm:Einsteinenergyestimates} Suppose that $N\geq 9$ and the decay and support
conditions \eqref{eq:deltacondition}--\eqref{eq:sharpHhigherlowderbeginsec6} hold. Then there is a
$\varepsilon_N>0$ and constants $C_N^{\prime\prime\prime}$, $d_1,\ldots,d_N$, depending only on $N$, $C_N$, $c$, $K$ and a lower positive bound for $\min{(\gamma,1-\gamma)}$, such that for all $t\in [0,T_*]$ and $\varepsilon<\varepsilon_N$,
\begin{equation*}
	Q_k(t)
	\leq
	8 Q_k(0)
	+
	M_k M
	+
	C^{\prime\prime\prime}_N\varepsilon
	\!\int_0^t
	\frac{Q_{k}(\tau)}{1+\tau}
	+
	\frac{Q_{k-1}(\tau)}{(1+\tau)^{1-d_k\varepsilon}}
	d\tau
	+
	M_k \sum_{|I|\leq k} \int_0^t \Vert Z^I  \widehat{T}  (\tau,\cdot) \Vert_{L^{2}} \, d\tau,
\end{equation*}
where $Q_k(t) := \sup_{0\leq \tau\leq t} E_k(\tau)^{1/2}$ and $Q_{-1}(0) \equiv 0$, and $M_1,\ldots,M_k$ are universal constants.
\end{theorem}

In the proof of Theorem \ref{thm:main2} in Section \ref{section:cty}, Proposition \ref{thm:Einsteinpointwiseestimates} will first be appealed to for the coupled Einstein--Vlasov system \eqref{eq:Tmunu}, \eqref{eq:Vlasov}, \eqref{eq:RE1}.  As a consequence the assumptions of Proposition \ref{prop:suppf} will be satisfied, which in turn will ensure that the assumptions of Proposition \ref{prop:higherordersharpdecay}  and hence of Theorem \ref{thm:Einsteinenergyestimates} are satisfied.

\subsection{Weak $L^\infty$ decay estimates} \label{subsec:weakdecay}
Here we assume the weak energy bounds
\eqref{eq:weakenergyboundstheorem61} and prove that this implies certain decay estimates.

\subsubsection{The weak decay estimates for the metric from Klainerman--Sobolev}
 \begin{lemma} {\it The  Klainerman--Sobolev inequalities with weights}
\begin{equation*}
| \phi(t,x)|w^{1/2}\leq \frac{C \sum_{|I|\leq 2} \|w^{1/2} Z^I
\phi(t,\cdot)\|_{L^2}
}{(1\!+\!t\!+\!r)(1\!+\!|t-r|)^{1/2}}.
\end{equation*}
\end{lemma}
For a proof see \cite{LR3}. Using this we get
\begin{proposition}\label{prop:weakdecayklainermansobolev} Suppose that the weak energy bounds \eqref{eq:weakenergyboundstheorem61} hold. Then
\begin{equation}\label{eq:weakdecayderh1}
|\pa Z^I h^1(t,x)|\leq \begin{cases}
C\varepsilon(1+t+r)^{-1} (1+|r-t|)^{-1-\gamma}(1+t)^\delta,
\qquad r>t\\
C\varepsilon(1+t+r)^{-1}(1+|r-t|)^{-1/2}(1+t)^\delta,\qquad r<t
\end{cases},\qquad |I|\leq N-2.
\end{equation}
Furthermore
\begin{equation}\label{eq:weakdecayh1}
|Z^I h^1(t,x)|\leq \begin{cases}
C\varepsilon(1+t+r)^{-1} (1+|r-t|)^{-\gamma}(1+t)^\delta,
\qquad r>t\\
C\varepsilon(1+t+r)^{-1}(1+|r-t|)^{1/2}(1+t)^\delta,\qquad r<t
\end{cases},\qquad |I|\leq N-2.
\end{equation}
The same estimates hold for $H_1$ in place of $h^1$, and for $h$ or $H$ in place of $h_1$ if $\gamma$ is replaced by $\delta$.
\end{proposition}
\begin{proof}
\eqref{eq:weakdecayh1} follows from integrating \eqref{eq:weakdecayderh1} in the $r-t$  direction from initial data, see Corollary 9.3 in \cite{LR3}.
\end{proof}

\subsubsection{The weak decay estimates for the metric from H\"ormander's $L^1$--$L^\infty$ estimate}
To get improve decay estimates in the interior we will use H\"ormander's $L^1$--$L^\infty$
estimates for the fundamental solution of $\Box$, see \cite{H1, L1}:
 \begin{lemma}\label{lemma:hormander} Suppose that $w$ is a solution of
 $\Box u=F$ (i,e. the flat Minkowski wave operator) with vanishing data $ u\big|_{t=0}=\pa_t u\big|_{t=0}=0$. Then
 \beq |u(t,x)|(1+t+|x|) \leq
C\sum_{|I|\leq
2}\int_0^t\int_{\mathbf{R}^3}\frac{|Z^I F(s,y)|}{1+s+|y|}\, dy\,ds,
 \eq
\end{lemma}
Also for the linear homoegenous solution we have from \cite{L1,LR3}:
\begin{lemma}\label{lemma:linearhomogenousdecay} If $v$ is the solution of
 $\Box v=0$, with data $v\big|_{t=0}=v_0$ and $ \pa_t v\big|_{t=0}=v_1 $
 then for any $\gamma>0$;
 \beq
 (1+t)|v(t,x)| \leq C\sup_x \big( (1+|x|)^{2+\gamma}
 ( |v_1(x)|+|\partial v_0(x)|) +(1+ |x|)^{1+\gamma}|v_0(x)|\big)
 \eq
\end{lemma}
For the proof below we will also use the following version of Hardy's inequality,
see Corollary 13.3 in \cite{LR3}
\begin{lemma} \label{lemma:hardy} For any $-1\leq a\leq 1$ and any $\phi\in C^\infty_0(\R^3)$
\begin{equation}\label{eq:poincareone}
\int \frac{|\phi|^2}{(1+|t-r|)^2}\,\frac{ w\, dx}{(1+t+r)^{1-a}}
\les \int |\pa \phi|^2\, \frac{ w\, dx}{(1+t+r)^{1-a}}.
\end{equation}
\end{lemma}
In the other direction we have
\begin{lemma} Let $\pa_q=(\pa_r-\pa_t)/2$ and let $\pab_\mu=\pa_\mu -L_\mu\pa_q $ be the projection of $\pa_\mu$ onto the tangent space
of the outgoing light cones. Then
\begin{equation}\label{eq:derbytransder}
|\partial\phi(t,x)|\les |\pa_q \phi(t,x)|+|\pab \phi(t,x)|
\end{equation}
and
\begin{equation}\label{eq:pabyZ}
(1+|t-r|)|\partial\phi(t,x)|+(1+t+r)|\overline{\partial}\phi(t,x)|\les \sum_{|I|=1} |Z^I \phi(t,x)|.
\end{equation}
\end{lemma}
\begin{proposition}\label{prop:weakdecayhormander} Suppose that $N\geq 4$ and the weak energy bounds \eqref{eq:weakenergyboundstheorem61} hold for
a solution of Einstein--Vlasov in wave coordinates. Then
\begin{equation}\label{eq:hormanderweakdecayh1}
|Z^I h^1(t,x)|\leq \frac{C\varepsilon(1+t)^{2\delta}}{(1+t+r) (1+q_+)^{\gamma}},
\qquad |I|\leq N-3.
\end{equation}
where $q=r-t$, $q_+=\max\{q,0\}$ and  $q_-=\max\{-q,0\}$. Moreover
\begin{equation}\label{eq:hormanderweakdecayderh1}
(1+|q|)\big|\pa Z^I h^1(t,x)\big|+(1+t+r)\big|\overline{\pa} Z^I h^1(t,x)\big|\leq
\frac{C\varepsilon(1+t)^{2\delta}}{(1+t+r)(1+q_+)^{\gamma}},\qquad |I|\leq N-4.
\end{equation}
The same estimates hold for $H_1$ in place of $h^1$, and $h$ or $H$ in place of $h^1$ if $\gamma$ is replaced by $2\delta$.
\end{proposition}
\begin{proof} First \eqref{eq:hormanderweakdecayderh1} is a consequence of \eqref{eq:hormanderweakdecayh1} using \eqref{eq:pabyZ} so it only remains to prove
\eqref{eq:hormanderweakdecayh1} for $r<t$.
Let $h^1_{\mu\nu}=v_{\mu\nu}+u_{\mu\nu}+\phi_{\mu\nu}$ where,
\begin{equation}
\Box v_{\mu\nu}=0,\qquad v_{\mu\nu}\big|_{t=0}=h^1_{\mu\nu}\big|_{t=0},\qquad \pa_t
 v_{\mu\nu}\big|_{t=0}=\pa_t
 h^1_{\mu\nu}\big|_{t=0}.
 \end{equation}
 and
 \begin{align}
 \Box u_{\mu\nu}&= -H^{\alpha\beta}\pa_\alpha\pa_\beta h_{\mu\nu}
 +F_{\mu\nu}(h)(\pa h,\pa h)-\Box h_{\mu\nu}^0,\qquad &u_{\mu\nu}\big|_{t=0}=\pa_t
 u_{\mu\nu}\big|_{t=0}=0,\\
\Box \phi_{\mu\nu}&=\widehat{T}_{\mu\nu},\qquad &\phi_{\mu\nu}\big|_{t=0}=\pa_t
 \phi_{\mu\nu}\big|_{t=0}=0,
 \end{align}
 We will prove \eqref{eq:hormanderweakdecayh1} for $r<t$ separately for each of $v,u,\phi$.
 For $v$ and $u$ the proof follows the proof in section 16 of \cite{LR3}.
 The estimate for the homogeneous linear part $v$ follows from using Lemma \ref{lemma:linearhomogenousdecay} with the estimate for $v_0$ and $v_1$ obtained from
 \eqref{eq:weakdecayderh1}-\eqref{eq:weakdecayh1} when $t=0$.

Then using the $L^\infty$ bounds \eqref{eq:weakdecayderh1}-\eqref{eq:weakdecayh1} for small number of vector fields  we have
 \beq
 |Z^I F_{\mu\nu}(h)(\pa h,\pa h)|\leq
 C\!\!\!\sum_{|J|+|K|\leq |I|\!\!\!\!\!\!}|\pa Z^J h|\, |\pa Z^K h|
 +C\!\!\! \sum_{|J|+|K|\leq |I|} \frac{|Z^J h|}{1+|q|}\, |\pa Z^K h|,
 \eq
 and since $H^{\alpha\beta}=-h_{\alpha\beta}+O(h^2)$,
  \beq
\big|Z^I\big( H^{\alpha\beta}\pa_\alpha\pa_\beta h_{\mu\nu}\big)\big|\leq C\!\!\!
\sum_{|J|+|K|\leq |I|+1,\, |J|\leq |I|} \frac{|Z^J h|}{1+|q|}\, |\pa
Z^K h|.
  \eq
  Now
  \beq
\int |\pa Z^J h|\, |\pa Z^K h|(s,y)\, dy \leq \sum_{|I|\leq N}\|\pa
Z^I h(s,\cdot)\|_{L^2}^2\leq C\varepsilon^2(1+s)^{2\delta},
  \eq
  since
  \beq
  \int |\pa Z^I h^0(s,y)|^2\,dy \leq C M^2 \int_0^\infty
 \frac{r^2\, dr}{(1+|t+r|)^4}\leq C^\prime M^2.
 \eq
 We write $h=h^0+h^1$ and estimate
 \beq
 \int \frac{|Z^J h^0(t,x)|^2}{(1+|t-r|)^2}\,dx \leq CM^2 \int_0^\infty
 \frac{r^2\, dr}{(1+|t+r|)^2(1+|t-r|)^2}\leq C^\prime M^2
 \eq
and by Lemma \ref{lemma:hardy}
 \beq
 \int \frac{|Z^J h^1(t,x)|^2}{(1+|t-r|)^2}\,dx
 \leq C\int |\pa Z^J h^1(t,x)|^2 w\, dx\leq C\varepsilon^2 (1+t)^{2\delta}
 \eq
 Hence
 \beq
 \int \frac{|Z^J h|}{1+|q|}\, |\pa
Z^K h|(s,y)\, dy \leq C\varepsilon^2 (1+t)^{2\delta}
 \eq
 Finally
 \beq
 |\Box h_{\mu\nu}^0|
 =\Big|\frac{M\de_{\mu\nu}}{r}(\pa_t+\pa_r)(\pa_t-\pa_r)\big(\chi(\tfrac{r}{1+t})\big)\Big|
 \leq \frac{CM H(r<3t/4)}{(1+t+r)^3},
 \eq
 where $H(r<3t/4)=1$ when $r<3t/4$ and $0$ otherwise.
 Hence
 \beq
 \|\, \Box h^0(t,\cdot)\|_{L^1}\leq M.
 \eq
 It now follows from Lemma \ref{lemma:hormander} that
 \beq
 |u_{\mu\nu}(t,x)|(1+t+|x|)\leq \int_0^t
 \frac{(\varepsilon^2 +M)\, ds}{(1+s)^{1-2\delta}}\leq
 C\varepsilon (1+s)^{2\delta},
 \eq
 which proves \eqref{eq:hormanderweakdecayh1} for $r<t$ also for $u$.
 It remains to prove the estimate also for $\phi$ but this also follows from
  Lemma \ref{lemma:hormander}:
 \beq
 |\phi_{\mu\nu}(t,x)|(1+t+|x|)\leq \int_0^t
 \frac{C\|\widehat{T}(s,\cdot)\|_{L^1}\, ds}{1+s}\leq \int_0^t
 \frac{C\varepsilon\, ds}{(1+s)^{1-\delta}}\leq
 C\varepsilon (1+s)^{\delta},
 \eq

\end{proof}

\subsubsection{The support and weak decay of matter}
The following Sobolev inequality will be used to obtain pointwise bounds for $\newhat{T}$ from the assumptions \eqref{eq:weakenergyboundstheorem61}.
\begin{lemma} \label{lem:Sobolevt3}
If $\supp \,\phi\subset \{(t,x);\, |x|\leq K+ct\}$ for some $K\geq 0$ and $0<c<1$, then
\begin{equation*}
	| \phi(t,x)|\les (1+t+r)^{-3} {\sum}_{|I| = 3} \| Z^I
	\phi(t,\cdot)\|_{L^1}.
\end{equation*}
\end{lemma}

\begin{proof}
	The proof proceeds by noting that, with $c' = (1+c)/2$, $\supp \,\phi (t,\cdot) \subset \{(t,x);\, |x|\leq c't\}$ for $t\geq 2K/(1-c)$.  The inequality for $t \leq 2K/(1-c)$ follows from the standard Sobolev inequality,
	\[
		\vert \phi(t,x) \vert \les \sum_{\vert I \vert = 3} \vert \partial^I \phi(t,\cdot) \vert.
	\]
	The proof for $t \geq 2K/(1-c)$ follows from the identity \eqref{eq:partialZ}.
\end{proof}

Lemma \ref{lem:Sobolevt3}, together with assumption \eqref{eq:supportofT} on the support of $\newhat{T}$ and the weak energy bounds \eqref{eq:weakenergyboundstheorem61}, gives
\begin{equation}\label{eq:weakdecayT}
|Z^I T(t,x)|\les \varepsilon (1+t)^{-3}, \qquad |I|\leq N-4.
\end{equation}

\subsection{The sharp decay estimates for the first order derivatives}\label{sec:sharpfirstorder}
Throughout the rest of this section we will assume that the weak decay estimates
\eqref{eq:hormanderweakdecayh1}-\eqref{eq:hormanderweakdecayderh1} hold for some
$0<8\delta<\gamma<1-8\delta$, $M\leq \varepsilon\leq 1$, along with the support condition \eqref{eq:supportofT} for $\widehat{T}$.  However we will not use the weak energy bounds \eqref{eq:weakenergyboundstheorem61}
any further.
\subsubsection{The sharp decay estimates for first order derivative of certain components from the wave coordinate condition}
 By \cite{L4,LR3} the wave coordinate condition can be written
\beq\pa_\mu \widehat{H}^{\mu\nu\!}\!
= W^{\nu}(h,\pa h)\qquad
\text{where}\quad  \widehat{H}^{\mu\nu}\!=\!H^{\mu\nu}\!-m^{\mu\nu}
\tr_m H_{\!}/2,\quad \tr H\!=m_{\alpha\beta} H^{\alpha\beta}
\eq
and $|W(h,\pa h)|\lesssim |h|\, |\pa h |$.
 Moreover
\begin{equation} \label{eq:wavehat1divergence}
\pa_\mu \widehat{H}_1^{\mu\nu\!}\!
= W^{\nu}(h,\pa h)-\pa_\mu \widehat{H}_0^{\mu\nu\!},\quad\text{where}\quad
\pa_\mu \widehat{H}_0^{\mu\nu\!}=2{\chi}^{\,\prime}\big(\tfrac{r}{1+t}\big) M(1\!+t)^{-2}\delta^{0\nu}.
\end{equation}
We first express the divergence in a null frame:
\begin{lemma}\label{lem:divergencenullframe} Let $\pa_q=(\pa_r\!-\pa_t)/2 $ and $\pa_s=(\pa_r\!+\pa_t)/2$. Then for any tensor $k^{\mu\nu}$:
\begin{equation}
\pa_q \big(L_\mu U_\nu
k^{\mu\nu}\big)\!=L_\mu U_\nu\pa_q
k^{\mu\nu}\!= \Lb_{\,\mu} U_\nu\pa_s
k^{\mu\nu} -A_\mu U_\nu \pa_A
k^{\mu\nu}+U_\nu \,\partial_\mu k^{\mu\nu}\!\!,
\qquad U\in\mathcal{N}.\label{eq:wavecoordinateframedivergence}
\end{equation}
\end{lemma}
\begin{proof} The proof follows expressing the divergence in a null frame
$\pa_\mu F^\mu\!=L_\mu\pa_q F^\mu -\Lb_{\,\mu} \pa_s F^\mu\!+A_\mu \pa_A F^\mu\!,$
and using that $\pa_q$ and $\pa_s$ commute with the frame, see \cite{L4}.
\end{proof}
Using this we get
\begin{lemma}\label{lem:wavecoordestimate0} We have
\begin{align}
 |\,\pa_q {H}_{LT}|
+|\,\pa_q  \trs H|
&\les |\overline{\pa} H|
+ |h|\, |\pa h|,
\label{eq:wavecoordinatenormaltangentialH}\\
 |\,\pa_q {H}_{1LT}|
+|\,\pa_q  \trs H_1 |
&\les |\overline{\pa} H_1|
+ |h|\, |\pa h|+M \big|\chi^\prime\big(\tfrac{r}{t+1}\big)\big| (1+t+r)^{-2},
\label{eq:wavecoordinatenormaltangential}
\end{align}
where $\chi^\prime(s)$, is a function supported when $1/4\leq s\leq 1/2$.
Moreover \eqref{eq:wavecoordinatenormaltangentialH} hold also for $h$ in place of $H$.
\end{lemma}
\begin{proof}
It follows from the previous lemmas that
$$
|\pa_q \widehat{H}_{1\,LU}|\lesssim |\overline{\pa} H_1|+|h|\,|\pa h|
+M\big|\chi^\prime\big(\tfrac{r}{t+1}\big)\big|(1+t+r)^{-2}.
$$
Picking $U=T$ respectively $U=\underline{L}$ gives
\eqref{eq:wavecoordinatenormaltangential}.
\end{proof}
\begin{proposition}\label{prop:wavecoorddecay}
With $H_{1UV}=H_1^{\mu\nu}U_\mu V_\nu$ and $\trs H_1=\delta^{AB} H_{1AB}$   we have for $T=\{L,S_1,S_2\}$
\begin{align}
 |\,\pa_q {H}_{1LT}|
+|\,\pa_q  \trs H_1 | &\les
\varepsilon(1+t+r)^{-2+2\delta}(1\!+q_+)^{-\gamma},
\label{eq:wavecoordinatederivative}
\\
|{H}_{1LT}| \!+|\trs{H}_{1}|
\les \varepsilon &
(1\!+t\!+r)^{-1-\gamma+2\delta}\!\! +
\varepsilon(1\!+t)^{-2+2\delta}(1\!+q_-)
\les \varepsilon(1\!+t\!+r)^{-1-\gamma+2\delta} (1\!+q_-)^{\gamma}.
\label{eq:wavecoordinatefunction}
\end{align}
The same estimates hold for $H$ in place of $H_1$ if $\gamma$ is replaced by $2\delta$,
\end{proposition}
\begin{proof} The proof follows \cite{LR3,L4}.  \eqref{eq:wavecoordinatefunction}
follows from integrating
\eqref{eq:wavecoordinatederivative}
 in the $t-r$ direction from initial data.
 When $|t-r|>t/8$ say the estimate ar estimates follow from Proposition \ref{prop:weakdecayhormander} so we may assume that $|t-r|<t/8$.
It then follows from Lemma \ref{lem:wavecoordestimate0} and Proposition \ref{prop:weakdecayhormander}
\beq
|\pa_q {H}_{1LT}|
+|\pa_q  \trs H_1 |\lesssim |\overline{\pa} H_1|+|h|\,|\pa h|
\lesssim \varepsilon(1+t+r)^{-2+2\delta}(1\!+q_+)^{-\gamma}
+\varepsilon^2(1+t+r)^{-2+2\delta}(1\!+q_+)^{-1-2\delta}.
\eq
\end{proof}

\subsubsection{The leading order behavious of the inhomogeneous term towards null infinity}
The inhomogeneous term in Einstein's equations can be written
\begin{align}
&F_{\mu\nu}= P(\pa_{\mu} h,\pa_{\nu} h) + Q_{\mu\nu}(\pa h,\pa
h)+G_{\mu\nu}(h)(\pa h,\pa h),\nn\\
&P(h, k) =
\frac{1}{2}  m^{\alpha\alpha^\prime}m^{\beta\beta^\prime} \
h_{\alpha\beta}\,  k_{\alpha^\prime\beta^\prime}
-\frac{1}{4}  m^{\alpha\alpha^\prime}
h_{\alpha\alpha^\prime} \,  m^{\beta\beta^\prime}
h_{\beta\beta^\prime}.
\end{align}
where
$G_{\mu\nu}(h)(\pa h,\pa h)$ is cubic
$$
\big| G_{\mu\nu}(h)(\pa h,\pa h)\big| \lesssim |h|\, |\pa h|^2.
$$
and $ Q_{\mu\nu}(\pa h,\pa
h)$ satisfy the standard null condition and hence
\beq \label{eq:nullcondest}
|Q(\pa h,\pa k)|\les |\overline{\pa} h|\,|\pa k|+|\pa h|\, |\overline{\pa} k|.
\eq
The main term $P(\pa_{\mu} h,\pa_{\nu} h)$ can be further analyzed as follows. First we note that
\beq\label{eq:PPnull}
\big| P(\pa_\mu h,\pa_\nu k)- L_{\mu}L_\nu P(\pa_q h,\pa_q h)\big|\les |\overline{\pa} h| \,|\pa k|
+|\pa h|\,|\overline{\pa} k|.
\eq

Expressing $P(h,k)=P_{\mathcal N}(h,k)$ in a null frame we have
\begin{multline}\label{eq:nullframeP}
P_\mathcal{N\,}(h,k)
=-\frac{1}{8}\big({h}_{LL}
{k}_{\underline{L}\underline{L}}
+{h}_{\underline{L}\underline{L}}
{k}_{{L}{L}}\big) -\frac{1}{4}\delta^{CD}\delta^{C^\prime
D^\prime}\big(2{h}_{CC^\prime}{k}_{DD^\prime}-
{h}_{CD} {k}_{C^\prime D^\prime}\big)\\
+\frac{1}{4}\delta^{CD}\big(2{h}_{C
L}{k}_{D\underline{L}} +2{h}_{C
\underline{L}}{k}_{D{L}}- {h}_{CD}
{k}_{L\underline{L}}-{h}_{L\underline{L}}
{k}_{CD} \big).
\end{multline}
Taking into account the wave coordinate condition this reduces in leading order to
 $P_{\!\mathcal{N}}(\pa_q h,\pa_q h)\sim P_{\!\mathcal{S}}(\pa_q h,\pa_q h)$,
 where
\beq
P_{\mathcal{S}} (D,E)= -\widehat{D}_{AB}\, \widehat{E}^{AB}\!/2,\quad A,B\in\mathcal{S},\quad
\text{where} \quad \widehat{D}_{AB}=D_{AB}-\delta_{AB}\trs D/2,\quad
\trs D=\delta^{AB}D_{AB}.
\eq
In fact, by \eqref{eq:nullframeP} we have
\beq\label{eq:tanP}
\big|P_{\mathcal N}(h,k)-P_{\mathcal S}(h,k)|\big|\les \big(|h\,|_{L\mathcal
T}+|\trs h|\big)|k| +|h\,|\big(|k\,|_{L\mathcal
T}+|\trs k_{}|\big),\quad\text{where}\quad |h|_{L\mathcal
T}=|h_{LL}|+|h_{LS_1}|+|h_{LS_2}|.
\eq
Also using \eqref{eq:wavecoordinatenormaltangentialH} and that fact the $H=-h+O(h^2)$ we get
\beq \big| P_{\mathcal N}(\pa_q h,\pa_q h)-P_{\mathcal S}(\pa_q h,\pa_q h)\big|
\les (|\overline{\pa} h|+|h||\pa h|) |\pa h|.
\eq
Summing up we have shown
\begin{lemma} \label{lemma:Pslash}
Let
\begin{equation}\label{eq:Pmunudef}
 {\slash \!\!\!\!} P_{\mu\nu}(\pa h,\pa k)=\overline{\chi}\big(\tfrac{\langle \,r-t\,\rangle }{t+r}\big)
 L_{\mu}L_\nu P_\mathcal{S\,}(\pa_q h,\pa_q h),
 \end{equation}
 where $\overline{\chi}_{\!}\!\in\!  C_0^\infty\!\!$ satisfies
 $\overline{\chi}({q})_{\!}\!=\!0$, when
$|{q}|\!\geq \!3_{\!}/4$ and $\overline{\chi}({q})_{\!}\!=\!1$, when $|{q}|\!\leq\! 1_{\!}/2$. Here $<q>=(1+|q|^2)^{1/2}$.
Then
\beq\big|F_{\mu\nu}(h)(\pa h,\pa h)- {\slash \!\!\!\!} P_{\mu\nu}(\pa h,\pa h) \big|
\les |\overline{\pa} h| |\pa h|+| h| |\pa h|^2\!\!, \qquad \text{when}\quad |q|\leq 1/2.
\label{eq:Festimateh}
\eq
\end{lemma}
Using \eqref{eq:hormanderweakdecayderh1}-\eqref{eq:hormanderweakdecayh1} and
\eqref{eq:weakdecayh1}-\eqref{eq:weakdecayderh1} we obtain
\begin{lemma} With notation as in the previous lemma we have
\beq\big|F_{\mu\nu}(h)(\pa h,\pa h)- {\slash \!\!\!\!} P_{\mu\nu}(\pa h,\pa h) \big|
\les\frac{\varepsilon^2}{(1\!+t\!+r)^{3-4\delta}(1\!+|q|)(1\!+q_+)^{4\delta}}.
\label{eq:Festimate}
\eq
\end{lemma}

\subsubsection{The leading order of the geometric wave operator towards null infinity}
Expanding in a null frame as in the proof of Lemma \ref{lem:divergencenullframe} using
\eqref{eq:pabyZ} we get
\begin{lemma} We have
\beq\label{eq:seconddernullframe}
\big|k^{\alpha\beta}\pa_\alpha\pa_\beta \phi\big|\les
\Big(\frac{| k_{LL}| }{1+|q|}
 + \frac{|k|  }{1+t+r}\Big)\sum_{|K|\leq 1}|\pa Z^K \phi|.
\eq
\end{lemma}
As a consequence we get
\begin{lemma} \label{lem:curvedwaveminuswaveasympschwarz} We have
\beq
\big|(\Boxr_g -\Box_0)\phi\big|\les
\frac{\varepsilon(1+q_+)^{-\gamma}}{(1+t+r)^{1+\gamma-2\delta}(1+|q|)^{1-\gamma}}\sum_{|K|\leq 1}|\pa Z^K \phi| ,
\eq
where the asymptotic Schwarzschild wave operator is given by
\beq\label{eq:box0def}
\Box_{\,0}=\big(m^{\alpha\beta}+H_0^{\alpha\beta}\big)\pa_\alpha\pa_\beta,
\quad\text{where}\quad
H_0^{\alpha\beta}=-\tfrac{M}{r}{\chi}\big(\tfrac{r}{1+t}\big)\,
\delta^{\alpha\beta}.
\eq
\end{lemma}
\begin{proof} We apply \eqref{eq:seconddernullframe} to $H_1^{\alpha\beta}\pa_\alpha\pa_\beta \phi$ using \eqref{eq:wavecoordinatefunction} and \eqref{eq:hormanderweakdecayh1}:
\begin{equation*}
\big|H_1^{\alpha\beta}\pa_\alpha\pa_\beta \phi\big|\les
\Big(\frac{| H_{1\,LL}| }{1+|q|}
 + \frac{|H_1|  }{1+t+r}\Big)\sum_{|K|\leq 1}|\pa Z^K \phi|
 \les \frac{C\varepsilon}{(1+t+r)^{1+\gamma-2\delta}(1+|q|)^{1-\gamma}(1+q_+)^\gamma}\sum_{|K|\leq 1}|\pa Z^K \phi|
\end{equation*}
\end{proof}
In spherical coordinates this takes the form
Here
\begin{equation}\label{eq:Box0insphericalcoord}
\Box_0\phi=\Big(-\pa_t^2+\triangle_x
-\tfrac{M}{r}{\chi}\big(\tfrac{r}{1+t}\big)\big(\pa_t^2+\triangle_x\big)\Big)\phi
=\frac{1}{r}\Big(-\pa_t^2+\pa_r^2-\tfrac{M}{r}{\chi}\big(\tfrac{r}{1+t}\big)(\pa_t^2+\pa_r^2)\Big)(r\phi)
+\Big(1-\tfrac{M}{r}{\chi}\big(\tfrac{r}{1+t}\big)\Big) \frac{1}{r^2}\triangle_\omega \phi
\end{equation}

\subsubsection{The leading order of the metric towards space like infinity.} Following \cite{LR3} we have defined $h^0_{\alpha\beta}=-H_0^{\alpha\beta}$ to be a function that picks up the leading behavior of the initial data at space like infinity
\beq \label{eq:h0def}
h^0_{\mu\nu}=\tfrac{M}{r}{\chi}\big(\tfrac{r}{1+t}\big)\,
\delta_{\mu\nu}
\eq
However it would perhaps have bene more natural to define it to be a solution of the homogeneous wave $\Box h^0=0$ (or even better $\Box_0 h^0=0$.) with data coinciding with this function at time $0$ in which case $h^0_{\mu\nu}=\chi(r-t) \tfrac{M}{r}\delta_{\mu\nu}$ which is equal to \eqref{eq:h0def} in the exterior when
$r\geq t+1$. We therefore think of  \eqref{eq:h0def} as an approximate solution to the
homogeneous wave equation. By \eqref{eq:Box0insphericalcoord} we have
\begin{equation}\label{eq:Box0h0}
\Box_0 h^0_{\mu\nu}=\frac{1}{r}\Big(-\pa_t^2+\pa_r^2
-\tfrac{M}{r}{\chi}\big(\tfrac{r}{1+t}\big)(\pa_t^2+\pa_r^2)\Big)\chi\big(\tfrac{r}{t+1}\big) M\delta_{\mu\nu}=\frac{M\delta_{\mu\nu}}{(1+t)^3} \Big(\chi_1^\prime\big(\tfrac{r}{t+1}\big)+\tfrac{1}{t+1} \chi_2^\prime\big(\tfrac{r}{t+1}\big)\Big)
\end{equation}
for some functions $\chi_i^\prime(s)$ that vanish when $s\geq 1/2$ or $s\leq 1/4$.
This in particular means that in the exterior and in the wave zone it is a solution of the wave operator $\Box_0$. Moreover
\beq\label{eq:BoxgBox0h0}
(\Boxr_g -\Box_0)h^0_{\mu\nu}= H_1^{\alpha\beta}\pa_\alpha \pa_\beta
h^0_{\mu\nu}=\frac{M\delta_{\mu\nu}}{(1+t+r)^3} H_1^{\alpha\beta} \chi_{\alpha\beta}\big(\tfrac{r}{t+1},\omega\big)
\eq
for some smooth function $\chi_{\alpha\beta}(s,\omega)$ supported when $s\geq 1/4$.
Hence
\beq
\big|(\Boxr_g-\Box_0) h^0_{\mu\nu}\big|\les \frac{\varepsilon M}{(1+t+r)^{4-2\delta}(1+q_+)^\gamma}
\eq
Summing up we have proved that
\begin{proposition}[Asymptotic Approximate Einstein's equations]
\label{prop:approxwaveequation}
When $|x|\geq ct$ we have
\begin{equation}\label{eq:approxwaveequation}
 \big|\,\Box_{\,0}\,
h^1_{\mu\nu}- {\slash \!\!\!\!}P_{\mu\nu}(\pa h,\pa h)\big|
\les \frac{\varepsilon^2}{(1\!+t\!+r)^{2+\gamma-4\delta} (1\!+|q|)^{2-\gamma}(1\!+q_+)^{4\delta}}.
\end{equation}
\end{proposition}

\subsubsection{The sharp decay estimates for first order derivatives from the wave equation}
We will now derive sharp estimates for the first order derivatives, following
\cite{L1,LR3} we have
\begin{lemma}\label{lem:transversalder} Let $D_t=\{(t,x);\, \big|t-|x|\big|\leq c_0 t \}$, for some constant
$0<c_0<1$ and let $\overline{w}(q)>0$ be an increasing positive weight $\overline{w}^{\,\prime}(q)\geq 0$. Then
\begin{multline}\label{eq:transversalder}
(1+t+|x|) \,|\pa \phi_{UV}(t,x)\, \overline{w}(q)| \les\!\sup_{0\leq
\tau\leq t}
\sum_{|I|\leq 1}\|\,Z^I\! \phi(\tau,\cdot)\, \overline{w}\|_{L^\infty}\\
+ \int_0^t\Big( (1+\tau)\|
\,(\Box_0 \phi)_{UV}(\tau,\cdot)\, \overline{w}\|_{L^\infty(D_\tau)} +\sum_{|I|\leq 2} (1+\tau)^{-1}
\| Z^I \phi(\tau,\cdot)\, \overline{w}\|_{L^\infty(D_\tau)}\Big)\, d\tau.
\end{multline}
\end{lemma}
\begin{proof}
 Since $\Box\phi=-r^{-1}(\pa_t^2-\pa_r^2)(r\phi)+r^{-2}\triangle_\omega \phi $, where $\triangle_\omega=\sum\Omega_{ij}^2$ and $
 |Z U|\leq C$, for $U\in \{ A,B, L,\underline{L}\}$,
it follows that
 \beq
 \big| \Box_0  {\phi}_{UV}
 -U^\mu V^\nu\Box_0 \phi_{\mu\nu}\big|
 \leq{r}^{-2}\,\, {\sum}_{|J|\leq 1} \,\,|Z^{J} \phi|.
 \eq
Using \eqref{eq:Box0insphericalcoord} we get
\begin{equation*}
\Box_0\phi
=\frac{1}{r}\Big(4\pa_s\pa_q
-2\tfrac{M}{r}{\chi}\big(\tfrac{r}{1+t}\big)(\pa_q^2+\pa_s^2)\Big)(r\phi)
+\Big(1-\tfrac{M}{r}{\chi}\big(\tfrac{r}{1+t}\big)\Big) \frac{1}{r^2}\triangle_\omega \phi,
\end{equation*}
where  $\pa_q=(\pa_r\!-\pa_t)/2 $ and $\pa_s=(\pa_r\!+\pa_t)/2$.
Hence
\beq
\Big|\big(4\pa_s -\tfrac{2M}{r}\pa_q\big)\pa_q (r\phi) -r \Box_0\phi  \Big|
\les {r}^{-1}\,\, {\sum}_{|J|\leq 2} \,\,|Z^{J} \phi|
\eq
so with $s=t+r$
\beq
\Big|\big(\pa_s -\tfrac{M}{s}\pa_q\big)\pa_q (r \phi_{UV})\Big|
\les r\big|(\Box_0 \phi)_{UV}\big|+ (t+r)^{-1}\,\, {\sum}_{|J|\leq 2} \,\,|Z^{J} \phi|,
\qquad |t-r|\leq c_0 t
\eq
Integrating this along the flow lines of the vector field $\big(\pa_s -\tfrac{M}{s}\pa_q\big)$ from the boundary of $D=\cup_{\tau\geq 0}D_\tau$
to any point inside $D$, using that $\overline{w}$ is decreasing along the flow lines, gives that for any $(t,x)\in D$
\begin{multline}
|\pa_q(r \phi_{UV}(t,x))\, \overline{w}| \les\!\sup_{0\leq
\tau\leq t}
\sum_{|I|\leq 1}\|\,Z^I \phi(\tau,\cdot)\, \overline{w}\|\\
+ \int_0^t\Big( (1+\tau)\|
\,(\Box_0 \phi)_{UV}(\tau,\cdot)\, \overline{w}\|_{L^\infty(D_\tau)} +\sum_{|I|\leq 2} (1+\tau)^{-1}
\| Z^I \phi(\tau,\cdot)\, \overline{w}\|_{L^\infty(D_\tau)}\Big)\, d\tau.
\end{multline}
The lemma now follows from \eqref{eq:derbytransder}, that the estimate is trivially true for $h$ and when
$|r-t|\geq c_0 t$.
\end{proof}
From Lemma \ref{lem:transversalder} and the estimate \eqref{eq:hormanderweakdecayh1} we get
\begin{lemma}\label{lem:transversaldersimplified} Let $D_t=\{(t,x);\, \big|t-|x|\big|\leq c_0 t \}$, for some constant
$0<c_0<1$, and $\overline{w}(q)=(1+q_+)^{1+\gamma^\prime}$, where $-1\leq \gamma^\prime<\gamma-2\delta$. Then
\begin{equation}\label{eq:transversaldersimplified}
(1+t+|x|) \,|\pa h^1_{UV}(t,x)\, \overline{w}(q)| \les\varepsilon + \int_0^t (1+\tau)\|
\,(\Box_0 h^1)_{UV}(\tau,\cdot)\, \overline{w}\|_{L^\infty(D_\tau)} d\tau.
\end{equation}
\end{lemma}
Using Lemma \ref{lem:transversaldersimplified}
and Proposition \ref{prop:approxwaveequation} we obtain
\begin{proposition}\label{prop:sharpfirstorderwave} If the weak energy bounds and initial bounds hold then we have for any $0\leq \gamma^\prime<\gamma-4\delta$
\begin{align}\label{eq:sharphTUder}
(1+t+r) (1+q_+)^{1+\gamma^\prime}\big| \pa h^1_{TU}\big|&\les \varepsilon ,\\
\label{eq:sharpHanyder}
(1+t+r) (1+q_+)^{1+\gamma^\prime}\big| \pa h^1\big|&\les \varepsilon
(1+\varepsilon\ln{(2+t)})\les \varepsilon (1+t)^\varepsilon.
\end{align}
The same estimates hold for $h$ in place of $h^1$ if $\gamma^\prime=0$.
\end{proposition}
\begin{proof} We want to apply Lemma \ref{lem:transversaldersimplified} to the decomposition in Proposition \ref{prop:approxwaveequation}. To prove \eqref{eq:sharphTUder} we note that
${\slash \!\!\!\!}P_{TU}=0$ in $D_t$.  Moreover $\newhat{T}=0$ in $D_t$, for $t \geq 2K/(1-c)$, if we pick $c_0$ so small that $c_0\leq 1-c'$, where $c' = (1+c)/2$.  Also Lemma \ref{lem:Sobolevt3} implies that $\newhat{T}(t,x)$ is uniformly bounded for $0 \leq t \leq 2K/(1-c)$.  From the preceeding lemmas it follows that all the terms in the right hand side
of \eqref{eq:transversalder} are bounded independently of $t$ by a constant
times $\varepsilon$ when $(U,V)=(U,T)$ and this proves \eqref{eq:sharphTUder}.
To prove \eqref{eq:sharpHanyder} we note that the only new term is ${\slash \!\!\!\!}P_{\mu\nu}$
which is controlled by
\beq
|{\slash \!\!\!\!}P(\pa h,\pa h)|\leq |\pa h_{TS}|^2\les \varepsilon^2 (1+t+r)^{-2}
(1+q_+)^{-2-2\gamma^\prime}
\eq
by the first part and multiplying by $(1+t)$ and integrating gives a logarithm.
\end{proof}

\subsection{The commutators and Lie derivatives} \label{subsec:commutation}
We will use Lie derivatives which will simplify the commutators very much
by removing the lower order terms. It was first observed in \cite{L4} that
one can get bounds from the wave coordinate condition for Lie derivatives.
Here we take it further and observe that the Lie derivative unlike vector fields preserve the geometric null
structure of not only the wave coordinate condition, but also of the nonlinear inhomogeneous terms of Einstein's equations and the commutators with the geometric wave operator.
\subsubsection{Modified Lie derivatives applied to the equations}
The Lie derivative applied to a $(r,s)$ tensor $K$ is defined by
\begin{equation} \label{eq:Liederiv}
{\mathcal L}_Z K^{\alpha_1\dots \alpha_r}_{\beta_1\dots \beta_s}
=Z K^{\alpha_1\dots \alpha_r}_{\beta_1\dots \beta_s}
-\pa_\gamma Z^{\alpha_1} K^{\gamma\dots \alpha_r}_{\beta_1\dots \beta_s}-\cdots
-\pa_\gamma Z^{\alpha_r} K^{\alpha_1\dots \gamma}_{\beta_1\dots \beta_s}
+\pa_{\beta_1} Z^{\gamma} K^{\alpha_1\dots \alpha_r}_{\gamma\dots \beta_s}+\cdots
+\pa_{\beta_s} Z^{\gamma} K^{\alpha_1\dots \alpha_r}_{\beta_1\dots \gamma}.
\end{equation}
Recall that Lie derivative satisfies Leibnitz rule. For the case of our vector fields
$\pa_\gamma Z^\beta$ are constant, which results in the following commutation properties.

\begin{proposition} \label{prop:Liecommutation}
	If $K$ is an (r,s) tensor then, with respect to the coordinate system $\{x^{\mu}\}$, the vector fields $Z = \partial_{x^{\mu}}, \Omega_{ij}, B_i, S$ satisfy,
	\begin{equation} \label{eq:Liepartial}
		{\mathcal L}_Z\pa_{\mu_1}\!\cdots\pa_{\mu_k} K^{\alpha_1\dots \alpha_r}_{\beta_1\dots \beta_s}
		=
		\pa_{\mu_1}\!\cdots\pa_{\mu_k} {\mathcal L}_Z K^{\alpha_1\dots \alpha_r}_{\beta_1\dots\beta_s},
	\end{equation}
	and
	\begin{equation} \label{eq:Liediv}
		{\mathcal L}_Z \pa_\mu K^{\mu\dots \alpha_r}_{\beta_1\dots \beta_s}
		=
		\pa_\mu {\mathcal L}_Z K^{\mu\dots \alpha_r}_{\beta_1\dots \beta_s}.
	\end{equation}
\end{proposition}

\begin{proof}
	From the definition \eqref{eq:Liederiv},
	\begin{align*}
		{\mathcal L}_Z\pa_{\mu_1}\!\cdots\pa_{\mu_k} K^{\alpha_1\dots \alpha_r}_{\beta_1\dots \beta_s}
		=
		\
		&
		Z \left( \pa_{\mu_1}\!\cdots\pa_{\mu_k} K^{\alpha_1\dots \alpha_r}_{\beta_1\dots \beta_s} \right)
		+
		\partial_{\mu_1} Z^{\gamma} \pa_{\gamma}\!\cdots\pa_{\mu_k} K^{\alpha_1\dots \alpha_r}_{\beta_1\dots \beta_s}
		+
		\ldots
		+
		\partial_{\mu_k} Z^{\gamma} \pa_{\mu_1}\!\cdots\pa_{\gamma} K^{\alpha_1\dots \alpha_r}_{\beta_1\dots \beta_s}
		\\
		&
		-
		\partial_{\gamma} Z^{\alpha_1} \pa_{\mu_1}\!\cdots\pa_{\mu_k} K^{\gamma\dots \alpha_r}_{\beta_1\dots \beta_s}
		-
		\ldots
		-
		\partial_{\gamma} Z^{\alpha_r} \pa_{\mu_1}\!\cdots\pa_{\mu_k} K^{\alpha_1 \dots \gamma}_{\beta_1\dots \beta_s}
		\\
		&
		+
		\partial_{\beta_1} Z^{\gamma} \pa_{\mu_1}\!\cdots\pa_{\mu_k} K^{\alpha_1\dots \alpha_r}_{\gamma\dots \beta_s}
		+
		\ldots
		+
		\partial_{\beta_s} Z^{\gamma} \pa_{\mu_1}\!\cdots\pa_{\mu_k} K^{\alpha_1\dots \alpha_r}_{\beta_1 \dots \gamma},
	\end{align*}
	and
	\begin{multline*}
		\pa_{\mu_1}\!\cdots\pa_{\mu_k} {\mathcal L}_Z K^{\alpha_1\dots \alpha_r}_{\beta_1\dots\beta_s}
		=
		\pa_{\mu_1}\!\cdots\pa_{\mu_k}
		\bigg[
		Z^{\gamma} \partial_{\gamma} \left( K^{\alpha_1\dots \alpha_r}_{\beta_1\dots\beta_s} \right)
		-
		\partial_{\gamma} Z^{\alpha_1} K^{\gamma \dots \alpha_r}_{\beta_1\dots\beta_s}
		-
		\ldots
		-
		\partial_{\gamma} Z^{\alpha_r} K^{\alpha_1 \dots \gamma}_{\beta_1\dots\beta_s}
		\\
		+
		\partial_{\beta_1} Z^{\gamma} K^{\alpha_1 \dots \alpha_r}_{\gamma \dots\beta_s}
		+
		\ldots
		+
		\partial_{\beta_s} Z^{\gamma} K^{\alpha_1 \dots \alpha_r}_{\beta_1\dots\gamma}
		\bigg].
	\end{multline*}
	The equality \eqref{eq:Liepartial} follows directly since $\partial_{x^{\alpha}} Z^{\beta}$ is constant for each of the vector fields $Z$.  The equality \eqref{eq:Liediv} follows directly from \eqref{eq:Liepartial}.
\end{proof}

Since, for an $(r,s)$ tensor $K$, the quantity $\pa_{\mu_1}\!\cdots\pa_{\mu_k} K^{\alpha_1\dots \alpha_r}_{\beta_1\dots \beta_s}$ appearing in \eqref{eq:Liepartial} is not a geometric object, its Lie derivative is defined formally, in the $\{x^{\mu}\}$ coordinate system, using the coordinate expression \eqref{eq:Liederiv}.  Alternatively, one could note that, in the $\{x^{\mu}\}$ coordinate system,
\[
	\pa_{\mu_1}\!\cdots\pa_{\mu_k} K^{\alpha_1\dots \alpha_r}_{\beta_1\dots \beta_s}
	=
	D_{\mu_1}\!\cdots D_{\mu_k} K^{\alpha_1\dots \alpha_r}_{\beta_1\dots \beta_s},
\]
where $D$ denotes the connection of the Minkowski metric, since the Christoffel symbols of $D$ with respect to the Cartesian coordinate system $\{ x^{\mu}\}$ vanish, $D_{\partial_{x^{\alpha}}} \partial_{x^{\beta}} = 0$.  One could then give a geometric proof Proposition \ref{prop:Liecommutation} using the fact that the curvature tensor of $D$ vanishes and $D^2 Z = 0$ for each of the vector fields $Z$.

Let the modified Lie derivative be defined by
\begin{equation}\label{eq:modifiedlie}
\widehat{\mathcal L}_Z K^{\alpha_1\dots \alpha_r}_{\beta_1\dots \beta_s}
={\mathcal L}_Z K^{\alpha_1\dots \alpha_r}_{\beta_1\dots \beta_s}
+\tfrac{r-s}{4}(\pa_\gamma Z^\gamma)K^{\alpha_1\dots \alpha_r}_{\beta_1\dots \beta_s}.
\end{equation}
With this definition $\widehat{\mathcal L}_Z m^{\alpha\beta}=0$ and
$\widehat{\mathcal L}_Z m_{\alpha\beta}=0$ for the vector fields in our collection as
the modified Lie derivative is defined so it commutes with contractions with the
 the Minkowski metric.
 Let $h_{\alpha\beta}$ and $k_{\alpha\beta}$ be $(0,2)$ tensors and
let $S_{\mu\nu}(\pa h,\pa k)$ be a $(0,2)$ tensor which is a quadratic form in the $(0,3)$ tensors $\pa h$ and $\pa k$ with two contractions with the Minkowski metric (in particular $P(\pa_\mu h,\pa_\nu h)$ or $Q_{\mu\nu}(\pa h,\pa k)$). Then
 \beq
 {\mathcal L}_Z\,\big(  S_{\mu\nu}(\pa h,\pa k)\big)
 =S_{\mu\nu}(\pa \widehat{\mathcal L}_Z h,k)+S_{\mu\nu}(\pa h,\pa \widehat{\mathcal L}_Z k).
 \eq
 Moreover
 \beq
 {\mathcal L}_Z \, \big(g^{\alpha\beta} \pa_\alpha\pa_\beta h_{\mu\nu}\big)
 =\big(\widehat{\mathcal L}_Z  g^{\alpha\beta} \big)\pa_\alpha\pa_\beta h_{\mu\nu}
 + g^{\alpha\beta} \pa_\alpha\pa_\beta \widehat{\mathcal L}_Z h_{\mu\nu}.
 \eq
 Let ${\mathcal L}_Z^I$ be a product of $|I|$ Lie derivatives with respect to $|I|$ vector fields $Z$.
 It follows that
 \begin{equation}
 \Boxr_g \widehat{\mathcal L}_Z^I h_{\mu\nu}
 =\big[ \Boxr_g \widehat{\mathcal L}_Z^I -{\mathcal L}_Z^I \Boxr_g\big]h_{\mu\nu}
 +{\mathcal L}_Z^I F_{\mu\nu}(H)(\pa h,\pa h)+{\mathcal L}_Z^I T_{\mu\nu}
 \end{equation}
 where
 \beq
 \big[ \Boxr_g \widehat{\mathcal L}_Z^I -{\mathcal L}_Z^I \Boxr_g\big]\phi_{\mu\nu}
 =-\!\!\!\!\!\sum_{J+K=I,\,|K|<|I|}\!\!\!\!\! \widehat{\mathcal L}_Z^J  H^{\alpha\beta} \,\pa_\alpha\pa_\beta \widehat{\mathcal L}_Z^K \phi_{\mu\nu}
 \end{equation}
 and
 \begin{equation}
 {\mathcal L}_Z^I F_{\mu\nu}(H)(\pa h,\pa h)=\!\!\!\!\!\sum_{J+K=I}\!\! P\big(\pa_\mu \widehat{\mathcal L}_Z^J h,\pa_\nu \widehat{\mathcal L}_Z^K h\big)
 +\!\!\!\!\!\sum_{J+K=I}\!\!Q_{\mu\nu}\big(\pa \widehat{\mathcal L}_Z^J h,\pa \widehat{\mathcal L}_Z^K h\big)\\
 +{\mathcal L}_Z^I G_{\mu\nu}(H)(\pa h,\pa h),
 \end{equation}
 where
 \beq
 \big| {\mathcal L}_Z^I G_{\mu\nu}(H)(\pa h,\pa h)\big|\les
 \sum_{I_1+I_{2}+\dots I_{k}=I,\, k\geq 3 } |\widehat{\mathcal L}_Z^{I_3} H|
 \cdots |\widehat{\mathcal L}_Z^{I_k} H|
 |\pa \widehat{\mathcal L}_Z^{I_1} h|\, |\pa \widehat{\mathcal L}_Z^{I_2} h|,
 \eq
 i.e. at least one factor of $|\widehat{\mathcal L}_Z^{I_k} H|$.
 Finally the wave coordinate condition
 \beq
 \pa_\mu  \widehat{\mathcal L}_Z \widehat{H}^{\mu\nu}=\big(\widehat{\mathcal L}_Z +\tfrac{\pa_\gamma Z^\gamma}{2}\big)\pa_\mu \widehat{H}^{\mu\nu}
 =\big(\widehat{\mathcal L}_Z +\tfrac{\pa_\gamma Z^\gamma}{2}\big) W^{\nu}(H,\pa h)
 \eq
 It follows that
 \beq
 \big| \pa_\mu  \widehat{\mathcal L}_Z^I \widehat{H}^{\mu\nu} \big|\les
 \sum_{I_1+\cdots+ I_k,\, k\geq 2} \big|  \widehat{\mathcal L}_Z^{I_2} {H}\big|\cdots
  \big|  \widehat{\mathcal L}_Z^{I_k} {H}\big|\,  \big|\pa  \widehat{\mathcal L}_Z^{I_1} H\big|, \label{eq:wavecoordinateLie}
 \eq
 where
 \begin{equation}
 \widehat{\mathcal L}_Z^{I}\widehat{H}{}^{\mu\nu}\!=\! \widehat{\mathcal L}_Z^{I}H^{\mu\nu}\!-m^{\mu\nu}
\tr_m  \widehat{\mathcal L}_Z^{I}H_{\!}/2,\qquad \tr  \widehat{\mathcal L}_Z^{I}H\!=m_{\alpha\beta}  \widehat{\mathcal L}_Z^{I}H^{\alpha\beta}\!\!
\label{eq:wavehatLie}
 \end{equation}

We have
\beq
\sum_{|I|\leq k} |Z^I K|\les \sum_{|I|\leq k}|\widehat{\mathcal L}_Z^I K|\les \sum_{|I|\leq k} |Z^I K|,
\eq
since the Lie derivative just adds lower order terms.

\subsubsection{Estimates from the wave coordinate condition}
It follows from Lemma \ref{lem:divergencenullframe} and \eqref{eq:wavecoordinateLie}
and the fact that $ \big|  \widehat{\mathcal L}_Z^{I_k} {h}\big|\les 1$ for small $|I_k|$
that
\begin{lemma}
\begin{align}
 |\pa_q \widehat{\mathcal L}_Z^I{H}|_{L\mathcal T}
+|\pa_q \trs\widehat{\mathcal L}_Z^I {H} |&\les
|\overline{\pa} \widehat{\mathcal L}_Z^I H|
+ \sum_{|J|+|K|\leq |I|} \big|  \widehat{\mathcal L}_Z^{J} {h}\big|
    \big|\pa  \widehat{\mathcal L}_Z^{K} h\big|
\label{eq:higherwavecoordinateLiederivativeH}
\\
 |\pa_q \widehat{\mathcal L}_Z^I{H}_1|_{L\mathcal T}
+|\pa_q  \trs\widehat{\mathcal L}_Z^I {H} |&\les
|\overline{\pa} \widehat{\mathcal L}_Z^I H_1|
+ \sum_{|J|+|K|\leq |I|} \big|  \widehat{\mathcal L}_Z^{J} {h}\big|
    \big|\pa  \widehat{\mathcal L}_Z^{K} h\big|+\big|\chi^\prime\big(\tfrac{r}{t+1}\big)\big| \frac{M}{(1+t+r)^2},
\label{eq:higherwavecoordinateLiederivative}
\end{align}
where $\chi^\prime(s)$, is a function supported when $1/4\leq s\leq 1/2$.
Moreover \eqref{eq:higherwavecoordinateLiederivativeH} hold also for $h$ in place of $H$.
\end{lemma}

\subsubsection{$L^\infty$ estimates from the wave coordinate condition}

For low derivatives \eqref{eq:higherwavecoordinateLiederivative} leads to the following
\begin{proposition}\label{prop:wavecoorddecayhighlow} For $|I|\leq N-4$ we have
\begin{align}
|\pa_q \widehat{\mathcal L}_Z^I{H}_1|_{L\mathcal T}
+|\pa_q  \trs\widehat{\mathcal L}_Z^I {H}_1 |&\les
\varepsilon(1+t+r)^{-2+2\delta}(1\!+q_+)^{-\gamma},
\label{eq:wavecoordinatederivativeLie}
\\
|\widehat{\mathcal L}_Z^I{H}_1|_{L\mathcal T}
+|\trs\widehat{\mathcal L}_Z^I {H}_1 |
\les \varepsilon &
(1\!+t\!+r)^{-1-\gamma+2\delta}\!\! +
\varepsilon(1\!+t)^{-2+2\delta}(1\!+q_-)
\les \varepsilon(1\!+t\!+r)^{-1-\gamma+2\delta} (1\!+q_-)^{\gamma}.
\label{eq:wavecoordinatefunctionLie}
\end{align}
The same estimates hold for $H$ in place of $H_1$ if $\gamma$ is replaced by $2\delta$.
\end{proposition}
The proof is the same as Proposition \ref{prop:wavecoorddecay}.

\subsubsection{Estimates for the inhomogeneous term}
First using the fact that $ \big|  \widehat{\mathcal L}_Z^{I_k} {h}\big|\les 1$ for small $|I_k|$
\beq
 \big| {\mathcal L}_Z^I G_{\mu\nu}(h)(\pa h,\pa h)\big|\les
 \sum_{|I_1|+|I_{2}|+|I_{3}|\leq |I|,\,  } |\widehat{\mathcal L}_Z^{I_3} H|
 |\pa \widehat{\mathcal L}_Z^{I_1} h|\, |\pa \widehat{\mathcal L}_Z^{I_2} h|,
 \eq
Secondly for any term satisfying classical null condition by \eqref{eq:nullcondest} we have
\beq
\Big|\sum_{J+K=I}\!\!Q_{\mu\nu}\big(\pa \widehat{\mathcal L}_Z^J h,\pa \widehat{\mathcal L}_Z^K h\big)\Big|\les \sum_{|J|+|K|\leq |I|}\!\! |\overline{\pa} \widehat{\mathcal L}_Z^J h|\, |{\pa} \widehat{\mathcal L}_Z^K h|.
\eq
Moreover by \eqref{eq:tanP}, \eqref{eq:PPnull} we have with ${\slash \!\!\!\!}P_{\mu\nu}$
as in \eqref{eq:Pmunudef}
\begin{equation*}
\sum_{J+K=I}\!\!\!\big|P\big(\pa_\mu \widehat{\mathcal L}_Z^J h,\pa_\nu \widehat{\mathcal L}_Z^K h\big)\big|\les\sum_{|J|+|K|\leq |I|\!\!\!\!\!\!\!\!\!\!\!}\big| {\slash \!\!\!\!}P_{\mu\nu}\big(\pa \widehat{\mathcal L}_Z^J h,\pa \widehat{\mathcal L}_Z^K h\big)\big|
+\!\!\!\!\sum_{|J|+|K|\leq |I|\!\!\!\!\!\!\!\!\!\!\!}
\big(|\pa_q \widehat{\mathcal L}_Z^J{h}|_{L{\mathcal T}}
+|\pa_q  \trs \widehat{\mathcal L}_Z^J {h} |\big)\,
|\pa \widehat{\mathcal L}_Z^K h|
\end{equation*}
Summing up we have the estimate
\begin{multline}
\big|{\mathcal L}_Z^I F_{\mu\nu}(h)(\pa h,\pa h)\big|\\
\les \sum_{|J|+|K|\leq |I|}\!\! \!\!\!\!\!\!\big| {\slash \!\!\!\!}P_{\mu\nu}\big(\pa \widehat{\mathcal L}_Z^J h, \widehat{\mathcal L}_Z^K h\big)\big|
+\sum_{|J|+|K|\leq |I|}\!\! |\overline{\pa} \widehat{\mathcal L}_Z^J h|\, |{\pa} \widehat{\mathcal L}_Z^K h|
+\sum_{|I_1|+|I_{2}|+|I_{3}|\leq |I|,\,  } |\widehat{\mathcal L}_Z^{I_3} h|
 |\pa \widehat{\mathcal L}_Z^{I_1} h|\, |\pa \widehat{\mathcal L}_Z^{I_2} h|,
\end{multline}
Dividing up into low and high derivatives we get
\begin{multline}\label{eq:inhomogeneousLieestimate}
\big|{\mathcal L}_Z^I F_{\mu\nu}(h)(\pa h,\pa h)\big|\\
\les \big(|\pa h|_{\mathcal SS}+|\overline{\pa} h|+|h|\,|\pa h|)\!\!\!\sum_{|J|\leq |I|}\!\! \!\!|\pa\widehat{\mathcal L}_Z^J h|
+|\pa h|\!\!\! \sum_{|J|\leq |I|}\!\! |\overline{\pa} \widehat{\mathcal L}_Z^J h|\,
+|\pa h|^2\!\!\! \sum_{|J|\leq |I|}\!\! |\widehat{\mathcal L}_Z^J h|
+\sum_{|K|\leq |I|/2\!\!\!\!\!\!\!\!\!\!\!\!\!\!\!\!\!}|\pa\widehat{\mathcal L}_Z^K h|\sum_{|J|\leq |I|-1\!\!\!\!\!\!\!\!\!\!\!} |\pa\widehat{\mathcal L}_Z^J h|\\
\les \frac{\varepsilon(1\!+\!q_+)^{-1}}{1+t+r}\!\! \sum_{|J|\leq |I|}\!\! \!\!|\pa\widehat{\mathcal L}_Z^J h|
+\frac{\varepsilon(1\!+\!t)^{2\delta}(1\!+\!q_+)^{-2\delta}\!\!\!}{(1+t+r)(1+|q|)}
\!\sum_{|J|\leq |I|}\!\! |\overline{\pa} \widehat{\mathcal L}_Z^J h|
+\frac{\varepsilon^2(1\!+\!t)^{4\delta}(1\!+\!q_+)^{-4\delta}\!\!\!}{(1+t+r)^2(1+|q|)^2}
\!\!\!\sum_{|J|\leq |I|}\!\! |\widehat{\mathcal L}_Z^J h|\\
+\sum_{|K|\leq |I|/2\!\!\!\!\!\!\!\!\!\!\!\!\!\!\!\!\!}|\pa\widehat{\mathcal L}_Z^K h|\sum_{|J|\leq |I|-1\!\!\!\!\!\!\!\!\!\!\!} |\pa\widehat{\mathcal L}_Z^J h|
\end{multline}

\subsubsection{Estimates of the wave operator applied to $h^0$. }
By \eqref{eq:BoxgBox0h0} we have
\beq\label{eq:BoxgBox0h0est}
\big|{\mathcal L}_Z^I\, \big(\Boxr_g -\Box_0\big)h^0_{\mu\nu}\big| \les \frac{M}{(1+t+r)^3}\sum_{|J|\leq |I|} \big|\widehat{\mathcal L}_Z^J H_1|
\eq
 and by \eqref{eq:Box0h0} we have
 \beq\label{eq:Box0h0est}
 \big| {\mathcal L}_Z^I\, \Box_0 h^0_{\mu\nu}\big|
 \les \frac{M}{(1+t+r)^3} \chi^\prime\big(\tfrac{r}{t+1}\big)
 \eq
 where $\chi^\prime(s)$ is supported in $1/4\leq s\leq 1$.

\subsubsection{Estimates of the wave commutator term}
By \eqref{eq:seconddernullframe} we have
\begin{equation*}
\big|\big[ \Boxr_g \widehat{\mathcal L}_Z^I -{\mathcal L}_Z^I\Boxr_g\big]\phi_{\mu\nu}\big|\les
\sum_{J+K=I,\,|K|<|I|\!\!\!\!\! \!\!\!\!\!\!\!\!\!\!\!\!\!\!\!\!\!\!} \big|\widehat{\mathcal L}_Z^J  H^{\alpha\beta}\pa_\alpha\pa_\beta \widehat{\mathcal L}_Z^K \phi_{\mu\nu}\big|
\les \sum_{\substack{|J|+|K|-1\leq |I|,\\ 1\leq |K|\leq |I|}}\!\!\!\!\! \big(\frac{\big|(\widehat{\mathcal L}_Z^J  H)_{LL}\big|}{1+|q|}+\frac{\big|\widehat{\mathcal L}_Z^J  H\big|}{1+t+r}\big)  \,\big|\pa\widehat{\mathcal L}_Z^K \phi_{\mu\nu}\big|.
 \end{equation*}
 Writing $H=H_0+H_1$ this can be divided up in the commutator with $\Box_0=\Box+H_0^{\alpha\beta}\pa_\alpha\pa_\beta$ and with $\Boxr_g-\Box_0=H_1^{\alpha\beta}\pa_\alpha\pa_\beta$.
 Since $H_0\sim r^{-1}$ we have
 \begin{equation}\label{eq:wavecommutatorH0Lie}
\big|\big[ \Box_0 \widehat{\mathcal L}_Z^I -{\mathcal L}_Z^I\Box_0\big]\phi_{\mu\nu}\big|
\les
\sum_{|J|\leq |I|}\!\!\!\! \big(\frac{\big|(\widehat{\mathcal L}_Z^J  H_0)_{LL}\big|\!}{1+|q|}+\frac{\big|\widehat{\mathcal L}_Z^J  H_0\big|}{1+t+r}\big)  \!\!\!\sum_{|K|\leq |I|}\!\!\!\!\big|\pa\widehat{\mathcal L}_Z^K \phi_{\mu\nu}\big|
\les \frac{M(1\!+\!|q|)^{-1}}{1\!+t+r}\!\sum_{|K|\leq |I|\!\!\!\!\!\!}\!\! \,\big|\pa\widehat{\mathcal L}_Z^K \phi_{\mu\nu}\big|.
 \end{equation}
 Similarly by \eqref{eq:wavecoordinatefunctionLie} and \eqref{eq:hormanderweakdecayh1}
\begin{equation*}
\sum_{|J|\leq |I|/2+1\!\!\!\!\! \!\!\!\!\!\!\!\!}\big(\frac{\big|(\widehat{\mathcal L}_Z^J  H_{\!1})_{LL}\big|\!}{1+|q|}+\frac{\big|\widehat{\mathcal L}_Z^J  H_{\!1}\big|}{1+t+r}\big)  \!\!\!\sum_{|K|\leq |I|}\!\!\!\!\big|\pa\widehat{\mathcal L}_Z^K \phi_{\mu\nu}\big|
\les \frac{\varepsilon(1+q_+)^{-\gamma}}{(1+t+r)^{1+\gamma-2\delta}(1+|q|)^{1-\gamma}}
\!\sum_{|K|\leq |I|\!\!\!\!\!\!}\!\! \,\big|\pa\widehat{\mathcal L}_Z^K \phi_{\mu\nu}\big|,
 \end{equation*}
 we conclude that
 \begin{multline}\label{eq:wavecommutatorH1Lie}
\big|\big[ (\Boxr_g \!-\Box_0)\widehat{\mathcal L}_Z^I -{\mathcal L}_Z^I (\Boxr_g \!-\Box_0) \big]\phi_{\mu\nu}\big|\\
\les
\frac{\varepsilon(1+q_+)^{-\gamma}}{(1\!+\!t\!+\!r)^{1+\gamma-2\delta}(1\!+\!|q|)^{1-\gamma}}
\!\sum_{|K|\leq |I|\!\!\!\!\!\!}\!\! \,\big|\pa\widehat{\mathcal L}_Z^K \phi_{\mu\nu}\big|\!
+\!\sum_{|K|\leq |I|/2\!\!\!\!\!\!\!\!\!\!\!\!\!}\!\!\big|\pa\widehat{\mathcal L}_Z^K \phi_{\mu\nu}\big|
\sum_{|J|\leq |I| \!\!\!\!\!\!} \big(\!\frac{\big|(\widehat{\mathcal L}_Z^J  H_1)_{LL}\big|\!}{1+|q|}+\frac{\big|\widehat{\mathcal L}_Z^J  H_1\big|\!}{1+t}\big) .
\end{multline}

\subsection{The sharp $L^\infty$ decay estimates for higher order low derivatives} \label{subsec:sharpdecayhighorder}
As in section 10 of \cite{LR3} using the methods in section \ref{sec:sharpfirstorder} we can inductively prove sharp
decay estimates also for higher order low derivatives.
As we have already proven the higher order weak decay estimates in
Proposition \ref{prop:weakdecayhormander} and the higher order sharp decay estimates for
components we control with the wave coordinate condition in Proposition
\ref{prop:wavecoorddecayhighlow} it only remains to generalize
Proposition \ref{prop:sharpfirstorderwave} to higher order.
In order to do that we will as before rely on the crucial Lemma \ref{lem:transversalder}
to control transversal derivatives in terms of tangential derivatives, which we control by
Proposition \ref{prop:weakdecayhormander}, and $\Box_0$ close to the light cone $|t-r|<1-c$. It therefore only remains to get control of $\Box_0 \widehat{\mathcal L}_Z^I h^1_{\mu\nu}$
close to the light cone $|t-r|<(1-c)t$, where $0<c<1$.  When $|t-r|<(1-c)t$ we have
by  \eqref{eq:Box0h0} and \eqref{eq:BoxgBox0h0}
\beq
{\mathcal L}_Z^I \Box_g h^1_{\mu\nu}
={\mathcal L}_Z^I F_{\mu\nu}
-{\mathcal L}_Z^I (\Box_g -\Box_0)h^0_{\mu\nu},\qquad |t-r|<(1-c)t,
\eq
where $\widehat{\mathcal L}_Z^I (\Box_g -\Box_0)h^0_{\mu\nu}$ is controlled by \eqref{eq:BoxgBox0h0est} using Proposition \ref{prop:weakdecayhormander}:
\beq\label{eq:BoxgBox0h0est2}
\big|{\mathcal L}_Z^I\, \big(\Boxr_g -\Box_0\big)h^0_{\mu\nu}\big| \les \frac{\varepsilon M}{(1+t+r)^{3-2\delta}(1+q_+)^{\gamma}}
\eq
and by \eqref{eq:inhomogeneousLieestimate}
\begin{multline}\label{eq:inhomogeneousLieestimate2}
\big|{\mathcal L}_Z^I F_{\mu\nu}(H)(\pa h,\pa h)\big|\\
\les \frac{\varepsilon(1\!+\!q_+)^{-1}}{1+t+r}\!\! \!\!\sum_{|J|\leq |I|}\!\! \!\!|\pa\widehat{\mathcal L}_Z^J h|
+\frac{\varepsilon^2(1\!+\!q_+)^{-4\delta}\!\!\!}{(1+t+r)^{3-4\delta}(1+|q|)}
+\frac{\varepsilon^3(1\!+\!q_+)^{-6\delta}\!\!\!}{(1+t+r)^{3-6\delta}(1+|q|)^3}
+\sum_{|K|\leq |I|/2\!\!\!\!\!\!\!\!\!\!\!\!\!\!\!\!\!}|\pa\widehat{\mathcal L}_Z^K h|\sum_{|J|\leq |I|-1\!\!\!\!\!\!\!\!\!\!\!} |\pa\widehat{\mathcal L}_Z^J h|
\end{multline}
It remains to estimate the difference $\Boxr_g-\Box_0$ and the commutators
\begin{equation*}
\Box_0 \widehat{\mathcal L}_Z^I h^1_{\mu\nu}
= {\mathcal L}_Z^I \Box_0 h^1_{\mu\nu}
+[ \Box_0\widehat{\mathcal L}_Z^I-{\mathcal L}_Z^I\Box_0] h^1_{\mu\nu}
={\mathcal L}_Z^I \Boxr_g h^1_{\mu\nu}
-{\mathcal L}_Z^I (\Boxr_g-\Box_0) h^1_{\mu\nu}
+[ \Box_0\widehat{\mathcal L}_Z^I-{\mathcal L}_Z^I\Box_0] h^1_{\mu\nu}
\end{equation*}
By
\eqref{eq:wavecommutatorH0Lie}
 \begin{equation}
\big|\big[ \Box_0 \widehat{\mathcal L}_Z^I- {\mathcal L}_Z^I \Box_0 \big]h^1_{\mu\nu}\big|
\les \frac{\varepsilon(1\!+\!|q|)^{-1}}{1\!+t+r}\!\sum_{|K|\leq |I|\!\!\!\!\!\!}\!\! \,\big|\pa\widehat{\mathcal L}_Z^K h^1_{\mu\nu}\big|.
 \end{equation}
and by
\beq\label{eq:wavecommutatorH1Lie2}
\big|\big[ (\Boxr_g \!-\Box_0)\widehat{\mathcal L}_Z^I -{\mathcal L}_Z^I (\Boxr_g \!-\Box_0) \big]h^1_{\mu\nu}\big|\les
\frac{\varepsilon(1+q_+)^{-\gamma}}{(1\!+\!t\!+\!r)^{1+\gamma-2\delta}(1\!+\!|q|)^{1-\gamma}}
\!\sum_{|K|\leq |I|\!\!\!\!\!\!}\!\! \,\big|\pa\widehat{\mathcal L}_Z^K h^1_{\mu\nu}\big|.
\eq
Since ${\mathcal L}_Z^I (\Boxr_g-\Box_0) h^1_{\mu\nu}
-\big[ (\Boxr_g \!-\Box_0)\widehat{\mathcal L}_Z^I-{\mathcal L}_Z^I (\Boxr_g \!-\Box_0) \big]h^1_{\mu\nu}
= (\Boxr_g-\Box_0) \widehat{\mathcal L}_Z^Ih^1_{\mu\nu}$, which can be stimated in the same way we obtain
\beq\label{eq:wavecommutatorH1Lie3}
\big|{\mathcal L}_Z^I( \Boxr_g \!-\Box_0)h^1_{\mu\nu}\big|\les
\frac{\varepsilon(1+q_+)^{-\gamma}}{(1\!+\!t\!+\!r)^{1+\gamma-2\delta}(1\!+\!|q|)^{1-\gamma}}
\!\sum_{|K|\leq |I|+1\!\!\!\!\!\!}\!\! \,\big|\pa\widehat{\mathcal L}_Z^K h^1_{\mu\nu}\big|\les
\frac{\varepsilon^2(1+q_+)^{-2\gamma}}{(1\!+\!t\!+\!r)^{2+\gamma-4\delta}(1\!+\!|q|)^{2-\gamma}}.
\eq
Summing up
\begin{multline}
\big| \Box_0 \widehat{\mathcal L}_Z^I h^1_{\mu\nu}\big|
\les \frac{\varepsilon(1\!+\!q_+)^{-1}}{1+t+r}\!\! \!\!\sum_{|J|\leq |I|}\!\! \!\!|\pa\widehat{\mathcal L}_Z^J h|
+\sum_{|K|\leq |I|/2\!\!\!\!\!\!\!\!\!\!\!\!\!\!\!\!\!}|\pa\widehat{\mathcal L}_Z^K h|\sum_{|J|\leq |I|-1\!\!\!\!\!\!\!\!\!\!\!} |\pa\widehat{\mathcal L}_Z^J h| \\
+\frac{\varepsilon^2(1\!+\!q_+)^{-4\delta}\!\!\!}{(1+t+r)^{3-4\delta}(1+|q|)}
+\frac{\varepsilon^2(1+q_+)^{-2\gamma}}{(1\!+\!t\!+\!r)^{2+\gamma-4\delta}(1\!+\!|q|)^{2-\gamma}}
\end{multline}

From Lemma \ref{lem:transversalder} and the estimate \eqref{eq:hormanderweakdecayh1} we get
\begin{lemma}\label{lem:transversaldersimplifiedLie} Let $D_t=\{(t,x);\, \big|t-|x|\big|\leq c_0 t \}$, for some constant
$0<c_0<1$, and $\overline{w}(q)=(1+q_+)^{1+\gamma^\prime}$, where $-1\leq \gamma^\prime<\gamma-2\delta$. Then
\begin{equation}\label{eq:transversaldersimplifiedLie}
(1+t+|x|) \,|\pa \widehat{\mathcal L}_Z^I
h^1(t,x)\, \overline{w}(q)| \les\varepsilon + \int_0^t (1+\tau)\|
\,\Box_0 \widehat{\mathcal L}_Z^I h^1(\tau,\cdot)\, \overline{w}\|_{L^\infty(D_\tau)} d\tau.
\end{equation}
\end{lemma}

\begin{proposition} \label{prop:hogherordersharpdecay} If the weak energy bounds and initial bounds hold then we have for any $0\leq \gamma^\prime<\gamma-2\delta$, and $|I|=k\leq N-5$, then there are constants $c_k$ such that
\beq\label{eq:sharpHhigherlowder}
\big| \pa \widehat{\mathcal L}_Z^I h^1\big|\leq c_k \varepsilon(1+t)^{c_k\varepsilon}  (1+t+r)^{-1}(1+q_+)^{-1-\gamma^\prime}.
\eq
The same estimates hold for $h$ in place of $h^1$ if $\gamma^\prime=0$.
\end{proposition}
 \begin{proof} Let
 $
 N_k(t)=(1+t) \sum_{|I|\leq k}\| \widehat{\mathcal L}_Z^I h^1(t,\cdot)\overline{w}\|_{L^\infty(D_t)}
 $. We will prove \eqref{eq:sharpHhigherlowder} by induction noting that its true for $k=0$
 by \eqref{eq:sharpHanyder}. Then by Lemma \ref{lem:transversaldersimplifiedLie} we have for
 $k\geq 1$
 \beq
 N_k(t)\les \varepsilon +\int_0^t \frac{\varepsilon}{1+\tau} N_k(\tau) d\tau
 +\int_0^t \frac{1}{1+\tau} N_{k-1}(\tau)^2 \, d\tau,
 \eq
 where the bounds \eqref{eq:weakdecayT} have been used, and by the induction hypothesis $N_k(\tau)^2 \les \varepsilon^2 c_{k-1}^2 , (1+\tau)^{2c_{k-1}\varepsilon}$ so for some $c_k\geq 4c_{k-1}$
 \beq
 N_k(t)\leq  c_k\int_0^t \frac{\varepsilon}{1+\tau} N_k(\tau) d\tau
 +c_k\, \varepsilon (1+t)^{2c_{k-1}\varepsilon} .
 \eq
 Using Gr\"onwall's lemma with $G$ denoting the integral we get
 $G'(t)\leq \varepsilon (1+t)^{-1}c_k\big( G(t)+\varepsilon(1+t)^{2c_{k-1}\varepsilon}\big) $ and multiplying with the integrating factor we get
 $\big( G(t) (1+t)^{-c_k\varepsilon}\big)^\prime
 \leq c_k\varepsilon^2 (1+t)^{2 c_{k-1}\varepsilon-c_k\varepsilon -1} $. Assuming that
 $c_k \geq 4c_{k-1}$ We get $G(t)\leq c_k \varepsilon(1+t)^{c_k\varepsilon}$
 and hence $N_k(t)\leq c_k^\prime \varepsilon(1+t)^{c_k\varepsilon}$.
\end{proof}

\subsection{The energy estimate}
\subsubsection{The basic energy estimate for the wave equation}
In \cite{LR3} the following energy estimate was proven:
\begin{lemma}
\label{lem:Decayenergy}
Let $\phi$ be a solution of the wave equation $\Boxr_g \phi=F$,
with the metric $g$ such that for $H^{\alpha\beta}=g^{\alpha\beta}-m^{\alpha\beta}$;
\begin{align}
&(1+|q|)^{-1} |H|_{LL} +|\pa H|_{LL}+|\overline{\pa} H|\leq
C\varepsilon' (1+t)^{-1},\nn\\
& (1+|q|)^{-1}\,|H|+ |\pa H|\leq C\varepsilon' (1+t)^{-\frac 12}
(1+|q|)^{-\frac 12}(1+q_-)^{-\mu}.\label{eq:metricdecay}
\end{align}
for some $\mu>0$. Set
$$
w=\begin{cases} (1+|r-t|)^{1+2\gamma},\,\,\, &r>t\\
1+(1+|r-t|)^{-2\mu},\quad &r\leq t\end{cases}\qquad\text{and}\qquad
w^\prime=\begin{cases} (1+2\gamma)(1+|r-t|)^{2\gamma},\,\,\, &r>t\\
2\mu (1+|r-t|)^{-1-2\mu},\quad &r\leq t\end{cases}
$$
Then for any $0<\ga\le 1$, and $0<\varepsilon'\leq \gamma/C_1$, we have
\begin{multline}\label{eq:firstenergy}
 \int_{\Si_{t}}\!\! |\pa\phi|^{2}\,w dx + \!\int_{0}^{t}\!\! \int_{\Si_{\tau}}
\!\!|\pab\phi|^{2}\,w^{\,\prime} dx d\tau \leq
 8\!\int_{\Si_{0}}\!\! |\pa \phi|^{2}\,w dx+
\!\int_0^t\!\!\frac{C\varepsilon}{1+\tau}\int_{\Si_{\tau}} \!\!  |\pa\phi|^{2}
 \, w \,dx \, d\tau\\
 +16 \int_0^t\Big(\int_{\Si_{\tau}}\!\!\! |F|^2 wdx\Big)^{\!1/2}\Big(\int_{\Si_{\tau}} \!\!\! |\pa \phi|^2 w dx\Big)^{\!1/2}\! d\tau
\end{multline}
\end{lemma}

\subsubsection{The lowest order energy estimate for Einstein's equations}
Let
\beq
E_k(t)=\sum_{|I|\leq k}\int_{\Sigma_t}|\pa Z^I h^1|^2 w dx\qquad\text{and}\qquad
S_k(t)=\sum_{|I|\leq k}\int_0^t\int_{\Sigma_t}|\pa Z^I h^1|^2 w^\prime dx d\tau,
\eq
 By Lemma \ref{lem:Decayenergy} we have
\beq
E_0(t)+S_0(t)\leq 8 E_0(0)+\int_0^t \frac{\varepsilon}{1+\tau}E_0(\tau)+\| F(\tau,\cdot) w^{1/2}\|_{L^2}E_0(\tau)^{1/2} \,d\tau,
\eq
where $F=\big|F_{\mu\nu}(h)(\pa h,\pa h)+\widehat{T}_{\mu\nu}-\Boxr_g h^0_{\mu\nu}\big|$, where with $\overline{h}=h_{TS}$, $T,S\in\mathcal{S}$, we have
\beq
|F_{\mu\nu}|\les |{\slash \!\!\!\!}P_{\mu\nu} (\pa h,\pa h)|+|\overline{\pa} h|\, |\pa h| +|h|\, |\pa h|^2\les \big(|\pa \overline{h}|+|\overline{\pa} h|+|h|\, |\pa h|\big)|\pa h|
\les \frac{\varepsilon |\pa h|}{(1+t+r)(1+q_+)}.
\eq
Writing $h=h^0+h^1$ we see that its enough to estimate
\beq
F^j=\frac{\varepsilon |\pa h^j|}{(1+t+r)(1+q_+)}
\eq
for $j=0,1$. We have
\beq
\|F^{1}(t,\cdot) w^{1/2}\|_{L^2}\les \varepsilon(1+t)^{-1}\|\pa h^1(t,\cdot) w^{1/2}\|_{L^2}=\varepsilon(1+t)^{-1}E_0(t)^{1/2},
 \eq
 and $F^0\les M\varepsilon(1+t+r)^{-3}(1+q_+)^{-1}$ so
 \beq
 \|F^0(t,\cdot)w^{1/2}\|_{L^2}
 \les \varepsilon M\Big(\int \frac{(1+q_+)^{2\gamma-1}}{(1+t+r)^6}\, r^2 dr\Big)^{1/2}
 \les \frac{M\varepsilon }{(1+t)^{2-\gamma}}.
 \eq
  As far as the energy estimate one could have picked $h^0_{\mu\nu}$ to satisfy $\Boxr_g h^0_{\mu\nu}=0$
one we wouldn't have to do anything further. However, since we didn't do this we will
estimate using \eqref{eq:BoxgBox0h0} and \eqref{eq:Box0h0}
\beq
\big| \Boxr_g h^0_{\mu\nu}\big|\leq \big|[\Boxr_g-\Box_0]h^0_{\mu\nu}\big|
+\big|[\Box_0 -\Box]h^0_{\mu\nu}|+\big|\Box h^0_{\mu\nu}|\leq \frac{C_0 |H_1| M}{(1+t+r)^3} +\frac{M^2\chi^\prime\big(\tfrac{r}{t+1}\big)}{(1+t+r)^4}
+\frac{M\chi^\prime\big(\tfrac{r}{t+1}\big)}{(1+t+r)^3}
\eq
and hence using Hardy's inequality
\beq
\| \Boxr_g h^0_{\mu\nu}(t,\cdot) w^{1/2}\|_{L^2}\leq CM (1+t)^{-2} \| \pa H_1(t,\cdot) w^{1/2}\|_{L^2}+ M_0 M(1+t)^{-3/2},
\eq
where $M_0$ is a universal constant.

Hence
\begin{multline}
E_0(t)\leq 8 E_0(0) +C'\varepsilon\int_0^t \frac{E_0(\tau)}{1+\tau}d\tau\\
+C'\varepsilon\int_0^t\!\!\frac{ME_0(\tau)^{1/2}}{(1+t)^{2-\gamma}}\, d\tau+ 16M_0\int_0^t\frac{ME_0(\tau)^{1/2}}{(1+\tau)^{3/2}} d\tau
+16\int_0^t \Vert \widehat{T}  (\tau,\cdot) \Vert_{L^{2}}E_0(\tau)^{1/2}\, d\tau,
\end{multline}
for some univeral constant $M_0$.

\subsection{Higher order $L^2$ Energy estimates}

For this section we have to make the following smallness assumption on $\varepsilon$:
\beq
c_{k'}\varepsilon\leq \delta
\eq
where $c_{k'}$ are the constants in Proposition \ref{prop:hogherordersharpdecay}.

\subsubsection{$L^2$ estimate of the inhomogeneous term}
It follows from \eqref{eq:inhomogeneousLieestimate} that with $k=|I|$ and $k^\prime=[k/2]+1$
we have
\begin{equation}
\big|{\mathcal L}_Z^I F_{\mu\nu}(h)(\pa h,\pa h)\big|
\les F_1^{k0}+F_1^{k1}+F_2^{k0}+F_2^{k1}+F_3^{k0}+F_3^{k1}+F_4^{k0}+F_4^{k1},
\end{equation}
where
\beq
F_1^{kj}=\frac{\varepsilon\!\!\!\!}{(1+t+r)(1+q_+)}\! \sum_{|J|\leq k}\!\! |\pa\widehat{\mathcal L}_Z^J h^j|,\qquad
F_2^{kj}=\frac{c_{k^\prime} \varepsilon(1+t)^{c_{k^\prime}\varepsilon} }{ (1+t+r)(1+q_+)}\sum_{|J|\leq k-1}\!\! \!\!|\pa\widehat{\mathcal L}_Z^J h^j|
\eq
\beq
F_3^{kj} =\frac{\varepsilon^2(1+t)^{4\delta}(1+q_+)^{-4\delta}\!\!\!}{(1+t+r)^2(1+|q|)^2} \sum_{|J|\leq |I|}\!\! |\widehat{\mathcal L}_Z^J h^j|,\qquad
F_4^{kj}=\frac{\varepsilon(1+t)^{2\delta}(1+q_+)^{-2\delta}\!\!\!}{(1+t+r)(1+|q|)}\sum_{|J|\leq |I|}\!\! |\overline{\pa} \widehat{\mathcal L}_Z^J h^j|\,
\eq
For $i=1,2$ we have
\beq
\Big(\int |F_1^{k1}|^2 w\, dx\Big)^{1/2} \les
\frac{\varepsilon}{1+t} E_{k}(t)^{1/2} ,\qquad
\Big(\int |F_2^{k1}|^2\,  w\, dx\Big)^{1/2} \leq
\frac{c_{k^\prime}\varepsilon(1+t)^{c_{k^\prime} \varepsilon} }{1+t}
E_{k-1}(t)^{1/2} .
\eq
We have
\beq
|F_1^{k0}|\les \frac{\varepsilon M\!\!\!\!}{(1+t+r)^{3}(1+q_+)},
\qquad
|F_2^{k0}|\les \frac{\varepsilon M(1+t)^{c_{k^\prime}\varepsilon}}{(1+t+r)^{3}(1+q_+)},
\eq
and hence
\beq
\Big(\int |F_1^{k0}|^2 w\, dx\Big)^{1/2} \les
\frac{\varepsilon M}{(1+t)^{2-\gamma}}  ,\qquad
\Big(\int |F_2^{k0}|^2\,  w\, dx\Big)^{1/2} \leq
\frac{c_{k^\prime}\varepsilon M }{(1+t)^{2-\gamma-c_{k'}\varepsilon}} .
\eq
For $i=3$ we will use Hardy's inequality Lemma \ref{lemma:hardy}
but first we divide it up into two terms for $j=0,1$:
\beq
|F_3^{k0}|\les \frac{\varepsilon^2  M (1+t)^{4\delta}(1+q_+)^{-4\delta}\!\!\!}{(1+t+r)^3(1+|q|)^2},
\eq
and hence
\beq
\Big(\int |F_3^{k0}|^2\, w\, dx \Big)^{1/2}
\les \frac{\varepsilon^2 M}{(1+t)^{2-4\delta} }.
\eq
By Hardy's inequality
\beq
\Big(\int |F_3^{k1}|^2\,  w\, dx \Big)^{1/2}
\les\sum_{|J|\leq |I|}\!\! \frac{\varepsilon}{(1+t)^{2-4\delta} }\Big( \int  \frac{|\widehat{\mathcal L}_Z^J h^1|^2}{(1+|q|)^2} wdx \Big)^{1/2}\les \frac{\varepsilon}{(1+t)^{2-4\delta} }
E_k(t)^{1/2}.
\eq
Moreover
\beq
|F_4^{k0}|\les \frac{\varepsilon M (1+t)^{2\delta}(1+q_+)^{-2\delta}\!\!\!}{(1+t+r)^3(1+|q|)},
\eq
and hence
\beq
\Big(\int |F_4^{k0}|^2\, w\, dx \Big)^{1/2}
\les \frac{\varepsilon M}{(1+t)^{3-2\gamma} }.
\eq
The last term $F_4^{k1}$ will be estimated differently in terms of the space term integral.
We have
\begin{equation*}
\int |F_4^{k1}|^2 w\, dx
\les \int
\frac{\varepsilon^2(1+t)^{4\delta}(1+q_+)^{-4\delta}\!\!\!}{(1+t+r)^2(1+|q|)^2}\sum_{|J|\leq |I|}\!\! |\overline{\pa} \widehat{\mathcal L}_Z^J h^1|^2\, w dx
\les
\frac{\varepsilon^2}{(1+t)^{2-4\delta}}
\int
\sum_{|J|\leq |I|}\!\! |\overline{\pa} \widehat{\mathcal L}_Z^J h^1|^2\, w^\prime  dx
\end{equation*}
It follows that
\begin{equation*}
\int \Big(\int |F_4^{k1}|^2  w\, dx\Big)^{1/2} E_k(\tau)^{1/2} \,d\tau
 \leq C
    \big(\varepsilon S_k(t)\big)^{1/2}\Big(\int_0^t \frac{\varepsilon E_k(\tau)\,d\tau}{(1+\tau)^{2-4\delta}} \Big)^{1/2}
    \leq
    \varepsilon S_k(t)
    +C^2\int_0^t \frac{\varepsilon E_k(\tau)\,d\tau}{(1+\tau)^{2-4\delta}} .
\end{equation*}
Summing up using that $\delta\leq 1/4$ we have
\begin{equation}\label{eq:inhomogeneousenergy}
\int_0^t \!\!\int \!\!\big|{\mathcal L}_Z^I F_{\mu\nu}\big|
\,\big|\pa{\mathcal L}_Z^I h^1\!\big| w\, dx d\tau
\leq C^{\prime}\!\!\int_0^t\!\!\!\Big( \frac{\varepsilon E_{k}(\tau)^{1/2}\!\!\!\!\!}{1+\tau} +
\frac{c_{k^\prime}\varepsilon (1\!+\!\tau)^{ c_{k^\prime} \varepsilon}\!\!\!\!\!\!}{1+\tau}
E_{k-1}(\tau)^{1/2}+\frac{\varepsilon M}{(1+t)^{2-\gamma-c_{k'}\varepsilon}}\Big)E_k(\tau)^{1/2} d\tau+\varepsilon S_k(t)
\end{equation}

\subsubsection{Equivalence of norms} The inhomogeneous terms contain factors
of $\pa \widehat{\mathcal L}_Z^I h$  which we estimate by writing $h=h_0+h^1$, and estimate the $L^2$ norm factors with $h^1$ in terms of the energy of $h^1$ whereas the $L^2$ norms of $h^0$ can be estimated directly. The commutator terms will in addition contain factors of
$\widehat{\mathcal L}_Z^I H$, where we can also write $H=H_0+H_1$ and estimate the factors with $H_0$ directly since it is explicit, and for the factors with $\widehat{\mathcal L}_Z^I H_1$ we first use Hardy's inequality to estimate them in terms of $\pa \widehat{\mathcal L}_Z^I H_1$.
However $H_1$  is only approximately equal to $-h^1$. We have that $H=-h+K(h)$, where $K(h)=O(h^2)$.
and hence $H_1=-h^1+K(h)-h^0-H_0$. Differentiating we see that to conclude that the norms or $H_1$ are approximately bounded by those of $h^1$ we have to estimate factors of the form
 $\widehat{\mathcal L}_Z^J h\,\, \pa \widehat{\mathcal L}_Z^K h$, with $|J|+|K|\leq |I|$,
 in $L^2$ with respect to the measure $w$. Again this can be estimated by writing $h=h^0+h^1$ and estimate the factors with $h^1$ in terms of the energy (after possibly using Hardy's inequality) and estimate the explicit factors with $h^0$ directly. The conclusion of this
 process is that
 \beq
 \| \pa \widehat{\mathcal L}_Z^I K(h)(t,\cdot)\,\, w^{1/2}\|_{L^2}\les \varepsilon \big(M+E_k(t)^{1/2}\big),\quad \text{if}\quad |I|\leq k.
 \eq
 Hence since $h^0=-H_0$ it follows that
 \beq
{\sum}_{|I|\leq k}  \| \pa \widehat{\mathcal L}_Z^I H_1(t,\cdot)\,\, w^{1/2}\|_{L^2}\les
\varepsilon M +E_k(t)^{1/2}
 \eq
 and similarly for the space time integrals of tangential components.

 \subsubsection{$L^2$ estimate of the wave operator applied to $h^0$}
 By \eqref{eq:BoxgBox0h0est} using Hardy's inequality we have
\beq
\|\big({\mathcal L}_Z^I\, \big(\Boxr_g -\Box_0\big)h^0\big)(t,\cdot) w^{1/2} \|_{L^2} \les \frac{M}{(1+t)^{2}} E_k(t)^{1/2}
\eq
 and by \eqref{eq:Box0h0est} we have
 \beq
 \sum_{|I|\leq k}\| \big({\mathcal L}_Z^I\, \Box_0 h^0\big)(t,\cdot) w^{1/2}\|_{L^2}^2
 \les M_k^2 \frac{M^2}{(1+t)^{3}},
 \eq
 for some universal constant $M_k$.

\subsubsection{$L^2$ estimates of the wave commutator }
It remains to estimate the commutator which by
\eqref{eq:wavecommutatorH0Lie} and \eqref{eq:wavecommutatorH1Lie} is bounded by
\begin{equation*}
\big|\big[ \Boxr_g \widehat{\mathcal L}_Z^I -{\mathcal L}_Z^I  \Boxr_g \big]h^1_{\mu\nu}\big|
\leq \big|\big[ \Box_0 \widehat{\mathcal L}_Z^I -{\mathcal L}_Z^I\Box_0\big]h^1_{\mu\nu}\big|
+\big|\big[ (\Boxr_g -\Box_0)\widehat{\mathcal L}_Z^I- {\mathcal L}_Z^I(\Boxr_g -\Box_0)\big]h^1_{\mu\nu}\big|
\les F_5^k + F_6^k+F_7^k
\end{equation*}
where
 \beq
 F_5^k=\frac{\varepsilon}{1\!+t}\!\sum_{|K|\leq |I|\!\!\!\!\!\!}\!\! \,\big|\pa\widehat{\mathcal L}_Z^K h^1_{\mu\nu}\big|\!
 \eq
 and
 \beq
F_6^{k}=\frac{\varepsilon(1+q_+)^{-\gamma}}{(1+t+r)^{2-2\delta}} \!\!
\sum_{|J|\leq |I|}\!\!\!\frac{\big|\widehat{\mathcal L}_Z^J  H_1\big|}{1+|q|},\qquad
F_7^{k}=\frac{\varepsilon(1+t)^{2\delta}(1+q_+)^{-\gamma}}{(1+t+r)(1+|q|)} \!\!
\sum_{|J|\leq |I|}\!\!\!\frac{\big|(\widehat{\mathcal L}_Z^J  H_1)_{LL}\big|}{1+|q|}\, \chi(\tfrac{r}{t+1}),
\eq
since
\begin{equation*}
\sum_{|J|\leq |I|,\, |K|\leq |I|/2+1\!\!\!\!\! \!\!\!\!\! } \big(\frac{\big|(\widehat{\mathcal L}_Z^J  H_1)_{LL}\big|\!\!}{1+|q|}+\frac{\big|\widehat{\mathcal L}_Z^J  H_1\big|\!}{1+t}\big)  \,\big|\pa\widehat{\mathcal L}_Z^K h^1_{\mu\nu}\big|
\les \frac{\varepsilon(1+t)^{2\delta}(1+q_+)^{-\gamma}\!\!\!\!}{(1+t+r)(1+|q|)}
\sum_{|J|\leq |I|}\!\!\!\big(\frac{\big|(\widehat{\mathcal L}_Z^J  H_1)_{LL}\big|}{1+|q|}+\frac{\big|\widehat{\mathcal L}_Z^J  H_1\big|}{1+t+r}\big).
 \end{equation*}
  We have
 \beq
\Big(\int |F_5^{k}|^2\,  w\, dx \Big)^{1/2}\les \frac{\varepsilon}{1+t }
E_k(\tau)^{1/2}.
\eq
By Hardy's inequality
\beq
\Big(\int |F_6^{k}|^2\,  w\, dx \Big)^{1/2}
\les\sum_{|J|\leq |I|}\!\! \frac{\varepsilon}{(1+t)^{2-2\delta} }\Big( \int  \frac{|\widehat{\mathcal L}_Z^J h^1|^2}{(1+|q|)^2} wdx \Big)^{1/2}\les \frac{\varepsilon}{(1+t)^{2-2\delta} }
E_k(\tau)^{1/2}.
\eq
Dealing with the last term $F_7^{k}$ requires a slight generalization of Hardy's inequality
Corollary 13.3 in \cite{LR3}:
\begin{corollary} \label{cor:Poinc}
Let $\gamma>0$ and $\mu>0$.
Then for any $-1\leq a\leq 1$ and any $\phi\in C^\infty_0(\R^3)$
if in addition  $a<2\min{(\gamma, \mu)}$, we have
\begin{equation}\label{eq:poincaretwo}
\int \frac{|\phi|^2}{(1+|q|)^2}\,\frac{(1+|q|)^{-a} \,
}{(1+t+|q|)^{1-a}}\, \frac {w\, dx}{(1+q_-)^{2\mu}}\les \int {|\pa
\phi|^2}\min\, (w', \frac{w}{(1+t+|q|)^{1-a}}) \, dx.
\end{equation}
\end{corollary}

The last term $F_7^{k}$ will be estimated differently in terms of the space time integral.
By Hardy's inequality and \eqref{eq:higherwavecoordinateLiederivative}
\begin{multline}
 \int |F_7^{k}|^2  w\, dx
\les \frac{\varepsilon^2}{1+\tau} \sum_{|J|\leq |I|}
\int
\frac{(1+q_+)^{-2\gamma}\!\!\!}{(1+\tau+r)^{1-4\delta}(1+|q|)^2}\frac{| (\widehat{\mathcal L}_Z^J H_1)_{LL}|^2}{(1+|q|)^2}\, w \,dx
\\
\les\frac{\varepsilon^2}{1+\tau} \sum_{|J|\leq |I|}
\int |\pa(\widehat{\mathcal L}_Z^J H_1)_{LL}|^2\, \min{\Big( w^\prime,\frac{w}{(1+\tau+|q|)^{1-4\delta}}\Big)}  dx \\
\les \!\frac{\varepsilon^2}{1\!+\!\tau}\Big( \sum_{|J|\leq |I|}\int
 |\overline{\pa} \widehat{\mathcal L}_Z^J H_1|^2w^\prime dx
+\!\!\!\!\!\!\! \!\!\sum_{|J|+|K|\leq |I|} \int\!\big|  \widehat{\mathcal L}_Z^{J} {h}\big|^2
    \big|\pa  \widehat{\mathcal L}_Z^{K} h\big|^2 \, \frac{w\, dx }{(1\!+\tau\!+\!r)^{1-4\delta}}
    +\int_{|x|\leq 3\tau\!/4} \!\!\!\!\!\!
    \frac{M^2 w^\prime dx \, \!\!\!}{\,\,\,\,\,(1\!+\!\tau)^4}\,\Big)
    \\
  \les
  \frac{\varepsilon^2}{1\!+\!\tau}\!\!\sum_{|J|\leq |I|}\int\!
 |\overline{\pa} \widehat{\mathcal L}_Z^J H_1|^2w^\prime\, dx
 + \frac{\varepsilon^4}{1\!+\!\tau}\!\!\sum_{|J|\leq |I|} \int\!\!
    \frac{(1+q_+)^{-4\delta}}{(1\!+\!\tau\!+\!r)^{3-8\delta} }
    \Big(\big|\pa  \widehat{\mathcal L}_Z^{J} h\big|^2\!
    +\frac{| \widehat{\mathcal L}_Z^{J} h|^2}{(1+|q|)^2}\Big) w\, dx
    +\frac{\varepsilon^2 M^2}{(1\!+\!\tau)^{3+2\mu}}.
\end{multline}
Here we again write $h=h^0+h^1$. We have
\beq
\sum_{|J|\leq |I|} \int\!\!
    \frac{(1+q_+)^{-4\delta}}{(1\!+\!\tau\!+\!r)^{3-8\delta} }
    \Big(\big|\pa  \widehat{\mathcal L}_Z^{J} h^0\big|^2\!
    +\frac{| \widehat{\mathcal L}_Z^{J} h^0|^2}{(1+|q|)^2}\Big) w\, dx
    \les
    \frac{M^2}{(1+\tau)^{3-2\gamma-4\delta} }
\eq
and by Hardy's inequality again
\beq
\sum_{|J|\leq |I|} \int\!\!
    \frac{(1+q_+)^{-4\delta}}{(1\!+\!\tau)^{3-8\delta} }
    \Big(\big|\pa  \widehat{\mathcal L}_Z^{J} h^1\big|^2\!
    +\frac{| \widehat{\mathcal L}_Z^{J} h^1|^2}{(1+|q|)^2}\Big) w\, dx
    \les
  \frac{1}{(1\!+\!\tau)^{3-8\delta} }   \sum_{|J|\leq |I|} \int\!\!
  \big|\pa  \widehat{\mathcal L}_Z^{J} h^1\big|^2\!
    w\, dx.
\eq
Hence
\begin{equation*}
 \int |F_7^{k}|^2  w\, dx
  \les
  \frac{\varepsilon^2}{1\!+\!\tau}\!\!\sum_{|J|\leq |I|}\int\!
 |\overline{\pa} \widehat{\mathcal L}_Z^J h^1|^2w^\prime\, dx
 + \frac{\varepsilon^2}{(1\!+\!\tau)^3 } \Big(  \sum_{|J|\leq |I|}\varepsilon^2 \int\!\!
  \big|\pa  \widehat{\mathcal L}_Z^{J} h^1\big|^2\!
    w\, dx+M^2\Big)
    +\frac{\varepsilon^4 M^2}{(1\!+\!\tau)^{4-2\gamma-4\delta}}.
\end{equation*}
It follows that
\begin{multline*}
\int \Big(\int |F_7^{k}|^2  w\, dx\Big)^{1/2} E_k(\tau)^{1/2} \,d\tau  \\
 \leq
    C\big(\varepsilon S_k(t)\big)^{1/2}\Big(\int_0^t \frac{\varepsilon}{1+\tau} E_k(\tau)\,d\tau\Big)^{1/2}
    +C\int_0^t \frac{\varepsilon\big(\varepsilon E_k(\tau)+M E_k(\tau)^{1/2})}{(1+\tau)^{3/2}} +\frac{\varepsilon^2 M E_k(\tau)^{1/2}}{(1\!+\!\tau)^{2-\gamma-2\delta}}\, d\tau\\
    \leq
    \varepsilon S_k(t)+\int_0^t \frac{C^\prime\varepsilon}{1+\tau} E_k(\tau)\,d\tau
    +\int_0^t \Big(\frac{C^\prime \varepsilon  }{(1+\tau)^{3/2}} +\frac{C^\prime\varepsilon^2 }{(1\!+\!\tau)^{2-\gamma-2\delta}}\Big)M E_k(\tau)^{1/2}\, d\tau.
\end{multline*}
Summing up we have
\begin{multline} \label{eq:wavecommutatorenergy}
\int_0^t \!\!\int \!\!\big| \big[ \Boxr_g \widehat{\mathcal L}_Z^I- {\mathcal L}_Z^I\Boxr_g\big]h^1_{\mu\nu}\big|
\,\big|\pa{\mathcal L}_Z^I h^1\big| w\, dx d\tau\\
 \leq
    \varepsilon S_k(t)+\int_0^t \frac{C^\prime\varepsilon}{1+\tau} E_k(\tau)\,d\tau
    +\int_0^t \Big(\frac{C^\prime \varepsilon  }{(1+\tau)^{3/2}} +\frac{C^\prime\varepsilon^2 }{(1\!+\!\tau)^{2-\gamma-2\delta}}\Big)M E_k(\tau)^{1/2}\, d\tau.
\end{multline}

\subsubsection{Higher order Energy $L^2$ estimates} Here we give the proof of Theorem \ref{thm:Einsteinenergyestimates} using the decay estimates proven in the previous section. We will argue by induction so we assume the energy estimate is true for $k-1$ and we will prove it for $k$.
Using the energy inequality Lemma \ref{lem:Decayenergy} we get from adding
up the energy contributions from the inhomogeneous term \eqref{eq:inhomogeneousenergy}, the commutator with the wave equation
\eqref{eq:wavecommutatorenergy} and the Vlasov matter
\begin{multline*}
E_k(t)+S_k(t)\leq 8 E_k(0)+32 \varepsilon S_k(t)
+16 M_k\!\!\int_0^t \frac{M E_k(\tau)^{1/2}}{(1+\tau)^{3/2}}\, d\tau
+C^{\prime\prime}\varepsilon \!\int_0^t\!\!\!
\Big( \frac{ E_{k}(\tau)}{1+\tau}+
\frac{c_{k^\prime}(1\!+\!\tau)^{ c_{k^\prime} \varepsilon}\!\!\!}{1+\tau}
E_{k-1}(\tau)\Big) d\tau\\
+C^{\prime\prime}\int_0^t \Big(\frac{\varepsilon^2}{(1+\tau)^{2-\gamma-\max\{2\delta,c_{k'}\varepsilon\}}}+\frac{\varepsilon}{(1+\tau)^{3/2}}\Big)M E_k(\tau)^{1/2}\, d\tau
+16\!\int_0^t\!\!\Big(\sum_{|I|\leq k\!\!\!\!}\Vert Z^I  T  (\tau,\cdot) \Vert_{L^{2}}^2\!\Big)^{\!\!1/2}\!E_k(\tau)^{1/2} \, d\tau
\end{multline*}
for some universal constant $M_k$.
We now choose $\varepsilon$ so small that $32\varepsilon\leq 1$ so that $S_k(t)$ in the right can be absorbed in
$S_k(t)$ on the left, so that $c_{k'}\varepsilon\leq 2\delta$ and so that
$C^{\prime\prime} \varepsilon \leq M_k$ we obtain
\begin{multline*}
E_k(t)\leq 8 E_k(0)
+32 M_k\!\!\int_0^t \frac{M E_k(\tau)^{1/2}}{(1+\tau)^{3/2}}\, d\tau
+C^{\prime\prime}\varepsilon \!\int_0^t\!\!\!
\Big( \frac{ E_{k}(\tau)}{1+\tau}+
\frac{c_{k^\prime}(1\!+\!\tau)^{ c_{k^\prime} \varepsilon}\!\!\!}{1+\tau}
E_{k-1}(\tau)\Big) d\tau\\
+C^{\prime\prime}\int_0^t \frac{\varepsilon^2 M E_k(\tau)^{1/2}}{(1+\tau)^{2-\gamma-2\delta}} \, d\tau
+16\!\int_0^t\!\!\Big(\sum_{|I|\leq k\!\!\!\!}\Vert Z^I  T  (\tau,\cdot) \Vert_{L^{2}}^2\!\Big)^{\!\!1/2}\!E_k(\tau)^{1/2} \, d\tau.
\end{multline*}

\section{The continuity argument and the proof of Theorem \ref{thm:main2}}
\label{section:cty}

The proof of Theorem \ref{thm:main2} is a direct consequence of Theorem \ref{thm:mainL2}, Proposition \ref{prop:suppf}, Proposition \ref{thm:Einsteinpointwiseestimates}, Proposition \ref{prop:higherordersharpdecay} and Theorem \ref{thm:Einsteinenergyestimates} as follows.

Let $T_*$ be the supremum of all times $T_1$ such that a solution of the reduced Einstein--Vlasov system \eqref{eq:Tmunu}, \eqref{eq:Vlasov}, \eqref{eq:RE1} attaining the given data exists for all $t\in [0,T_1]$ and satisfies
\begin{equation} \label{eq:bootstrap}
	E_N(t)^{\frac{1}{2}} \leq C_N \varepsilon (1+t)^{\delta},
	\qquad
	\sum_{\vert I \vert \leq N-1} \Vert Z^I \widehat{T}(t,\cdot) \Vert_{L^1}
	\leq
	C_N \varepsilon,
\end{equation}
for all $t\in [0,T_1]$, where $\delta>0$ is such that $\delta < \gamma < 1-8\delta$ and $C_N$ is a fixed large constant, to be determined, depending only on $N$, $\delta$ and $\supp(f_0)$.  Recall that $\widehat{T}^{\mu \nu} = T^{\mu \nu} - \frac{1}{2} \tr T g^{\mu \nu}$.  Clearly the set of such $T_1$ is non empty by the local existence theorem of Choquet-Bruhat \cite{ChBr71} (see also the textbook of Ringstr\"{o}m \cite{Ri}) and so $T_* > 0$.  Suppose $T_*<\infty$.

By Proposition \ref{thm:Einsteinpointwiseestimates} the pointwise bounds
\begin{equation} \label{eq:hpointwise}
	| h(t,x)|\leq \frac{C_N' \varepsilon(1+t)^{2\delta}}{(1+t+r)(1+q_+)^{2\delta}},
	\qquad
	\sum_{\vert I \vert + \vert J \vert \leq N-4}
	|\partial^I Z^J \partial h(t,x)|\leq \frac{C_N' \varepsilon(1+t)^{2\delta}}{(1+t+r) (1+\vert q \vert) (1+q_+)^{2\delta}},
\end{equation}
hold for some constant $C_N'$ depending only on $C_N$ and on $N$.  In particular the assumptions of Proposition \ref{prop:suppf} are satisfied, and so, provided $\varepsilon$ is sufficiently small,
\[
	\supp( \widehat{T}^{\mu \nu}) \subset \{ (t,x) \mid \vert x \vert \leq c t + K\},
\]
for some $0<c<1$, $K \geq 0$.

The assumptions of Proposition \ref{prop:higherordersharpdecay} are now satisfied and so Proposition \ref{prop:higherordersharpdecay} and Theorem \ref{thm:Einsteinenergyestimates} then imply that for $\varepsilon<\varepsilon_N$ we have
\begin{equation} \label{eq:energybounds}
	Q_k(t)
	\leq
	8 Q_k(0)
	+
	M_k M
	+
	C^{\prime\prime\prime}_N\varepsilon
	\!\int_0^t
	\frac{Q_{k}(\tau)}{1+\tau}
	+
	\frac{Q_{k-1}(\tau)}{(1+\tau)^{1-d_k\varepsilon}}
	d\tau
	+
	M_k \sum_{|I|\leq k} \int_0^t \Vert Z^I \widehat{T} (\tau,\cdot) \Vert_{L^{2}} \, d\tau,
\end{equation}
for each $k=0,1,\ldots,N$ and for all $t\in [0,T_*]$, where $Q_{-1} \equiv 0$. Here $\varepsilon_N$, $C_N'''$, $d_1,\ldots,d_N$ are constants which depend only on $C_N^\prime$, on $N$ and on $c$ and $K$ and a lower positive bound for $\min{\{\gamma,1-\gamma\}}$, whereas  $M_0,\ldots,M_N$ are universal constants (which in particular do not depend on $C_N$).

The pointwise bounds \eqref{eq:hpointwise} in particular imply that
\[
	\sum_{\vert I \vert \leq N-4}
	|Z^I \Gamma(t,x)|
	\leq
	\frac{C_N' \varepsilon}{(1+t)^{1+a}},
\]
for $t\in [0,T_*]$ and $\vert x \vert \leq c t + K$, with $a = 2-2\delta$.  The assumptions of Theorem \ref{thm:mainL2} are therefore satisfied so
\begin{align*}
	\sum_{\vert I \vert \leq k}
	\Vert Z^I T(t,\cdot) \Vert_{L^2}
	&
	\lesssim
	\frac{\mathcal{V}_{k}(1 + D_N \varepsilon)}{(1+t)^{\frac{3}{2}}}
	+
	D_N \mathbb{D}_{\left\lfloor \frac{k}{2} \right\rfloor +1}
	\left(
	\frac{E_{k-1}(t)^{\frac{1}{2}}}{(1+t)^{1+a}}
	+
	\frac{1}{(1+t)^{\frac{3}{2}}}
	\int_{0}^t
	\frac{E_{k}(s)^{\frac{1}{2}}}{(1+s)^{\frac{1}{2}}}
	ds
	\right)
	\\
	&
	\lesssim
	\frac{\mathcal{V}_{k}(1 + D_N \varepsilon)}{(1+t)^{\frac{3}{2}}}
	+
	D_N \mathbb{D}_{\left\lfloor \frac{k}{2} \right\rfloor +1}
	\frac{Q_{k}(t)}{(1+t)},
\end{align*}
for all $k=0,1,\ldots,N$ and for all $t\in [0,T_*]$, where the constant $D_N$ depends on $C_N$.

Now the $L^1$ bounds \eqref{eq:bootstrap} and the Sobolev inequality, Lemma \ref{lem:Sobolevt3}, imply that
\[
	\sum_{\vert I \vert \leq N-4} \Vert Z^I \widehat{T}(t,\cdot) \Vert_{L^{\infty}}
	\leq
	\frac{C_N \varepsilon}{(1+t)^{3}},
\]
and so,
\begin{multline*}
	\sum_{\vert I \vert \leq k}k \Vert Z^I \widehat{T} (t,\cdot) \Vert_{L^2}
	\lesssim
	\Big(
	1
	+
	\sum_{\vert J \vert \leq \left\lfloor \frac{k}{2} \right\rfloor +1} \Vert \psi Z^J h^1 (t,\cdot) \Vert_{L^{\infty}}
	\Big)
	\sum_{\vert I \vert \leq k} \Vert Z^I T (t,\cdot) \Vert_{L^2}
	\\
	+
	\sum_{\vert J \vert \leq \left\lfloor \frac{k}{2} \right\rfloor}
	\Vert Z^J T (t,\cdot) \Vert_{L^{\infty}}
	\sum_{1\leq \vert I \vert \leq k} \Vert \psi Z^I h^1 (t,\cdot) \Vert_{L^2},
\end{multline*}
where $\psi(t,x)$ is the indicator function of the set $\{ \vert x \vert \leq c t + K\}$.  Since $\sum \Vert Z^J T (t,\cdot) \Vert_{L^{\infty}} \leq 2\sum \Vert Z^J \widehat{T} (t,\cdot) \Vert_{L^{\infty}}$ provided $\varepsilon$ is sufficiently small, it therefore follows that,
\[
	\sum_{\vert I \vert \leq k }\Vert Z^I \widehat{T} (t,\cdot) \Vert_{L^2}
	\leq
	\frac{\mathcal{V}_{k}(C + D_N \varepsilon)}{(1+t)^{\frac{3}{2}}}
	+
	D_N \mathbb{D}_{\left\lfloor \frac{k}{2} \right\rfloor +1}
	\frac{Q_{k}(t)}{(1+t)}
	+
	D_N \varepsilon \frac{Q_{k}(t)}{(1+t)},
\]
for $k=0,1,\ldots,N$, where the constant $C$ is independent of $C_N$ and the constant $D_N$ depends on $C_N$.  Inserting into \eqref{eq:energybounds} and using the fact that,
\[
	Q_N(0) + \mathbb{D}_{\left\lfloor N/2 \right\rfloor +1} + \mathcal{V}_N + M < \varepsilon,
\]
and making $M_k$ and $C_N'''$ larger if necessary gives,
\begin{equation} \label{eq:Qkpreinduction}
	Q_k(t)
	\leq
	M_k \varepsilon
	+
	C_N''' \varepsilon
	\int_0^t
	\frac{Q_{k}(\tau)}{1+\tau}
	+
	\frac{Q_{k-1}(\tau)}{(1+\tau)^{1-d_k\varepsilon}}
	d\tau,
\end{equation}
for $k=0,1,\ldots,N$.  It follows from an inductive argument that the bound \eqref{eq:Qkpreinduction} implies that
\begin{equation} \label{eq:Qkinduction}
	Q_k(t)
	\leq
	(M_0 + M_1 + \ldots + M_k) \varepsilon (1+t)^{(d_1+\ldots+d_k + (k+1)C_N''')\varepsilon},
\end{equation}
for all $t\in [0,T_*]$ and $k=0,1,\ldots,N$, using the following form of the Gr\"{o}nwall inequality.
\begin{lemma} \label{lem:Gronwall2}
	\!For $t\!>\!0$ and continuous functions $v,a,b: [0 ,t]\! \to \mathbb{R}$ such that $a\!\geq \!0$ and $b$ is non-decreasing, if
	\[
		v(s) \leq \int_0^s a(s') v(s') ds' + b(s),
	\]
	for $s\in [0,t]$, then
	\[
		v(s)
		\leq
		b(s)
		e^{\int_0^{s} a(s') ds'}.
	\]
\end{lemma}
Indeed, recall that $Q_{-1} \equiv 0$ and so, from the bound \eqref{eq:Qkpreinduction} with $k=0$, it follows from the Lemma \ref{lem:Gronwall2} with $a(s) = C_N''' \varepsilon(1+s)^{-1}$ and $b(s) = M_0 \varepsilon$ that,
\[
	Q_0(t) \leq M_0 \varepsilon (1+t)^{C_N'''\varepsilon}.
\]
Now suppose \eqref{eq:Qkinduction} holds for some $0 \leq k \leq N-1$.  Then, since
\begin{multline*}
	C_N''' \varepsilon \int_0^t \frac{Q_k (\tau)}{(1+\tau)^{1-d_{k+1} \varepsilon}} d \tau
	\leq
	C_N''' \varepsilon^2 (M_0 + \ldots + M_k)
	\int_0^t (1+\tau)^{(d_1+\ldots+d_{k+1} + (k+1)C_N''')\varepsilon - 1} d \tau
	\\
	\leq
	\frac{C_N''' \varepsilon^2 (M_0 + \ldots + M_k)
	(1+t)^{(d_1+\ldots+d_{k+1} + (k+1)C_N''')\varepsilon}
	}{(d_1+\ldots+d_{k+1} + (k+1)C_N''')\varepsilon}
	\leq
	(M_0 + \ldots + M_k)\varepsilon
	(1+t)^{(d_1+\ldots+d_{k+1} + (k+1)C_N''')\varepsilon},
\end{multline*}
it follows from \eqref{eq:Qkpreinduction} and Lemma \ref{lem:Gronwall2} with $a(s) = C_N''' \varepsilon(1+s)^{-1}$ and $b(s) = M_{k+1} \varepsilon + (M_0 + \ldots + M_k)\varepsilon (1+s)^{(d_1+\ldots+d_{k+1} + (k_1)C_N''')\varepsilon}$ that
\begin{multline*}
	Q_{k+1}(t)
	\leq
	\left( M_{k+1} \varepsilon
	+
	(M_0 + \ldots + M_k)\varepsilon
	(1+t)^{(d_1+\ldots+d_{k+1} + (k+1)C_N''')\varepsilon}
	\right)
	(1+t)^{C_N''' \varepsilon}
	\\
	\leq
	(M_0 + \ldots + M_{k+1})\varepsilon
	(1+t)^{(d_1+\ldots+d_{k+1} + (k+2)C_N''')\varepsilon}.
\end{multline*}

Theorem \ref{thm:mainL2} moreover implies that,
\begin{align*}
	\sum_{\vert I \vert \leq N-1} \Vert Z^I \widehat{T}(t,\cdot) \Vert_{L^1}
	\lesssim
	\varepsilon (C + \varepsilon D_N)
	+
	\varepsilon D_N \frac{Q_{N-2}(t)}{(1+t)^{\frac{1}{2} - 2 \delta}}
	+
	\varepsilon D_N \int_0^t \frac{Q_{N}(s)}{(1+t)^{\frac{3}{2} - 2 \delta}} ds
	,
\end{align*}
where the constant $C$ is independent of $C_N$ and the constant $D_N$ depends on $C_N$.  Inserting the above bound for $Q_N$ then gives
\[
	\sum_{\vert I \vert \leq N-1} \Vert Z^I \widehat{T}(t,\cdot) \Vert_{L^1}
	\lesssim
	C \varepsilon + D_N \varepsilon^2,
\]
provided $\varepsilon$ is sufficiently small, for some new $C$, $D_N$ as above.  Now, as above,
\begin{multline*}
	\sum_{\vert I \vert \leq N-1} \Vert Z^I \widehat{T} (t,\cdot) \Vert_{L^1}
	\lesssim
	\Big(
	1
	+
	\sum_{\vert J \vert \leq \left\lfloor \frac{N}{2} \right\rfloor +1} \Vert \psi Z^J h^1 (t,\cdot) \Vert_{L^{\infty}}
	\Big)
	\sum_{\vert I \vert \leq N-1} \Vert Z^I T (t,\cdot) \Vert_{L^1}
	\\
	+
	\sum_{\vert J \vert \leq \left\lfloor \frac{N}{2} \right\rfloor}
	\Vert Z^J T (t,\cdot) \Vert_{L^{\infty}}
	\sum_{1\leq \vert I \vert \leq N-1} \Vert \psi Z^I h^1 (t,\cdot) \Vert_{L^1},
\end{multline*}
and, since,
\[
	\sum_{1\leq \vert I \vert \leq N-1} \Vert \psi Z^I h^1 (t,\cdot) \Vert_{L^1}
	\leq
	\Vert \psi^{\frac{1}{2}} \Vert_{L^2}
	\sum_{1\leq \vert I \vert \leq N-1} \Vert \psi^{\frac{1}{2}} Z^I h^1 (t,\cdot) \Vert_{L^2}
	\leq
	(1+t)^{\frac{3}{2}} \frac{1}{1+t} E_{\vert I \vert}(t)^{\frac{1}{2}},
\]
where the equality \eqref{eq:partialZ} was used, it follows that
\[
	\sum_{\vert I \vert \leq N-1} \Vert Z^I \widehat{T} (t,\cdot) \Vert_{L^1}
	\leq
	(C + \varepsilon C_N''') (C \varepsilon + D_N \varepsilon^2)
	+
	\frac{C + \varepsilon C_N}{(1+t)^3} (1+t)^{\frac{1}{2}} Q_N(t)
	,
\]
and so,
\begin{equation} \label{eq:Thimproved}
	\sum_{\vert I \vert \leq N-1} \Vert Z^I \widehat{T} (t,\cdot) \Vert_{L^1}
	\leq
	C' \varepsilon + D_N' \varepsilon^2,
\end{equation}
where $D_N'$ depends on $C_N$ and $C'$ does not.

It follows from the bound \eqref{eq:Qkinduction} with $k=N$ and the bound \eqref{eq:Thimproved}, provided the constant $C_N$ is chosen so that $C_N \geq \max \{ 2(M_0+\ldots+M_N), 4C \}$ and $\varepsilon$ is chosen so that $\varepsilon < \min\{ \frac{\delta}{2} (d_1+\ldots + d_N + (N+1)C_N''')^{-1}, \frac{C_N}{4D_N'} \}$, that the bounds
\[
	E_N(t)^{\frac{1}{2}} \leq \frac{C_N}{2} \varepsilon (1+t)^{\frac{\delta}{2}},
	\qquad
	\sum_{\vert I \vert \leq N-1} \Vert Z^I \widehat{T}(t,\cdot) \Vert_{L^1}
	\leq
	\frac{C_N}{2} \varepsilon,
\]
hold for all $t \in [0,T_*]$.  Appealing once again to the local existence theorem, this contradicts the maximality of $T_*$ and hence the solution exists and the estimates hold for all $t\in[0,\infty)$.


\begin{thebibliography}{99}


 \bibitem{A} S. Alinhac, \newblock
\emph{ An example of blowup at infinity for a quasilinear wave equation\/.}
 Asterisque \textbf{284} (2003), 1-91.
 	\bibitem{An} H.\@ Andr\'{e}asson, \emph{The Einstein--Vlasov System/Kinetic Theory}, Living Rev.\@ Relativity, \textbf{14} (2011).
	\bibitem{BaDe} C.\@ Bardos, P.\@ Degond, \emph{Global Existence for the Vlasov--Poisson System in 3 Space Variables with Small Initial Data}, Anal.\@ Non Lin\'{e}aire, \textbf{2} (1985) 101-118.
	\bibitem{Bi} L. Bieri, \emph{Extensions of the Stability Theorem of the Minkowski Space in General Relativity, Solutions of the Vacuum Einstein Equations}, American Mathematical Society, Boston (2009).
	\bibitem{Ca} S.\@ Calogero, \emph{Global Classical Solutions to the 3D Vlasov--Nordstr\"{o}m System}, Comm.\@ Math.\@ Phys.\@ \textbf{266} (2006) 343-353.
	 \bibitem{ChBr} Y.\@ Choquet-Bruhat, \emph{Th\'{e}reme d'existence pour certains syst\`{e}mes d'\'{e}quations aux d\'{e}riv\'{e}es partielles non lin\'{e}aires} Acta Math.\@, \textbf{88} (1952) 141-225.
	 \bibitem{ChBr71} Y.\@ Choquet-Bruhat, \emph{Probl\`{e}me de Cauchy pour le syst\`{e}me int\'{e}gro-diff\'{e}rentiel d'Einstein-Liouville} Ann. Inst. Fourier, \textbf{21} (1971) 181-201.
     \bibitem{ChBr00} Y. Choquet-Bruhat, \newblock
\emph{ The null condition and asymptotic expansions for the Einstein's equations\/.}
 Ann. Phys. (Leipzig) \textbf{9} (2000), 258-266.
	 \bibitem{ChBrGe} Y. Choquet-Bruhat and R. Geroch \emph{Global Aspects of the Cauchy Problem in General Relativity} Comm. Math. Phys., \textbf{14} (1969) 329-335.
	 \bibitem{C1} D. Christodoulou, \newblock{\em Global solutions of nonlinear hyperbolic equations for small initial data\/.} \newblock   Comm. Pure Appl. Math. \textbf{39} (1986) 267-282.
 	\bibitem{ChKl} D.\@ Christodoulou and S.\@ Klainerman, \emph{The Global Nonlinear Stability of the Minkowski Space} Princeton Mathematical Series, Vol.\@ 41 (Princeton University Press, 1993).
 	\bibitem{Da} M.\@ Dafermos, \emph{A Note on the Collapse of Small Data Self-Gravitating Massless Collisionless Matter} J.\@ Hyperbol.\@ Differ.\@ Equations, \textbf{3} (2006) 905-961.
	\bibitem{Fa} D.\@ Fajman, \emph{The nonvacuum Einstein flow on surfaces of negative curvature and nonlinear stability}, Commun. Math. Phys. $\mathbf{353}$ (2017) 561-583.
	\bibitem{FaJoSm} D.\@ Fajman, J.\@ Joudioux and J.\@ Smulevici, \emph{A Vector Field Method for Relativistic Transport Equations with Applications} arXiv:1510.04939.
	\bibitem{FaJoSm17} D.\@ Fajman, J.\@ Joudioux and J.\@ Smulevici, \emph{Sharp Asymptotics for Small Data Solutions of the Vlasov--Nordstr\"{o}m System in Three Dimensions} arXiv:1704.05353.
	\bibitem{FaJoSm172} D.\@ Fajman, J.\@ Joudioux and J.\@ Smulevici, \emph{The stability of Minkowski space for the Einstein--Vlasov system}, preprint.
	\bibitem{Fr} H.\@ Friedrich, \emph{On the Existence of $n$-geodesically Complete or Future Complete Solutions of Einstein's Field Equations with Smooth Asymptotic Structure} Commun.\@ Math.\@ Phys.\@ \textbf{107} (1986) 587-609.
	\bibitem{GlSt} R.\@ T.\@ Glassey and W.\@ A.\@ Strauss, \emph{Absence of Shocks in an Initially Dilute Collisionless Plasma}, Comm.\@ Math.\@ Phys.\@ \textbf{113} (1987) 191-208.
	\bibitem{H1}  L. H\"ormander, \newblock {\em The lifespan of classical solutions of nonlinear hyperbolic equations\/.}\newblock Pseudodifferential operators, 214--280, Lecture Notes in Math., 1256, Springer, 1987.
	\bibitem{Hu} C.\@ Huneau, \emph{Stability of Minkowski Space--time with a Translation Space--like Killing Field} arXiv:1511.07002.
	\bibitem{HwReVe} H.\@ J.\@ Hwang, A.\@ D.\@ Rendall, J.\@ J.\@ L.\@ Vel\'{a}zquez, \emph{Optimal Gradient Estimates and Asymptotic Behaviour for the Vlasov--Poisson System with Small Initial Data}, Arch.\@ Rational Mech.\@ Anal.\@ \textbf{200} (2011) 313-360.
	\bibitem{J1} F. John, \newblock{\em Blow-up for quasilinear wave equations in three space dimensions\/.} \newblock Comm. Pure Appl. Math. \textbf{34} (1981), no. 1, 29--51.
 \bibitem{J2} F. John, \newblock
\emph{Blow-up of radial solutions of $u\sb {tt}\!=\!c\sp 2(u\sb
t)\Delta u$ in three space dimensions\/.}
  Mat. Apl. Comput. \textbf{4} (1985), no. 1, 3--18.
	\bibitem{K1}  S. Klainerman, \newblock{\em Long time behavior of solutions
to nonlinear wave equations}, Proceed. ICM, Warsaw, (1982), 1209-1215.
	\bibitem{Kl} S.\@ Klainerman, \emph{The Null Condition and Global Existence to Nonlinear Wave Equations} Nonlinear Systems of Partial Differential Equations in Applied Mathematics, Part 1; Santa Fe, N.\@M.\@, 1984 Lectures in Appl.\@ Math.\@ \textbf{23} (1986) 293-326.
	\bibitem{KlNi} S.\@ Klainerman and F.\@ Nicol\`{o}, \emph{The Evolution Problem in General Relativity} Vol.\@ 23 Lectures in Appl.\@ Math.\@, Birkh\"{a}user Boston Inc.\@, Boston, MA, (2003).
	\bibitem{LeMa} P.\@ G.\@ LeFloch and Y.\@ Ma, \emph{The Global Nonlinear Stability of Minkowski Space for Self-Gravitating Massive Fields} arXiv:1511.03324.
	\bibitem{Lo} J.\@ Loizelet, \emph{Solutions globales des \'{e}quations d'Einstein--Maxwell}, Ann.\@ Fac.\@ Sci.\@ Toulouse Math.\@ \textbf{18} (2009) 565-610.
	\bibitem{L1} H. Lindblad, \newblock {\em On the lifespan of solutions of nonlinear wave equations with small initial data\/.} \newblock Comm. Pure Appl. Math \textbf{43} (1990), 445--472.
\bibitem{L2} H. Lindblad,
\emph{ Global solutions of nonlinear wave equations\/.}  Comm. Pure Appl. Math. \textbf{45} (9) (1992), 1063--1096.
 \bibitem{L3}  H.\! Lindblad,
 \emph{ Global solutions \!of quasilinear wave
equations\/.}  Amer.\@ J.\@ Math.\@ (2008), 115-157.
	\bibitem{L4} H. Lindblad \newblock { \emph{On the asymptotic behavior of solutions to Einstein's vacuum equations in wave coordinates.}} \newblock Comm. Math. Phys. {\bf 353}, (2017), No 1, 135-184
	\bibitem{LR1}  H. Lindblad and I. Rodnianski, \newblock {\em The weak null condition for Einstein's equations\/.}\newblock C. R. Math. Acad. Sci. Paris 336 (2003), no. 11, 901--906.
	\bibitem{LR2} H. Lindblad and I. Rodnianski, \newblock {\em Global existence for the Einstein vacuum equations in wave coordinates\/.} \newblock Comm. Math. Phys. \textbf{256} (2005), no. 1, 43--110.
	\bibitem{LR3} H. Lindblad and I. Rodnianski, \emph{The Global Stability of Minkowski Space-time in Harmonic Gauge}, Ann.\@ of Math.\@ \textbf{171} (2010) 1401-1477.
	\bibitem{LiPe} P.\@ L.\@ Lions and B.\@ Perthame, \emph{Propagation of Moments and Regularity for the 3-dimensional Vlasov--Poisson System}, Invent.\@ Math.\@ \textbf{105} (1991) 415-430.
	\bibitem{Mo} G.\@ Moschidis, \emph{A Proof of the Instability of AdS for the Einstein--Null Dust System with an Inner Mirror} arXiv:1704.08681.
	\bibitem{Pf} K.\@ Pfaffelmoser, \emph{Global Classical Solutions of the Vlasov--Poisson System in Three Dimensions for General Initial Data}, J.\@ Diff.\@ Eq.\@ \textbf{95} (1992) 281-303.
  	\bibitem{ReRe} G. Rein and A. D. Rendall, \emph{Global Existence of Solutions of the Spherically Symmetric Vlasov--Einstein System with Small Initial Data}, Commun. Math. Phys. $\mathbf{150}$ (1992) 561-583.
	\bibitem{Ri} H.\@ Ringstr\"{o}m, \emph{On the Topology and Future Stability of the Universe}, Oxford Mathematical Monographs, Oxford University Press, (2013).
	\bibitem{Sp} J.\@ Speck, \emph{The Global Stability of the Minkowski Spacetime Solution to the Einstein--Nonlinear Electromagnetic System in Wave Coordinates}, Anal.\@ PDE \textbf{7} (2014) 771-901.
	\bibitem{Sm} J.\@ Smulevici, \emph{Small Data Solutions of the Vlasov--Poisson System and the Vector Field Method}, Ann.\@ PDE \textbf{2:11} (2016).
	\bibitem{ShYa} R. Schoen and S. Yau,
\newblock {\em On the proof of the positive mass conjecture in
general relativity\/.}
\newblock Comm. Math. Phys. \textbf{65} (1979) 45--76.
	\bibitem{Ta} M.\@ Taylor, \emph{The Global Nonlinear Stability of Minkowski Space for the Massless Einstein--Vlasov System}, Ann.\@ PDE \textbf{3:9} (2017).
	\bibitem{Wi} E. Witten, \newblock {\em A new proof of the positive mass theorem}, \newblock Comm. Math. Phys. \textbf{80} (1981) 381--402.
	\bibitem{Zi} N. Zipser, \emph{Extensions of the Stability Theorem of the Minkowski Space in General Relativity, Solutions of the Einstein--Maxwell Equations}, American Mathematical Society, Boston (2009).

 \end{thebibliography}
\end{document}